\documentclass[10pt,journal,onecolumn,letterpaper]{IEEEtran}
\usepackage[utf8]{inputenc} 
\usepackage[T1]{fontenc}
\usepackage{url}
\usepackage{ifthen}
\usepackage{cite}
\usepackage{braket}
\usepackage{float}
\usepackage[cmex10]{amsmath} 
\usepackage{color,graphicx,amsmath,amssymb,amsthm,epsfig,mathrsfs,cite}
\usepackage[dvipsnames]{xcolor}
\usepackage[font=small]{subcaption}
\usepackage[font=small]{caption}
\makeatletter
\renewcommand\p@subfigure{\thefigure}
\makeatother
\usepackage{eso-pic}
\usepackage{tikz}
 \usetikzlibrary{decorations.pathmorphing} 
\usepackage{stackengine}
\usepackage{microtype}
\usepackage[linesnumbered,ruled,vlined]{algorithm2e}

\newcommand{\graphprod}{\mathop{\scalebox{3.5}{$\times$}}}
\usetikzlibrary{tikzmark,fit}
\usepackage{nicematrix}
\usepackage{multirow}
\usepackage{adjustbox}
\usepackage{pgfplots} 
\usepackage{balance}
\usepackage{amssymb}
\usepackage{amsthm}
\usepackage{mathtools}
\usepackage{qcircuit}
\usepackage{blkarray}
\SetKwInput{KwInput}{Input}
\SetKwInput{KwVars}{Variables}
\SetKwInput{KwInits}{Initialization}
\SetKwInput{KwOutput}{Output}
\theoremstyle{definition}
\newtheorem{definition}{Definition}

\newtheorem{lemma}{Lemma}

\newtheorem{example}{Example}
\newtheorem{remark}{Remark}

 \pgfplotsset{compat=1.18}
\begin{document}
\title{Fast Fault-Tolerant Decoders for Hypergraph Product and Lifted-Product Codes} 
\author{%
\IEEEauthorblockN{Asit Kumar Pradhan, Nithin Raveendran, David Declercq and Bane Vasi\'c}
\IEEEauthorblockA{\textit{Department of Electrical and Computer Engineering, The University of Arizona, Tucson, AZ, 85721 USA}}\\
\IEEEauthorblockA{Email: pradhanasit07@gmail.com, \{nithin, daviddeclercq\}@arizona.edu, vasic@ece.arizona.edu}
}
\maketitle
\begin{abstract}
We design low-complexity, fault-tolerant decoders for quantum low-density parity-check (QLDPC) codes with the goal of reducing decoding latency. We target two major bottlenecks of decoding under the \emph{circuit-level} noise model: (i) post-processing via order-statistics decoding (OSD), and (ii) the large number of auxiliary variable nodes commonly introduced to represent CNOT-induced correlations during syndrome extraction. 
Our key observation is that propagating CNOT faults (\emph{hook errors}) create \emph{stabilizer-induced} trapping sets (TSs) that are intrinsic to hypergraph-product (HGP) and lifted-product (LP) constructions. Therefore, instead of modeling each such fault with an explicit correlation node and relying on OSD to clean up the resulting failures, we design message-passing decoders that resolve the corresponding \emph{stabilizer-induced} TSs directly. We obtain these decoders by deriving  QLDPC decoders from decoders for the parent classical LDPC codes and using them collectively to correct broad families of \emph{stabilizer-induced} TSs.
For CNOT faults that manifest primarily as syndrome errors, we show that their effect is equivalent to a data error together with syndrome-bit measurement errors. Consequently, given repeated measurements and a decoding graph that already includes nodes representing syndrome-bit errors, no distinct variable node is needed for each CNOT fault. Using a \emph{phenomenological} Tanner graph with nodes representing only data errors and syndrome-bit errors, simulations on LP codes show a significant improvement in logical error rate relative to BP+OSD. Moreover, the improved error correction relies on low-complexity finite-state iterative decoders, targeted for fast and low-latency FPGA implementation.  

\end{abstract}


\section{Introduction}

Quantum information is intrinsically fragile, so quantum error-correcting codes (QECCs) must reliably extract error syndromes using low-depth circuits. Low-depth syndrome extraction is essential to prevent additional faults from accumulating during each error-correction cycle. Achieving such shallow circuits requires sparse parity-check matrices, so that each check acts on only a small number of qubits.
For robust error correction, the code distance should grow with the total number of physical qubits while preserving sparsity. Topological codes, such as surface codes~\cite{bravyi1998quantum}, meet both of these criteria---sparse checks and increasing distance---but suffer from a vanishing rate: the ratio of logical to physical qubits approaches zero as the system size increases. As a result, they incur very large qubit overheads for large-scale quantum algorithms.

These limitations have motivated the development of quantum low-density parity-check (QLDPC) codes, which aim to combine sparse parity checks with a non-vanishing rate and a distance that scales with system size. Constructions such as hypergraph-product (HGP)~\cite{Tillich_2014}, balanced-product~\cite{Breuckmann_2021}, lifted-product (LP)~\cite{LP_codes}, and quantum Tanner codes~\cite{leverrier2022quantum} represent significant progress toward this goal.

Realizing constant-overhead fault-tolerant quantum computation with QLDPC codes faces several remaining challenges that can be grouped into three main categories: (1) code--architecture co-design, (2) efficient and scalable decoding, and (3) parallel processing of logical gates. In code--architecture co-design, QLDPC codes and hardware architectures must be developed jointly so that (i) the inherently nonlocal connectivity of the codes can be efficiently embedded in the physical device, and (ii) syndrome extraction can be accelerated through highly parallel gate execution, particularly when gate times are relatively low. In efficient and scalable decoding, decoders should achieve error-correction performance that improves with increasing code distance, while also running faster than the syndrome-extraction process so that decoding does not become the bottleneck of the error-correction cycle. In parallel logical gates, one must implement highly parallel, low-overhead, fault-tolerant logical gate operations on QLDPC codes.
The co-design of codes and architectures has been explored in several works. For instance, the challenge of limited local connectivity in superconducting qubits is addressed in~\cite{bravyi2024high}, while the issue of slow gate times in neutral-atom platforms is tackled in~\cite{xu2024constant}. Low-overhead, fault-tolerant implementations of logical gates are studied in~\cite{cohen2022low,zhou2025low}, and their parallel execution is investigated in~\cite{xu2025fast}.

In this work, we focus on efficient and scalable decoding of QLDPC codes. The importance of fast decoding is emphasized in~\cite{barbara_syndrome_backlog}, which shows that a decoder significantly slower than syndrome extraction can negate the advantages of quantum algorithms. This point is further illustrated in~\cite{webster2026pinnaclearchitecturereducingcost}, where increasing the error-correction cycle from 1 to 100 microseconds (with the decoding time as the dominant contribution) increases the number of physical qubits required to break RSA in one week from about $1.5\times10^5$ to $5\times10^7$, assuming a physical error rate of $10^{-3}$.

\subsection{Existing Approaches to Designing Low-Complexity Decoders}
\label{sec:related_work}
The challenges in designing fast decoders with low logical error rates for QLDPC codes are best understood by revisiting the decoding of classical LDPC codes, since QLDPC codes are defined by a pair of dual-containing classical LDPC codes.
Separately, for 2D topological codes, matching-based decoders (minimum-weight perfect matching and variants) provide a highly efficient alternative~\cite{SurfaceCodes_DecodingMWPM}; recent work has generalized matching decoders to broader families of translationally-invariant 2D codes~\cite{tan2026generalized}.

In the non-asymptotic regime, low-complexity iterative message-passing decoders for classical LDPC codes can fail when the error is supported on certain subgraphs of the Tanner graph called trapping sets (TSs)~\cite{Shashi_trapping_sets,VNC_2014_Book}. The most harmful TSs are often induced by a subset of variable nodes on the support of a low-weight codewords. Such induced subgraphs contain fewer odd-degree check nodes than variable nodes, a property that often leads to decoding failure when the error is confined to that subgraph.

In classical LDPC codes, error floors due to TS-related decoding failures are typically mitigated in two steps. First, code constructions such as progressive-edge growth~\cite{PEG} reduce the number of short cycles, which in turn reduce the number of harmful TSs. Second, failures caused by the remaining TSs are mitigated by modifying the message-passing rules~\cite{offset_min_sum,FAID_diversity}.
For QLDPC codes, this strategy is fundamentally limited: because the two defining classical LDPC codes are dual-containing, the parity checks of one code are codewords of the other. As a result, subgraphs induced by their support inevitably act as TSs; these are referred to as \emph{stabilizer-induced} TSs~\cite{Nithin2021_QTS}.

In~\cite{BP_OSD}, TS-related failures are reduced using a post-processing \emph{order statistics decoder} (OSD) applied after message passing. Although OSD can reduce the decoding-failure rate, it substantially increases complexity. Specifically, it requires solving a linear system whose size scales with the blocklength, yielding cubic complexity. In addition, increasing the OSD order enlarges the greedy search over candidate error patterns, leading to exponential growth in complexity with order. Consequently, OSD-based post-processing is typically impractical for QLDPC decoding. Relatedly, for HGP codes under small-set-flip (SSF) decoding,~\cite{stambler2023addressing} analyzes stopping/trapping-set-like failure mechanisms and proposes a modified SSF decoder to avoid them.

To overcome this limitation, several low-complexity variants of OSD have been proposed in~\cite{panteleev2021degenerate,BP_OSD,LSD}, which reduce implementation cost without significantly degrading error-correction performance. Another approach, proposed in~\cite{inactivation_decoder}, eliminates the exponential search step by detecting from the decoder dynamics when the error lies in a TS. In~\cite{chytas2025enhanced,Refined_BP}, the authors adapt techniques such as normalization and bias optimization, which modify message-passing rules to mitigate TS-related failures without relying on post-processing. In particular, the decoding algorithm in~\cite{chytas2025enhanced} employs multiple decoders that differ in their message-update rules at the variable nodes to reduce decoding failures.

Neural-network-based approaches have also been investigated. In~\cite{pradhanISTC,pruned_neural_BP}, neural networks are trained to learn message-passing rules that resolve TS-related failures. Building on a more detailed understanding of the decoding dynamics within \emph{stabilizer-induced} TSs, \cite{pradhan2025_TS_left_right_schedule_decoder} proposes a scheduled message-passing decoder designed to avoid TS-induced failures under a certain technical condition. This scheduled decoder can be combined with other decoders that do not depend on post-processing to further boost their performance. For example, in the approach of~\cite{chytas2025enhanced}, the number of parallel decoders needed to achieve a target logical error rate can be reduced by optimizing the scheduled decoder, rather than relying solely on a fully parallel decoder.

The decoding approaches discussed so far assume an error-free syndrome extraction circuit and focus exclusively on optimizing decoders to avoid TS-related failures.
Although this helps to clarify how TS-related decoding failures arise and how to prevent them, it is not sufficient for realistic hardware, where the syndrome extraction circuit itself can be faulty and introduce additional errors.
The errors introduced during the syndrome extraction process fall into two main categories: (1) errors that propagate through the circuit and create correlated data errors, and (2) errors that do not propagate but instead corrupt the measured syndrome by effectively flipping some of the bits of the syndrome vector.
To account for these errors, additional variable nodes are added to the decoding graph.
For propagating errors, these nodes represent the resulting correlations and provide additional information to the decoder.
For non-propagating errors that appear as flipped syndrome bits, they allow the decoder to identify error patterns that are consistent with the faulty syndrome, helping it converge.
However, these additional nodes can also exacerbate the non-convergence of BP decoding: they introduce many short cycles in the decoding graph, and combinations of these cycles can create new TSs beyond those already present in the original Tanner graph of the code.
As in the case of perfect syndrome extraction, many decoding algorithms handle these failures using a post-processing OSD step, at the cost of increased complexity.
In the imperfect syndrome extraction setting, this complexity increase is even more significant: the performance of the message-passing decoder is further degraded by the newly introduced short cycles, so the OSD must correct more residual errors. Moreover, the OSD complexity scales cubically with the number of variable nodes, which has now grown. This complexity growth can be mitigated by using more efficient OSD implementations~\cite{panteleev2021degenerate,BP_OSD,LSD}, as in the perfect syndrome extraction setting.
Decoding approaches that completely avoid post-processing can also be adapted to the imperfect syndrome extraction case. In~\cite{chen2025improved}, the authors design several neural-network-based decoders, each parameterized by a different bias (i.e., a different initial channel LLR) and run them in parallel to collectively reduce decoding failures. In~\cite{muller2025improved}, it is shown that running these decoders sequentially -- feeding the output of one decoder as the input to the next --can further reduce decoding failures; this sequential strategy is known as \emph{relay} BP. Although these decoders avoid post-processing, running many of them sequentially still increases the decoding time, as reported in~\cite{muller2025improved}, even if the overhead can be substantially smaller than that of post-processing.
Other list-based strategies also avoid post-processing while improving robustness. In~\cite{ye2025beam}, the decoder runs a small number of BP iterations, identifies the least reliable variable node (e.g., the node with the smallest-magnitude LLR), fixes its value, and then runs two BP instances corresponding to the two possible fixed values. This branching is repeated within each BP instance until the number of active instances reaches a preset list size; at that point, instances with the least reliable scores are pruned. The process continues until one of the remaining BP instances converges.
The \emph{beam-search} decoder avoids the sequential strategy considered in \emph{relay} BP but (like \emph{relay} BP) relies on additional nodes representing CNOT faults for convergence.
In~\cite{leverrier2026approximating}, the authors consider list decoding on the code trellis, aggregating the probabilities of error candidates that correspond to the same syndrome. Both approaches report improved performance compared to BP+OSD, while avoiding explicit OSD-style post-processing.
Unlike the message-passing decoders discussed above, this trellis-based decoder approximates coset decoding. While it maintains a decoding complexity that is linear in blocklength (as in BP), it can have a large constant factor due to the larger state space of the trellis.

\subsection{Limitations of Existing Decoders and Motivation for Our Approach}
\label{sec:challenges_and_contribution}
The discussion in Section~\ref{sec:related_work} highlights that many existing decoders focus on reducing or eliminating the additional complexity introduced by post-processing, rather than reducing the complexity of the underlying message-passing decoder.
When the syndrome-extraction circuit is faulty, post-processing becomes more complex because the decoding graph must include extra variable nodes, as discussed earlier. These additional Tanner-graph nodes represent faults during syndrome extraction and increase the cost of both OSD and message passing. Although they do not change the asymptotic scaling of message passing, they can significantly increase the constant factor, and thus the practical decoding cost.

To quantify this overhead, consider the syndrome-extraction circuit of a QLDPC code from~\cite{bravyi2024high}. A failed CNOT gate can induce one of $15$ two-qubit Pauli error operators; modeling these possibilities typically requires adding $15$ variable nodes per potentially faulty CNOT.
The total number of CNOT gates scales with the number of check nodes (proportional to the blocklength) times the check degree $d_\mathrm{c}$. Hence, the number of additional variable nodes scales as $15 d_\mathrm{c} \times$ blocklength. For the bivariate-bicycle codes in~\cite{bravyi2024high}, which have $d_\mathrm{c}=6$, this corresponds to a factor-$15\times 6 = 90$ increase in the number of variable nodes. This substantial overhead directly impacts the decoding cost and, consequently, the overall runtime of quantum algorithms and the physical-qubit estimates reported in~\cite{webster2026pinnaclearchitecturereducingcost}.

In this work, we ask whether it is possible to (i) avoid adding extra nodes that account for circuit faults and (ii) avoid reliance on post-processing for resolving decoding failures.
The main contributions of this paper are summarized as follows.
\begin{enumerate}
\item In Section~\ref{sec:hook_error}, we connect propagating circuit faults (\emph{hook errors}) to \emph{stabilizer-induced} TSs and explain why introducing nodes that represent hook errors can substantially reduce the OSD order needed to correct such faults.
\item In Section~\ref{sec:proposed_decoders}, we characterize broad families of \emph{stabilizer-induced} TSs in HGP and LP codes based on TSs in their parent classical codes, extending the characterization in~\cite{pradhan2025_TS_left_right_schedule_decoder} beyond the products of tree-like subgraphs. We also relate these structures to \emph{circuit-level} faults by quantifying how many CNOT failures suffice to induce them, which helps identify the dominant TSs that determine the slope of the logical error-rate curve in the error-floor regime. Our TS enumeration can also be used to sample error patterns to train neural-network decoders such as those in~\cite{gu2026scalable}.
\item In Section~\ref{sec:syndrome_error_inducing_CNOT_faults}, we show that CNOT faults that manifest as syndrome errors can be reinterpreted as data errors together with syndrome measurement errors. Consequently, given repeated measurements and a decoding graph that already includes nodes representing syndrome measurement errors, one does not need a distinct variable node for each such CNOT failure to retain the same error-correction capability.
\item In Section~\ref{sec:proposed_decoders}, we derive design principles to avoid \emph{stabilizer-induced} TS-related decoding failures without adding explicit nodes to represent CNOT failures, using decoders that resolve TSs in the parent classical codes. In Section~\ref{sec:simulation_results}, we validate this approach by deriving  QLDPC decoders from decoders of the parent codes that resolve the corresponding classical TSs.


\end{enumerate}

\section{Preliminaries}
\label{sec:prelim}
\subsection{Notations}
We use bold lowercase and uppercase letters to denote vectors and matrices, respectively. The $i$-th entry of a vector $\mathbf{a}$ is written as $\mathbf{a}(i)$, and the $(i,j)$-th entry of a matrix $\mathbf{A}$ is written as $A(i,j)$. The $i$-th row and $j$-th column of $\mathbf{A}$ are denoted by $\mathbf{A}(i,:)$ and $\mathbf{A}(:,j)$, respectively.
The Hamming weight of a vector $\mathbf{a}$, defined as the number of nonzero entries, is denoted by $\lvert\mathbf{a}\rvert$.
Calligraphic letters are used to denote sets; the cardinality of a set $\mathcal{S}$ is denoted by $\lvert \mathcal{S} \rvert$.
Script letters are used to denote graphs and subgraphs. Set-theoretic operations, such as intersection and union between two graphs, refer to applying these operations to the underlying node and edge sets that define the graphs or subgraphs.
\subsection{Linear Codes}
\label{sec:linear_codes}
A $(n,k)$-binary linear code is a $k$-dimensional subspace of the $n$-dimensional vector space $\mathbb{F}_2^n$, specified as the null space of a binary matrix $\mathbf{H} \in \mathbb{F}_2^{(n-k)\times n}$, called the parity-check matrix.
The associated codespace $\mathcal{U}$ is
\begin{equation*}
    \mathcal{U} = \{\mathbf{u} \in \mathbb{F}_2^n : \mathbf{H}\mathbf{u} = \mathbf{0} \pmod 2\}.
\end{equation*}
The minimum distance $d_{\mathcal{U}}$ of $\mathcal{U}$ is defined as the minimum Hamming weight of any nonzero codeword, i.e.,
\begin{equation*}
    d_{\mathcal{U}} = \min_{\mathbf{u} \in \mathcal{U}, \\ \mathbf{u}\neq \boldsymbol{0}} \lvert\mathbf{u}\rvert.
\end{equation*}
A parity-check matrix can be represented by a bipartite graph known as a Tanner graph\cite{VNC_2014_Book}. Let the Tanner graph corresponding to the parity-check matrix $\mathbf{H}$ be denoted by $\mathcal{G}(\mathcal{V} \cup \mathcal{C}, \mathcal{E})$, where $\mathcal{V}$ is the set of variable nodes (corresponding to the columns of $\mathbf{H}$), $\mathcal{C}$ is the set of check nodes (corresponding to the rows of $\mathbf{H}$), and $\mathcal{E}$ is the set of edges of $\mathcal{G}$.
The set of check nodes connected to a variable node $v_{\mathrm{j}} \in \mathcal{V}$ is denoted by $\mathcal{N}_{v_{\mathrm{j}}}$, and the set of variable nodes connected to a check node $c_{\mathrm{i}} \in \mathcal{C}$ is denoted by $\mathcal{N}_{c_{\mathrm{i}}}$. An edge exists between a check node $c_{\mathrm{i}} \in \mathcal{C}$ and a variable node $v_{\mathrm{j}} \in \mathcal{V}$ if the $(i,j)$-th entry of $\mathbf{H}$ is 1.
An example of a Tanner graph is shown in Figure~\ref{fig:classical_Tanner_graph}.
\begin{figure}
    \centering
   \input{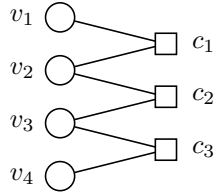}
    \caption{An example of a Tanner graph. Circular nodes represent variable nodes, and rectangular nodes represent check nodes. Each check node imposes a linear constraint on its neighboring variable nodes. For instance, check node $c_1$ is connected to variable nodes $v_1$ and $v_2$, which implies that in any valid codeword $v_1 + v_2 \equiv 0 \pmod{2}$.}
    \label{fig:classical_Tanner_graph}
\end{figure}
The dual of the linear code $\mathcal{U}$, denoted by $\mathcal{U}^{\perp}$, is defined as the code whose codewords are orthogonal to that of $\mathcal{U}$ with associated codebook
\begin{equation*}
    \mathcal{U}^{\perp}=\{\mathbf{u}': \mathbf{u}'^{\mathsf{T}}\mathbf{u}=\boldsymbol{0} \pmod 2\}.
\end{equation*}

\subsection{Pauli Group}
\label{sec:pauli_group}

The set of Pauli matrices $\mathcal{P} = \{I, X, Y, Z\}$, where
\[
I = \begin{bmatrix} 1 & 0 \\ 0 & 1 \end{bmatrix},\;
X = \begin{bmatrix} 0 & 1 \\ 1 & 0 \end{bmatrix},\;
Y = \begin{bmatrix} 0 & -i \\ i & 0 \end{bmatrix},\;
Z = \begin{bmatrix} 1 & 0 \\ 0 & -1 \end{bmatrix},
\]
forms a group\cite{Nielsen}. The Pauli group $\mathcal{P}$ acts on the 2-dimensional complex Hilbert space $\mathbb{C}^2$. Larger groups can be formed by taking the $n$-fold tensor products of Pauli matrices, for $n \in \mathbb{N}$. We denote this by $\mathcal{P}^n = \{I, X, Y, Z\}^{\otimes n}.$ The group $\mathcal{P}^n$ acts on a $2^n$-dimensional complex vector space.

\subsection{Stabilizer Codes}
\label{sec:stabilizer_codes}
Let $\mathcal{S}$ be an Abelian subgroup of $\mathcal{P}^n$, generated by $\{S_1, S_2, \ldots, S_r\}$, with each $S_i \in \mathcal{P}^n \setminus I^{\otimes n}$. The codespace $\mathcal{U}$ is defined as the subspace of $\mathbb{C}^{2^n}$ on which all elements of $\mathcal{S}$ act trivially, $\mathcal{U} = \bigl\{\ket{\psi} \in \mathbb{C}^{2^n} : S_i \ket{\psi} = \ket{\psi} \text{ for } 1 \le i \le r\bigr\}$\cite{Gottesman97}.
Thus, $\mathcal{U}$ is the simultaneous $+1$ eigenspace of the operators in $\mathcal{S}$, and has dimension $2^{n-r}$. Elements of $\mathcal{S}$ are called \emph{stabilizers}. For operators $A$ and $B$, we say that they \emph{commute} if $AB=BA$ and \emph{anticommute} if $AB=-BA$.
An operator in $\mathcal{P}^n$ that maps any codeword to another codeword is called a logical operator. The set of logical operators, denoted by $\mathcal{L}$, is
\[
\mathcal{L} = \bigl\{L \in \mathcal{P}^n : L\ket{\psi} \in \mathcal{U} \text{ for all } \ket{\psi} \in \mathcal{U}\bigr\}.
\]
Equivalently, $\mathcal{L}$ can be characterized as the subset of $\mathcal{P}^n$ consisting of operators that commute with every stabilizer.

The stabilizer group $\mathcal{S}$ is a normal subgroup of $\mathcal{L}$, and the quotient group $\mathcal{L}/\mathcal{S}$ describes the operators that act non-trivially on (at least) one codeword. Because any two elements of $\mathcal{P}^n$ either commute or anticommute, whether an operator maps a codeword to another codeword can be determined from its commutation relations with the stabilizers.
Any operator in $\mathcal{L}/\mathcal{S}$ acts non-trivially on the codespace while remaining undetectable by the stabilizers. The minimum distance of the code defined by $\mathcal{S}$ is the minimum weight of a non-trivial element of $\mathcal{L}/\mathcal{S}$.
There is an isomorphism between the single-qubit Pauli group and $\mathbb{F}_2^2$ under component wise addition, given by
\[
I \mapsto \begin{bmatrix}0 & 0\end{bmatrix},\;
X \mapsto \begin{bmatrix}1 & 0\end{bmatrix},\;
Y \mapsto \begin{bmatrix}1 & 1\end{bmatrix},\;
Z \mapsto \begin{bmatrix}0 & 1\end{bmatrix}.
\]
This induces a binary representation of $n$-qubit Pauli operators and, in particular, a matrix representation of the stabilizer generators. The binary representation of the stabilizer generators is called the parity-check matrix and is written as
\[
\mathbf{H} = \bigl[\,\mathbf{H}_{\mathrm{x}} \mid \mathbf{H}_{\mathrm{z}}\bigr],
\]
where $\mathbf{H}_{\mathrm{x}}$ and $\mathbf{H}_{\mathrm{z}}$ are binary matrices corresponding to the $X$- and $Z$-type stabilizer generators, respectively, in the binary description of the stabilizer generators.
 Since the stabilizer generators commute, the parity-check matrix satisfies
\begin{equation}
\label{eq:commutavity_condition}
\mathbf{H}_{\mathrm{x}} \mathbf{H}_{\mathrm{z}}^\mathsf{T} + \mathbf{H}_{\mathrm{x}}^\mathsf{T} \mathbf{H}_{\mathrm{z}} = \mathbf{0} \pmod{2}.
\end{equation}

\subsection{CSS Codes}
\label{sec:CSS_codes}

Consider stabilizer codes whose stabilizer generators can be divided into two groups such that the generators in one group are $n$-fold tensor products of $I$ and $X$, while the generators in the other group are $n$-fold tensor products of $I$ and $Z$. The binary representations of the stabilizer generators of these codes have the form
\begin{equation*}
    \mathbf{H}=\begin{bmatrix}
        \mathbf{H}_{\mathrm{x}} & \mathbf{0}\\
        \mathbf{0} & \mathbf{H}_{\mathrm{z}}
    \end{bmatrix}.
\end{equation*}
The commutativity condition in~\eqref{eq:commutavity_condition} reduces to
\begin{equation}
\label{eq:CSS_condition}
  \mathbf{H}_{\mathrm{x}}\mathbf{H}_{\mathrm{z}}^{\mathsf{T}}=\mathbf{0},
\end{equation}
which makes it possible to derive stabilizer codes from two classical codes that satisfy this condition. These are called \emph{CSS} codes, after Calderbank, Shor, and Steane, who first highlighted the connection between stabilizer codes and classical codes\cite{calderbank1996quantum_exists}.
Let $\mathcal{U}_{\mathrm{z}}$ and $\mathcal{U}_{\mathrm{x}}$ denote the classical codes with parity-check matrices $\mathbf{H}_{\mathrm{x}}$ and $\mathbf{H}_{\mathrm{z}}$, respectively. Since the Pauli logical operators must commute with the stabilizers, the binary representations of the $X$-type and $Z$-type logical operators are, respectively,
\begin{equation*}
    \mathcal{L}_{\mathrm{x}} = \mathcal{U}_{\mathrm{x}}/\mathcal{U}_{\mathrm{z}}^{\perp}, \qquad
    \mathcal{L}_{\mathrm{z}}= \mathcal{U}_{\mathrm{z}}/\mathcal{U}_{\mathrm{x}}^{\perp},
\end{equation*}
where $\mathcal{U}_{\mathrm{x}}^{\perp}$ and $\mathcal{U}_{\mathrm{z}}^{\perp}$ are the dual codes of $\mathcal{U}_{\mathrm{x}}$ and $\mathcal{U}_{\mathrm{z}}$, respectively.
The minimum distance of a CSS code is
\begin{equation*}
    d=\min(d_{\mathrm{x}},d_{\mathrm{z}}),
\end{equation*}
where
\begin{equation*}
    d_{\mathrm{x}}=\min_{\mathbf{u} \in \mathcal{L}_{\mathrm{x}}, \\ \mathbf{u} \neq \boldsymbol{0}} \lvert\mathbf{u}\rvert, \qquad
    d_{\mathrm{z}}=\min_{\mathbf{u} \in \mathcal{L}_{\mathrm{z}},\\ \mathbf{u} \neq \boldsymbol{0}} \lvert\mathbf{u}\rvert.
\end{equation*}

\begin{example}
    \label{ex:CSS_code}
  Consider the five-qubit stabilizer code characterized by the stabilizer group $S$, generated by $\mathcal{S}=<X_1X_2X_3X_4X_5,Z_1Z_2,Z_2Z_3Z_4Z_5>$.
  In the shorthand notation, $Z_1Z_2$ represents the five-fold tensor product $Z \otimes Z \otimes I \otimes I \otimes I$, where $Z_\mathrm{i}$ indicates the $i$-th element and identity operators are omitted.
   Since the stabilizer group generators of $\mathcal{S}$ have either $X$ or $Z$ operators, the code associated with it is a CSS code.
   In the binary representation, the $X$ and $Z$ parity-check matrices of the CSS code are given by
   \begin{equation*}
       \label{eq:PCM_in_CSS_code_ex}
       \mathbf{H}_{\mathrm{x}}=\begin{bmatrix}
           1 & 1 & 1 & 1 & 1
       \end{bmatrix} \text{ and }
       \mathbf{H}_{\mathrm{z}}=\begin{bmatrix}
           1 & 1 & 0 & 0 & 0\\
           0 & 1 & 1 & 1 & 1
       \end{bmatrix}.
   \end{equation*}
  The binary forms of the logical Pauli operators are represented by the row space of
   \begin{equation*}
       \label{eq:logical_operator_CSS_ex}
       \mathbf{L}_{\mathrm{x}} = \mathbf{L}_{\mathrm{z}}=\begin{bmatrix}
           0 & 0 & 1 & 1 & 0\\
           0 & 0 & 0 & 1 & 1
       \end{bmatrix}.
   \end{equation*}
  Since there are no weight-one vectors in the row spaces of $L_{\mathrm{x}}$ and $\mathbf{L}_{\mathrm{z}}$, and there exists a vector of weight two, the code defined by the stabilizer group $\mathcal{S}$ has distance two.
\end{example}

\subsection{Quantum LDPC Codes}
\label{sec:QLDPC_codes}
QLDPC codes are a class of CSS codes whose stabilizer weights do not scale with blocklength. There exists several families of QLDPC codes that offer different trade-offs between distance, rate, and connectivity constraints. Although the results in this paper are qualitatively applicable to most QLDPC code families, they are formally derived and validated only for HGP~\cite{Tillich_2014} and LP~\cite{LP_codes} codes. HGP and LP codes are constructed from two classical codes, so we first provide a brief overview of classical LDPC codes before detailing their construction.

\subsubsection{Classical quasi-cyclic LDPC Codes}

\label{sec:classical_LDPC_codes}
Since LP codes are constructed from two classical quasi-cyclic LDPC codes, we restrict the discussion to quasi-cyclic LDPC codes.
The parity-check matrices of quasi-cyclic codes admit a compact matrix representation, referred to as a base matrix. The elements of a base matrix take values from the quotient polynomial ring $\mathcal{R}_{\gamma} \coloneqq \mathbb{F}_2[x]/(x^\gamma-1)$. It is important to note that the quotient polynomial ring $\mathcal{R}_{\gamma}$ is isomorphic to the ring of binary circulant matrices of dimension $\gamma \times \gamma$. The parity-check matrix $\mathbf{H} \in \mathbb{F}_2^{\gamma m\times \gamma n}$ associated with a base matrix $\mathbf{B} \in \mathcal{R}_{\gamma}^{m \times n}$ is obtained by substituting each entry of $\mathbf{B}$ with its corresponding circulant-matrix representation.
The base matrix has a graphical representation similar to that of a parity-check matrix, called the base graph. Let $\underline{\mathcal{G}}(\underline{\mathcal{V}}\cup \underline{\mathcal{C}}, \underline{\mathcal{E}})$ be the graphical representation of the base matrix $\mathbf{B}$, where $\underline{\mathcal{V}}$ and $\underline{\mathcal{C}}$ are the sets of left and right nodes of $\underline{\mathcal{G}}$, corresponding to the columns and rows of the base matrix, respectively. To distinguish them from the nodes and edges of the Tanner graph, the nodes and edges of the base graphs are denoted with an underline. There is an edge between $\underline{v}_{\mathrm{j}} \in \underline{\mathcal{V}}$ and $\underline{c}_{\mathrm{i}} \in \underline{\mathcal{C}}$ if $\mathbf{B}(\mathrm{i},\mathrm{j})$ is a nonzero element of $\mathcal{R}_\gamma$.
To obtain the Tanner graph for a code from a base matrix $\mathbf{B}$, we first expand $\mathbf{B}$ into its associated parity-check matrix and then follow the procedure described in Section~\ref{sec:linear_codes} to construct the Tanner graph.
Alternatively, the Tanner graph can be obtained directly from the base graph using the copy-and-permute operation\cite{Divsalar_protograph_LDPC}. The copy-and-permute operation involves two steps. In the copy step, the base graph is copied $\gamma$ times. Denote the $t$-th copy of the base graph by $\underline{\mathcal{G}}(t)$ and its left nodes, right nodes, and edges by $\underline{\mathcal{V}}(t)$, $\underline{\mathcal{C}}(t)$, and $\underline{\mathcal{E}}(t)$, respectively.
During the permutation step, the edges of the base-graph copies are rearranged. Specifically, an edge $\underline{e}$ that connects the variable node $\underline{v}_{\mathrm{j}}(t) \in \underline{\mathcal{V}}(t)$ to the check node $\underline{c}_{\mathrm{i}}(t) \in \underline{\mathcal{C}}(t)$ becomes an edge (which we still denote by $\underline{e}$) connecting the variable node $\underline{v}_\mathrm{j}(t)$ to the check node $\underline{c}_\mathrm{i}(\pi_{\mathrm{i},\mathrm{j}}(t))$, where $\pi_{\mathrm{i},\mathrm{j}}$ is a cyclic permutation of $[\gamma]$ determined by the circulant-matrix representation of $\mathbf{B}(\mathrm{i},\mathrm{j})$. For a detailed description of the copy-and-permute operation, see \cite{Divsalar_protograph_LDPC}.

\subsubsection{Lifted-product Codes} 
\label{sec:LP_codes}
Given two base matrices $\mathbf{B}_1 \in \mathcal{R}_{\gamma}^{m_1 \times n_1}$ and $\mathbf{B}_2 \in \mathcal{R}_{\gamma}^{m_2 \times n_2}$ corresponding to two classical LDPC codes, the LP construction, introduced by Panteleev and Kalachev\cite{LP_codes}, produces two base matrices $\mathbf{B}_{\mathrm{x}} \in \mathcal{R}_{\gamma}^{m_1n_2\times (m_1m_2+n_1n_2)}$ and $\mathbf{B}_{\mathrm{z}}\in \mathcal{R}_{\gamma}^{m_2n_1\times (m_1m_2+n_1n_2)}$ given by
\begin{align}
\label{eq:LP_construction}
\mathbf{B}_{\mathrm{x}} & = \begin{bmatrix}
\mathbf{B}_1 \otimes \mathbf{I}_{n_{2}} & \mathbf{I}_{m_{1}} \otimes \mathbf{B}_2^{*}
\end{bmatrix},\\
\mathbf{B}_{\mathrm{z}} & = \begin{bmatrix}
\mathbf{I}_{n_{1}} \otimes \mathbf{B}_2 & \mathbf{B}_1^{*} \otimes \mathbf{I}_{m_{2}} \\
\end{bmatrix},
\label{eq:lifted_product_base_matrices}
\end{align}
where $\mathbf{B}^{*}(i,j)= \big(\mathbf{B}(j,i)\big)^{*}$ and, for any $\mathbf{B}(i,j) \in \mathcal{R}_{\gamma}$,
\begin{equation}
\big(\mathbf{B}(i,j)\big)^{*}=
\begin{cases}
\big(\mathbf{B}(i,j)\big)^{-1}, & \text{if } \mathbf{B}(i,j) \in \mathcal{R}_{\gamma} \setminus \{0\},\\
0, & \text{if } \mathbf{B}(i,j)=0.
\end{cases}
\label{eq:circulant_conjugate}
\end{equation}
Since each nonzero ring element used in $\mathbf{B}_1$ and $\mathbf{B}_2$ corresponds to a circulant-permutation matrix, its inverse is equal to its transpose.
Define $M_{\mathrm{x}}\coloneqq \gamma m_1n_2$, $M_{\mathrm{z}}\coloneqq \gamma m_2n_1$, and $Q \coloneqq \gamma(m_1m_2+n_1n_2)$.
The parity-check matrices $\mathbf{H}_{\mathrm{x}} \in \mathbb{F}_2^{M_{\mathrm{x}} \times Q}$ and $\mathbf{H}_{\mathrm{z}} \in \mathbb{F}_2^{M_{\mathrm{z}}\times Q}$ corresponding to the base matrices $\mathbf{B}_{\mathrm{x}}$ and $\mathbf{B}_{\mathrm{z}}$ satisfy the condition in \eqref{eq:CSS_condition}. Therefore, $\big(\mathbf{H}_{\mathrm{x}}, \mathbf{H}_{\mathrm{z}}\big)$ together define a CSS code.
When the elements of $\mathbf{B}_{i}$, for $i \in \{1,2\}$, are chosen from $\mathbb{F}_2$ instead of $\mathcal{R}_{\gamma}$, the matrices $\mathbf{B}_{i}$ can be interpreted as parity-check matrices of linear codes. In this case, the two matrices in \eqref{eq:LP_construction} can be interpreted as parity-check matrices of two classical codes that satisfy the CSS criterion in \eqref{eq:CSS_condition}. This modified construction of QLDPC codes is called the hypergraph product (HGP) construction, introduced by Tillich and Zémor\cite{Tillich_2014}.
Equivalently, LP codes can be constructed by taking the graph product of the Base graphs of the two classical LDPC codes. Since some aspects of the decoder behavior are easier to explain using the graph product viewpoint, we next describe the formulation of LP codes as a graph product.

\textbf{Graphical description}: 
The graph product of two base graphs, $\underline{\mathcal{G}}^1(\underline{\mathcal{V}}^1\cup\underline{\mathcal{C}}^1,\underline{\mathcal{E}}^1)$ and $\underline{\mathcal{G}}^2(\underline{\mathcal{V}}^2\cup\underline{\mathcal{C}}^2,\underline{\mathcal{E}}^2)$, yields a third graph $\underline{\mathcal{G}}_{\mathrm{LP}}$ whose left nodes are given by
\begin{equation*}
\label{eq:variable_nodes}
\underline{\mathcal{Q}} = \left( \underline{\mathcal{V}}^1 \bigtimes \underline{\mathcal{V}}^2 \right) \cup \left( \underline{\mathcal{C}}^1 \bigtimes \underline{\mathcal{C}}^2 \right),
\end{equation*}
and whose right nodes are given by
\begin{equation}
\label{eq:check_nodes}
\underline{\mathcal{C}} = \left( \underline{\mathcal{C}}^1 \bigtimes \underline{\mathcal{V}}^2 \right) \cup \left( \underline{\mathcal{V}}^1 \bigtimes \underline{\mathcal{C}}^2 \right).
\end{equation}
Nodes and edges of the base graphs are distinguished by their corresponding superscripts. Moreover, we use a prime on subscripts (e.g., $\mathrm{i}'$, $\mathrm{j}'$) for indices associated with the first base graph $\underline{\mathcal{G}}^1$, while indices associated with the second base graph $\underline{\mathcal{G}}^2$ are written without a prime.
Left-side nodes of the form $\underline{v}^1_{\mathrm{j}'}\,\underline{v}^2_{\mathrm{j}}$, with $\underline{v}_{\mathrm{j}'}^1 \in \underline{\mathcal{V}}^1$ and $\underline{v}_{\mathrm{j}}^2 \in \underline{\mathcal{V}}^2$, are called VV-type variable nodes, while those of the form $\underline{c}^1_{\mathrm{i}'}\,\underline{c}^2_{\mathrm{i}}$, with $\underline{c}_{\mathrm{i}'}^1 \in \underline{\mathcal{C}}^1$ and $\underline{c}_{\mathrm{i}}^2 \in \underline{\mathcal{C}}^2$, are called CC-type variable nodes.
Right-side nodes of the form $\underline{c}^1_{\mathrm{i}'}\,\underline{v}^2_{\mathrm{j}}$, with $\underline{c}_{\mathrm{i}'}^1 \in \underline{\mathcal{C}}^1$ and $\underline{v}_{\mathrm{j}}^2 \in \underline{\mathcal{V}}^2$, are called $X$-stabilizers or $X$-checks, while those of the form $\underline{v}^1_{\mathrm{j}'}\,\underline{c}^2_{\mathrm{i}}$, with $\underline{v}_{\mathrm{j}'}^1 \in \underline{\mathcal{V}}^1$ and $\underline{c}_{\mathrm{i}}^2 \in \underline{\mathcal{C}}^2$, are called $Z$-stabilizers or $Z$-checks. The collections of $X$-checks and $Z$-checks are denoted by $\underline{\mathcal{C}}_{\mathrm{x}}$ and $\underline{\mathcal{C}}_{\mathrm{z}}$, respectively. From \eqref{eq:check_nodes}, it follows that
\begin{equation*}
\label{eq:X_checks_and_Z_checks}
\underline{\mathcal{C}}_{\mathrm{x}} = \underline{\mathcal{C}}^1 \bigtimes \underline{\mathcal{V}}^2 \quad \text{and} \quad \underline{\mathcal{C}}_{\mathrm{z}} = \underline{\mathcal{V}}^1 \bigtimes \underline{\mathcal{C}}^2.
\end{equation*}
In the product graph $\underline{\mathcal{G}}_{\mathrm{LP}}$, there is an edge between an $X$-check $\underline{c}^1_{\mathrm{i}'}\,\underline{v}^2_{\mathrm{j}}$ and a VV-type variable node $\underline{v}^1\,\underline{v}^2$ if $\underline{v}^1 \in \mathcal{N}_{\underline{c}^1_{\mathrm{i}'}}$ and $\underline{v}^2 = \underline{v}^2_{\mathrm{j}}$. Similarly, there is an edge between an $X$-check $\underline{c}^1_{\mathrm{i}'}\,\underline{v}^2_{\mathrm{j}}$ and a CC-type variable node $\underline{c}^1\,\underline{c}^2$ if $\underline{c}^1 = \underline{c}^1_{\mathrm{i}'}$ and $\underline{c}^2 \in \mathcal{N}_{\underline{v}^2_{\mathrm{j}}}$. Therefore, the set of neighboring variable nodes of the $X$-check $\underline{c}^1_{\mathrm{i}'}\,\underline{v}^2_{\mathrm{j}}$ is given by
\begin{equation*}
    \label{eq:X_neighbors}
    \mathcal{N}_{\underline{c}^1_{\mathrm{i}'}\underline{v}^2_{\mathrm{j}}} = \left( \mathcal{N}_{\underline{c}^1_{\mathrm{i}'}} \bigtimes \underline{v}^2_{\mathrm{j}} \right) \cup \left( \underline{c}^1_{\mathrm{i}'} \bigtimes \mathcal{N}_{\underline{v}^2_{\mathrm{j}}} \right).
\end{equation*}
Similarly, the set of neighboring variable nodes of a $Z$-type check $\underline{v}^1_{\mathrm{j}'}\,\underline{c}^2_{\mathrm{i}}$ is given by
\begin{equation*}
    \label{eq:Z_neighbors}
    \mathcal{N}_{\underline{v}^1_{\mathrm{j}'}\underline{c}^2_{\mathrm{i}}} = \left( \underline{v}^1_{\mathrm{j}'} \bigtimes \mathcal{N}_{\underline{c}^2_{\mathrm{i}}} \right) \cup \left( \mathcal{N}_{\underline{v}^1_{\mathrm{j}'}} \bigtimes \underline{c}^2_{\mathrm{i}} \right).
\end{equation*}
\begin{definition}
    \label{def:induced_subgraph}
    The subgraph induced by a set of nodes in a graph consists of those nodes, their neighboring nodes, and all edges connecting the inducing nodes to their neighbors.
\end{definition}
Let the subgraphs induced by $\underline{\mathcal{C}}_{\mathrm{x}}$ and $\underline{\mathcal{C}}_{\mathrm{z}}$ in $\underline{\mathcal{G}}_{\mathrm{LP}}$ be denoted by $\underline{\mathscr{G}}_{\mathrm{x}}$ and $\underline{\mathscr{G}}_{\mathrm{z}}$, respectively. Observe that $\underline{\mathcal{G}}_{\mathrm{x}}$ and $\underline{\mathcal{G}}_{\mathrm{z}}$ are base graphs corresponding to the base matrices $\mathbf{B}_{\mathrm{x}}$ and $\mathbf{B}_{\mathrm{z}}$ given in \eqref{eq:LP_construction}, assuming $\underline{\mathscr{G}}_1$ and $\underline{\mathscr{G}}_2$ are base graphs corresponding to the base matrices $\mathbf{B}_1$ and $\mathbf{B}_2$, respectively. The product of two nodes, one from $\underline{\mathscr{G}}_1$ and the other from $\underline{\mathscr{G}}_2$, is illustrated in Figure~\ref{fig:node_types}.
\begin{figure}
    \centering
    \tikzstyle{cnode}=[circle,minimum size=0.55 cm,draw]
\tikzstyle{scnode}=[circle,minimum size=0.15 cm,draw]
\tikzstyle{zscnode}=[circle,minimum size=0.15 cm,draw,fill=red]
\tikzstyle{zrnode}=[rectangle,draw,minimum width=0.4 cm,minimum height=0.4cm,outer sep=0pt]
\tikzstyle{cgnode}=[circle,draw]
\tikzstyle{crnode}=[circle,draw]
\tikzstyle{conode}=[rectangle,draw]
\tikzstyle{cpnode}=[circle,draw]
\tikzstyle{rnode}=[rectangle,draw,minimum width=0.4 cm,minimum height=0.4cm,outer sep=0pt]
\tikzstyle{srnode}=[rectangle,draw,minimum width=0.1 cm,minimum height=0.1cm,outer sep=0pt]
\tikzstyle{prnode}=[rectangle,rounded corners,fill=blue!50,text width=4.5em,text centered,outer sep=0pt]
\tikzstyle{prnodebig}=[rectangle,rounded corners,fill=blue!50,text width=7em,text centered,outer sep=0pt]
\tikzstyle{prnodesimple}=[rectangle,draw,text width=4.5em,text centered,outer sep=0pt]
\tikzstyle{bigsnake}=[fill=green!50,snake=snake,segment amplitude=4mm, segment length=4mm, line after snake=1mm]
\tikzstyle{smallsnake}=[snake=snake,segment amplitude=0.7mm, segment length=4mm, line after snake=1mm]
\definecolor{color1}{rgb}{1,0.2,0.3}
\definecolor{color2}{rgb}{0.4,0.5,0.7}
\definecolor{color3}{rgb}{0.1,0.8,0.5}
\definecolor{color4}{rgb}{0.5,0.3,1}
\definecolor{color5}{rgb}{0.5,1,1}
\definecolor{color6}{rgb}{0.8,0.3,0.6}
\definecolor{color7}{rgb}{0.6,0.4,0.3}

\begin{tikzpicture}[every node/.style={scale=1}]
\begin{scope}[node distance=1.5cm,semithick]
\node[cnode] (v1) {};
\node (m1) [right of = v1,xshift=-0.75cm] {$\times$};
\node[cnode] (v2) [right of =m1,xshift=-0.75cm]{};
\node (arrow1)[right of =v2,xshift=-0.65cm]{$=$};
\node[cnode] (vv)[right of=arrow1,xshift=-0.65cm] {};
\node [scnode](vv_in)[right of=arrow1,xshift=-0.65cm]{};
\node [below of=v2,yshift=0.8 cm]{\textbf{VV-type}};
\node[rnode] (c1) [right of =vv]{};
\node (m2) [right of = c1,xshift=-0.75cm] {$\times$};
\node[rnode] (c2) [right of =m2,xshift=-0.75cm]{};
\node (arrow2)[right of =c2,xshift=-0.65cm]{$=$};
\node[rnode] (cc)[right of=arrow2,xshift=-0.65cm] {};
\node[srnode] (cc_in)[right of=arrow2,xshift=-0.65cm] {};
\node [below of=c2,yshift=0.8 cm]{\textbf{CC-type}};

\node[rnode] (c3) [below of =v1]{};
\node (m3) [right of = c3,xshift=-0.75cm] {$\times$};
\node[cnode] (v3) [right of =m3,xshift=-0.75cm]{};
\node (arrow3)[right of =v3,xshift=-0.65cm]{$=$};
\node[rnode] (cv)[right of=arrow3,xshift=-0.65cm] {};
\node [scnode](cv_in)[right of=arrow3,xshift=-0.65cm]{};
\node [below of=v3,yshift=0.8 cm]{\textbf{CV-type}};
\node[cnode] (v4) [right of =cv]{};
\node (m4) [right of = v4,xshift=-0.75cm] {$\times$};
\node[rnode] (c4) [right of =m4,xshift=-0.75cm]{};
\node (arrow4)[right of =c4,xshift=-0.65cm]{$=$};
\node[cnode] (vc)[right of=arrow4,xshift=-0.65cm] {};
\node [srnode](vc_in)[right of=arrow4,xshift=-0.65cm]{};
\node [below of=c4,yshift=0.8 cm]{\textbf{VC-type}};
\end{scope}
\end{tikzpicture}
    \caption{Node types in the Tanner graphs of HGP/LP codes. VV-type variable nodes are depicted with a circle inside another circle because they are obtained by multiplying two variable nodes, each represented by a single circle. Analogously, CC-type variable nodes, $X$-checks, and $Z$-checks are represented, respectively, by a square inside another square, a circle inside a square, and a square inside a circle.}
    \label{fig:node_types}
\end{figure}
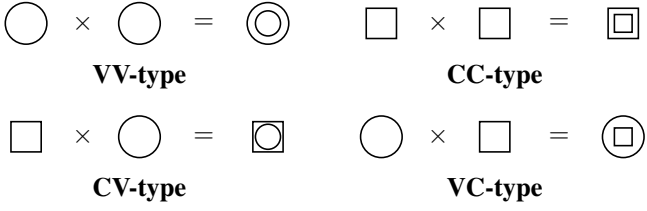

\subsection{Syndrome Measurement Circuit}
\label{sec:syndrome_extraction}
Encoded states lie in the $+1$ eigenspace of the stabilizers, so checking whether a faulty state is in the codespace amounts to verifying that it lies in the $+1$ eigenspace.
The circuit used to test whether a quantum state is in the $+1$ eigenspace of a given stabilizer generator is called a detector.
A detector outputs $+1$ for states in the $+1$ eigenspace and $-1$ for states in the $-1$ eigenspace.
Taken together, the detectors corresponding to all stabilizer generators of a code form the syndrome measurement circuit\cite{Nielsen}.

To understand the structure and operation of a measurement circuit, consider the stabilizer group $S = \{I, Z\}$, which stabilizes the trivial state $\ket{0}$.
Prepare an ancilla qubit in the state $\ket{+}$ and apply a CZ gate between the encoded qubit and the ancilla.
If the encoded qubit is in the $+1$ eigenstate—here, the state $\ket{0}$—the CZ gate leaves the ancilla unchanged, and the joint state is $\ket{0} \otimes (\ket{0} + \ket{1})$.
Otherwise, the CZ gate flips the phase of the ancilla, producing $\ket{0} \otimes (\ket{0} - \ket{1})$.
Note that after applying the CZ gate, the encoded qubit and ancilla remain in a product state rather than becoming entangled.
Thus, by applying a Hadamard gate to the ancilla qubit and then measuring it in the $Z$ basis, one can determine whether the encoded state lies in the codespace.
The corresponding measurement circuit is shown in Figure~\ref
{fig:stab_measurement} (a).
 \begin{figure}
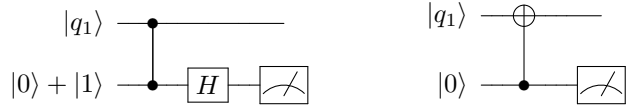

 \centering
 \begin{subfigure}{0.24\textwidth}
     \input{Figures/one_qubit_Z_stab}
     \subcaption{Measurement of $Z$ stabilizer that stabilizes the $\ket{0}$ state.}
     \label{fig:one_qubit_Z_stab}
 \end{subfigure}
 \hfill
   \centering
 \begin{subfigure}{0.24\textwidth}
     \input{Figures/1qubit_X_stabilizer_measurement}
     \subcaption{Measurement of $X$ operator that stabilizes the $\ket{+}$ state.}
     \label{fig:one_qubit_X_stab}
 \end{subfigure}
 \caption{Measurement of stabilizers that acts on one qubit. }
 \label{fig:stab_measurement}
 \end{figure}
 To extend the above notion of syndrome measurement for $Z$-stabilizers from a single qubit to multiple qubits, recall that a stabilizer is the tensor product of Pauli matrices and the encoded state is the tensor product of the eigenvectors of these Pauli matrices such that the product of their respective eigenvalues equals one. 
 \begin{figure}
     \centering
     \input{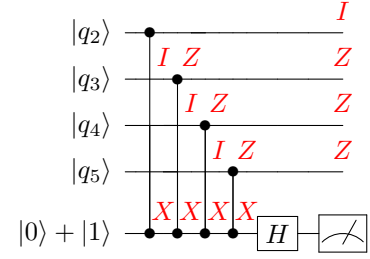}
     \caption{Measurement of the $Z_2Z_3Z_4Z_5$ operator.}
     \label{fig:Z_stabilizer_measurement}
 \end{figure}
Therefore, to identify whether a faulty encoded state belongs to the $+1$ eigen space, initially apply CZ gates sequentially between the ancilla, prepared in the $\ket{0}+\ket{1}$ state, and the qubits on which the $Z$ stabilizer acts non-trivially. Then measure the ancilla.
Figure~\ref{fig:Z_stabilizer_measurement} illustrates the measurement circuit for the stabilizer generator $Z_2Z_3Z_4Z_5$ from the stabilizer group $S$ in Example~\ref{ex:CSS_code}.
Analogously, it can be shown that the measurement of a $X$ stabilizer involves initializing an ancilla qubit in the state $\ket{0},$ followed by applying CNOT gates sequentially between the ancilla and each qubit that the $X$ stabilizer affects non-trivially, with the ancilla qubit acting as the control qubit.
The measurement of the eigenvalue corresponding to the $X$ operator is shown in Figure~\ref{fig:one_qubit_X_stab}.
\subsection{Detector-Error Matrix}
\label{sec:detector_error_Tanner_graph}

Given a specified set of fault locations together with all possible errors associated with each location, the detector–error matrix indicates whether a given error causes the codeword states to exit the codespace.
Let $\mathcal{K}$ denote the set of fault locations in a circuit, and let $F_k$ for $k \in \mathcal{K}$ denote the error at the $k$-th location, with $F_k$ taking the values in the set $\mathcal{F}_k$.

The error detector matrix associated with the syndrome-measurement circuit for $R$ rounds of measurement is a binary matrix of dimension $Rr \times \sum_{k \in \mathcal{K}} |\mathcal{F}_k|$, where the rows correspond to the detectors and the columns correspond to errors in the circuit. The $(i,j)$-th entry of the error–detector matrix is $1$ if the error corresponding to the $j$-th column causes the $i$-th detector to output $-1$.
To determine whether a specific error changes a detector output from $+1$ to $-1$, one must consider the detector output when only that error is present and no other errors occur in the circuit. This can be done by propagating the error through the controlled gates according to the following rules\cite{Nielsen}. For CNOT, we use the convention that the tensor-product ordering is \emph{(target $\otimes$ control)}, consistent with Fig.~\ref{fig:propagation_rules}:
\begin{gather}
(I \otimes X)\, \mathrm{CNOT} = (X \otimes X)\, \mathrm{CNOT},\\
(I \otimes Z)\, \mathrm{CNOT} = (I \otimes Z)\, \mathrm{CNOT},\\
(Z \otimes I)\, \mathrm{CNOT} = (Z \otimes Z)\, \mathrm{CNOT},\\[0.4em]
(I \otimes X)\, \mathrm{CZ} = (Z \otimes X)\, \mathrm{CZ},\quad
(X \otimes I)\, \mathrm{CZ} = (X \otimes Z)\, \mathrm{CZ},\\
(I \otimes Z)\, \mathrm{CZ} = (I \otimes Z)\, \mathrm{CZ},\quad
(Z \otimes I)\, \mathrm{CZ} = (Z \otimes I)\, \mathrm{CZ}.
\end{gather}

The above propagation rules are illustrated in Figure~\ref{fig:propagation_rules}.
Note that an ancilla error can propagate to multiple data qubits. Such errors are referred to as \emph{hook errors}\cite{Nielsen}. In Figure~\ref{fig:Z_stabilizer_measurement}, an $X$ error in the ancilla propagates to qubits $q_3$, $q_4$, and $q_5$, manifesting as $Z$ errors on these qubits.
The error detector matrix has $rR$ rows, where $r$ denotes the number of generators in the stabilizer group and $R$ is the number of syndrome-extraction rounds. Each stabilizer-generator measurement is associated with a single detector. Consequently, for $R$ rounds of syndrome extraction, there are $rR$ detectors in total.
The Tanner graph associated with the error--detector matrix is commonly referred to as the \emph{circuit-level Tanner graph} (or \emph{detector--error Tanner graph}).

\begin{definition}
The effective minimum distance is the smallest number of physical faults whose propagation through the entire syndrome extraction circuit results in  a nontrivial logical Pauli operator.
\end{definition}
\begin{figure}
    \centering
    \input{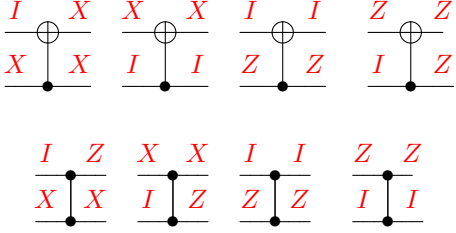}
    \caption{Propagation of errors through CNOT and CZ gates.}
    \label{fig:propagation_rules}
\end{figure}
\begin{figure*}
    \centering
    \input{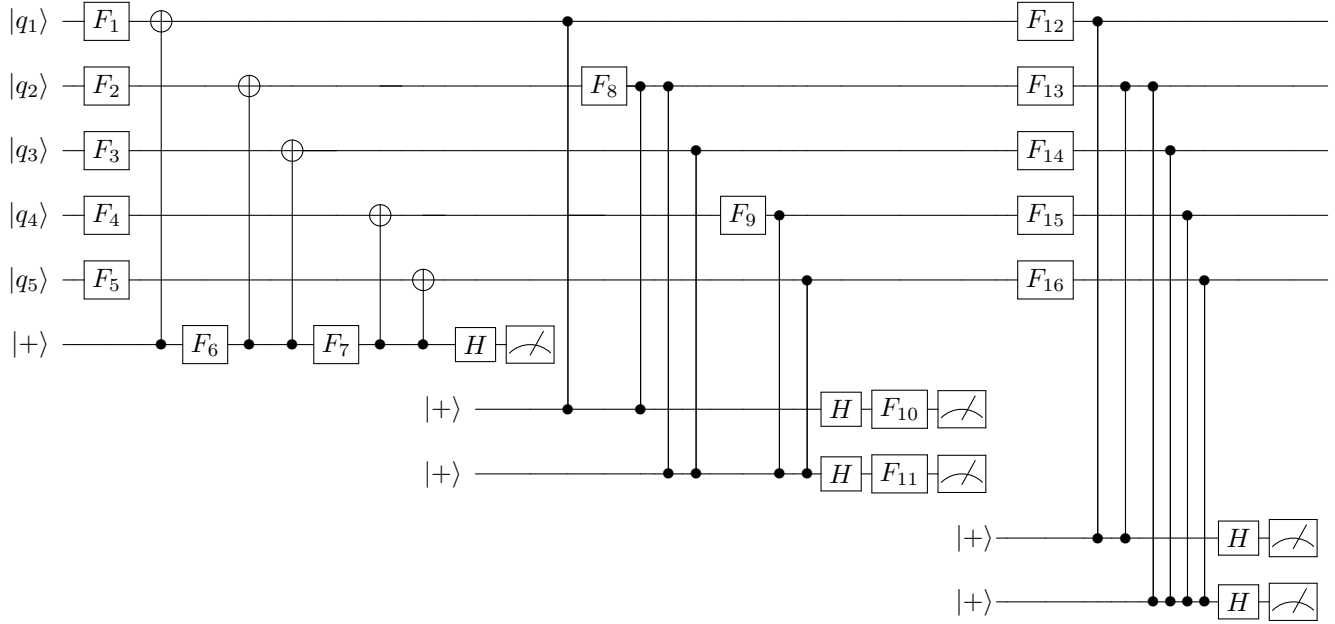}
    \caption{Syndrome measurement circuit for the CSS code in \ref{ex:CSS_code}: The fault locations are shown as a set of random gates $F_i$, for $i \in [16]$, which are Pauli matrices with fixed probabilities. The $X$ stabilizer is measured once, and the $Z$ stabilizers are measured twice. }
    \label{fig:example_error_detector_circuit}
\end{figure*}
\begin{figure*}
    \centering
    \input{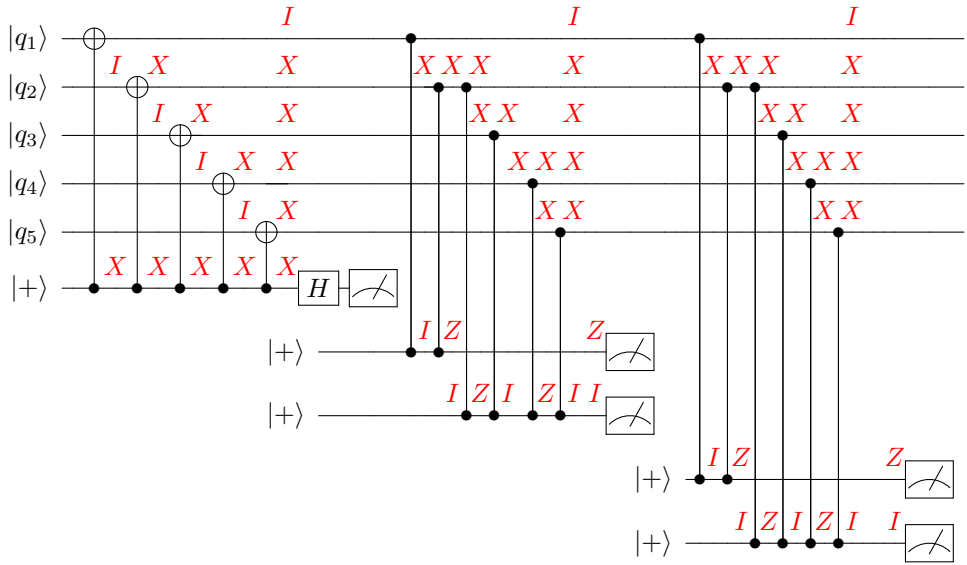}
    \caption{Determination of the group of detectors whose outputs vary when $F_6=X$ in the circuit shown in Figure~\ref{fig:example_error_detector_circuit}: The errors on the inputs and outputs of the controlled $Z$  and CNOT gates are displayed, except in the cases where both inputs are identity. All measurements are performed in the $Z$ basis. Observe that error $F_6=X$ propagates to qubits $\ket{q_2},\ket{q_3},\ket{q_4},\ket{q_5}$ and is detected by the second and fourth detectors. These detectors are associated with the measurements of stabilizer $Z_1Z_2$ from Example~\ref{ex:CSS_code}. }
    \label{fig:error_propagation_illustration}
\end{figure*}
\begin{example}
    \label{ex:error_detector_matrix}
    Consider the syndrome measurement circuit corresponding to the code in Example~\ref{ex:CSS_code}, in which the $X$ stabilizer is measured once and the $Z$ stabilizers are measured twice (see Figures~\ref{fig:example_error_detector_circuit} and~\ref{fig:error_propagation_illustration}). The fault locations are shown as a set of random gates $F_i$ for $i \in \{1,2,\ldots,16\}$. Each $F_i$ takes values in $\{I,X,Z\}$ according to a specified probability distribution.
The error-detector matrix corresponding to the detectors associated with the $Z$ stabilizers is given by
\begin{align}
\label{eq:error_det_mat}
\resizebox{0.91\hsize}{!}{$
\begin{blockarray}{c@{\hspace*{-3pt}}c@{\hspace*{-2.5pt}}c@{\hspace*{-2.5pt}}c@{\hspace*{-2.5pt}}c@{\hspace*{-2.5pt}}c@{\hspace*{-2.5pt}}c@{\hspace*{-2pt}}c@{\hspace*{-2pt}}c@{\hspace*{-2pt}}c@{\hspace*{-2pt}}c@{\hspace*{-2pt}}c@{\hspace*{-2pt}}c@{\hspace*{-2pt}}c@{\hspace*{-2pt}}c@{\hspace*{-2pt}}c }
{\scriptstyle X_1} & {\scriptstyle X_2} & {\scriptstyle X_3} & {\scriptstyle X_4} & {\scriptstyle X_5} & {\scriptstyle X_6} & {\scriptstyle X_7} & {\scriptstyle X_8} & {\scriptstyle X_9} & {\scriptstyle X_{10}} & {\scriptstyle X_{11}} & {\scriptstyle X_{12}} & {\scriptstyle X_{13}} & {\scriptstyle X_{14}} & {\scriptstyle X_{15}} & {\scriptstyle X_{16}}\\[ -0.25ex]
\begin{block}{[cccccccccccccccc]}
1 & 1 & 0 & 0 & 0 & 1 & 0 & 1 & 0 & 1 & 0 & 0 & 0 & 0 & 0 & 0 \\
0 & 1 & 1 & 1 & 1 & 0 & 0 & 1 & 1 & 0 & 1 & 0 & 0 & 0 & 0 & 0\\
1 & 1 & 0 & 0 & 0 & 1 & 0 & 1 & 0 & 0 & 0 & 1 & 1 & 0 & 0 & 0\\
0 & 1 & 1 & 1 & 1 & 0 & 0 & 1 & 1 & 0 & 0 & 0 & 1 & 1 & 1 & 1 \\
\end{block}%
\end{blockarray}.$}
\end{align}
The column labeled $X_k$ represents the detector outputs when the error at the $k$-th location is $X$, i.e., when $F_k = X$. Note that when $F_k = Z$, the outputs of the detectors associated with the $Z$ stabilizers do not change; hence, errors with $F_k = Z$ are omitted from the error-detector matrix in~\eqref{eq:error_det_mat}.
The first two rows correspond to detectors linked to the first round of measurements of the stabilizers $Z_1Z_2$ and $Z_2Z_3Z_4Z_5$, respectively. The third and fourth rows correspond to their measurements in the second round.
Figure~\ref{fig:error_propagation_illustration} illustrates how to determine the entries of a column by inserting an $X$ fault at the 6-th location and propagating it to the end of the circuit. Since an $X$ error at the 6-th location flips the outputs of the first and third detectors, the corresponding entries are 1 in the column labeled $X_6$. In contrast, an $X$ error at the 7-th location is not detected. Consequently, the column labeled $X_7$ consists entirely of zeros. This implies that the effective distance of the circuit is strictly less than two.
An alternative representation of the error-detector matrix in~\eqref{eq:error_det_mat} is
\begin{align}
\label{eq:alt_error_det_mat}
\resizebox{0.91\hsize}{!}{$
\begin{blockarray}{c@{\hspace*{-2.5pt}}c@{\hspace*{-2.5pt}}c@{\hspace*{-2.5pt}}c@{\hspace*{-2.5pt}}c@{\hspace*{-2.5pt}}c@{\hspace*{-2.5pt}}c@{\hspace*{-2.5pt}}c@{\hspace*{-2.5pt}}c@{\hspace*{-2pt}}c@{\hspace*{-2.5pt}}c@{\hspace*{-2.5pt}}c@{\hspace*{-2.5pt}}c@{\hspace*{-2.5pt}}c@{\hspace*{-2.5pt}}c@{\hspace*{-2.5pt}}c }
{\scriptstyle X_1} & {\scriptstyle X_2} & {\scriptstyle X_3} & {\scriptstyle X_4} & {\scriptstyle X_5} & {\scriptstyle X_6} & {\scriptstyle X_7} & {\scriptstyle X_8} & {\scriptstyle X_9} & {\scriptstyle X_{10}} & {\scriptstyle X_{11}} & {\scriptstyle X_{12}} & {\scriptstyle X_{13}} & {\scriptstyle X_{14}} & {\scriptstyle X_{15}} & {\scriptstyle X_{16}} \\[ -0.25ex]
\begin{block}{[cccccccccccccccc]}
1 & 1 & 0 & 0 & 0 & 1 & 0 & 1 & 0 & 1 & 0 & 0 & 0 & 0 & 0 & 0\\
0 & 1 & 1 & 1 & 1 & 0 & 0 & 1 & 1 & 0 & 1 & 0 & 0 & 0 & 0 & 0\\
0 & 0 & 0 & 0 & 0 & 0 & 0 & 0 & 0 & 1 & 0 & 1 & 1 & 0 & 0 & 0\\
0 & 0 & 0 & 0 & 0 & 0 & 0 & 0 & 0 & 0 & 1 & 0 & 1 & 1 & 1 & 1\\
\end{block}%
\end{blockarray}. $}
\end{align}
In this alternative representation, the first two rows are identical to those in~\eqref{eq:error_det_mat}. The third row is the modulo-two sum of the measurement outcomes for $Z_1Z_2$ from the first and second rounds, and the fourth row is the modulo-two sum of the measurement outcomes for $Z_2Z_3Z_4Z_5$ from the two rounds.
The minimum effective distance is the same for both representations, as they have the same null space~\cite{jens_eiset2024designing}. However, the alternative representation is sparser, which makes it more amenable to iterative decoding.
\end{example}

\subsection{Noise Models and Decoders}
\label{sec:circuit_level_noise_model}
We consider three noise models: (1) capacity-error model, (2) phenomenological noise model, and (3) circuit-level noise model\cite{SurfaceCodes_DecodingMWPM}. In the capacity-error model, each data qubit independently undergoes a Pauli error $X$, $Y$, or $Z$ with probability $p/3$, while all gates and syndrome measurements are assumed to be perfect.
In the phenomenological noise model, each data qubit independently undergoes a Pauli error $X$, $Y$, or $Z$ with probability $p/3$, and each measurement outcome is independently flipped with probability $p$; there are no CNOT gate errors.


In the circuit-level noise model, each data qubit independently undergoes a Pauli error $X$, $Y$, or $Z$ with probability $p/3$. Each CNOT gate independently fails with probability $p$; conditioned on failure, it applies one of the $15$ nontrivial two-qubit Pauli errors uniformly at random (so each occurs with probability $p/15$). Finally, each measurement outcome is independently flipped with probability $p$.
We primarily focus on the circuit-level noise model, and refer to the capacity-error and phenomenological noise models only when needed for comparison.

\subsection{Decoders}
\label{sec:decoder}
Given a measured syndrome, a decoder outputs an estimate of the underlying error up to stabilizer equivalence; if it cannot do so, we declare a decoding failure. In this work, we focus on iterative message-passing decoders. In particular, we use the min-sum decoder~\cite{mct} and finite-alphabet iterative decoders (FAIDs)~\cite{FAID_diversity}. We first describe the min-sum message-update rules for classical LDPC codes; the extension to QLDPC codes is straightforward. The FAID rules are given in Appendix~\ref{sec:FAID}.

For a Tanner graph $\mathcal{G}(\mathcal{V}\cup\mathcal{C},\mathcal{E})$, let $\boldsymbol{\sigma}\in\{0,1\}^{|\mathcal{C}|}$ denote the measured syndrome, where $\sigma_j=1$ indicates that check node $c_{\mathrm{j}}$ is unsatisfied. Let $\gamma_i$ denote an a-priori reliability value associated with variable node $v_{\mathrm{i}}\in\mathcal{V}$ (in the same spirit as the channel value $C$ used in Appendix~\ref{sec:FAID}).

In the min-sum decoder, the message from variable node $v_{\mathrm{i}}$ to check node $c_{\mathrm{j}}\in\mathcal{N}_{v_{\mathrm{i}}}$ is denoted by $m_{v_{\mathrm{i}}\rightarrow c_{\mathrm{j}}}$, and the message from $c_{\mathrm{j}}$ to $v_{\mathrm{i}}$ is denoted by $m_{c_{\mathrm{j}}\rightarrow v_{\mathrm{i}}}$. Note that the symbol $m$ is also used elsewhere in this paper to denote the number of rows in the base matrix of an LP code; the meaning is always clear from context. We initialize $m_{v_{\mathrm{i}}\rightarrow c_{\mathrm{j}}}=\gamma_i$ and $m_{c_{\mathrm{j}}\rightarrow v_{\mathrm{i}}}=0$.
In the update equations below, we omit the iteration index: in each iteration, $m_{c_{\mathrm{j}}\rightarrow v_{\mathrm{i}}}$ is computed from the previous-iteration values of $m_{v_{\mathrm{k}}\rightarrow c_{\mathrm{j}}}$, and then $m_{v_{\mathrm{i}}\rightarrow c_{\mathrm{j}}}$ is computed using the updated check-to-variable messages.
\begin{equation}
\label{eq:minsum_c2v}
 m_{c_{\mathrm{j}}\rightarrow v_{\mathrm{i}}}
 = (1-2\sigma_j)
 \prod_{v_{\mathrm{k}}\in\mathcal{N}_{c_{\mathrm{j}}}\setminus v_{\mathrm{i}}}\! \mathrm{sgn}\bigl(m_{v_{\mathrm{k}}\rightarrow c_{\mathrm{j}}}\bigr)
 \min_{v_{\mathrm{k}}\in\mathcal{N}_{c_{\mathrm{j}}}\setminus v_{\mathrm{i}}}\! \bigl|m_{v_{\mathrm{k}}\rightarrow c_{\mathrm{j}}}\bigr|.
\end{equation}
\begin{equation}
\label{eq:minsum_v2c}
 m_{v_{\mathrm{i}}\rightarrow c_{\mathrm{j}}}
 = \gamma_i + \sum_{c_{\mathrm{j}'}\in\mathcal{N}_{v_{\mathrm{i}}}\setminus c_{\mathrm{j}}}\! m_{c_{\mathrm{j}'}\rightarrow v_{\mathrm{i}}}.
\end{equation}
After each iteration, we form an a-posteriori reliability value
\begin{equation}
\label{eq:minsum_aposteriori_llr}
\Gamma_i = \gamma_i + \sum_{c_{\mathrm{j}}\in\mathcal{N}_{v_{\mathrm{i}}}} m_{c_{\mathrm{j}}\rightarrow v_{\mathrm{i}}}.
\end{equation}
and make a hard decision $\hat{e}_i=1$ if $\Gamma_i<0$ and $\hat{e}_i=0$ otherwise. We stop when the hard decision $\hat{\mathbf{e}}$ is consistent with the measured syndrome, i.e.,
\begin{equation}
\label{eq:minsum_stopping}
\mathbf{H}\hat{\mathbf{e}} \equiv \boldsymbol{\sigma} \pmod 2.
\end{equation}
 \section{Why Does OSD Work Better on Detector-Error Tanner Graphs?}
 \label{sec:motivation_for_proposed_decoder}
The \emph{detector-error} Tanner graph facilitates the convergence of BP+OSD by introducing additional variable nodes that represent faults during syndrome measurement, which can substantially reduce the logical error rate. These extra nodes allow the decoder to capture correlations induced by the measurement circuit, but they also dramatically increase decoding complexity.
This raises a central question: are nodes representing faults essential for achieving a low decoding-failure rate, or can one design decoders that operate on a smaller Tanner graph (without a dedicated variable node for every fault) while retaining comparable performance?
To answer this, we first ask how the additional nodes help BP+OSD: do they improve BP, OSD, or both, and through what mechanism?
Here, we focus on failures in which the decoder converges to a logical error. We separate the discussion by fault type: (i) CNOT faults that introduce \emph{hook errors} (propagating data errors) and (ii) CNOT faults that introduce syndrome-bit errors together with a non-propagating data error.

\subsection{CNOT faults and stabilizer-induced TSs}
\label{sec:hook_error}
In the capacity error model, as discussed in Section~\ref{sec:circuit_level_noise_model}, BP decoders can fail to converge on QLDPC codes due to \emph{stabilizer-induced} trapping sets~\cite{Nithin2021_QTS,pradhan2025_TS_left_right_schedule_decoder}. Here, we connect this phenomenon with \emph{hook errors} during syndrome measurement.
Consider the measurement of an $X$-type stabilizer generator of an HGP code in Fig.~\ref{fig:stab_induced_TS}, where CNOTs are applied in the order indicated on the edges. If the CNOT associated with edge label $3$ fails and introduces an $X$ error on the ancilla, that error propagates to the data qubits $v^1_1v^2, v^1_2v^2, v^1_3v^2,$ and $v^1_4v^2$ (Fig.~\ref{fig:hook_error_propagation}). The resulting data-error pattern is supported on the \emph{stabilizer-induced} TS associated with $c^1v^2$, and BP is unlikely to converge.
This mechanism makes \emph{stabilizer-induced} TSs particularly harmful in the \emph{circuit-level} error model: a single gate fault can activate (see Definition~\ref{def:active_TS}) a \emph{stabilizer-induced} TS.
\begin{figure*}
\begin{subfigure}[t!]{0.4\textwidth}
        \centering
        \tikzstyle{cnode}=[circle,minimum size=0.65 cm,draw,blue,scale=0.58]
\tikzstyle{scnode}=[circle,minimum size=0.4 cm,draw,blue,scale=0.58]
\tikzstyle{gcnode}=[circle,minimum size=0.55 cm,draw,green,scale=0.58]
\tikzstyle{sgcnode}=[circle,minimum size=0.15 cm,draw,green,scale=0.58]
\tikzstyle{secnode}=[circle,minimum size=0.55 cm,draw,Sepia,scale=0.58]
\tikzstyle{ssecnode}=[circle,minimum size=0.15 cm,draw,Sepia,scale=0.58]
\tikzstyle{ocnode}=[circle,minimum size=0.55 cm,draw,orange,scale=0.58]
\tikzstyle{socnode}=[circle,minimum size=0.15 cm,draw,orange,scale=0.58]
\tikzstyle{brcnode}=[circle,minimum size=0.55 cm,draw,brown,scale=0.58]
\tikzstyle{sbrcnode}=[circle,minimum size=0.15 cm,draw,brown,scale=0.58]

\tikzstyle{olrnode}=[rectangle,draw,minimum width=0.4 cm,minimum height=0.4cm,outer sep=0pt,olive,scale=0.58]
\tikzstyle{solrnode}=[rectangle,draw,minimum width=0.1 cm,minimum height=0.1cm,outer sep=0pt,olive,scale=0.58]
\tikzstyle{rnode}=[rectangle,draw,minimum width=0.5 cm,minimum height=0.5cm,outer sep=0pt,scale=0.58]
\tikzstyle{srnode}=[rectangle,draw,minimum width=0.3 cm,minimum height=0.3cm,outer sep=0pt,scale=0.58]
\tikzstyle{crnode}=[rectangle,draw,minimum width=0.4 cm,minimum height=0.4cm,outer sep=0pt,cyan,scale=0.58]
\tikzstyle{scrnode}=[rectangle,draw,minimum width=0.2 cm,minimum height=0.1cm,outer sep=0pt,cyan,scale=0.58]
\tikzstyle{Myrnode}=[rectangle,draw,minimum width=0.4 cm,minimum height=0.4cm,outer sep=0pt,Mulberry,scale=0.58]
\tikzstyle{sMyrnode}=[rectangle,draw,minimum width=0.1 cm,minimum height=0.1cm,outer sep=0pt,Mulberry,scale=0.58]
\tikzstyle{vrnode}=[rectangle,draw,minimum width=0.4 cm,minimum height=0.4cm,outer sep=0pt,violet,scale=0.58]
\tikzstyle{svrnode}=[rectangle,draw,minimum width=0.1 cm,minimum height=0.1cm,outer sep=0pt,violet,scale=0.58]
\tikzstyle{prnode}=[rectangle,draw,minimum width=0.4 cm,minimum height=0.4cm,outer sep=0pt,purple,scale=0.58]
\tikzstyle{sprnode}=[rectangle,draw,minimum width=0.1 cm,minimum height=0.1cm,outer sep=0pt,purple,scale=0.58]
\tikzstyle{yrnode}=[rectangle,draw,minimum width=0.4 cm,minimum height=0.4cm,outer sep=0pt,yellow,scale=0.58]
\tikzstyle{syrnode}=[rectangle,draw,minimum width=0.1 cm,minimum height=0.1cm,outer sep=0pt,yellow,scale=0.58]

\definecolor{color1}{rgb}{1,0.2,0.3}
\definecolor{color2}{rgb}{0.4,0.5,0.7}
\definecolor{color3}{rgb}{0.1,0.8,0.5}
\definecolor{color4}{rgb}{0.5,0.3,1}
\definecolor{color5}{rgb}{0.5,1,1}
\definecolor{color6}{rgb}{0.8,0.3,0.6}
\definecolor{color7}{rgb}{0.6,0.4,0.3}

\begin{tikzpicture}[every node/.style={scale=0.7}]
\begin{scope}[node distance=2cm,semithick]

\node[ocnode] (v1) {};
\node (lv1)[above of =v1,yshift=-1.55cm]{\Large $v_1^1$};
\node[cnode] (v2) [right of =v1,xshift=-0 cm]{};
\node (lv2)[above of =v2,yshift=-1.55cm]{\Large$v_2^1$};
\node[olrnode] (c1) [below of =v1,xshift=1 cm]{};
\node (lc1)[left of =c1,xshift=1.55cm]{\Large$c^1$};
\node[gcnode] (v3) [below of =c1,xshift=-1 cm]{};
\node (lv3)[below of =v3,yshift=1.55cm]{\Large$v_3^1$};
\node[secnode] (v4) [right of =v3,xshift=0 cm]{};
\node (lv4)[below of =v4,yshift=1.55cm]{\Large$v_4^1$};
\draw (v1) -- (c1);
\draw (c1) -- (v2);
\draw (c1) -- (v3);
\draw (c1) -- (v4);

\node (prod_symbol) [below of =v4,xshift=-0.6cm,yshift=-1.25cm]{$\graphprod$};

\node[gcnode] (vc3) [right of=v4,xshift=-0.25cm]{};
\node[sMyrnode] (vc3_in) at (vc3.center) {};
\node[secnode] (vc4) [right of=vc3,xshift=-0 cm]{};
\node[sMyrnode] (vc4_in) at (vc4.center) {};
\node[olrnode] (cc1) [above of=vc3,xshift=1 cm]{};
\node[sMyrnode] (cc1_in) at (cc1.center) {};
\node (lcc1)[left of =cc1,xshift=1.4cm]{\Large$c^1c^2_1$};
\node[ocnode] (vc1) [above of=cc1,xshift=-1 cm]{};
\node[sMyrnode] (vc1_in)at (vc1.center) {};
\node[cnode] (vc2) [above of=cc1,xshift=1 cm]{};
\node[sMyrnode] (vc2_in) at (vc2.center) {};

\draw (cc1) -- (vc1);
\draw (cc1) -- (vc2);
\draw (cc1) -- (vc3);
\draw (cc1) -- (vc4);

\node[secnode] (vv1) [right of=vc4,yshift=-1 cm]{};
\node[ssecnode] (vv1_in) at (vv1.center){};
\node (lvv1)[below of =vv1,xshift=0cm,yshift=1.6cm]{\Large $v^1_{3}v^2$};
\node[gcnode] (vv2) [below of=vv1,yshift=0cm,xshift=1.2cm]{};
\node[ssecnode] (vv2_in) at (vv2.center) {};
\node (lvv2)[below of =vv2,yshift=1.55cm,xshift=0cm]{\Large $v^1_{4}v^2$};

\draw (vv1) -- (vc4);
\draw (vv2) -- (vc3);

\node[cnode] (vv6) [right of=vc2,yshift=0.75cm]{};
\node[ssecnode] (vv6_in) at (vv6.center) {};
\node (lvv6)[left of =vv6,yshift=-0cm,xshift=2.1cm,yshift=0.45cm]{\Large$v^1_{2}v^2$};
\node[ocnode] (vv7) [above of=vv6,yshift=-0.25cm,xshift=1.2cm]{};
\node[ssecnode] (vv7_in) at (vv7.center) {};
\node (lvv7)[left of =vv7,yshift=0.45cm,xshift=1.95cm]{\Large$v^1_{1}v^2$};

\draw (vv6) -- (vc2);
\draw (vv7) -- (vc1);

\node[secnode] (vc9) [right of=vc4,xshift=3.5cm]{};
\node[sMyrnode] (vc9_in) at (vc9.center) {};
\node[gcnode] (vc10) [right of=vc9,xshift=0cm]{};
\node[sMyrnode] (vc10_in) at (vc10.center) {};
\node[olrnode] (cc3) [above of=vc10,xshift=-1cm]{};
\node[sMyrnode] (cc3_in) at (cc3.center) {};
\node (lcc3)[above of =cc3,xshift=0.55cm,yshift=-1.9cm]{\Large$c^1c^2_{2}$};
\node[cnode] (vc11) [above of=cc3,xshift=-1cm]{};
\node[sMyrnode] (vc11_in) at (vc11.center) {};
\node[ocnode] (vc12) [above of=cc3,xshift=1cm]{};
\node[sMyrnode] (vc12_in) at (vc12.center) {};

\draw (cc3) -- (vc9);
\draw (cc3) -- (vc10);
\draw (cc3) -- (vc11);
\draw (cc3) -- (vc12);

\node[olrnode] (cc5) [right of=cc3,xshift=1 cm,yshift=-0.5cm]{};
\node[sMyrnode] (cc5_in) at (cc5.center){};
\node (lcc5)[right of =cc5,xshift=-1.4cm,yshift=-0cm]{\Large$c^1c_{3}^2$};
\node[ocnode] (vc17) [above of=cc5,xshift=1cm,yshift=0.75cm]{};
\node[sMyrnode] (vc17_in) at (vc17.center) {};
\node[cnode] (vc18) [above of=cc5,xshift=-1cm,yshift=0.75cm]{};
\node[sMyrnode] (vc18_in) at (vc18.center){};
\node[gcnode] (vc19) [below of=cc5,xshift=1cm,yshift=0 cm]{};
\node[sMyrnode] (vc19_in) at (vc19.center) {};
\node[secnode] (vc20) [below of=cc5,xshift=-1cm,yshift=0 cm]{};
\node[sMyrnode] (vc20_in) at (vc20.center) {};

\draw (cc5) -- (vc17);
\draw (cc5) -- (vc18);
\draw (cc5) -- (vc19);
\draw (cc5) -- (vc20);

\draw (vv1) --(vc9);
\draw (vv1.0) --(vc20.270);

\draw (vv2) --(vc10);
\draw (vv2.0) --(vc19);

\draw (vv6) --(vc11);
\draw (vv7) --(vc17.90);
\draw (vv6) --(vc18.120);
\draw (vv7) --(vc12);

\node[secnode] (v9) [below of=vv2,xshift=-0.25cm,yshift=-1 cm]{};
\node (lv9) [below of=v9,yshift=1.55cm]{\Large$v^2$};
\node[vrnode] (c3) [left of=v9,xshift=-0.5cm]{};
\node (lc3) [left of=c3,xshift=1.55cm]{\Large$c^2_1$};
\node[prnode] (c4) [right of=v9,xshift=0.5cm,yshift=0.5cm]{};
\node (lc4) [right of=c4,xshift=-1.55cm]{\Large$c^2_2$};
\node[Myrnode] (c5) [right of=v9,xshift=2.1cm,yshift=-0.5cm]{};
\node (lc5) [right of=c5,xshift=-1.55cm]{\Large$c^2_3$};
\draw (v9) -- (c3);
\draw (v9) -- (c4);
\draw (v9) -- (c5);

\node[olrnode] (cv1) [right of =cc1,xshift=1.5cm]{};
\node[ssecnode] (cv1_in) at (cv1.center) {};
\draw[dashed] (cv1) -- node[above,midway,yshift=-0 cm]{\large$1$} (cc1);
\draw[dashed] (cv1) -- node[above,midway,yshift=-0 cm]{\large$2$} (cc3);
\draw[dashed] (cv1) -- node[below,near start,yshift=0 cm]{\large$3$} (cc5);
\draw[dashed] (cv1) -- node[left,midway,xshift=0 cm]{\large$4$} (vv1);
\draw[dashed] (cv1) -- node[right,near start,xshift=-0 cm]{\Large$5$} (vv2);
\draw[dashed] (cv1) -- node[left,midway,xshift=0 cm]{\large$6$} (vv6);
\draw[dashed] (cv1) -- node[right,near start,xshift=-0 cm]{\large$7$} (vv7);

\end{scope}    
\end{tikzpicture}
        \subcaption{TS induced by the $X$-type stabilizer generator $c^1v^2$, i.e., the graph product of check node $c^1$ and variable node $v^2$. Dotted lines indicate the edges connecting $c^1v^2$ to its neighboring variable nodes, and the annotations specify the order in which the corresponding CNOT gates are applied during syndrome measurement. Solid lines show the induced subgraph within $\mathscr{G}_{\mathrm{z}}$.}
        \label{fig:stab_induced_TS}
    \end{subfigure}
    \hfill
\begin{subfigure}[t!]{0.28\textwidth}
\centering
         \input{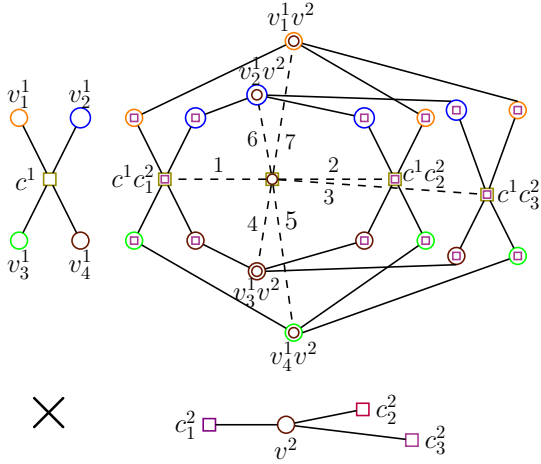}
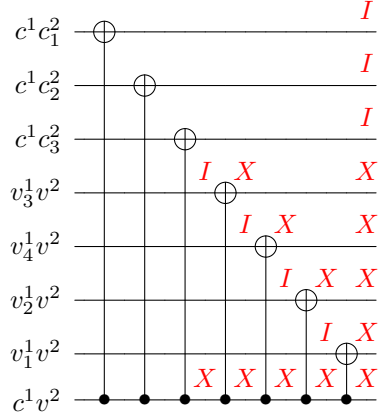
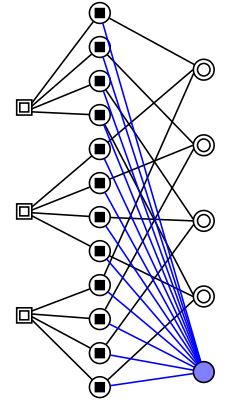
          \subcaption{A hook error activating the \emph{stabilizer-induced} TS in Figure~\ref{fig:stab_induced_TS}. The CNOT gate corresponding to the edge with label $3$ fails and introduces an $X$ error in the ancilla. This $X$ error subsequently spreads to the data qubits $v^1_1v^2, v^1_2v^2, v^1_3v^2,$ and $v^1_4v^2$, which is enough to induce the TS~\cite{pradhan2025_TS_left_right_schedule_decoder}.}
        \label{fig:hook_error_propagation}
    \end{subfigure}
   \hfill
   \begin{subfigure}[t!]{0.25\textwidth}
\centering
        \tikzstyle{cnode}=[circle,minimum size=0.55 cm,draw]
\tikzstyle{scnode}=[circle,minimum size=0.15 cm,draw]
\tikzstyle{bscnode}=[circle,minimum size=0.15 cm,draw,fill=black]
\tikzstyle{bsrnode}=[rectangle,draw,minimum width=0.1 cm,minimum height=0.1cm,outer sep=0pt,fill=black]
\tikzstyle{cgnode}=[circle,draw]
\tikzstyle{cblnode}=[circle,draw,minimum size=0.55 cm,fill=blue!50]
\tikzstyle{conode}=[rectangle,draw]
\tikzstyle{cpnode}=[circle,draw]
\tikzstyle{rnode}=[rectangle,draw,minimum width=0.4 cm,minimum height=0.4cm,outer sep=0pt]
\tikzstyle{srnode}=[rectangle,draw,minimum width=0.1 cm,minimum height=0.1cm,outer sep=0pt]
\tikzstyle{prnode}=[rectangle,rounded corners,fill=blue!50,text width=4.5em,text centered,outer sep=0pt]
\tikzstyle{prnodebig}=[rectangle,rounded corners,fill=blue!50,text width=7em,text centered,outer sep=0pt]
\tikzstyle{prnodesimple}=[rectangle,draw,text width=4.5em,text centered,outer sep=0pt]
\tikzstyle{bigsnake}=[fill=green!50,snake=snake,segment amplitude=4mm, segment length=4mm, line after snake=1mm]
\tikzstyle{smallsnake}=[snake=snake,segment amplitude=0.7mm, segment length=4mm, line after snake=1mm]
\definecolor{color1}{rgb}{1,0.2,0.3}
\definecolor{color2}{rgb}{0.4,0.5,0.7}
\definecolor{color3}{rgb}{0.1,0.8,0.5}
\definecolor{color4}{rgb}{0.5,0.3,1}
\definecolor{color5}{rgb}{0.5,1,1}
\definecolor{color6}{rgb}{0.8,0.3,0.6}
\definecolor{color7}{rgb}{0.6,0.4,0.3}

\begin{tikzpicture}[every node/.style={scale=0.5}]
\begin{scope}[node distance=2cm,semithick]

\node[rnode] (cc1) {};
\node[srnode] (cc1_in){};
\node [rnode] (cc2) [below of=cc1,yshift=-0.75cm] {};
\node[srnode] (cc2_in)[below of=cc1,yshift=-0.75 cm]{};

\node [rnode] (cc3) [below of=cc2,yshift=-0.75 cm] {};
\node[srnode] (cc3_in)[below of=cc2,yshift=-0.75 cm]{};

\node [cnode] (vc1) [right of=cc1,yshift=2.5 cm] {};
\node [bsrnode] (vc1_in) [right of=cc1,yshift=2.5cm] {};

\node [cnode] (vc2) [below of=vc1,yshift=1.1 cm] {};
\node [bsrnode] (vc2_in) [below of=vc1,yshift=1.1 cm] {};

\node [cnode] (vc3) [below of=vc2,yshift=1.1 cm] {};
\node [bsrnode] (vc3_in) [below of=vc2,yshift=1.1 cm] {};

\node [cnode] (vc4) [below of=vc3,yshift=1.1 cm] {};
\node [bsrnode] (vc4_in) [below of=vc3,yshift=1.1 cm] {};

\node [cnode] (vc5) [below of=vc4,yshift=1.1cm] {};
\node [bsrnode] (vc5_in) [below of=vc4,yshift=1.1cm] {};

\node [cnode] (vc6) [below of=vc5,yshift=1.1 cm] {};
\node [bsrnode] (vc6_in) [below of=vc5,yshift=1.1 cm] {};

\node [cnode] (vc7) [below of=vc6,yshift=1.1 cm] {};
\node [bsrnode] (vc7_in) [below of=vc6,yshift=1.1 cm] {};

\node [cnode] (vc8) [below of=vc7,yshift=1.1 cm] {};
\node [bsrnode] (vc8_in) [below of=vc7,yshift=1.1 cm] {};

\node [cnode] (vc9) [below of=vc8,yshift=1.1cm] {};
\node [bsrnode] (vc9_in) [below of=vc8,yshift=1.1 cm] {};

\node [cnode] (vc10) [below of=vc9,yshift=1.1cm] {};
\node [bsrnode] (vc10_in) [below of=vc9,yshift=1.1cm] {};

\node [cnode] (vc11) [below of=vc10,yshift=1.1cm] {};
\node [bsrnode] (vc11_in) [below of=vc10,yshift=1.1cm] {};

\node [cnode] (vc12) [below of=vc11,yshift=1.1 cm] {};
\node [bsrnode] (vc12_in) [below of=vc11,yshift=1.1 cm] {};
\draw (cc1.30) -- (vc1);
\draw (cc1.10) -- (vc2);
\draw (cc1.-10) -- (vc3);
\draw (cc1.-30) -- (vc4);

\draw (cc2.30) -- (vc5);
\draw (cc2.10) -- (vc6);
\draw (cc2.-10) -- (vc7);
\draw (cc2.-30) -- (vc8);

\draw (cc3.30) -- (vc9);
\draw (cc3.10) -- (vc10);
\draw (cc3.-10) -- (vc11);
\draw (cc3.-30) -- (vc12);

\node [cnode] (vv1) [right of=vc1,yshift=-1.5cm,xshift=0.75cm] {};
\node [scnode] (vv1_in) [right of=vc1,yshift=-1.5cm,xshift=0.75cm] {};

\node [cnode] (vv2) [below of=vv1,yshift=0 cm] {};
\node[scnode] (vv2_in)[below of=vv1,yshift=0 cm]{};

\node [cnode] (vv3) [below of=vv2,yshift=0 cm] {};
\node[scnode] (vv3_in)[below of=vv2,yshift=0 cm]{};

\node [cnode] (vv4) [below of=vv3,yshift=0 cm] {};
\node[scnode] (vv4_in)[below of=vv3,yshift=0 cm]{};

\draw(vc1) -- (vv1);
\draw(vc5) -- (vv1.180);
\draw(vc9) -- (vv1.200);
\draw(vc2) -- (vv2.170);
\draw(vc6) -- (vv2.180);
\draw(vc10) -- (vv2.200);
\draw(vc3) -- (vv3.170);
\draw(vc7) -- (vv3.180);
\draw(vc11) -- (vv3.200);
\draw(vc4) -- (vv4.170);
\draw(vc8) -- (vv4.180);
\draw(vc12) -- (vv4.200);
\node [cblnode] (f1)[below of=vv4]{};
\draw[blue] (vc1)--(f1);
\draw[blue] (vc2)--(f1);
\draw[blue] (vc3)--(f1);
\draw[blue] (vc4)--(f1);
\draw[blue] (vc5)--(f1);
\draw[blue] (vc6)--(f1);
\draw[blue] (vc7)--(f1);
\draw[blue] (vc8)--(f1);
\draw[blue] (vc9)--(f1);
\draw[blue] (vc10)--(f1);
\draw[blue] (vc11)--(f1);
\draw[blue] (vc12)--(f1);
\end{scope}    
\end{tikzpicture}
        \caption{The TS induced by $c^1v^2$ incorporates an extra variable node (colored blue) symbolizing an $X$-error in the ancilla following the execution of the CNOT associated with the edge labeled $3$. The inner shapes of the checks that do not satisfy the measured syndrome are depicted in black.  OSD-$1$ is capable of addressing this TS caused by a hook error. }
        \label{fig:circuit_level_Tanner_graph_hook_error}
    \end{subfigure}
    \caption{Example illustrating of how a hook error activates \emph{stabilizer-induced} TSs.}
    \label{fig:coneection_of_hook_error_to_TS}
\end{figure*}

Next, we examine the factors that contribute to the improved decoding failure rate of the BP+OSD decoder when it is run on the \emph{detector-error} Tanner graph. To this end, we first delineate the working principles of the BP+OSD decoder.
Let the number of columns in the detector-error parity-check matrix be denoted by $N$ and its rank by $r$. 
In the BP+OSD decoding framework, the BP decoder is first run for a specified number of iterations on the \emph{detector-error} Tanner graph to produce a soft value (e.g., an LLR) for each of the $N$ variable nodes. The variable nodes are then ranked by reliability, e.g., by the magnitude of their soft values. The order-$t$ OSD step next permutes the columns of the detector-error parity-check matrix according to this reliability ordering and performs Gaussian elimination to obtain a full-rank systematic form. Using this systematic form, OSD enumerates all test patterns of weight at most $t$ on a subset of $N-r$ less-reliable positions; for each test pattern, it computes the corresponding values on the remaining $r$ positions so that the resulting length-$N$ error estimate satisfies the measured syndrome. This generates $\sum_{i=0}^{t}\binom{N-r}{i}$ candidate solutions. Finally, OSD selects the candidate with the largest a priori probability (equivalently, the best metric with respect to the BP soft values).

With the above explanation of the BP+OSD decoding approach, let us re-examine the trapping-set example depicted in Figure~\ref{fig:coneection_of_hook_error_to_TS}. This scenario helps to clarify the advantage of adding additional variable nodes to handle hook errors during the decoding process. To simplify the discussion, assume that the variable nodes that form the TS are not part of the $r$ variable nodes with the highest soft values. This is reasonable because the error estimate tends to oscillate between iterations when the support of the error pattern lies within a \emph{stabilizer-induced} TS.
In cases where the error pattern contains only one hook error, as illustrated in Figure~\ref{fig:circuit_level_Tanner_graph_hook_error}, it is likely to become a candidate solution for OSD-1. However, if no additional node representing the hook error is present, the true error pattern is considered as a candidate solution only when the OSD order is three or greater. An OSD of order three is sufficient, as there exists a corresponding weight-$3$ error pattern that matches the syndrome, although the hook error in Figure~\ref{fig:hook_error_propagation} spreads to four qubits.

Note that stabilizer generators alone are not the only source of TSs; their linear combinations can also induce them~\cite{pradhan2025_TS_left_right_schedule_decoder}. As the size of such linear combinations grows, the required OSD order to resolve the corresponding TSs increases in both cases—whether or not additional nodes are included to account for hook errors. In particular, this increase is much faster when additional variable nodes representing hook errors are absent.
The resulting increase in decoder complexity is illustrated in Fig.~\ref{fig:linear_comb_stab_inducing_TS} through an example in which a pair of hook errors creates a TS induced by two stabilizer generators. 

Consider the TS formed by two $X$-type stabilizer generators, as shown in Fig.~\ref{fig:linear_comb_stab_inducing_TS}. Assume that the CNOT operations associated with the edges labeled $3$ fail during the syndrome measurement and introduce $X$ errors in the corresponding ancilla qubits. These $X$ errors subsequently propagate to data qubits $d_5,d_6,d_7,d_{11},d_{12},$ and $d_{13}$, collectively forming an error pattern that activates the TS depicted in Fig.~\ref{fig:linear_comb_stab_inducing_TS}.
This shows that, in the \emph{circuit-level} error model, the TS in Fig.~\ref{fig:linear_comb_stab_inducing_TS} can be activated by just two hook errors, instead of the six data errors required in the capacity error model. In Fig.~\ref{fig:circuit_level_Tanner_graph_for_two_hook_errors}, the two blue variable nodes represent these $X$-type errors, which arise after the CNOT operations. A variable node associated with a hook error is adjacent to a check node if that check node can detect the corresponding hook error. Thus, the adjacent check nodes of these variable nodes are determined by whether the check node detects the associated hook error.
Running the OSD-$2$ decoder on the Tanner graph, which includes the two blue variable nodes for hook errors, is likely to produce candidate solutions that include the true error pattern. However, omitting these nodes from the Tanner graph would require an OSD decoder of order six or higher to  resolve the TS. Given the exponential increase in OSD complexity with order, it is advantageous to incorporate variable nodes that capture correlations from hook errors.
   
 
 \begin{figure*}
        \begin{subfigure}[t!]{0.4\linewidth}
        \centering
        \input{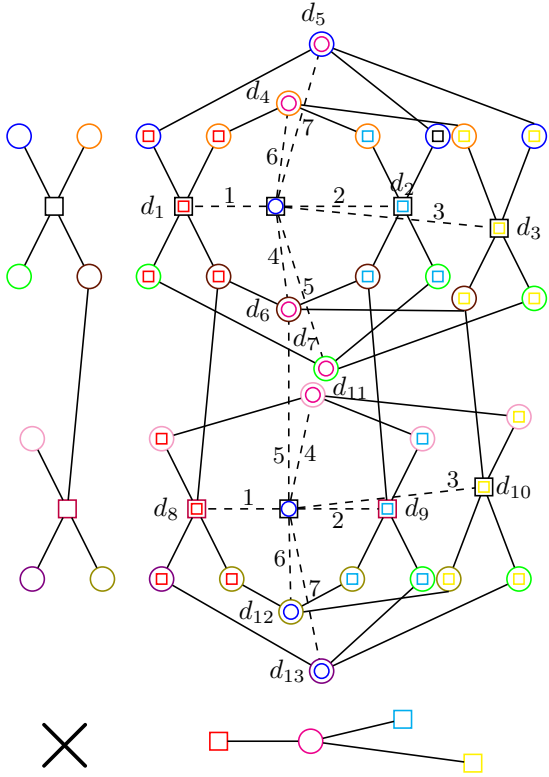}
        \subcaption{ An example of a TS created by two $X$-type stabilizers, formed through the graph product of the subgraph induced by two check nodes of one parent classical code and the subgraph induced by a variable node from the other parent classical code.  The dotted edges showed the edges connected to the $X$-type stabilizers. Solid lines denote the edges of the subgraph induced by these $X$-type stabilizer generators. The labels on these dotted lines indicate the order in which the CNOT operations associated with them are carried out during the stabilizer measurement.}
        \label{fig:linear_comb_stab_inducing_TS}
    \end{subfigure}
    \hfill
\begin{subfigure}[t!]{0.3\linewidth}
        \centering
        \input{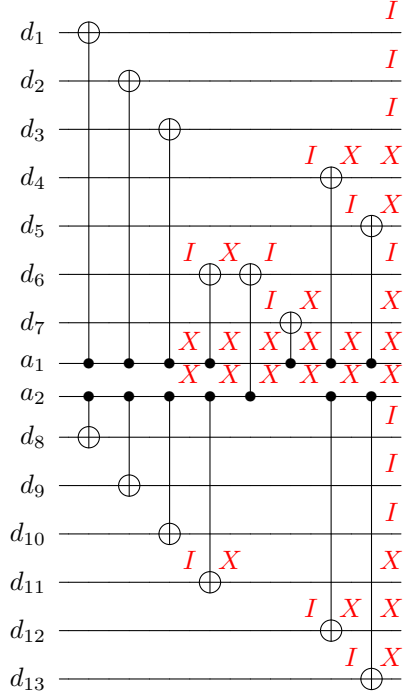}
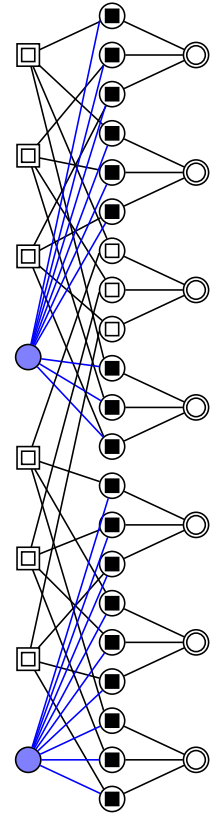
        \subcaption{Hook errors activating the TS shown in Fig.~\ref{fig:linear_comb_stab_inducing_TS}. Both CNOT gates corresponding to the edges labeled $3$ in Fig.~\ref{fig:linear_comb_stab_inducing_TS} fail, causing $X$ errors in the respective ancillas. These $X$ errors then propagate to the data qubits that participate in later CNOT operations with those ancillas. Here, the $X$-error that propagated from ancilla $a_1$ to data $d_6$ is canceled by the $X$ error that propagates from ancilla $a_2$ to data $d_6$.}
        \label{fig:two_hook_error_propagating_to_mul_data_error}  
    \end{subfigure}
    \hfill
    \begin{subfigure}[t!]{0.23\linewidth}
        \centering
        \tikzstyle{cnode}=[circle,minimum size=0.45 cm,draw]
\tikzstyle{blcnode}=[circle,minimum size=0.45 cm,draw,fill=blue!50]
\tikzstyle{scnode}=[circle,minimum size=0.2 cm,draw]
\tikzstyle{zscnode}=[circle,minimum size=0.08 cm,draw,fill=red]
\tikzstyle{zrnode}=[rectangle,draw,minimum width=0.2 cm,minimum height=0.2cm,outer sep=0pt]
\tikzstyle{cgnode}=[circle,draw]
\tikzstyle{crnode}=[circle,draw]
\tikzstyle{conode}=[rectangle,draw]
\tikzstyle{cpnode}=[circle,draw]
\tikzstyle{rnode}=[rectangle,draw,minimum width=0.4 cm,minimum height=0.4cm,outer sep=0pt]
\tikzstyle{srnode}=[rectangle,draw,minimum width=0.1 cm,minimum height=0.1cm,outer sep=0pt]
\tikzstyle{bsrnode}=[rectangle,draw,minimum width=0.1 cm,minimum height=0.1cm,outer sep=0pt,fill=black]
\begin{tikzpicture}[every node/.style={scale=0.74}]
\begin{scope}[node distance=1 cm,semithick]
\node[cnode] (vc1) {}; 
\node[bsrnode] (vc1_in) {}; 
\node[cnode] (vc2) [below of=vc1,yshift=0.3cm]{}; 
\node[bsrnode] (vc2_in) [below of=vc1,yshift=0.3cm]{}; 
\node[cnode] (vc3) [below of=vc2,yshift=0.3cm]{}; 
\node[bsrnode] (vc3_in) [below of=vc2,yshift=0.3cm]{}; 
\node[cnode] (vc4) [below of=vc3,yshift=0.3cm]{}; 
\node[bsrnode] (vc4_in) [below of=vc3,yshift=0.3cm]{}; 
\node[cnode] (vc5) [below of=vc4,yshift=0.3cm]{}; 
\node[bsrnode] (vc5_in) [below of=vc4,yshift=0.3cm]{}; 
\node[cnode] (vc6) [below of=vc5,yshift=0.3cm]{}; 
\node[bsrnode] (vc6_in) [below of=vc5,yshift=0.3cm]{}; 
\node[cnode] (vc7) [below of=vc6,yshift=0.3cm]{}; 
\node[srnode] (vc7_in) [below of=vc6,yshift=0.3cm]{}; 
\node[cnode] (vc8) [below of=vc7,yshift=0.3cm]{}; 
\node[srnode] (vc8_in) [below of=vc7,yshift=0.3cm]{}; 
\node[cnode] (vc9) [below of=vc8,yshift=0.3cm]{}; 
\node[srnode] (vc9_in) [below of=vc8,yshift=0.3cm]{}; 
\node[cnode] (vc10) [below of=vc9,yshift=0.3cm]{}; 
\node[bsrnode] (vc10_in) [below of=vc9,yshift=0.3cm]{}; 
\node[cnode] (vc11) [below of=vc10,yshift=0.3cm]{}; 
\node[bsrnode] (vc11_in) [below of=vc10,yshift=0.3cm]{}; 
\node[cnode] (vc12) [below of=vc11,yshift=0.3cm]{}; 
\node[bsrnode] (vc12_in) [below of=vc11,yshift=0.3cm]{}; 
\node[cnode] (vc13) [below of=vc12,yshift=0.3cm]{}; 
\node[bsrnode] (vc13_in) [below of=vc12,yshift=0.3cm]{}; 
\node[cnode] (vc14) [below of=vc13,yshift=0.3cm]{}; 
\node[bsrnode] (vc14_in) [below of=vc13,yshift=0.3cm]{}; 
\node[cnode] (vc15) [below of=vc14,yshift=0.3cm]{}; 
\node[bsrnode] (vc15_in) [below of=vc14,yshift=0.3cm]{}; 
\node[cnode] (vc16) [below of=vc15,yshift=0.3cm]{}; 
\node[bsrnode] (vc16_in) [below of=vc15,yshift=0.3cm]{}; 
\node[cnode] (vc17) [below of=vc16,yshift=0.3cm]{}; 
\node[bsrnode] (vc17_in) [below of=vc16,yshift=0.3cm]{}; 
\node[cnode] (vc18) [below of=vc17,yshift=0.3cm]{}; 
\node[bsrnode] (vc18_in) [below of=vc17,yshift=0.3cm]{}; 
\node[cnode] (vc19) [below of=vc18,yshift=0.3cm]{}; 
\node[bsrnode] (vc19_in) [below of=vc18,yshift=0.3cm]{}; 
\node[cnode] (vc20) [below of=vc19,yshift=0.3cm]{}; 
\node[bsrnode] (vc20_in) [below of=vc19,yshift=0.3cm]{}; 
\node[cnode] (vc21) [below of=vc20,yshift=0.3cm]{}; 
\node[bsrnode] (vc21_in) [below of=vc20,yshift=0.3cm]{}; 
\node[cnode] (vv1) [right of=vc2,xshift=0.5cm]{}; 
\node[scnode] (vv1_in) [right of=vc2,xshift=0.5cm]{}; 
\draw (vv1)--(vc1);
\draw (vv1)--(vc2);
\draw (vv1)--(vc3);
\node[cnode] (vv2) [right of=vc5,xshift=0.5cm]{}; 
\node[scnode] (vv2_in) [right of=vc5,xshift=0.5cm]{}; 
\draw (vv2)--(vc4);
\draw (vv2)--(vc5);
\draw (vv2)--(vc6);
\node[cnode] (vv3) [right of=vc8,xshift=0.5cm]{}; 
\node[scnode] (vv3_in) [right of=vc8,xshift=0.5cm]{}; 
\draw (vv3)--(vc7);
\draw (vv3)--(vc8);
\draw (vv3)--(vc9);
\node[cnode] (vv4) [right of=vc11,xshift=0.5cm]{}; 
\node[scnode] (vv4_in) [right of=vc11,xshift=0.5cm]{}; 
\draw (vv4)--(vc10);
\draw (vv4)--(vc11);
\draw (vv4)--(vc12);
\node[cnode] (vv5) [right of=vc14,xshift=0.5cm]{}; 
\node[scnode] (vv5_in) [right of=vc14,xshift=0.5cm]{}; 
\draw (vv5)--(vc13);
\draw (vv5)--(vc14);
\draw (vv5)--(vc15);
\node[cnode] (vv6) [right of=vc17,xshift=0.5cm]{}; 
\node[scnode] (vv6_in) [right of=vc17,xshift=0.5cm]{}; 
\draw (vv6)--(vc16);
\draw (vv6)--(vc17);
\draw (vv6)--(vc18);

\node[cnode] (vv7) [right of=vc20,xshift=0.5cm]{}; 
\node[scnode] (vv7_in) [right of=vc20,xshift=0.5cm]{}; 
\draw (vv7)--(vc19);
\draw (vv7)--(vc20);
\draw (vv7)--(vc21);
\node[rnode] (cc1)[left of =vc2,xshift=-0.5cm]{};
\node[srnode] (cc1_in)[left of =vc2,xshift=-0.5cm]{};
\draw (cc1)--(vc1);
\draw (cc1)--(vc4);
\draw (cc1)--(vc7);
\draw (cc1)--(vc10.130);
\node[rnode] (cc2)[below of =cc1,yshift=-0.8cm]{};
\node[srnode] (cc2_in)[below of =cc1,yshift=-0.8cm]{};
\draw (cc2)--(vc2);
\draw (cc2)--(vc5);
\draw (cc2)--(vc8);
\draw (cc2)--(vc11.130);
\node[rnode] (cc3)[below of =cc2,yshift=-0.8cm]{};
\node[srnode] (cc3_in)[below of =cc2,yshift=-0.8cm]{};
\draw (cc3)--(vc3);
\draw (cc3)--(vc6);
\draw (cc3)--(vc9);
\draw (cc3)--(vc12.130);

\node[blcnode] (f1)[below of =cc3,yshift=-0.8cm]{};
\draw[blue] (f1)--(vc1.200);
\draw [blue](f1)--(vc2.230);
\draw [blue](f1)--(vc3.240);
\draw [blue](f1)--(vc4);
\draw [blue](f1)--(vc5);
\draw [blue](f1)--(vc6);
\draw [blue](f1)--(vc10);
\draw [blue](f1)--(vc11);
\draw [blue](f1)--(vc12);
\node[rnode] (cc4)[below of =f1,yshift=-0.8cm]{};
\node[srnode] (cc4_in)[below of =f1,yshift=-0.8cm]{};
\draw (cc4)--(vc7.190);
\draw (cc4)--(vc13);
\draw (cc4)--(vc16);
\draw (cc4)--(vc19.130);

\node[rnode] (cc5)[below of =cc4,yshift=-0.8cm]{};
\node[srnode] (cc5_in)[below of =cc4,yshift=-0.8cm]{};
\draw (cc5)--(vc8.190);
\draw (cc5)--(vc14);
\draw (cc5)--(vc17);
\draw (cc5)--(vc20.130);
\node[rnode] (cc6)[below of =cc5,yshift=-0.8cm]{};
\node[srnode] (cc6_in)[below of =cc5,yshift=-0.8cm]{};
\draw (cc6)--(vc9.190);
\draw (cc6)--(vc15);
\draw (cc6)--(vc18);
\draw (cc6)--(vc21.130);

\node[blcnode] (f2)[below of =cc6,yshift=-0.8cm]{};
\draw[blue] (f2)--(vc13);
\draw [blue](f2)--(vc14);
\draw [blue](f2)--(vc15);
\draw [blue](f2)--(vc16);
\draw [blue](f2)--(vc17);
\draw [blue](f2)--(vc18);
\draw [blue](f2)--(vc19);
\draw [blue](f2)--(vc20);
\draw [blue](f2)--(vc21);
\end{scope}
\end{tikzpicture}
        \subcaption{\emph{detector-error} Tanner graph featuring two additional variable nodes (colored blue) to represent the $X$ errors introduced after the faulty execution of CNOTs corresponding to the edges labeled $3$ in Figure~\ref{fig:linear_comb_stab_inducing_TS}. The inner shapes of the checks that do not satisfy the measured syndrome are depicted in black.}
        \label{fig:circuit_level_Tanner_graph_for_two_hook_errors}
    \end{subfigure}
    \caption{Example illustrating how a combination of hook errors induce TSs.}
    \label{fig:connection_between_hook_errors_and_TS}
    \end{figure*}
Figure~\ref{fig:connection_between_hook_errors_and_TS} illustrates how CNOT faults that create hook errors propagate to correlated data errors and thus activate a \emph{stabilizer-induced} TS in the \emph{detector-error} Tanner graph.
The above discussion explains why fault-representing nodes help the OSD stage: by explicitly modeling the hook-error correlations, the true error pattern can appear at much lower OSD order. However, adding these nodes can also worsen the convergence of the BP because it introduces additional short cycles\cite{pacenti2026trapping,bhatnagar2026impulse}.
These nodes representing faults become unnecessary if BP itself is modified to resolve the corresponding \emph{stabilizer-induced} TSs. An approach in this direction is~\cite{pradhan2025_TS_left_right_schedule_decoder}, which characterizes the TSs of HGP/LP codes as graph products of parent-code subgraphs and shows that alternating VV/CC message updates can ensure convergence when the relevant subgraphs are cycle-free. Since good classical codes typically contain many cycles, this guarantees convergence only for a restricted class of TSs~\cite{mct}. Here, we propose decoders that resolve a broader class of \emph{stabilizer-induced} TSs by using decoders for the underlying parent code, thereby avoiding the need to add fault nodes for correlation-inducing events.

\subsection{CNOT Faults and Syndrome-Measurement Errors}
\label{sec:cnot_faults_and_syndrome_error}
 \begin{figure*}
        \input{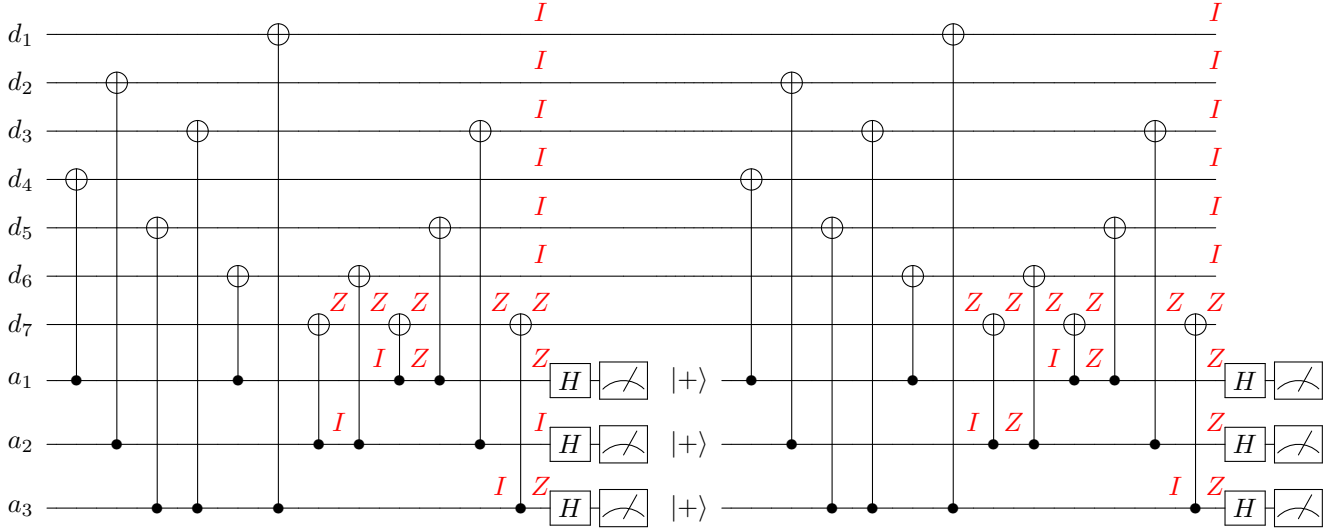}
        \caption{Circuit errors leading to an erroneous syndrome (and why repeating the measurement helps): A faulty CNOT between $d_7$ and $a_2$ introduces a $Z$ error on $d_7$ \,during syndrome extraction, so in the first round it is detected only by the first and third stabilizers; in contrast, a pre-existing $Z$ error on $d_7$ (present before measurement) would be detected by all three generators. Repeating the syndrome measurement in a second round makes the syndrome history distinguish this CNOT-induced fault from data errors that were present before measurement.}
        \label{fig:discrepancy_between_error_and_syndrome}
    \end{figure*}
In this section, we discuss CNOT faults that occur during stabilizer measurement and do not propagate to other data qubits, but nevertheless corrupt the measured syndrome. In \emph{detector-error} Tanner graphs, such faults are often modeled by introducing additional variable nodes.

Whenever a stabilizer generator acts nontrivially on a data qubit, the measurement circuit applies a two-qubit gate (e.g., a CNOT) between that qubit and an ancilla. If a data error is introduced midway through the measurement circuit, it is detected only by those stabilizer checks whose relevant two-qubit gates occur after the error. As a result, the measured syndrome can differ from the syndrome of the final data-error pattern.
We illustrate this effect using the Steane code in Fig.~\ref{fig:discrepancy_between_error_and_syndrome}. Assume all CNOT gates are perfect except the gate between data qubit $d_7$ and ancilla $a_2$, and that this faulty CNOT induces a $Z$ error on $d_7$. In the first round, this error is detected by only the first and third stabilizers, even though a pre-existing $Z$ error on $d_7$ would be detected by all three. Feeding such an inconsistent syndrome to a decoder can lead to miscorrection; in this example, a maximum-likelihood decoder would apply a correction on $d_5$, turning a single-qubit error into a two-qubit error.

The corresponding \emph{detector-error} matrix is
\begin{equation}
    \label{eq:det_error_matrix}
        \begin{pNiceArray}{ccccccccccccccc}
   \Block[fill=gray!75]{3-15}{}0 & 0 & 0 & 1 & 1 & 1 & 1 & 1 & 0 & 0 & 0 & 0 & 0 & 0 & 0 \\
        0 & 1 & 1 & 0 & 0 & 1 & 1 & 0 & 0 & 0 & 0 & 0 & 0 & 0 & 0\\
        1 & 0 & 1 & 0 & 1 & 0 & 1  & 1 & 0 & 0 & 0 & 0 & 0   & 0 &  {0}\\
     \Block[fill=cyan!75]{3-15}{} 0 & 0 & 0 & 0 & 0 & 0 & 0 & 0 & 0 & 0 & 0 & 1 & 1 & 1 & 1\\
        0 & 0 & 0 & 0 & 0 & 0 & 0  & 1 & 0 & 1 & 1 & 0 & 0 & 1 & 1 \\
        0 & 0 & 0 & 0 & 0 & 0 & 0 & 0 & 1 & 0 & 1 & 0 & 0 & 0  & {1}
\end{pNiceArray},
\end{equation}
where the cyan block corresponds to the first round of stabilizer measurement. In that round, the eighth column (the $Z$ error introduced after the faulty CNOT) coincides with the fifth column, so these two faults are indistinguishable from a single-round syndrome; in particular, adding explicit fault nodes alone does not resolve the ambiguity when the code is not single-shot. Repeating the measurement mitigates this issue. In~\eqref{eq:det_error_matrix}, the two highlighted blocks together correspond to two rounds, and the eighth column becomes unique, allowing an optimal decoder to identify the correct fault.
This motivates the use of  nodes to represent such faults together with repeated measurements. 
The question now is whether a decoder can estimate errors caused by CNOT failures that lead to erroneous syndromes without the need for these additional variable nodes. 
We return to this question in Section~\ref{sec:syndrome_error_inducing_CNOT_faults}.

\section{Stabilizer-Induced TSs with Multiple Copies of TSs from a Parent Code }
\label{sec:proposed_decoders}
Iterative message-passing decoders fail in the error-floor region of classical LDPC codes due to short cycles and their combinations in the Tanner graph, known as classical TSs. An HGP/LP code inherits these structures from its parent classical codes, resulting in multiple isomorphic copies of each classical TS in the derived HGP/LP code.
Certain collections of these isomorphic copies form large \emph{stabilizer-induced} TSs that arise from linear combinations of stabilizer generators. Hook errors during the measurement of these stabilizer generators can introduce errors in multiple copies simultaneously, increasing the likelihood that enough copies become active (in the sense defined below) to cause a decoding failure.
It is therefore crucial to analyze the decoding dynamics of these large \emph{stabilizer-induced} TSs and to design decoders that can avoid decoding failures caused by them. In this section, we focus on \emph{stabilizer-induced} TSs that contain multiple copies of a TS from one of the parent classical LDPC codes.
Due to the similarity between the TSs induced by $Z$-type and $X$-type stabilizer generators, we restrict our attention to those induced by the $X$-type stabilizer generators. This extends the study of \emph{stabilizer-induced} TSs considered in~\cite{pradhan2025_TS_left_right_schedule_decoder}, which focused only on TSs that are the graph product of tree-like subgraphs from the Tanner graphs of the parent classical LDPC codes, by allowing one of the subgraphs involved in the graph product to have cycles. Throughout this section, we view TSs as candidate failure mechanisms for a specific iterative decoder. Next, we formally define these TSs.

\begin{definition}
An $(a,b)$ TS of an LDPC code is a subgraph of the Tanner graph induced by a set of $a$ variable nodes that contains $b$ odd-degree check nodes. The critical number of a TS (with respect to a given iterative decoder) is the smallest number of variable nodes that must initially be in error to cause a decoding failure.
\end{definition}

\begin{definition}(\textbf{Active TSs})
\label{def:active_TS}
     A classical TS is said to be \emph{active} if the number of erroneous variable nodes it contains is at least its critical number.
\end{definition}

\subsection{Stabilizer-induced TSs as graph products of a parent-code node and a TS}
\label{sec:node_TS_graph_product}
We begin with the simplest \emph{stabilizer-induced} product TS family, in which one factor of the graph product is a single node (a check node or a variable node) from one parent code and the other factor is a TS from the other parent code. Example~\ref{ex:classical_and_stab_induced_TS} illustrates how such structures arise and motivates the decoder in Algorithm~\ref{alg:derived_decoders_from_parent_decoders}.

\begin{example}
    \label{ex:classical_and_stab_induced_TS}
    \begin{figure*}
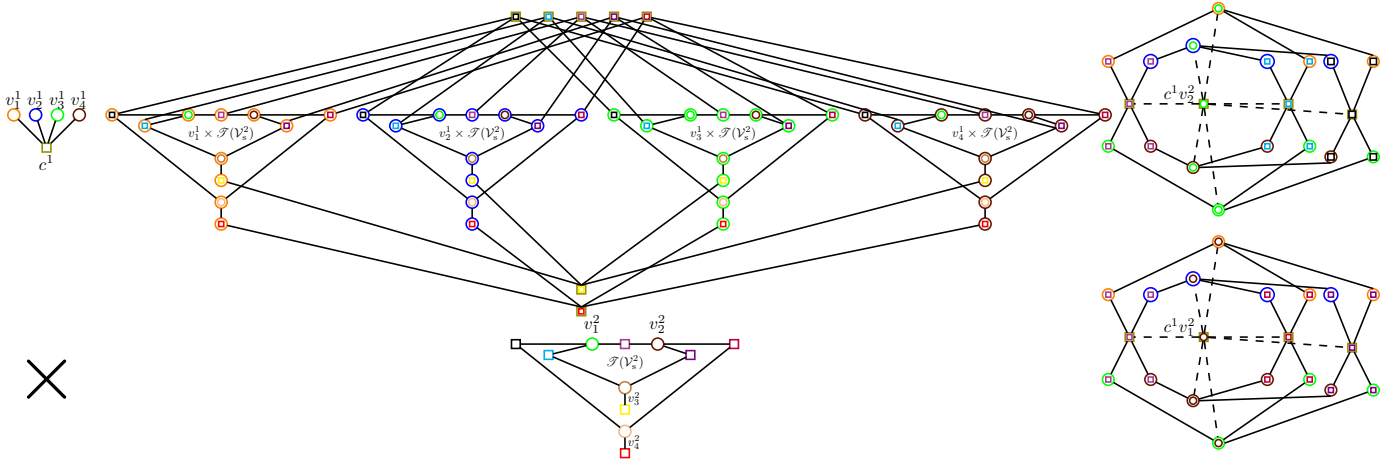

    \begin{subfigure}[t!]{0.78\textwidth}
        \centering
        \input{Figures/classical_and_stab_induced_TS/copies_of_VV_type_classical_TS}
        \subcaption{The figure shows how the graph product of a degree-four check node, $c^1$, from $\mathscr{G}_1$ and a TS, $\mathscr{T}(\mathcal{V}^2_{\mathrm{s}})$, from $\mathscr{G}_2$ results in a graph that has four identical copies of TS $\mathscr{T}(\mathcal{V}^2_{\mathrm{s}})$. These identical copies are given by $v^1_\mathrm{i} \times \mathscr{T}(\mathcal{V}^2_\mathrm{s})$ for $v^1_{\mathrm{i}}\in \mathcal{N}_{c^1}$. Note that part of the resulting graph from the graph product that contains only $Z$-stabilizers is shown. This graph is induced by CV-type stabilizer generators $c^1v_1^2, c^1v_2^2,c^1v^2_3$ and $c^1v_4^2$ within $\mathscr{G}_{\mathrm{z}}$. 
        To see that indeed the graph is induced by these stabilizer generators, the subgraphs induced by $X$-type stabilizer generators $c^1v_1^2$ and $c^1v_2^2$ within $\mathscr{G}_{\mathrm{z}}$ are explicitly shown in Figure~\ref{fig:stab_induced_TS_inside_the_bigger_TS}.
        As this TS is induced by four $X$-type stabilizer generators, it can be induced by hook errors introduced during their measurement. Similar to how the CC-type variable node is shown to connect with two stabilizer generators in Figure~\ref{fig:stab_induced_TS_inside_the_bigger_TS}, other CC-type variable nodes in the upper part of the figure can also be shown to connect with two stabilizer generators among those that induce this subgraph.}
        \label{fig:multiple_copies_of_classical_TS_inside_stab_induced_TS}
    \end{subfigure}
        \hfill
        \begin{subfigure}[t!]{0.2\textwidth}
            \tikzstyle{cnode}=[circle,minimum size=0.65 cm,draw,blue,scale=0.33]
\tikzstyle{scnode}=[circle,minimum size=0.4 cm,draw,blue,scale=0.33]
\tikzstyle{gcnode}=[circle,minimum size=0.55 cm,draw,green,scale=0.33]
\tikzstyle{sgcnode}=[circle,minimum size=0.15 cm,draw,green,scale=0.33]
\tikzstyle{secnode}=[circle,minimum size=0.55 cm,draw,Sepia,scale=0.33]
\tikzstyle{ssecnode}=[circle,minimum size=0.15 cm,draw,Sepia,scale=0.33]
\tikzstyle{ocnode}=[circle,minimum size=0.55 cm,draw,orange,scale=0.33]
\tikzstyle{socnode}=[circle,minimum size=0.15 cm,draw,orange,scale=0.33]
\tikzstyle{brcnode}=[circle,minimum size=0.55 cm,draw,brown,scale=0.33]
\tikzstyle{sbrcnode}=[circle,minimum size=0.15 cm,draw,brown,scale=0.33]

\tikzstyle{olrnode}=[rectangle,draw,minimum width=0.4 cm,minimum height=0.4cm,outer sep=0pt,olive,scale=0.33]
\tikzstyle{solrnode}=[rectangle,draw,minimum width=0.1 cm,minimum height=0.1cm,outer sep=0pt,olive,scale=0.33]
\tikzstyle{rnode}=[rectangle,draw,minimum width=0.5 cm,minimum height=0.5cm,outer sep=0pt,scale=0.33]
\tikzstyle{srnode}=[rectangle,draw,minimum width=0.3 cm,minimum height=0.3cm,outer sep=0pt,scale=0.33]
\tikzstyle{crnode}=[rectangle,draw,minimum width=0.4 cm,minimum height=0.4cm,outer sep=0pt,cyan,scale=0.33]
\tikzstyle{scrnode}=[rectangle,draw,minimum width=0.2 cm,minimum height=0.1cm,outer sep=0pt,cyan,scale=0.33]
\tikzstyle{Myrnode}=[rectangle,draw,minimum width=0.4 cm,minimum height=0.4cm,outer sep=0pt,Mulberry,scale=0.33]
\tikzstyle{sMyrnode}=[rectangle,draw,minimum width=0.1 cm,minimum height=0.1cm,outer sep=0pt,Mulberry,scale=0.33]
\tikzstyle{vrnode}=[rectangle,draw,minimum width=0.4 cm,minimum height=0.4cm,outer sep=0pt,violet,scale=0.33]
\tikzstyle{svrnode}=[rectangle,draw,minimum width=0.1 cm,minimum height=0.1cm,outer sep=0pt,violet,scale=0.33]
\tikzstyle{prnode}=[rectangle,draw,minimum width=0.4 cm,minimum height=0.4cm,outer sep=0pt,purple,scale=0.33]
\tikzstyle{sprnode}=[rectangle,draw,minimum width=0.1 cm,minimum height=0.1cm,outer sep=0pt,purple,scale=0.33]
\tikzstyle{yrnode}=[rectangle,draw,minimum width=0.4 cm,minimum height=0.4cm,outer sep=0pt,yellow,scale=0.33]
\tikzstyle{syrnode}=[rectangle,draw,minimum width=0.1 cm,minimum height=0.1cm,outer sep=0pt,yellow,scale=0.33]

\definecolor{color1}{rgb}{1,0.2,0.3}
\definecolor{color2}{rgb}{0.4,0.5,0.7}
\definecolor{color3}{rgb}{0.1,0.8,0.5}
\definecolor{color4}{rgb}{0.5,0.3,1}
\definecolor{color5}{rgb}{0.5,1,1}
\definecolor{color6}{rgb}{0.8,0.3,0.6}
\definecolor{color7}{rgb}{0.6,0.4,0.3}

\begin{tikzpicture}[every node/.style={scale=0.85}]
\begin{scope}[node distance=2cm,semithick]

\node[gcnode] (vc3) {};
\node[sMyrnode] (vc3_in) at (vc3.center){};
\node[secnode] (vc4) [right of=vc3,xshift=-0 cm]{};
\node[sMyrnode] (vc4_in) [right of=vc3,xshift=-0 cm]{};
\node[olrnode] (cc1) [above of=vc3,xshift=1 cm]{};
\node[sMyrnode] (cc1_in) [above of=vc3,xshift=1 cm]{};
\node (lcc1)[left of =cc1,xshift=1.6cm]{};
\node[ocnode] (vc1) [above of=cc1,xshift=-1 cm]{};
\node[sMyrnode] (vc1_in) [above of=cc1,xshift=-1 cm]{};
\node[cnode] (vc2) [above of=cc1,xshift=1 cm]{};
\node[sMyrnode] (vc2_in) [above of=cc1,xshift=1 cm]{};

\draw (cc1) -- (vc1);
\draw (cc1) -- (vc2);
\draw (cc1) -- (vc3);
\draw (cc1) -- (vc4);

\node[secnode] (vv1) [right of=vc4,yshift=-1 cm]{};
\node[ssecnode] (vv1_in) [right of=vc4,yshift=-1 cm]{};
\node (lvv1)[below of =vv1,xshift=0cm,yshift=1.8cm]{};
\node[gcnode] (vv2) [below of=vv1,yshift=0cm,xshift=1.2cm]{};
\node[ssecnode] (vv2_in) [below of=vv1,yshift=0cm,xshift=1.2cm]{};
\node (lvv2)[below of =vv2,yshift=1.7cm,xshift=0cm]{};

\draw (vv1) -- (vc4);
\draw (vv2) -- (vc3);

\node[cnode] (vv6) [right of=vc2,yshift=0.75cm]{};
\node[ssecnode] (vv6_in) [right of=vc2,yshift=0.75cm]{};
\node (lvv6)[left of =vv6,yshift=-0cm,xshift=2cm,yshift=0.23cm]{};
\node[ocnode] (vv7) [above of=vv6,yshift=-0.25cm,xshift=1.2cm]{};
\node[ssecnode] (vv7_in) [above of=vv6,yshift=-0.25cm,xshift=1.2cm]{};
\node (lvv7)[left of =vv7,yshift=0.25cm,xshift=1.9cm]{};

\draw (vv6) -- (vc2);
\draw (vv7) -- (vc1);

\node[secnode] (vc9) [right of=vc4,xshift=3.5cm]{};
\node[sprnode] (vc9_in) [right of = vc4,xshift=3.5cm]{};
\node[gcnode] (vc10) [right of=vc9,xshift=0cm]{};
\node[sprnode] (vc10_in) [right of = vc9,xshift=0cm]{};
\node[olrnode] (cc3) [above of=vc10,xshift=-1cm]{};
\node[sprnode] (cc3_in) [above of = vc10,xshift=-1cm]{};
\node (lcc3)[above of =cc3,xshift=0.35cm,yshift=-1.9cm]{};
\node[cnode] (vc11) [above of=cc3,xshift=-1cm]{};
\node[sprnode] (vc11_in) [above of = cc3,xshift=-1cm]{};
\node[ocnode] (vc12) [above of=cc3,xshift=1cm]{};
\node[sprnode] (vc12_in) [above of = cc3,xshift=1cm]{};

\draw (cc3) -- (vc9);
\draw (cc3) -- (vc10);
\draw (cc3) -- (vc11);
\draw (cc3) -- (vc12);

\node[olrnode] (cc5) [right of=cc3,xshift=1 cm,yshift=-0.5cm]{};
\node[svrnode] (cc5_in) [right of = cc3,xshift=1 cm,yshift=-0.5cm]{};
\node (lcc5)[right of =cc5,xshift=-1.5cm,yshift=-0cm]{};
\node[ocnode] (vc17) [above of=cc5,xshift=1cm,yshift=0.5cm]{};
\node[svrnode] (vc17_in) [above of = cc5,xshift=1cm,yshift=0.5cm]{};
\node[cnode] (vc18) [above of=cc5,xshift=-1cm,yshift=0.5cm]{};
\node[svrnode] (vc18_in) [above of = cc5,xshift=-1cm,yshift=0.5cm]{};
\node[gcnode] (vc19) [below of=cc5,xshift=1cm,yshift=0 cm]{};
\node[svrnode] (vc19_in) [below of = cc5,xshift=1cm,yshift=0 cm]{};
\node[secnode] (vc20) [below of=cc5,xshift=-1cm,yshift=0 cm]{};
\node[svrnode] (vc20_in) [below of = cc5,xshift=-1cm,yshift=0 cm]{};

\draw (cc5) -- (vc17);
\draw (cc5) -- (vc18);
\draw (cc5) -- (vc19);
\draw (cc5) -- (vc20);

\draw (vv1) --(vc9);
\draw (vv1.0) --(vc20.270);

\draw (vv2) --(vc10);
\draw (vv2.0) --(vc19);

\draw (vv6) --(vc11);
\draw (vv7) --(vc17.90);
\draw (vv6) --(vc18.120);
\draw (vv7) --(vc12);

\node[olrnode] (cv1) [right of =cc1,xshift=1.5cm]{};
\node[ssecnode] (cv1_in) [right of =cc1,xshift=1.5cm]{};
\draw[dashed] (cv1) -- node[above,midway,yshift=-0.1cm]{} (cc1);
\draw[dashed] (cv1) -- node[above,midway,yshift=-0.1cm]{} (cc3);
\draw[dashed] (cv1) -- node[below,near start,yshift=0.1cm]{} (cc5);
\draw[dashed] (cv1) -- node[left,midway,xshift=0.1cm]{} (vv1);
\draw[dashed] (cv1) -- node[right,near start,xshift=-0.1cm]{} (vv2);
\draw[dashed] (cv1) -- node[left,midway,xshift=0.1cm]{} (vv6);
\draw[dashed] (cv1) -- node[right,near start,xshift=-0.1cm]{} (vv7);
\node[scale=0.75] [left of =cv1,xshift=1.5 cm,yshift=0.2cm]{$c^1v^2_1$};

\node[gcnode] (vc21) [above of=vc1,yshift=5cm] {};
\node[sMyrnode] (vc21_in) at (vc21.center){};
\node[secnode] (vc22) [right of=vc21,xshift=-0 cm]{};
\node[sMyrnode] (vc22_in) [right of=vc21,xshift=-0 cm]{};
\node[olrnode] (cc6) [above of=vc21,xshift=1 cm]{};
\node[sMyrnode] (cc6_in) [above of=vc21,xshift=1 cm]{};
\node[ocnode] (vc23) [above of=cc6,xshift=-1 cm]{};
\node[sMyrnode] (vc23_in) [above of=cc6,xshift=-1 cm]{};
\node[cnode] (vc24) [above of=cc6,xshift=1 cm]{};
\node[sMyrnode] (vc24_in) [above of=cc6,xshift=1 cm]{};

\draw (cc6) -- (vc21);
\draw (cc6) -- (vc22);
\draw (cc6) -- (vc23);
\draw (cc6) -- (vc24);

\node[secnode] (vv8) [right of=vc22,yshift=-1 cm]{};
\node[sgcnode] (vv8_in) [right of=vc22,yshift=-1 cm]{};
\node[gcnode] (vv9) [below of=vv8,yshift=0cm,xshift=1.2cm]{};
\node[sgcnode] (vv9_in) [below of=vv8,yshift=0cm,xshift=1.2cm]{};

\draw (vv8) -- (vc22);
\draw (vv9) -- (vc21);

\node[cnode] (vv10) [right of=vc24,yshift=0.75cm]{};
\node[sgcnode] (vv10_in) [right of=vc24,yshift=0.75cm]{};
\node[ocnode] (vv11) [above of=vv10,yshift=-0.25cm,xshift=1.2cm]{};
\node[sgcnode] (vv11_in) [above of=vv10,yshift=-0.25cm,xshift=1.2cm]{};

\draw (vv10) -- (vc24);
\draw (vv11) -- (vc23);

\node[secnode] (vc25) [right of=vc22,xshift=3.5cm]{};
\node[scrnode] (vc25_in) [right of = vc22,xshift=3.5cm]{};
\node[gcnode] (vc26) [right of=vc25,xshift=0cm]{};
\node[scrnode] (vc26_in) [right of = vc25,xshift=0cm]{};
\node[olrnode] (cc7) [above of=vc26,xshift=-1cm]{};
\node[scrnode] (cc7_in) [above of = vc26,xshift=-1cm]{};
\node[cnode] (vc27) [above of=cc7,xshift=-1cm]{};
\node[scrnode] (vc27_in) [above of = cc7,xshift=-1cm]{};
\node[ocnode] (vc28) [above of=cc7,xshift=1cm]{};
\node[scrnode] (vc28_in) [above of = cc7,xshift=1cm]{};

\draw (cc7) -- (vc25);
\draw (cc7) -- (vc26);
\draw (cc7) -- (vc27);
\draw (cc7) -- (vc28);

\node[olrnode] (cc8) [right of=cc7,xshift=1 cm,yshift=-0.5cm]{};
\node[srnode] (cc8_in) [right of = cc7,xshift=1 cm,yshift=-0.5cm]{};
\node[ocnode] (vc29) [above of=cc8,xshift=1cm,yshift=0.5cm]{};
\node[srnode] (vc29_in) [above of = cc8,xshift=1cm,yshift=0.5cm]{};
\node[cnode] (vc30) [above of=cc8,xshift=-1cm,yshift=0.5cm]{};
\node[srnode] (vc30_in) [above of = cc8,xshift=-1cm,yshift=0.5cm]{};
\node[gcnode] (vc31) [below of=cc8,xshift=1cm,yshift=0 cm]{};
\node[srnode] (vc31_in) [below of = cc8,xshift=1cm,yshift=0 cm]{};
\node[secnode] (vc32) [below of=cc8,xshift=-1cm,yshift=0 cm]{};
\node[srnode] (vc32_in) [below of = cc8,xshift=-1cm,yshift=0 cm]{};

\draw (cc8) -- (vc29);
\draw (cc8) -- (vc30);
\draw (cc8) -- (vc31);
\draw (cc8) -- (vc32);

\draw (vv8) --(vc25);
\draw (vv8.0) --(vc32.270);

\draw (vv9) --(vc26);
\draw (vv9.0) --(vc31);

\draw (vv10) --(vc27);
\draw (vv11) --(vc29.90);
\draw (vv10) --(vc30.120);
\draw (vv11) --(vc28);

\node[olrnode] (cv2) [right of =cc6,xshift=1.5cm]{};
\node[sgcnode] (cv2_in) [right of =cc6,xshift=1.5cm]{};
\draw[dashed] (cv2) -- node[above,midway,yshift=-0.1cm]{} (cc6);
\draw[dashed] (cv2) -- node[above,midway,yshift=-0.1cm]{} (cc7);
\draw[dashed] (cv2) -- node[below,near start,yshift=0.1cm]{} (cc8);
\draw[dashed] (cv2) -- node[left,midway,xshift=0.1cm]{} (vv8);
\draw[dashed] (cv2) -- node[right,near start,xshift=-0.1cm]{} (vv9);
\draw[dashed] (cv2) -- node[left,midway,xshift=0.1cm]{} (vv10);
\draw[dashed] (cv2) -- node[right,near start,xshift=-0.1cm]{} (vv11);
\node[scale=0.75] [left of =cv2,xshift=1.5 cm,yshift=0.2cm]{$c^1v^2_2$};
\end{scope}    
\end{tikzpicture}
            \subcaption{\emph{Stabilizer-induced} TSs induced by the $X$-type stabilizer generators $c^1v^2_1$ and $c^1v^2_2$ within $\mathscr{G}_{\mathrm{z}}$. All variable nodes in these two induced subgraphs also appear in Fig.~\ref{fig:multiple_copies_of_classical_TS_inside_stab_induced_TS}. The CC-type variable node represented by \input{Figures/classical_and_stab_induced_TS/shape_CC_type_node_olive_Myr} is common to both subgraphs.}
            \label{fig:stab_induced_TS_inside_the_bigger_TS}
        \end{subfigure}
        \caption{Illustration of a \emph{stabilizer-induced} TS containing multiple copies of a classical TS from $\mathscr{G}_2$.  }
        \label{fig:classical_and_stab_included_TS}
    \end{figure*}

Consider the product of a degree-four check node, $c^1$, from $\mathscr{G}_1$ with a classical TS, $\mathscr{T}(\mathcal{V}^2_{\mathrm{s}})$, from $\mathscr{G}_2$, as shown in Figure~\ref{fig:multiple_copies_of_classical_TS_inside_stab_induced_TS}. To simplify the pictorial representation, $X$-type (or CV-type) generators and their edges have been omitted. Since $c^1$ has degree four, the graph contains four identical copies of $\mathscr{T}(\mathcal{V}^2_{\mathrm{s}})$. These isomorphic copies are denoted by $v^1_{\mathrm{i}} \times \mathscr{T}(\mathcal{V}^2_{\mathrm{s}})$ for $v^1_{\mathrm{i}} \in \mathcal{N}_{c^1}$. In each copy, the variable nodes from $\mathscr{T}(\mathcal{V}^2_{\mathrm{s}})$ are transformed into VV-type variable nodes, while the check nodes are transformed into VC-type check nodes. The edges are therefore represented as edges between VV-type variable nodes and VC-type check nodes. This observation is formalized in Part~1 of Lemma~\ref{lemma:stab_induced_TS_containg_classical_TS}, described later in this section. Furthermore, note that two isomorphic copies of $\mathscr{T}(\mathcal{V}^2_{\mathrm{s}})$ do not share nodes or edges; this is formalized in Part~2 of Lemma~\ref{lemma:stab_induced_TS_containg_classical_TS}.

The graph in Figure~\ref{fig:multiple_copies_of_classical_TS_inside_stab_induced_TS} is induced by X-type stabilizer generators (CV-type checks) $c^1 v_1^2$, $c^1 v^2_2$, $c^1 v^2_3$, and $c^1 v_4^2$. This can be verified by checking that the variable nodes in the graph coincide exactly with those present in the subgraphs induced by these generators. To facilitate this verification, the nodes in the figure are color-coded, since the X-type stabilizer generators themselves are not shown in Figure~\ref{fig:multiple_copies_of_classical_TS_inside_stab_induced_TS}. As an example, consider the subgraph induced by the X-type stabilizer generator $c^1 v^2_1$, shown in Figure~\ref{fig:stab_induced_TS_inside_the_bigger_TS}.
In this subgraph, the outer circle in the representation of VV-type variable nodes is derived from the four neighboring variable nodes of $c^1$, represented as \tikzstyle{cnode}=[circle,minimum size=0.65 cm,draw,blue,scale=0.44]
\begin{tikzpicture}[every node/.style={scale=1}]
\begin{scope}[node distance=1.5 cm,semithick]
\node[cnode](bc){};
\end{scope}
\end{tikzpicture}, \tikzstyle{gcnode}=[circle,minimum size=0.65 cm,draw,green,scale=0.44]
\begin{tikzpicture}[every node/.style={scale=1}]
\begin{scope}[node distance=1.5 cm,semithick]
\node[gcnode](gc){};
\end{scope}
\end{tikzpicture}, \tikzstyle{ocnode}=[circle,minimum size=0.65 cm,draw,orange,scale=0.44]
\begin{tikzpicture}[every node/.style={scale=1}]
\begin{scope}[node distance=1.5 cm,semithick]
\node[ocnode](oc){};
\end{scope}
\end{tikzpicture} and \tikzstyle{secnode}=[circle,minimum size=0.65 cm,draw,Sepia,scale=0.44]
\begin{tikzpicture}[every node/.style={scale=1}]
\begin{scope}[node distance=1.5 cm,semithick]
\node[secnode](sec){};
\end{scope}
\end{tikzpicture}. The inner circle is derived from $v^2_1$, represented by \tikzstyle{gcnode}=[circle,minimum size=0.65 cm,draw,green,scale=0.44]
\begin{tikzpicture}[every node/.style={scale=1}]
\begin{scope}[node distance=1.5 cm,semithick]
\node[gcnode](gc){};
\end{scope}
\end{tikzpicture}. Similarly, the outer rectangle in the representation of CC-type variable nodes is derived from $c^1$, represented as \tikzstyle{olrnode}=[rectangle,draw,minimum width=0.1 cm,minimum height=0.1cm,outer sep=0pt,olive,scale=1]
\begin{tikzpicture}[every node/.style={scale=1}]
\begin{scope}[node distance=1.5 cm,semithick]
\node[olrnode](olr){};
\end{scope}
\end{tikzpicture}, and the inner rectangle is derived from one of the three neighboring check nodes of $v^2_1$, represented as \tikzstyle{brnode}=[rectangle,draw,minimum width=0.1 cm,minimum height=0.1cm,outer sep=0pt,black,scale=1]
\begin{tikzpicture}[every node/.style={scale=1}]
\begin{scope}[node distance=1.5 cm,semithick]
\node[brnode](br){};
\end{scope}
\end{tikzpicture}, \tikzstyle{crnode}=[rectangle,draw,minimum width=0.1 cm,minimum height=0.1cm,outer sep=0pt,cyan,scale=1]
\begin{tikzpicture}[every node/.style={scale=1}]
\begin{scope}[node distance=1.5 cm,semithick]
\node[crnode](Myr){};
\end{scope}
\end{tikzpicture}, and \tikzstyle{Myrnode}=[rectangle,draw,minimum width=0.1 cm,minimum height=0.1cm,outer sep=0pt,Mulberry,scale=1]
\begin{tikzpicture}[every node/.style={scale=1}]
\begin{scope}[node distance=1.5 cm,semithick]
\node[Myrnode](Myr){};
\end{scope}
\end{tikzpicture}.
From this, one can verify that the variable nodes connected to $c^1 v^2_1$ are in $\mathscr{T}(c^1) \times \mathscr{T}(\mathcal{V}^2_{\mathrm{s}})$, which is shown in Figure~\ref{fig:multiple_copies_of_classical_TS_inside_stab_induced_TS}. Although not explicitly stated for brevity, the analogous properties hold for variable nodes connected to X-type stabilizers $c^1 v^2_2$, $c^1 v^2_3$, and $c^1 v^2_4$. These can be used to confirm that these four X-type stabilizer generators indeed induce the graph in Figure~\ref{fig:multiple_copies_of_classical_TS_inside_stab_induced_TS}. This observation is formalized in Part~3 of Lemma~\ref{lemma:stab_induced_TS_containg_classical_TS}. Finally, observe that every CC-type variable node is connected to one VC-type check node in each of the isomorphic copies of $\mathscr{T}(\mathcal{V}^2_{\mathrm{s}})$, as formalized in Part~6 of Lemma~\ref{lemma:stab_induced_TS_containg_classical_TS}.
\end{example}

When a \emph{stabilizer-induced} TS contains multiple isomorphic copies of a classical TS, we call a copy \emph{active} if it contains at least the critical number of errors for that underlying classical TS. We call the overall \emph{stabilizer-induced} TS \emph{active} if at least $\left\lceil |\mathcal{N}_{c^1}|/2 \right\rceil$ of these copies are active (and analogously, in the dual case, if at least $\left\lceil |\mathcal{N}_{v^2}|/2 \right\rceil$ copies are active for the relevant variable node $v^2$).
Intuitively, decoders derived from the parent-code decoders can resolve such \emph{stabilizer-induced} TSs without modification as long as the number of active copies is strictly less than $\left\lceil |\mathcal{N}_{c^1}|/2 \right\rceil$ (resp., $\left\lceil |\mathcal{N}_{v^2}|/2 \right\rceil$); when this threshold is met or exceeded, the derived decoder may fail and additional measures are needed.
The following lemma formalizes the observations made in Example~\ref{ex:classical_and_stab_induced_TS} about how multiple instances of a given TS arise in HGP codes, inherited from their parent classical LDPC codes.

\begin{lemma}
\label{lemma:stab_induced_TS_containg_classical_TS}
Consider an HGP code constructed from the graph product of two classical LDPC codes, represented by regular Tanner graphs $\mathscr{G}_1$ and $\mathscr{G}_2$. Let $\mathcal{V}^2_{\mathrm{s}} \subset \mathcal{V}^2$ denote a subset of variable nodes of $\mathscr{G}_2$, and let $\mathscr{T}(\mathcal{V}^2_{\mathrm{s}})$ be the subgraph induced by these nodes. For a check node $c^1$ in $\mathscr{G}_1$, denote by $\mathscr{T}(c^1)$ the subgraph it induces. 
The subgraph induced by CV-type stabilizer generators that are obtained from $c^1$ and $\mathcal{V}^2_{\mathrm{s}}$ is defined as
\[ \mathscr{T}_{\mathrm{stab}} \triangleq \bigl( \mathscr{T}(c^1) \times \mathscr{T}(\mathcal{V}^2_{\mathrm{s}}) \bigr) \cap \mathscr{G}_{\mathrm{z}}, \]
where $\mathscr{G}_{\mathrm{z}}$ is the Tanner graph corresponding to the $Z$-stabilizer generators. 
\begin{enumerate}

\item The subgraph $\mathscr{T}_{\mathrm{stab}}$ contains $\vert\mathcal{N}_{c^1}\rvert$ pairwise isomorphic copies of $\mathscr{T}(\mathcal{V}^2_{\mathrm{s}})$, one for each neighboring variable node $v^1_{\mathrm{i}'}\in \mathcal{N}_{c^1}$. 
For a fixed $v^1_{\mathrm{i}'} \in \mathcal{N}_{c^1}$, each variable node $v^2\in \mathscr{T}(\mathcal{V}^2_{\mathrm{s}})$ forms a VV-type variable node $v^1_{\mathrm{i}'}v^2$, and each check node $c^2_{\mathrm{j}}\in \mathscr{T}(\mathcal{V}^2_{\mathrm{s}})$ forms a VC-type check node $v^1_{\mathrm{i}'}c^2_{\mathrm{j}}$.

\item These isomorphic copies do not share any common vertices, i.e., 
$\bigl(v_{\mathrm{i}'}^1 \times \mathscr{T}(\mathcal{V}^2_{\mathrm{s}})\bigr) \cap 
\bigl(v_{\mathrm{k}'}^1 \times \mathscr{T}(\mathcal{V}^2_{\mathrm{s}})\bigr) = \emptyset$
for all $v_{\mathrm{i}'}^1, v_{\mathrm{k}'}^1 \in \mathcal{N}_{c^1}$ with $v_{\mathrm{i}'}^1 \neq v_{\mathrm{k}'}^1$.

\item The subgraph $\mathscr{T}_{\mathrm{stab}}$ is induced by the CV-type stabilizer generators contained in $c^1 \times \mathcal{V}^2_{\mathrm{s}}$. 

\item Assume that \(\mathscr{T}(\mathcal{V}^2_{\mathrm{s}})\) is a TS in the Tanner graph \(\mathscr{G}_2\).
Then there exists a measurement schedule such that a single CNOT failure in the measurement circuit of each
CV-type stabilizer generator in the set $c^1 \times \mathcal{V}^2_{\mathrm{s}}$ can corrupt all
variable nodes in $\ell$ of the isomorphic copies of $\mathscr{T}(\mathcal{V}^2_{\mathrm{s}})$
contained in $\mathscr{T}_{\mathrm{stab}}$, for some $\ell \in \{1,2,\dots,\lvert\mathcal{N}_{c^1}\rvert\}$.

\item  Let $\mathcal{C}^2_{\mathrm{o}}$ denote the set of check nodes that have an odd degree in $\mathscr{T}(\mathcal{V}^2_{\mathrm{s}})$. 
The stabilizer support (i.e., the set of variable nodes on which the corresponding stabilizer acts non-trivially) generated by linear combinations of
 stabilizer generators in $c^1 \times \mathcal{V}^2_{\mathrm{s}}$ then contains all VV-type variable nodes in $\mathscr{T}_{\mathrm{stab}}$, as well as all CC-type variable nodes in $c^1 \times \mathcal{C}^2_{\mathrm{o}}$.
\item Each CC-type variable node $c^1c^2$ is connected to the VC-type check node $v^1_{\mathrm{i}'}c^2$ in the isomorphic copy $v^1_{\mathrm{i}'} \times \mathscr{T}(\mathcal{V}^2_{\mathrm{s}})$, for all $v^1_{\mathrm{i}'} \in \mathcal{N}_{c^1}$.


\end{enumerate}
\end{lemma}

\begin{proof}\begin{enumerate}

 \item Consider a VV-type variable node $v^1 v^2$ in subgraph $\mathscr{T}_{\mathrm{stab}}$, where $v^1 \in \mathcal{N}_{c^1}$ and $v^2 \in \mathscr{T}(\mathcal{V}^2_{\mathrm{s}})$. The neighborhood of $v^1 v^2$ in $\mathscr{G}_{\mathrm{z}}$ is identical to the neighborhood of $v^2$ in $\mathscr{G}_2$ because the set of check nodes adjacent to $v^1 v^2$ in $\mathscr{G}_{\mathrm{z}}$ is given by $v^1 \times \mathcal{N}_{v^2}$. Consequently, for each $v^1 \in \mathcal{N}_{c^1}$, the subgraph $v^1 \times \mathscr{T}(\mathcal{V}^2_{\mathrm{s}})$ is isomorphic to $\mathscr{T}(\mathcal{V}^2_{\mathrm{s}})$, yielding $\lvert\mathcal{N}_{c^1}\rvert$ such isomorphic copies in $\mathscr{T}_{\mathrm{stab}}$.

\item From the definition of the graph product, it follows that the subgraphs $(v_{\mathrm{i}'}^1 \times \mathscr{T}(\mathcal{V}^2_{\mathrm{s}}))$ and $(v_{\mathrm{k}'}^1 \times \mathscr{T}(\mathcal{V}^2_{\mathrm{s}}))$ have no common vertices whenever $v^1_{\mathrm{i}'} \neq v^1_{\mathrm{k}'}$. Hence, $(v_{\mathrm{i}'}^1 \times \mathscr{T}(\mathcal{V}^2_{\mathrm{s}})) \cap (v_{\mathrm{k}'}^1 \times \mathscr{T}(\mathcal{V}^2_{\mathrm{s}})) = \emptyset$.

\item To show that the subgraph $\mathscr{T}_{\mathrm{stab}}$ is induced by the CV-type check nodes in $c^1 \times \mathcal{V}^2_{\mathrm{s}}$, it suffices to verify that the set of variable nodes adjacent to these check nodes coincides with the set of variable nodes in the subgraph. 
From the construction of the subgraph, variable nodes within subgraph $\mathscr{T}_{\mathrm{stab}}$ are given by
\begin{align}
\label{eq:VN_subgraph}
\mathcal{V}(\mathscr{T}_{\mathrm{stab}})
&= \Bigl(\underbrace{\mathcal{N}_{c^1} \times \mathcal{V}^2_{\mathrm{s}}}_{\text{VV-type}}\Bigr)
\,\cup\,
\Bigl(\underbrace{c^1 \times\Bigl( \bigcup_{v^2 \in \mathcal{V}^2_{\mathrm{s}}} \bigcup_{c^2 \in \mathcal{N}_{v^2}} c^2\Bigr)\Bigr)}_{\text{CC-type}}.
\end{align}
 Rewriting this as a union over $v^2 \in \mathcal{V}^2_{\mathrm{s}}$ yields
\begin{align}
\label{eq:VN_check_nodes}
\mathcal{V}(\mathscr{T}_{\mathrm{stab}})
&= \bigcup_{v^2 \in \mathcal{V}^2_{\mathrm{s}}}
\underbrace{\bigl(\mathcal{N}_{c^1} \times v^2\bigr)}_{\text{VV-type neighbors}}
\,\cup\,
\bigcup_{v^2 \in \mathcal{V}^2_{\mathrm{s}}}
\underbrace{\bigl(c^1 \times \mathcal{N}_{v^2}\bigr)}_{\text{CC-type neighbors}}.
\end{align}
By construction, the variable nodes described in \eqref{eq:VN_check_nodes} are precisely the neighbors of the check nodes in $c^1 \times \mathcal{V}^2_{\mathrm{s}}$. Comparing \eqref{eq:VN_check_nodes} with \eqref{eq:VN_subgraph} shows that this set of neighbors coincides with the set of variable nodes of $\mathscr{T}_{\mathrm{stab}}$. Hence $\mathscr{T}_{\mathrm{stab}}$ is exactly the subgraph induced by the check-node set $c^1 \times \mathcal{V}^2_{\mathrm{s}}$.

\item We want to show that there exists a measurement schedule such that a single CNOT failure in the measurement circuit of each CV-type stabilizer generator in the set $c^1 \times \mathcal{V}_{\mathrm{s}}^2$ can corrupt all VV-type variable nodes in $\ell$ isomorphic copies of $\mathscr{T}(\mathcal{V}_{\mathrm{s}}^2)$.
Fix distinct neighbors $v^1_1,\dots,v^1_{\ell}\in \mathcal{N}_{c^1}$. Consider a stabilizer generator of the form $c^1 v^2_{\mathrm{i}}$ with $v^2_{\mathrm{i}}\in \mathcal{V}^2_{\mathrm{s}}$. Its measurement circuit consists of an ancilla qubit and CNOT gates from the ancilla to all variable nodes connected to it. Since the order of these CNOTs can be chosen arbitrarily, we may choose a schedule in which the $\ell$ CNOT gates
\[
\text{CNOT}(\text{ancilla}\to v^1_{\mathrm{i'}} v^2_{\mathrm{i}}), \quad \forall\, \mathrm{i}' \in [\ell]
\]
are executed last. Assume that exactly one CNOT in the measurement circuit for this stabilizer generator fails and introduces an $X$-error on the ancilla immediately before these last $\ell$ CNOTs. Then, by the hook-error propagation rule from Section~\ref{sec:hook_error}, this $X$-error is copied to the $\ell$ data qubits $v^1_1 v^2_{\mathrm{i}},\dots,v^1_{\ell} v^2_{\mathrm{i}}$.
Now repeat this construction for every stabilizer generator in the set $c^1 \times \mathcal{V}^2_{\mathrm{s}}$. For each $v^2_{\mathrm{i}}\in \mathcal{V}^2_{\mathrm{s}}$, choose the same type of schedule (with the CNOTs to $v^1_1 v^2_{\mathrm{i}},\dots,v^1_{\ell} v^2_{\mathrm{i}}$ executed last) and assume one such CNOT failure.
Therefore, after performing the measurement of all these stabilizer generators, the full set of erroneous VV-type variable nodes is
\begin{align*}
\bigcup_{\mathrm{i}'\in[\ell]} \bigl(v^1_{\mathrm{i}'} \times \mathcal{V}^2_{\mathrm{s}}\bigr),
\end{align*}
which is the union of VV-type variable nodes in $\ell$ isomorphic copies of $\mathscr{T}(\mathcal{V}^2_{\mathrm{s}})$ contained in $\mathscr{T}_{\mathrm{stab}}$.

\item Consider the set of CV-type stabilizer generators in
\[
c^1 \times \mathcal{V}^2_{\mathrm{s}}
= \{\, c^1 v^2_{\mathrm{i}} : v^2_{\mathrm{i}} \in \mathcal{V}^2_{\mathrm{s}}\,\}.
\]
The stabilizer support generated by linear combinations of these generators is the set of variable nodes
that appear in the support of an odd number of generators in $c^1 \times \mathcal{V}^2_{\mathrm{s}}$.
Fix any $v^2_{\mathrm{i}} \in \mathcal{V}^2_{\mathrm{s}}$. The VV-type variable nodes in the support of the
generator $c^1 v^2_{\mathrm{i}}$ are $\mathcal{N}_{c^1} \times v^2_{\mathrm{i}}$, i.e., the VV-type
variable nodes of the form $v^1 v^2_{\mathrm{i}}$ with $v^1 \in \mathcal{N}_{c^1}$.
Now take two stabilizer generators $c^1 v^2_{\mathrm{i}}$ and $c^1 v^2_{\mathrm{k}}$ in
$c^1 \times \mathcal{V}^2_{\mathrm{s}}$, with $v^2_{\mathrm{i}} \neq v^2_{\mathrm{k}}$. Their VV-type
supports,
\[
\mathcal{N}_{c^1} \times v^2_{\mathrm{i}}
\quad\text{and}\quad
\mathcal{N}_{c^1} \times v^2_{\mathrm{k}},
\]
are disjoint. Hence, no VV-type variable node in $\mathscr{T}_{\mathrm{stab}}$ can appear in more than one
stabilizer generator in $c^1 \times \mathcal{V}^2_{\mathrm{s}}$. Therefore, each VV-type variable node in
$\mathscr{T}_{\mathrm{stab}}$ appears in exactly one generator of $c^1 \times \mathcal{V}^2_{\mathrm{s}}$,
and thus appears in the support of an odd number of generators. It follows that the stabilizer support
contains all VV-type variable nodes present in $\mathscr{T}_{\mathrm{stab}}$.
We now determine the CC-type variable nodes in $\mathscr{T}_{\mathrm{stab}}$ that lie in the support of an
odd number of generators in $c^1 \times \mathcal{V}^2_{\mathrm{s}}$. A CC-type variable node in
$\mathscr{T}_{\mathrm{stab}}$ has the form $c^1 c^2_{\mathrm{j}}$, where $c^2_{\mathrm{j}}$ is a check node
in $\mathscr{T}(\mathcal{V}^2_{\mathrm{s}})$. Such a node belongs to the support of a generator
$c^1 v^2_{\mathrm{i}} \in c^1 \times \mathcal{V}^2_{\mathrm{s}}$ if and only if the check node
$c^2_{\mathrm{j}}$ is adjacent to $v^2_{\mathrm{i}}$ in the Tanner graph of the second code, i.e.,
\[
 c^1 c^2_{\mathrm{j}} \in \operatorname{supp}(c^1 v^2_{\mathrm{i}})
\quad\Longleftrightarrow\quad
 c^2_{\mathrm{j}} \in \mathcal{N}_{v^2_{\mathrm{i}}}.
\]
Equivalently, the number of generators in $c^1 \times \mathcal{V}^2_{\mathrm{s}}$ whose support contains
$c^1 c^2_{\mathrm{j}}$ is exactly the degree of $c^2_{\mathrm{j}}$ in $\mathscr{T}(\mathcal{V}^2_{\mathrm{s}})$.
Thus, $c^1 c^2_{\mathrm{j}}$ is contained in the support of an odd number of these stabilizer generators if and
only if the degree of $c^2_{\mathrm{j}}$ in $\mathscr{T}(\mathcal{V}^2_{\mathrm{s}})$ is odd. By definition,
the set of such check nodes is $\mathcal{C}^2_{\mathrm{o}}$. Therefore, a CC-type variable node
$c^1 c^2_{\mathrm{j}}$ lies in the stabilizer support if and only if $c^2_{\mathrm{j}} \in \mathcal{C}^2_{\mathrm{o}}$.
This characterizes completely which VV- and CC-type variable nodes are in the stabilizer support, with the role
of VV-nodes determined by the non-intersecting support of stabilizer generators and the role of CC-nodes
determined by odd-degree check nodes in $\mathscr{T}(\mathcal{V}^2_{\mathrm{s}})$.

\item The set of VC-type check nodes connected to the CC-type variable node \(c^1 c^2\) is
\(\mathcal{N}_{c^1} \times c^2\).
Therefore, the VC-type check node $v^1_{\mathrm{i}'}c^1$ is connected to $c^1c^2$ for $v^1_{\mathrm{i}'}\in \mathcal{N}_{c^1}.$
Recall that the isomorphic copies of \(\mathscr{T}(\mathcal{V}_{\mathrm{s}}^2)\) within
\(\mathscr{T}_{\mathrm{stab}}\) are given by
\[
v^1_{\mathrm{i}'}\,\times \,\mathscr{T}(\mathcal{V}_{\mathrm{s}}^2), \quad v^1_{\mathrm{i}'} \in \mathcal{N}_{c^1}.
\]
Since the check node \(c^2\) is contained in \(\mathscr{T}(\mathcal{V}_{\mathrm{s}}^2)\), it follows that for each
\(v^1_{\mathrm{i}'} \in \mathcal{N}_{c^1}\), the VC-type check node \(v^1_{\mathrm{i}'} c^2\) belongs to the
corresponding isomorphic copy \(v^1_{\mathrm{i}'}\times \mathscr{T}(\mathcal{V}_{\mathrm{s}}^2)\).
\end{enumerate}

\end{proof}
\begin{remark}
Part~4 of Lemma~\ref{lemma:stab_induced_TS_containg_classical_TS} analyzes a particular syndrome-measurement schedule and shows that, for the family of \emph{stabilizer-induced} TSs considered there, the number of CNOT failures needed to activate a fixed number of isomorphic copies of the inherited TS $\mathscr{T}(\mathcal{V}^2_{\mathrm{s}})$ inside $c^1\times \mathscr{T}(\mathcal{V}^2_{\mathrm{s}})$ does not increase with that number. A natural question is whether this non-scaling is specific to the chosen schedule. In fact, many other CNOT orderings yield the same behavior, and it is unclear whether there exists any schedule for which the minimum number of CNOT failures required to activate a fixed number of such copies grows with the number of copies.
Even if schedules exist that suppress activation of a particular $c^1\times \mathscr{T}(\mathcal{V}^2_{\mathrm{s}})$, optimizing a single schedule to suppress all such \emph{stabilizer-induced} TSs simultaneously (over all relevant $c^1\in\mathcal{C}^1$) is challenging: a schedule that suppresses one instance may inadvertently facilitate activation of another, and each parent code typically contains many such checks. Consequently, we do not attempt to optimize the measurement schedule. Instead, we design decoders that resolve \emph{stabilizer-induced} TSs of the form $c^1 \times \mathscr{T}(\mathcal{V}^2_{\mathrm{s}})$, making schedule optimization unnecessary to avoid decoding failures due to these TSs in our setting.

Finally, for HGP codes, the effective distance is invariant under the choice of measurement schedule~\cite{manes2025distance}. Consequently, by resolving these \emph{stabilizer-induced} TSs directly, the decoders considered in this work exploit this distance-preserving property. Recent work extends effective-distance analyses to other code families~\cite{tan2025effective}; it would be interesting to extend the present study of TS-related decoding failures to those families in the future.
\end{remark}

In an HGP code, the Tanner graphs $\mathscr{G}_{\mathrm{x}}$ and $\mathscr{G}_{\mathrm{z}}$ associated with the $X$- and $Z$-stabilizer generators arise as graph products of two parent LDPC codes with Tanner graphs $\mathscr{G}_1$ and $\mathscr{G}_2$. Lemma~\ref{lemma:stab_induced_TS_containg_classical_TS} shows that a classical TS from $\mathscr{G}_2$ can reappear inside a subgraph induced by an $X$-type stabilizer in $\mathscr{G}_{\mathrm{z}}$, due to the presence of multiple isomorphic copies of $\mathscr{G}_2$ in $\mathscr{G}_{\mathrm{z}}$. Similarly, $\mathscr{G}_1^\mathsf{T}$ also appears multiple times in $\mathscr{G}_{\mathrm{z}}$, so isomorphic copies of a classical TS from $\mathscr{G}_1^\mathsf{T}$ can occur within an $X$-type \emph{stabilizer-induced} TS.
The difference between \emph{stabilizer-induced} TSs that contain multiple copies inherited from $\mathscr{G}_2$ and those inherited from $\mathscr{G}_1^\mathsf{T}$ lies in the role of the variable nodes: for TSs inherited from $\mathscr{G}_2$, the variable nodes of the classical TS become VV-type nodes and CC-type nodes connect different copies, whereas for TSs inherited from $\mathscr{G}_1^\mathsf{T}$ the roles are reversed. These observations are formalized in Appendix~\ref{sec:lemma2}.

To illustrate this, consider Figure~\ref{fig:multiple_copies_of_classical_TS_inside_stab_induced_TS_CC}, which depicts the graph product of the classical TS $\mathscr{T}(\mathcal{C}^1_\mathrm{s})$ from $\mathscr{G}_1^\mathsf{T}$ with a variable node $v^2$ from $\mathscr{G}_2$. We denote this graph by $\mathscr{T}(\mathcal{C}^1_\mathrm{s}) \times \mathscr{T}(v^2)$. The graph contains three isomorphic, pairwise disjoint copies, namely $\mathscr{T}(\mathcal{C}^1_\mathrm{s}) \times c^2_{\mathrm{j}}$ for $c^2_{\mathrm{j}} \in \mathcal{N}_{v^2}$. It is generated by CV-type stabilizer generators $c^1_1 v^2, c^1_2 v^2$, and $c^1_3 v^2$ (see Example~\ref{ex:classical_and_stab_induced_TS} for a color-coded verification). These properties are formally stated in Lemma~\ref{lemma:stab_induced_TS_containg_classical_TS_CC} in Appendix~\ref{sec:lemma2}.
Analogous structural results for \emph{stabilizer-induced} TSs extend to LP codes and can be established by a similar argument combined with \cite[Lemma~4]{pradhan2025_TS_left_right_schedule_decoder}.
\begin{figure*}
   \begin{subfigure}{0.6\textwidth}
   \centering
       \input{Figures/classical_and_stab_induced_TS_CC/copies_of_CC_Type_classical_TS}
       \subcaption{The figure shows how the graph product of a TS $\mathscr{T}(\mathcal{C}^1_{\mathrm{s}})$ from $\mathscr{G}_1^\mathsf{T}$ with a degree-three variable node $v^2$ from $\mathscr{G}_2$ results in a graph that contains three isomorphic copies of the TS. These copies are $\mathscr{T}(\mathcal{C}^1_{\mathrm{s}})\times c^2_{\mathrm{j}}$ for $c^2_{\mathrm{j}} \in \mathcal{N}_{v^2}$. To reduce clutter, only the part of the resulting graph containing $Z$-stabilizers is shown. The same graph is induced in $\mathscr{G}_{\mathrm{z}}$ by the $X$-type stabilizer generators $c_1^1v^2, c_2^1v^2$, and $c_3^1v^2$; see Fig.~\ref{fig:stab_induced_TS_inside_bigger_TS} for two induced subgraphs. Since this structure is induced by CV-type stabilizer generators, it can be activated by hook errors introduced during their measurement.}
       \label{fig:multiple_copies_of_classical_TS_inside_stab_induced_TS_CC}
       \end{subfigure}
       \hfill
       \begin{subfigure}{0.3\textwidth}
       \input{Figures/classical_and_stab_induced_TS_CC/stabilizer_induced_TS}
           \subcaption{\emph{stabilizer-induced} subgraphs generated by the CV-type stabilizer generators $c^1_1v^2$ and $c^1_2v^2$ within $\mathscr{G}_{\mathrm{z}}$. Observe that all variable nodes in these induced subgraphs also appear in Fig.~\ref{fig:multiple_copies_of_classical_TS_inside_stab_induced_TS_CC}. The CC-type variable node represented by \tikzstyle{gcnode}=[circle,minimum size=0.5 cm,draw,green,scale=0.6]
\tikzstyle{sblnode}=[circle,minimum size=0.15 cm,draw,blue,scale=0.6]
\begin{tikzpicture}[every node/.style={scale=1}]
\begin{scope}[node distance=1.5 cm,semithick]
\node[gcnode](olr){};
\node[sblnode] (sbc) at (olr.center){};
\end{scope}
\end{tikzpicture} is shared by both induced subgraphs.}
           \label{fig:stab_induced_TS_inside_bigger_TS}
       \end{subfigure}
        \caption{Illustration of a \emph{stabilizer-induced} TS containing multiple copies of a classical TS from $\mathscr{G}_1^\mathsf{T}$.}
        \label{fig:classical_and_stab_included_TS_CC}
   \end{figure*}

Lemma~\ref{lemma:stab_induced_TS_containg_classical_TS} quantifies, for the families considered there, how many CNOT failures during syndrome measurement are required to activate the corresponding \emph{stabilizer-induced} TSs. As discussed in Section~\ref{sec:motivation_for_proposed_decoder}, addressing such large \emph{stabilizer-induced} structures via OSD can require increasing the OSD order with the number of inducing stabilizer generators, which is not viable for HGP/LP decoding.
Next, we illustrate a decoding approach that builds on good decoders for the parent LDPC codes and can resolve TSs like the one shown in Fig.~\ref{fig:multiple_copies_of_classical_TS_inside_stab_induced_TS_CC} without relying on OSD and without requiring optimization of the measurement schedule.

\subsubsection{Deriving QLDPC Decoders from Parent Code Decoders}
\label{sec:derived_decoders}

We first describe the structural properties of HGP/LP codes that allow us to derive decoders from those of the parent classical codes. The Tanner graph \(\mathscr{G}_\mathrm{z}\) of an HGP (respectively, LP) code contains \(|\mathcal{V}_1|\) (respectively, \(\lvert\underline{\mathcal{V}}_1\rvert\)) pairwise disjoint copies of \(\mathscr{G}_2\) (respectively, \(\underline{\mathscr{G}}_2\)); the union of their variable-node sets is exactly the set of VV-type variable nodes of \(\mathscr{G}_\mathrm{z}\). Since every check node in \(\mathscr{G}_\mathrm{z}\) is adjacent to at least one VV-type variable node, we can apply a parent-code decoder \(D_2\), originally defined on \(\mathscr{G}_2\), independently to each embedded copy to estimate errors on VV-type variable nodes. Equivalently, this defines a decoder that operates on the subgraph \(\mathscr{G}_\mathrm{z}(\mathcal{VV})\) induced by VV-type variable nodes, which we denote by \(D_{\mathrm{vv}}\).
Deriving decoders for HGP codes from decoders of their parent codes has also been considered in~\cite{quintavalle2022reshape}. The parent-code decoders used in~\cite{quintavalle2022reshape} are closer in spirit to bounded-distance decoding, in contrast to the iterative message-passing decoders considered here.

Similarly, starting from a decoder $D_1^\mathsf{T}$ defined on \(\mathscr{G}_1^{\mathsf{T}}\), we obtain a decoder \(D_{\mathrm{cc}}\) acting on the subgraph \(\mathscr{G}_\mathrm{z}(\mathcal{CC})\) induced by CC-type variable nodes. We view \(D_{\mathrm{vv}}(\cdot)\) and \(D_{\mathrm{cc}}(\cdot)\) as functions that take a syndrome vector as input and output error estimates on VV-type and CC-type variable nodes, respectively. The overall decoder for \(\mathscr{G}_\mathrm{z}\) alternates between running \(D_{\mathrm{vv}}\) on \(\mathscr{G}_\mathrm{z}(\mathcal{VV})\) and \(D_{\mathrm{cc}}\) on \(\mathscr{G}_\mathrm{z}(\mathcal{CC})\) until the estimated syndrome matches the measured syndrome; see Algorithm~\ref{alg:derived_decoders_from_parent_decoders}.

This alternating procedure is effective only under the following separation condition: if an error pattern has support on both VV-type and CC-type variable nodes, then no check node is adjacent to both a VV-type variable node in the support and a CC-type variable node in the support. For the TSs considered in this section, we enforce this condition via intermediate steps (Steps~6--12 and Steps~18--24), which flips selected variable-node estimates so that each check node is adjacent exclusively to VV-type errors or exclusively to CC-type errors, but never to both.

The next two lemmas formalize (i) when an error pattern with support exclusively on VV-type (or CC-type) variable nodes that matches the syndrome exists, and (ii) why the preprocessing steps in Algorithm~\ref{alg:derived_decoders_from_parent_decoders} guarantee its existence for the \emph{stabilizer-induced} TS families studied here. Following construction in \eqref{eq:LP_construction}, the entry of the error vector \(\mathbf{e}\) associated with the VV-type variable node \(v^1_{\mathrm{i}'}v^2_{\mathrm{i}}\) is the \((n_2(i'-1)+i)\)-th entry, and the entry associated with the CC-type variable node \(c^1_{\mathrm{j}'}c^2_{\mathrm{j}}\) is the \((n_1n_2 + (j'-1)m_2 + j)\)-th entry of \(\mathbf{e}\).

\begin{algorithm}[!ht]
\footnotesize
\DontPrintSemicolon
   \KwInput{Measured syndrome $\boldsymbol{\sigma}$, Maximum number of iterations $Maxiter$}
  \KwOutput{Error estimate $\hat{\mathbf{e}}$}
  \KwData{Tanner graph $\mathscr{G}_{\mathrm{z}}$, Decoders $D_{\mathrm{vv}}(\cdot)$ and $D_{\mathrm{cc}}(\cdot)$. }
  \KwVars{Intermediate error estimate on VV-type (respectively, CC-type) variable nodes $\tilde{\mathbf{e}}_{\mathrm{vv}}$ (respectively, $\tilde{\mathbf{e}}_{\mathrm{cc}}$),  Intermediate error estimate $\tilde{\mathbf{e}}=\begin{bmatrix}\tilde{\mathbf{e}}_{\mathrm{vv}} & \tilde{\mathbf{e}}_{\mathrm{cc}}\end{bmatrix}$, Mismatched syndrome $\tilde{\boldsymbol{\sigma}}$, Estimated syndrome $\hat{\boldsymbol{\sigma}}$}
  \KwInits{$\tilde{\boldsymbol{\sigma}} =\boldsymbol{\sigma}$, $\hat{\mathbf{e}}=\boldsymbol{0}$}
  
  Run decoder $D_{\mathrm{cc}}$ on Tanner graph $\mathscr{G}_{\mathrm{z}}(\mathcal{CC})$ to obtain error estimate $\hat{\mathbf{e}}_{\mathrm{cc}}$ given syndrome $\tilde{\boldsymbol{\sigma}}$, i.e., $\hat{\mathbf{e}}_{\mathrm{cc}}=D_{\mathrm{cc}}(\tilde{\boldsymbol{\sigma}})$.\;
   Update the estimated error $\hat{\mathbf{e}} \leftarrow \hat{\mathbf{e}} + \begin{bmatrix}
      \boldsymbol{0} & \hat{\mathbf{e}}_{\mathrm{cc}}
  \end{bmatrix}$.\;
   Compute the updated mismatched syndrome $\tilde{\boldsymbol{\sigma}} \leftarrow \tilde{\boldsymbol{\sigma}}+\mathbf{H}_{\mathrm{z}}\begin{bmatrix}
      \boldsymbol{0}\\
       \hat{\mathbf{e}}_{\mathrm{cc}}^\mathsf{T}
  \end{bmatrix}$.\;
  \If{$\tilde{\boldsymbol{\sigma}}=\boldsymbol{0}$}
  {
    Declare that decoding is successful and \textbf{Exit}
  }
  \For{ each $c^1_{\mathrm{\ell}'}\in \mathcal{C}^1$}
  {
    \If{there exist more than \(\left\lfloor\left\lvert\mathcal{N}_{c^1_{\mathrm{\ell}'}}\right\rvert/2\right\rfloor\) variable nodes \(v^1_{\mathrm{i}'} \in \mathcal{N}_{c^1_{\mathrm{\ell}'}}\) such that, for each of these nodes, there exists at least one VC-type check of the form \(v^1_{\mathrm{i}'} c^2_\mathrm{j}\), for some \(c^2_\mathrm{j} \in \mathcal{C}^2\), that is unsatisfied (i.e., $\tilde{\boldsymbol{\sigma}}(m_2(i'-1)+j)=1$)}
    {
        Randomly choose an unsatisfied VC-type check of the form $v_{\mathrm{k}'}^1c^2_{\mathrm{j}}$ for $v_{\mathrm{k}'}^1  \in \mathcal{N}_{c^1_{\mathrm{\ell}'}} $ and $c^2_{\mathrm{j}} \in \mathcal{C}^2$ \;
        \For{ each $c_{\mathrm{\ell}}^2\in \mathcal{C}^2$}
        {
            \If{$\tilde{\boldsymbol{\sigma}}(m_2(k'-1)+l)=1$}
            {
               Flip the error estimate on CC-type variable node $c^1_{\mathrm{\ell}'}c^2_{\mathrm{\ell}}$, i.e., $\hat{\mathbf{e}}(n_1n_2+(l'-1)m_2+l) \leftarrow \hat{\mathbf{e}}({n_1n_2+\mathbf{e}}((l'-1)m_2+l) \oplus 1. $ \;
               Update the syndrome vector accordingly $\tilde{\boldsymbol{\sigma}}=\tilde{\boldsymbol{\sigma}}\oplus \mathbf{H}(:,n_1n_2+(l'-1)m_2+l)$.\;
            }
        }
    }
  }
  Run decoder $D_{\mathrm{vv}}$ on Tanner graph $\mathscr{G}_{\mathrm{z}}(\mathcal{VV})$ to obtain error estimate $\hat{\mathbf{e}}_{\mathrm{vv}}$ given syndrome $\tilde{\boldsymbol{\sigma}}$, i.e., $\hat{\mathbf{e}}_{\mathrm{vv}}=D_{\mathrm{vv}}(\tilde{\boldsymbol{\sigma}})$.\;
  Update the estimated error $\hat{\mathbf{e}} \leftarrow \hat{\mathbf{e}} + \begin{bmatrix}
      \hat{\mathbf{e}}_{\mathrm{vv}} & \boldsymbol{0}
  \end{bmatrix}$.\;
  Compute the updated mismatched syndrome $\tilde{\boldsymbol{\sigma}} \leftarrow \tilde{\boldsymbol{\sigma}}+\mathbf{H}_{\mathrm{z}}\begin{bmatrix}
      \hat{\mathbf{e}}_{\mathrm{vv}}^\mathsf{T}\\ \boldsymbol{0}
  \end{bmatrix}$.\;
   \If{$\tilde{\boldsymbol{\sigma}}=\boldsymbol{0}$}
  {
    Declare that decoding is successful and \textbf{Exit}
  }
 \For{ each $v^2_{\mathrm{k}}\in \mathcal{V}^2$}
  {
    \If{there exist more than \(\left\lfloor\left\lvert\mathcal{N}_{v^2_{\mathrm{k}}}\right\rvert/2\right\rfloor\) check nodes \(c^2_{\mathrm{j}} \in \mathcal{N}_{v^2_{\mathrm{k}}}\) such that, for each of these nodes, there exists at least one VC-type check of the form \(v^1_{\mathrm{i}'} c^2_\mathrm{j}\), for some \(v^1_{\mathrm{i}'} \in \mathcal{V}^1\), that is unsatisfied (i.e., $\tilde{\boldsymbol{\sigma}}(m_2(i'-1)+j)=1$)}
    {
         Randomly choose an unsatisfied VC-type check of the form $v_{\mathrm{i}'}^1c^2_{\mathrm{\ell}}$ for $c_{\mathrm{\ell}}^2  \in \mathcal{N}_{v^2_{\mathrm{k}}} $ and $v^1_{\mathrm{i}'} \in \mathcal{V}^1$ \;
        \For{each $v^1_{\mathrm{k}'} \in \mathcal{V}^1$}
        {
            \If{$\tilde{\boldsymbol{\sigma}}(m_2(k'-1)+l)=1$}
            {
                Flip error estimate on $v^1_{\mathrm{k}'}v^2_{\mathrm{k}}$, i.e.,$\hat{\mathbf{e}}(n_2(k'-1)+k)\leftarrow\hat{\mathbf{e}}(n_2(k'-1)+k)\oplus 1.$\;
                Update syndrome $\tilde{\boldsymbol{\sigma}}=\tilde{\boldsymbol{\sigma}}\oplus \mathbf{H}((:,(n_2(k'-1)+k).$\;
            }
        }
    }
   
  }
   Repeat Steps $1$ to $5$.\;
  \caption{Decoder that resolves \emph{stabilizer-induced} TSs when the TS is the graph product of a single check (respectively, a single variable) node from one parent code and a TS from the other parent code.}
  \label{alg:derived_decoders_from_parent_decoders}
\end{algorithm}

Figure~\ref{fig:decoding_illustration_of_stab_induced_TS_with_multiple_classical_TS} illustrates, by means of an example, the operation of the decoder described in Algorithm~\ref{alg:derived_decoders_from_parent_decoders}. These ideas are formalized in Lemmas~\ref{lemma:error_exclusively_on_VV_condition} and~\ref{lemma:decoder_resolving_classical_plus_stab_TS}. In particular, Part~2 of Lemma~\ref{lemma:error_exclusively_on_VV_condition} shows that Steps~$7$ and~$19$ are sufficient to determine whether there exists an error pattern exclusively on VV-type or CC-type variable nodes that is consistent with the measured syndrome. When such an error pattern does not exist, Lemma~\ref{lemma:decoder_resolving_classical_plus_stab_TS} shows that the preprocessing steps in Steps~$8$--$12$ or Steps~$20$--$24$ ensure the existence of such an error pattern. The decoder described in Algorithm~\ref{alg:derived_decoders_from_parent_decoders} then resolves the TSs considered in Lemma~\ref{lemma:stab_induced_TS_containg_classical_TS} and Lemma~\ref{lemma:stab_induced_TS_containg_classical_TS_CC}.

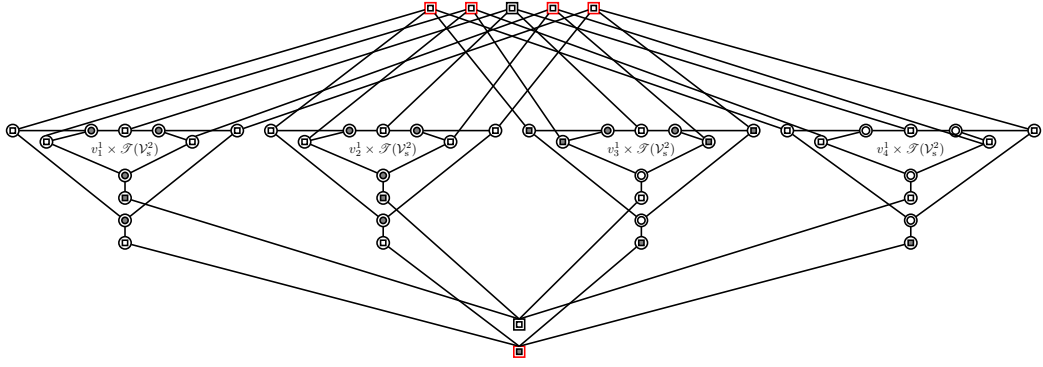
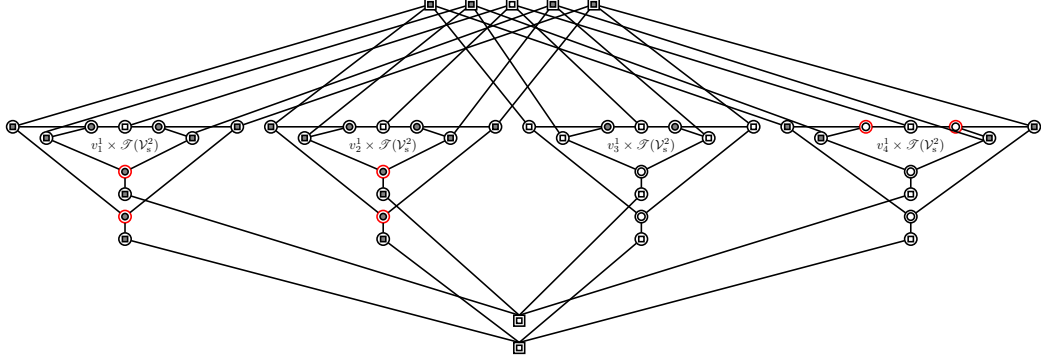
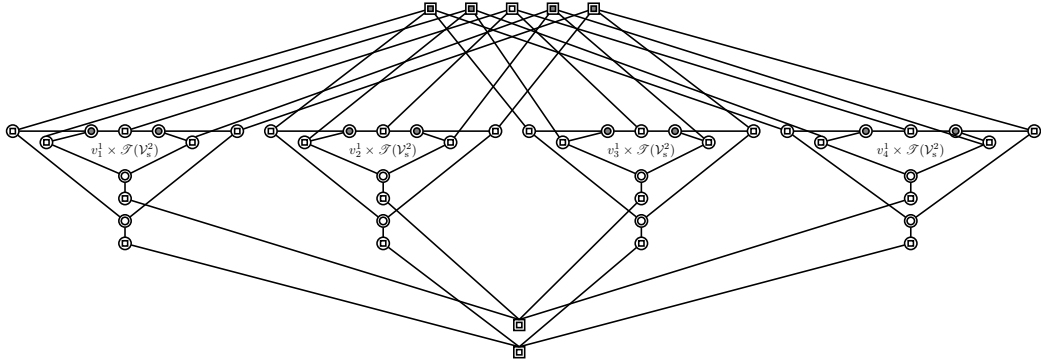
\begin{figure*}
    \begin{subfigure}{0.99\textwidth}
    \centering
    \tikzstyle{cnode}=[circle,minimum size=0.55 cm,draw,black,scale=0.33]
\tikzstyle{scnode}=[circle,minimum size=0.15 cm,draw,black,scale=0.33]
\tikzstyle{sb_cnode}=[circle,minimum size=0.15 cm,draw,black,fill=gray,scale=0.28]
\tikzstyle{rnode}=[rectangle,draw,minimum width=0.4 cm,minimum height=0.4cm,outer sep=0pt,scale=0.4]
\tikzstyle{rrnode}=[rectangle,red,draw,minimum width=0.4 cm,minimum height=0.4cm,outer sep=0pt,scale=0.4]
\tikzstyle{srnode}=[rectangle,draw,minimum width=0.1 cm,minimum height=0.1cm,outer sep=0pt,scale=0.33]
\tikzstyle{sb_rnode}=[rectangle,draw,minimum width=0.1 cm,minimum height=0.1cm,outer sep=0pt,fill=gray,scale=0.33]

\tikzstyle{hexagon}=[regular polygon, regular polygon sides=6, shape aspect=0.5, minimum width=1.5cm, minimum height=1cm, draw,shape border rotate=60,scale=0.33]
\tikzstyle{Ellipse}=[ellipse,minimum width=3cm,minimum height=1.25cm,draw,scale=0.33]
\tikzstyle{smallEllipse}=[ellipse,minimum width=1.75cm,minimum height=1.25cm,draw,scale=0.33]
\definecolor{color1}{rgb}{1,0.2,0.3}
\definecolor{color2}{rgb}{0.4,0.5,0.7}
\definecolor{color3}{rgb}{0.1,0.8,0.5}
\definecolor{color4}{rgb}{0.5,0.3,1}
\definecolor{color5}{rgb}{0.5,1,1}
\definecolor{color6}{rgb}{0.8,0.3,0.6}
\definecolor{color7}{rgb}{0.6,0.4,0.3}

\begin{tikzpicture}[every node/.style={scale=0.9}]
\begin{scope}[node distance=1.5 cm,semithick]

\node[cnode](cv3){};
\node[srnode](cv3_in) at (cv3.center) {};
\node[cnode](vv1)[right of=cv3,xshift=2cm]{};
\node[sb_cnode](vv1_in) at (vv1.center) {};
\node[cnode](cv2)[right of=cv3,yshift=-0.5cm]{};
\node[srnode](cv2_in) at (cv2.center) {};
\node[cnode](cv1)[right of=vv1,yshift=0 cm]{};
\node[srnode](cv2_in) at (cv1.center) {};
\node[cnode](vv2)[right of=cv1,xshift=0cm]{};
\node[sb_cnode](vv2_in) at (vv2.center) {};
\node[cnode](cv4)[right of=vv2,yshift=-0.5cm]{};
\node[srnode](cv4_in) at (cv4.center) {};
\node[cnode](cv5)[right of=vv2,xshift=2cm]{};
\node[srnode](cv5_in) at (cv5.center) {};
\node[cnode] (vv3) [below of=cv1,yshift=-0.5cm]{};
\node[sb_cnode](vv3_in) at (vv3.center) {};
\node[cnode] (vv4) [below of=cv1,yshift=-2.5cm]{};
\node[sb_cnode](vv4_in) at (vv4.center) {};
\node[cnode] (cv6) [below of=vv3,yshift=0.5 cm]{};
\node[sb_rnode](cv6_in) at (cv6.center) {};
\node[cnode] (cv7) [below of=vv4,yshift=0.5 cm]{};
\node[srnode](cv7_in) at (cv7.center) {};
\draw (vv1)--(cv1);
\draw (vv1)--(cv2);
\draw (vv1)--(cv3);
\draw (vv2)--(cv1);
\draw (vv2)--(cv4);
\draw (vv2)--(cv5);
\draw (vv3)--(cv2);
\draw (vv3)--(cv4);
\draw (vv4)--(cv3);
\draw (vv4)--(cv5);
\draw (vv3)--(cv6);
\draw (vv4)--(cv7);


\node[cnode](cv10)[right of =cv5]{};
\node[srnode](cv10_in) at (cv10.center) {};
\node[cnode](vv5)[right of=cv10,xshift=2cm]{};
\node[sb_cnode](vv5_in) at (vv5.center) {};
\node[cnode](cv9)[right of=cv10,yshift=-0.5cm]{};
\node[srnode](cv9_in) at (cv9.center) {};
\node[cnode](cv8)[right of=vv5,yshift=0 cm]{};
\node[srnode](cv8_in) at (cv8.center) {};
\node[cnode](vv6)[right of=cv8,xshift=0cm]{};
\node[sb_cnode](vv6_in) at (vv6.center) {};
\node[cnode](cv11)[right of=vv6,yshift=-0.5cm]{};
\node[srnode](cv11_in) at (cv11.center) {};
\node[cnode](cv12)[right of=vv6,xshift=2cm]{};
\node[srnode](cv12_in) at (cv12.center) {};
\node[cnode] (vv7) [below of=cv8,yshift=-0.5cm]{};
\node[sb_cnode](vv7_in) at (vv7.center) {};
\node[cnode] (vv8) [below of=cv8,yshift=-2.5cm]{};
\node[sb_cnode](vv8_in) at (vv8.center) {};
\node[cnode] (cv13) [below of=vv7,yshift=0.5 cm]{};
\node[sb_rnode](cv13_in) at (cv13.center) {};
\node[cnode] (cv14) [below of=vv8,yshift=0.5 cm]{};
\node[srnode](cv14_in) at (cv14.center) {};
\draw (vv5)--(cv8);
\draw (vv5)--(cv9);
\draw (vv5)--(cv10);
\draw (vv6)--(cv8);
\draw (vv6)--(cv11);
\draw (vv6)--(cv12);
\draw (vv7)--(cv9);
\draw (vv7)--(cv11);
\draw (vv8)--(cv10);
\draw (vv8)--(cv12);
\draw (vv7)--(cv13);
\draw (vv8)--(cv14);

\node[cnode](cv17)[right of =cv12]{};
\node[sb_rnode](cv17_in) at (cv17.center) {};
\node[cnode](vv9)[right of=cv17,xshift=2cm]{};
\node[sb_cnode](vv9_in) at (vv9.center) {};
\node[cnode](cv16)[right of=cv17,yshift=-0.5cm]{};
\node[sb_rnode](cv16_in) at (cv16.center) {};
\node[cnode](cv15)[right of=vv9,yshift=0 cm]{};
\node[srnode](cv15_in) at (cv15.center) {};
\node[cnode](vv10)[right of=cv15,xshift=0cm]{};
\node[sb_cnode](vv10_in) at (vv10.center) {};
\node[cnode](cv18)[right of=vv10,yshift=-0.5cm]{};
\node[sb_rnode](cv18_in) at (cv18.center) {};
\node[cnode](cv19)[right of=vv10,xshift=2cm]{};
\node[sb_rnode](cv19_in) at (cv19.center) {};
\node[cnode] (vv11) [below of=cv15,yshift=-0.5cm]{};
\node[scnode](vv11_in) at (vv11.center) {};
\node[cnode] (vv12) [below of=cv15,yshift=-2.5cm]{};
\node[scnode](vv12_in) at (vv12.center) {};
\node[cnode] (cv20) [below of=vv11,yshift=0.5 cm]{};
\node[srnode](cv20_in) at (cv20.center) {};
\node[cnode] (cv21) [below of=vv12,yshift=0.5 cm]{};
\node[sb_rnode](cv21_in) at (cv21.center) {};
\draw (vv9)--(cv15);
\draw (vv9)--(cv16);
\draw (vv9)--(cv17);
\draw (vv10)--(cv15);
\draw (vv10)--(cv18);
\draw (vv10)--(cv19);
\draw (vv11)--(cv16);
\draw (vv11)--(cv18);
\draw (vv12)--(cv17);
\draw (vv12)--(cv19);
\draw (vv11)--(cv20);
\draw (vv12)--(cv21);

\node[cnode](cv24)[right of =cv19]{};
\node[srnode](cv24_in) at (cv24.center) {};
\node[cnode](vv13)[right of=cv24,xshift=2cm]{};
\node[scnode](vv13_in) at (vv13.center) {};
\node[cnode](cv23)[right of=cv24,yshift=-0.5cm]{};
\node[srnode](cv23_in) at (cv23.center) {};
\node[cnode](cv22)[right of=vv13,xshift=0.5 cm]{};
\node[srnode](cv22_in) at (cv22.center) {};
\node[cnode](vv14)[right of=cv22,xshift=0.5cm]{};
\node[scnode](vv14_in) at (vv14.center) {};
\node[cnode](cv25)[right of=vv14,yshift=-0.5cm]{};
\node[srnode](cv25_in) at (cv25.center) {};
\node[cnode](cv26)[right of=vv14,xshift=2cm]{};
\node[srnode](cv26_in) at (cv26.center) {};
\node[cnode] (vv15) [below of=cv22,yshift=-0.5cm]{};
\node[scnode](vv15_in) at (vv15.center) {};
\node[cnode] (vv16) [below of=cv22,yshift=-2.5cm]{};
\node[scnode](vv16_in) at (vv16.center) {};
\node[cnode] (cv27) [below of=vv15,yshift=0.5 cm]{};
\node[srnode](cv27_in) at (cv27.center) {};
\node[cnode] (cv28) [below of=vv16,yshift=0.5 cm]{};
\node[sb_rnode](cv28_in) at (cv28.center) {};
\draw (vv13)--(cv22);
\draw (vv13)--(cv23);
\draw (vv13)--(cv24);
\draw (vv14)--(cv22);
\draw (vv14)--(cv25);
\draw (vv14)--(cv26);
\draw (vv15)--(cv23);
\draw (vv15)--(cv25);
\draw (vv16)--(cv24);
\draw (vv16)--(cv26);
\draw (vv15)--(cv27);
\draw (vv16)--(cv28);

\node[rnode] (cc1)[below of=cv14,xshift=5 cm,yshift=-1.5cm]{};
\node[srnode] (cc1_in) at (cc1.center){};
\node[rrnode] (cc2)[below of=cc1,yshift=0.5 cm]{};
\node[sb_rnode] (cc2_in) at (cc2.center){};
\draw (cv6) -- (cc1.90);
\draw (cv7) -- (cc2.90);
\draw (cv13) -- (cc1.90);
\draw (cv14) -- (cc2.90);
\draw (cv20) -- (cc1.90);
\draw (cv21) -- (cc2.90);
\draw (cv27) -- (cc1.90);
\draw (cv28) -- (cc2.90);
\node[rrnode] (cc3)[above of=vv6,xshift=0.5 cm,yshift=3cm]{};
\node[srnode](cc3_in) at (cc3.center) {};
\node[rrnode] (cc4)[right of=cc3,xshift=0 cm,yshift=0cm]{};
\node[srnode](cc4_in) at (cc4.center) {};
\node[rnode] (cc5)[right of=cc4,xshift=0 cm,yshift=0cm]{};
\node[srnode](cc5_in) at (cc5.center) {};
\node[rrnode] (cc6)[right of=cc5,xshift=0 cm,yshift=0cm]{};
\node[srnode](cc6_in) at (cc6.center) {};
\node[rrnode] (cc7)[right of=cc6,xshift=0 cm,yshift=0cm]{};
\node[srnode](cc7_in) at (cc7.center) {};
\node[scale=0.55] [below of =cv1,yshift=0 cm,yshift=1cm]{$v^1_1\times\mathscr{T}(\mathcal{V}^2_\mathrm{s})$};
\node[scale=0.55] [below of =cv8,yshift=0 cm,yshift=1cm]{$v^1_2\times\mathscr{T}(\mathcal{V}^2_\mathrm{s})$};
\node[scale=0.55] [below of =cv15,yshift=0 cm,yshift=1cm]{$v^1_3\times\mathscr{T}(\mathcal{V}^2_\mathrm{s})$};
\node[scale=0.55] [below of =cv22,yshift=0 cm,yshift=1cm]{$v^1_4\times\mathscr{T}(\mathcal{V}^2_\mathrm{s})$};
\draw (cv3) --(cc3);
\draw (cv10) --(cc3);
\draw (cv17) --(cc3);
\draw (cv24) --(cc3);
\draw (cv2.90) --(cc4);
\draw (cv9) --(cc4);
\draw (cv16) --(cc4);
\draw (cv23.90) --(cc4);
\draw (cv1) --(cc5);
\draw (cv8) --(cc5);
\draw (cv15) --(cc5);
\draw (cv22) --(cc5);

\draw (cv4.60) --(cc6);
\draw (cv11) --(cc6);
\draw (cv18) --(cc6);
\draw (cv25.200) --(cc6);
\draw (cv5) --(cc7);
\draw (cv12) --(cc7);
\draw (cv19) --(cc7);
\draw (cv26) --(cc7);

\end{scope}    
\end{tikzpicture}
    \subcaption{Iteration~$1$: In the first two copies of the classical TS (from the left), all VV-type variable nodes are in error; in the third copy, only two VV-type variable nodes are in error; in the fourth copy, none are in error. The decoder is initialized with $\hat{\mathbf{e}}=\mathbf{0}$, so the mismatched and measured syndromes coincide. For every $v^1_{\mathrm{i}'}\in\mathcal{N}_{c^1}$ there is at least one unsatisfied check in $v^1_{\mathrm{i}'} \times \mathscr{T}(\mathcal{V}^2_{\mathrm{s}})$, so the condition in Step~7 of Algorithm~\ref{alg:derived_decoders_from_parent_decoders} is not satisfied. The decoder therefore selects a neighbor of $c^1$ (here $v^1_3$) and, following Steps~8--12, flips the CC-type variable nodes adjacent to the unsatisfied checks $v^1_3 c^2_{\mathrm{j}}$ in the subgraph $v^1_3 \times \mathscr{T}(\mathcal{V}^2_{\mathrm{s}})$. The CC-type variable nodes to be flipped are outlined in red.}
    \label{fig:iteration1}
    \end{subfigure}
    \hfill \begin{subfigure}{0.99\textwidth}
    \centering
    \tikzstyle{cnode}=[circle,minimum size=0.55 cm,draw,black,scale=0.33]
\tikzstyle{rcnode}=[circle,minimum size=0.55 cm,draw,red,scale=0.33]
\tikzstyle{scnode}=[circle,minimum size=0.15 cm,draw,black,scale=0.33]
\tikzstyle{sb_cnode}=[circle,minimum size=0.15 cm,draw,black,fill=gray,scale=0.28]
\tikzstyle{rnode}=[rectangle,draw,minimum width=0.4 cm,minimum height=0.4cm,outer sep=0pt,scale=0.4]
\tikzstyle{srnode}=[rectangle,draw,minimum width=0.1 cm,minimum height=0.1cm,outer sep=0pt,scale=0.33]
\tikzstyle{sb_rnode}=[rectangle,draw,minimum width=0.1 cm,minimum height=0.1cm,outer sep=0pt,fill=gray,scale=0.33]

\tikzstyle{hexagon}=[regular polygon, regular polygon sides=6, shape aspect=0.5, minimum width=1.5cm, minimum height=1cm, draw,shape border rotate=60,scale=0.33]
\tikzstyle{Ellipse}=[ellipse,minimum width=3cm,minimum height=1.25cm,draw,scale=0.33]
\tikzstyle{smallEllipse}=[ellipse,minimum width=1.75cm,minimum height=1.25cm,draw,scale=0.33]
\definecolor{color1}{rgb}{1,0.2,0.3}
\definecolor{color2}{rgb}{0.4,0.5,0.7}
\definecolor{color3}{rgb}{0.1,0.8,0.5}
\definecolor{color4}{rgb}{0.5,0.3,1}
\definecolor{color5}{rgb}{0.5,1,1}
\definecolor{color6}{rgb}{0.8,0.3,0.6}
\definecolor{color7}{rgb}{0.6,0.4,0.3}

\begin{tikzpicture}[every node/.style={scale=0.9}]
\begin{scope}[node distance=1.5 cm,semithick]

\node[cnode](cv3){};
\node[sb_rnode](cv3_in) at (cv3.center) {};
\node[cnode](vv1)[right of=cv3,xshift=2cm]{};
\node[sb_cnode](vv1_in) at (vv1.center) {};
\node[cnode](cv2)[right of=cv3,yshift=-0.5cm]{};
\node[sb_rnode](cv2_in) at (cv2.center) {};
\node[cnode](cv1)[right of=vv1,yshift=0 cm]{};
\node[srnode](cv1_in) at (cv1.center) {};
\node[cnode](vv2)[right of=cv1,xshift=0cm]{};
\node[sb_cnode](vv2_in) at (vv2.center) {};
\node[cnode](cv4)[right of=vv2,yshift=-0.5cm]{};
\node[sb_rnode](cv4_in) at (cv4.center) {};
\node[cnode](cv5)[right of=vv2,xshift=2cm]{};
\node[sb_rnode](cv5_in) at (cv5.center) {};
\node[rcnode] (vv3) [below of=cv1,yshift=-0.5cm]{};
\node[sb_cnode](vv3_in) at (vv3.center) {};
\node[rcnode] (vv4) [below of=cv1,yshift=-2.5cm]{};
\node[sb_cnode](vv4_in) at (vv4.center) {};
\node[cnode] (cv6) [below of=vv3,yshift=0.5 cm]{};
\node[sb_rnode](cv6_in) at (cv6.center) {};
\node[cnode] (cv7) [below of=vv4,yshift=0.5 cm]{};
\node[sb_rnode](cv7_in) at (cv7.center) {};
\draw (vv1)--(cv1);
\draw (vv1)--(cv2);
\draw (vv1)--(cv3);
\draw (vv2)--(cv1);
\draw (vv2)--(cv4);
\draw (vv2)--(cv5);
\draw (vv3)--(cv2);
\draw (vv3)--(cv4);
\draw (vv4)--(cv3);
\draw (vv4)--(cv5);
\draw (vv3)--(cv6);
\draw (vv4)--(cv7);


\node[cnode](cv10)[right of =cv5]{};
\node[sb_rnode](cv10_in) at (cv10.center) {};
\node[cnode](vv5)[right of=cv10,xshift=2cm]{};
\node[sb_cnode](vv5_in) at (vv5.center) {};
\node[cnode](cv9)[right of=cv10,yshift=-0.5cm]{};
\node[sb_rnode](cv9_in) at (cv9.center) {};
\node[cnode](cv8)[right of=vv5,yshift=0 cm]{};
\node[srnode](cv8_in) at (cv8.center) {};
\node[cnode](vv6)[right of=cv8,xshift=0cm]{};
\node[sb_cnode](vv6_in) at (vv6.center) {};
\node[cnode](cv11)[right of=vv6,yshift=-0.5cm]{};
\node[sb_rnode](cv11_in) at (cv11.center) {};
\node[cnode](cv12)[right of=vv6,xshift=2cm]{};
\node[sb_rnode](cv12_in) at (cv12.center) {};
\node[rcnode] (vv7) [below of=cv8,yshift=-0.5cm]{};
\node[sb_cnode](vv7_in) at (vv7.center) {};
\node[rcnode] (vv8) [below of=cv8,yshift=-2.5cm]{};
\node[sb_cnode](vv8_in) at (vv8.center) {};
\node[cnode] (cv13) [below of=vv7,yshift=0.5 cm]{};
\node[sb_rnode](cv13_in) at (cv13.center) {};
\node[cnode] (cv14) [below of=vv8,yshift=0.5 cm]{};
\node[sb_rnode](cv14_in) at (cv14.center) {};
\draw (vv5)--(cv8);
\draw (vv5)--(cv9);
\draw (vv5)--(cv10);
\draw (vv6)--(cv8);
\draw (vv6)--(cv11);
\draw (vv6)--(cv12);
\draw (vv7)--(cv9);
\draw (vv7)--(cv11);
\draw (vv8)--(cv10);
\draw (vv8)--(cv12);
\draw (vv7)--(cv13);
\draw (vv8)--(cv14);

\node[cnode](cv17)[right of =cv12]{};
\node[srnode](cv17_in) at (cv17.center) {};
\node[cnode](vv9)[right of=cv17,xshift=2cm]{};
\node[sb_cnode](vv9_in) at (vv9.center) {};
\node[cnode](cv16)[right of=cv17,yshift=-0.5cm]{};
\node[srnode](cv16_in) at (cv16.center) {};
\node[cnode](cv15)[right of=vv9,yshift=0 cm]{};
\node[srnode](cv15_in) at (cv15.center) {};
\node[cnode](vv10)[right of=cv15,xshift=0cm]{};
\node[sb_cnode](vv10_in) at (vv10.center) {};
\node[cnode](cv18)[right of=vv10,yshift=-0.5cm]{};
\node[srnode](cv18_in) at (cv18.center) {};
\node[cnode](cv19)[right of=vv10,xshift=2cm]{};
\node[srnode](cv19_in) at (cv19.center) {};
\node[cnode] (vv11) [below of=cv15,yshift=-0.5cm]{};
\node[scnode](vv11_in) at (vv11.center) {};
\node[cnode] (vv12) [below of=cv15,yshift=-2.5cm]{};
\node[scnode](vv12_in) at (vv12.center) {};
\node[cnode] (cv20) [below of=vv11,yshift=0.5 cm]{};
\node[srnode](cv20_in) at (cv20.center) {};
\node[cnode] (cv21) [below of=vv12,yshift=0.5 cm]{};
\node[srnode](cv21_in) at (cv21.center) {};
\draw (vv9)--(cv15);
\draw (vv9)--(cv16);
\draw (vv9)--(cv17);
\draw (vv10)--(cv15);
\draw (vv10)--(cv18);
\draw (vv10)--(cv19);
\draw (vv11)--(cv16);
\draw (vv11)--(cv18);
\draw (vv12)--(cv17);
\draw (vv12)--(cv19);
\draw (vv11)--(cv20);
\draw (vv12)--(cv21);

\node[cnode](cv24)[right of =cv19]{};
\node[sb_rnode](cv24_in) at (cv24.center) {};
\node[rcnode](vv13)[right of=cv24,xshift=2cm]{};
\node[scnode](vv13_in) at (vv13.center) {};
\node[cnode](cv23)[right of=cv24,yshift=-0.5cm]{};
\node[sb_rnode](cv23_in) at (cv23.center) {};
\node[cnode](cv22)[right of=vv13,xshift=0.5 cm]{};
\node[srnode](cv22_in) at (cv22.center) {};
\node[rcnode](vv14)[right of=cv22,xshift=0.5cm]{};
\node[scnode](vv14_in) at (vv14.center) {};
\node[cnode](cv25)[right of=vv14,yshift=-0.5cm]{};
\node[sb_rnode](cv25_in) at (cv25.center) {};
\node[cnode](cv26)[right of=vv14,xshift=2cm]{};
\node[sb_rnode](cv26_in) at (cv26.center) {};
\node[cnode] (vv15) [below of=cv22,yshift=-0.5cm]{};
\node[scnode](vv15_in) at (vv15.center) {};
\node[cnode] (vv16) [below of=cv22,yshift=-2.5cm]{};
\node[scnode](vv16_in) at (vv16.center) {};
\node[cnode] (cv27) [below of=vv15,yshift=0.5 cm]{};
\node[srnode](cv27_in) at (cv27.center) {};
\node[cnode] (cv28) [below of=vv16,yshift=0.5 cm]{};
\node[srnode](cv28_in) at (cv28.center) {};
\draw (vv13)--(cv22);
\draw (vv13)--(cv23);
\draw (vv13)--(cv24);
\draw (vv14)--(cv22);
\draw (vv14)--(cv25);
\draw (vv14)--(cv26);
\draw (vv15)--(cv23);
\draw (vv15)--(cv25);
\draw (vv16)--(cv24);
\draw (vv16)--(cv26);
\draw (vv15)--(cv27);
\draw (vv16)--(cv28);

\node[rnode] (cc1)[below of=cv14,xshift=5 cm,yshift=-1.5cm]{};
\node[srnode] (cc1_in) at (cc1.center){};
\node[rnode] (cc2)[below of=cc1,yshift=0.5 cm]{};
\node[srnode] (cc2_in) at (cc2.center){};
\draw (cv6) -- (cc1.90);
\draw (cv7) -- (cc2.90);
\draw (cv13) -- (cc1.90);
\draw (cv14) -- (cc2.90);
\draw (cv20) -- (cc1.90);
\draw (cv21) -- (cc2.90);
\draw (cv27) -- (cc1.90);
\draw (cv28) -- (cc2.90);
\node[rnode] (cc3)[above of=vv6,xshift=0.5 cm,yshift=3cm]{};
\node[sb_rnode](cc3_in) at (cc3.center) {};
\node[rnode] (cc4)[right of=cc3,xshift=0 cm,yshift=0cm]{};
\node[sb_rnode](cc4_in) at (cc4.center) {};
\node[rnode] (cc5)[right of=cc4,xshift=0 cm,yshift=0cm]{};
\node[srnode](cc5_in) at (cc5.center) {};
\node[rnode] (cc6)[right of=cc5,xshift=0 cm,yshift=0cm]{};
\node[sb_rnode](cc6_in) at (cc6.center) {};
\node[rnode] (cc7)[right of=cc6,xshift=0 cm,yshift=0cm]{};
\node[sb_rnode](cc7_in) at (cc7.center) {};

\draw (cv3) --(cc3);
\draw (cv10) --(cc3);
\draw (cv17) --(cc3);
\draw (cv24) --(cc3);
\draw (cv2.90) --(cc4);
\draw (cv9) --(cc4);
\draw (cv16) --(cc4);
\draw (cv23.90) --(cc4);
\draw (cv1) --(cc5);
\draw (cv8) --(cc5);
\draw (cv15) --(cc5);
\draw (cv22) --(cc5);

\draw (cv4.60) --(cc6);
\draw (cv11) --(cc6);
\draw (cv18) --(cc6);
\draw (cv25.200) --(cc6);
\draw (cv5) --(cc7);
\draw (cv12) --(cc7);
\draw (cv19) --(cc7);
\draw (cv26) --(cc7);
\node[scale=0.55] [below of =cv1,yshift=0 cm,yshift=1cm]{$v^1_1\times\mathscr{T}(\mathcal{V}^2_\mathrm{s})$};
\node[scale=0.55] [below of =cv8,yshift=0 cm,yshift=1cm]{$v^1_2\times\mathscr{T}(\mathcal{V}^2_\mathrm{s})$};
\node[scale=0.55] [below of =cv15,yshift=0 cm,yshift=1cm]{$v^1_3\times\mathscr{T}(\mathcal{V}^2_\mathrm{s})$};
\node[scale=0.55] [below of =cv22,yshift=0 cm,yshift=1cm]{$v^1_4\times\mathscr{T}(\mathcal{V}^2_\mathrm{s})$};

\end{scope}    
\end{tikzpicture}
    \subcaption{Iteration~$2$: In the first iteration, only CC-type variable nodes connected to unsatisfied checks in the structure \(v^1_3 \times \mathscr{T}(\mathcal{V}^2_{\mathrm{s}})\) are flipped. Except for the one in the lower part, these nodes are not actually erroneous and thus become part of the mismatched error (inner black shapes), together with the VV-type variable nodes that are part of mismatched error in the first iteration. Flipping them toggles the status of their neighboring checks: previously satisfied checks become unsatisfied and vice versa. After this step, only the third copy of the classical TS has no unsatisfied checks. By Part~1 of Lemma~\ref{lemma:error_exclusively_on_VV_condition}, this implies the existence of an error pattern supported solely on VV-type variable nodes whose syndrome equals the mismatched syndrome. In the second iteration, a decoder that can correct any error pattern within a classical TS is applied to the VV-type variable nodes. Consequently, for each structure \(v^1_{\mathrm{i}'} \times \mathscr{T}(\mathcal{V}_{\mathrm{s}}^2)\), with \(v^1_{\mathrm{i}'} \in \mathcal{N}_{c^1}\), the decoder estimates an error pattern on VV-type nodes (indicated by red borders) whose support matches the mismatched syndrome within that structure.}
    \label{fig:iteration2}
    \end{subfigure}
    \hfill \begin{subfigure}{0.99\textwidth}
    \centering
    \tikzstyle{cnode}=[circle,minimum size=0.55 cm,draw,black,scale=0.33]
\tikzstyle{scnode}=[circle,minimum size=0.15 cm,draw,black,scale=0.33]
\tikzstyle{sb_cnode}=[circle,minimum size=0.15 cm,draw,black,fill=gray,scale=0.28]
\tikzstyle{rnode}=[rectangle,draw,minimum width=0.4 cm,minimum height=0.4cm,outer sep=0pt,scale=0.4]
\tikzstyle{srnode}=[rectangle,draw,minimum width=0.1 cm,minimum height=0.1cm,outer sep=0pt,scale=0.33]
\tikzstyle{sb_rnode}=[rectangle,draw,minimum width=0.1 cm,minimum height=0.1cm,outer sep=0pt,fill=gray,scale=0.33]

\tikzstyle{hexagon}=[regular polygon, regular polygon sides=6, shape aspect=0.5, minimum width=1.5cm, minimum height=1cm, draw,shape border rotate=60,scale=0.33]
\tikzstyle{Ellipse}=[ellipse,minimum width=3cm,minimum height=1.25cm,draw,scale=0.33]
\tikzstyle{smallEllipse}=[ellipse,minimum width=1.75cm,minimum height=1.25cm,draw,scale=0.33]
\definecolor{color1}{rgb}{1,0.2,0.3}
\definecolor{color2}{rgb}{0.4,0.5,0.7}
\definecolor{color3}{rgb}{0.1,0.8,0.5}
\definecolor{color4}{rgb}{0.5,0.3,1}
\definecolor{color5}{rgb}{0.5,1,1}
\definecolor{color6}{rgb}{0.8,0.3,0.6}
\definecolor{color7}{rgb}{0.6,0.4,0.3}

\begin{tikzpicture}[every node/.style={scale=0.9}]
\begin{scope}[node distance=1.5 cm,semithick]

\node[cnode](cv3){};
\node[srnode](cv3_in) at (cv3.center) {};
\node[cnode](vv1)[right of=cv3,xshift=2cm]{};
\node[sb_cnode](vv1_in) at (vv1.center) {};
\node[cnode](cv2)[right of=cv3,yshift=-0.5cm]{};
\node[srnode](cv2_in) at (cv2.center) {};
\node[cnode](cv1)[right of=vv1,yshift=0 cm]{};
\node[srnode](cv2_in) at (cv1.center) {};
\node[cnode](vv2)[right of=cv1,xshift=0cm]{};
\node[sb_cnode](vv2_in) at (vv2.center) {};
\node[cnode](cv4)[right of=vv2,yshift=-0.5cm]{};
\node[srnode](cv4_in) at (cv4.center) {};
\node[cnode](cv5)[right of=vv2,xshift=2cm]{};
\node[srnode](cv5_in) at (cv5.center) {};
\node[cnode] (vv3) [below of=cv1,yshift=-0.5cm]{};
\node[scnode](vv3_in) at (vv3.center) {};
\node[cnode] (vv4) [below of=cv1,yshift=-2.5cm]{};
\node[scnode](vv4_in) at (vv4.center) {};
\node[cnode] (cv6) [below of=vv3,yshift=0.5 cm]{};
\node[srnode](cv6_in) at (cv6.center) {};
\node[cnode] (cv7) [below of=vv4,yshift=0.5 cm]{};
\node[srnode](cv7_in) at (cv7.center) {};
\draw (vv1)--(cv1);
\draw (vv1)--(cv2);
\draw (vv1)--(cv3);
\draw (vv2)--(cv1);
\draw (vv2)--(cv4);
\draw (vv2)--(cv5);
\draw (vv3)--(cv2);
\draw (vv3)--(cv4);
\draw (vv4)--(cv3);
\draw (vv4)--(cv5);
\draw (vv3)--(cv6);
\draw (vv4)--(cv7);


\node[cnode](cv10)[right of =cv5]{};
\node[srnode](cv10_in) at (cv10.center) {};
\node[cnode](vv5)[right of=cv10,xshift=2cm]{};
\node[sb_cnode](vv5_in) at (vv5.center) {};
\node[cnode](cv9)[right of=cv10,yshift=-0.5cm]{};
\node[srnode](cv9_in) at (cv9.center) {};
\node[cnode](cv8)[right of=vv5,yshift=0 cm]{};
\node[srnode](cv8_in) at (cv8.center) {};
\node[cnode](vv6)[right of=cv8,xshift=0cm]{};
\node[sb_cnode](vv6_in) at (vv6.center) {};
\node[cnode](cv11)[right of=vv6,yshift=-0.5cm]{};
\node[srnode](cv11_in) at (cv11.center) {};
\node[cnode](cv12)[right of=vv6,xshift=2cm]{};
\node[srnode](cv12_in) at (cv12.center) {};
\node[cnode] (vv7) [below of=cv8,yshift=-0.5cm]{};
\node[scnode](vv7_in) at (vv7.center) {};
\node[cnode] (vv8) [below of=cv8,yshift=-2.5cm]{};
\node[scnode](vv8_in) at (vv8.center) {};
\node[cnode] (cv13) [below of=vv7,yshift=0.5 cm]{};
\node[srnode](cv13_in) at (cv13.center) {};
\node[cnode] (cv14) [below of=vv8,yshift=0.5 cm]{};
\node[srnode](cv14_in) at (cv14.center) {};
\draw (vv5)--(cv8);
\draw (vv5)--(cv9);
\draw (vv5)--(cv10);
\draw (vv6)--(cv8);
\draw (vv6)--(cv11);
\draw (vv6)--(cv12);
\draw (vv7)--(cv9);
\draw (vv7)--(cv11);
\draw (vv8)--(cv10);
\draw (vv8)--(cv12);
\draw (vv7)--(cv13);
\draw (vv8)--(cv14);

\node[cnode](cv17)[right of =cv12]{};
\node[srnode](cv17_in) at (cv17.center) {};
\node[cnode](vv9)[right of=cv17,xshift=2cm]{};
\node[sb_cnode](vv9_in) at (vv9.center) {};
\node[cnode](cv16)[right of=cv17,yshift=-0.5cm]{};
\node[srnode](cv16_in) at (cv16.center) {};
\node[cnode](cv15)[right of=vv9,yshift=0 cm]{};
\node[srnode](cv15_in) at (cv15.center) {};
\node[cnode](vv10)[right of=cv15,xshift=0cm]{};
\node[sb_cnode](vv10_in) at (vv10.center) {};
\node[cnode](cv18)[right of=vv10,yshift=-0.5cm]{};
\node[srnode](cv18_in) at (cv18.center) {};
\node[cnode](cv19)[right of=vv10,xshift=2cm]{};
\node[srnode](cv19_in) at (cv19.center) {};
\node[cnode] (vv11) [below of=cv15,yshift=-0.5cm]{};
\node[scnode](vv11_in) at (vv11.center) {};
\node[cnode] (vv12) [below of=cv15,yshift=-2.5cm]{};
\node[scnode](vv12_in) at (vv12.center) {};
\node[cnode] (cv20) [below of=vv11,yshift=0.5 cm]{};
\node[srnode](cv20_in) at (cv20.center) {};
\node[cnode] (cv21) [below of=vv12,yshift=0.5 cm]{};
\node[srnode](cv21_in) at (cv21.center) {};
\draw (vv9)--(cv15);
\draw (vv9)--(cv16);
\draw (vv9)--(cv17);
\draw (vv10)--(cv15);
\draw (vv10)--(cv18);
\draw (vv10)--(cv19);
\draw (vv11)--(cv16);
\draw (vv11)--(cv18);
\draw (vv12)--(cv17);
\draw (vv12)--(cv19);
\draw (vv11)--(cv20);
\draw (vv12)--(cv21);

\node[cnode](cv24)[right of =cv19]{};
\node[srnode](cv24_in) at (cv24.center) {};
\node[cnode](vv13)[right of=cv24,xshift=2cm]{};
\node[sb_cnode](vv13_in) at (vv13.center) {};
\node[cnode](cv23)[right of=cv24,yshift=-0.5cm]{};
\node[srnode](cv23_in) at (cv23.center) {};
\node[cnode](cv22)[right of=vv13,xshift=0.5 cm]{};
\node[srnode](cv22_in) at (cv22.center) {};
\node[cnode](vv14)[right of=cv22,xshift=0.5cm]{};
\node[sb_cnode](vv14_in) at (vv14.center) {};
\node[cnode](cv25)[right of=vv14,yshift=-0.5cm]{};
\node[srnode](cv25_in) at (cv25.center) {};
\node[cnode](cv26)[right of=vv14,xshift=2cm]{};
\node[srnode](cv26_in) at (cv26.center) {};
\node[cnode] (vv15) [below of=cv22,yshift=-0.5cm]{};
\node[scnode](vv15_in) at (vv15.center) {};
\node[cnode] (vv16) [below of=cv22,yshift=-2.5cm]{};
\node[scnode](vv16_in) at (vv16.center) {};
\node[cnode] (cv27) [below of=vv15,yshift=0.5 cm]{};
\node[srnode](cv27_in) at (cv27.center) {};
\node[cnode] (cv28) [below of=vv16,yshift=0.5 cm]{};
\node[srnode](cv28_in) at (cv28.center) {};
\draw (vv13)--(cv22);
\draw (vv13)--(cv23);
\draw (vv13)--(cv24);
\draw (vv14)--(cv22);
\draw (vv14)--(cv25);
\draw (vv14)--(cv26);
\draw (vv15)--(cv23);
\draw (vv15)--(cv25);
\draw (vv16)--(cv24);
\draw (vv16)--(cv26);
\draw (vv15)--(cv27);
\draw (vv16)--(cv28);

\node[rnode] (cc1)[below of=cv14,xshift=5 cm,yshift=-1.5cm]{};
\node[srnode] (cc1_in) at (cc1.center){};
\node[rnode] (cc2)[below of=cc1,yshift=0.5 cm]{};
\node[srnode] (cc2_in) at (cc2.center){};
\draw (cv6) -- (cc1.90);
\draw (cv7) -- (cc2.90);
\draw (cv13) -- (cc1.90);
\draw (cv14) -- (cc2.90);
\draw (cv20) -- (cc1.90);
\draw (cv21) -- (cc2.90);
\draw (cv27) -- (cc1.90);
\draw (cv28) -- (cc2.90);
\node[rnode] (cc3)[above of=vv6,xshift=0.5 cm,yshift=3cm]{};
\node[sb_rnode](cc3_in) at (cc3.center) {};
\node[rnode] (cc4)[right of=cc3,xshift=0 cm,yshift=0cm]{};
\node[sb_rnode](cc4_in) at (cc4.center) {};
\node[rnode] (cc5)[right of=cc4,xshift=0 cm,yshift=0cm]{};
\node[srnode](cc5_in) at (cc5.center) {};
\node[rnode] (cc6)[right of=cc5,xshift=0 cm,yshift=0cm]{};
\node[sb_rnode](cc6_in) at (cc6.center) {};
\node[rnode] (cc7)[right of=cc6,xshift=0 cm,yshift=0cm]{};
\node[sb_rnode](cc7_in) at (cc7.center) {};

\draw (cv3) --(cc3);
\draw (cv10) --(cc3);
\draw (cv17) --(cc3);
\draw (cv24) --(cc3);
\draw (cv2.90) --(cc4);
\draw (cv9) --(cc4);
\draw (cv16) --(cc4);
\draw (cv23.90) --(cc4);
\draw (cv1) --(cc5);
\draw (cv8) --(cc5);
\draw (cv15) --(cc5);
\draw (cv22) --(cc5);

\draw (cv4.60) --(cc6);
\draw (cv11) --(cc6);
\draw (cv18) --(cc6);
\draw (cv25.200) --(cc6);
\draw (cv5) --(cc7);
\draw (cv12) --(cc7);
\draw (cv19) --(cc7);
\draw (cv26) --(cc7);
\node[scale=0.55] [below of =cv1,yshift=0 cm,yshift=1cm]{$v^1_1\times\mathscr{T}(\mathcal{V}^2_\mathrm{s})$};
\node[scale=0.55] [below of =cv8,yshift=0 cm,yshift=1cm]{$v^1_2\times\mathscr{T}(\mathcal{V}^2_\mathrm{s})$};
\node[scale=0.55] [below of =cv15,yshift=0 cm,yshift=1cm]{$v^1_3\times\mathscr{T}(\mathcal{V}^2_\mathrm{s})$};
\node[scale=0.55] [below of =cv22,yshift=0 cm,yshift=1cm]{$v^1_4\times\mathscr{T}(\mathcal{V}^2_\mathrm{s})$};

\end{scope}    
\end{tikzpicture}
    \subcaption{Convergence to a stabilizer: In iteration~$2$, the variable nodes with a red outer border are identified as erroneous. Among these nodes, we update the set of estimated errors by removing those that belong to the mismatched error set and adding those that do not belong to it. From Part~5 of Lemma~\ref{lemma:stab_induced_TS_containg_classical_TS}, it follows that the variable nodes in the mismatched error form a stabilizer, which shows that the decoder has converged to a stabilizer.}
    \label{fig:iteration3}
    \end{subfigure}
    \caption{Decoding iterations are shown for the \emph{stabilizer-induced} TS in Fig.~\ref{fig:multiple_copies_of_classical_TS_inside_stab_induced_TS}, which contains multiple classical TSs. The decoder consists of two classical LDPC decoders: a decoder that acts on CC-type variable nodes and a second decoder that acts on VV-type variable nodes. The latter takes as input the mismatched syndrome (discrepancies between the measured and estimated syndromes) produced by the first decoder and is derived from a decoder designed to resolve classical TSs in the underlying classical LDPC code. Variable nodes with mismatched errors (discrepancies between actual and estimated errors) are shaded gray, as are the check nodes associated with mismatched syndromes.}
    \label{fig:decoding_illustration_of_stab_induced_TS_with_multiple_classical_TS}
\end{figure*}

\begin{lemma}
    \label{lemma:error_exclusively_on_VV_condition}
  Consider the \emph{stabilizer-induced} TS $\mathscr{T}_{\mathrm{stab}} \triangleq \mathscr{T}(c^1) \times \mathscr{T}(\mathcal{V}^2_{\mathrm{s}})$ in the Tanner graph $\mathscr{G}_{\mathrm{z}}$ of the HGP/LP code defined in Lemma~\ref{lemma:stab_induced_TS_containg_classical_TS}, where $c^1$ is a check node in the parent Tanner graph $\mathscr{G}_1$, and $\mathcal{V}^2_{\mathrm{s}}$ is a subset of variable nodes $\mathcal{V}^2$ in the parent Tanner graph $\mathscr{G}_2$.
Furthermore, consider an error pattern $\bar{\mathbf{e}}\in \mathbb{F}_2^q/\mathcal{C}_{\mathrm{x}}^{\perp}$, where $q=m_1m_2+n_1n_2$ is the number of qubits (variable nodes) and $\mathcal{C}_{\mathrm{x}}^{\perp}$ is the row space of $\mathbf{H}_{\mathrm{x}}$. 
Assume that $\operatorname{supp}(\bar{\mathbf{e}})$ does not fully contain the support of any logical operator. We ask whether there exists an error pattern $\tilde{\mathbf{e}}$ and a vector $\mathbf{s}\in\mathcal{C}_{\mathrm{x}}^{\perp}$ such that $\tilde{\mathbf{e}} = \bar{\mathbf{e}} + \mathbf{s}$ and $\operatorname{supp}(\tilde{\mathbf{e}})$ lies exclusively on the VV-type variable nodes.
\begin{enumerate}
    \item If there exists at least one variable node $v^1_{\mathrm{i}'} \in \mathcal{N}_{c^1}$ such that all checks in $v^1_{\mathrm{i}'} \times \mathscr{T}(\mathcal{V}^2_{\mathrm{s}})$ are satisfied (i.e., there is no unsatisfied check in $v^1_{\mathrm{i}'} \times \mathscr{T}(\mathcal{V}^2_{\mathrm{s}})$), then such an error pattern $\tilde{\mathbf{e}}$ exists.
     \item If no such error pattern $\tilde{\mathbf{e}}$ exists, then for every variable node $v^1_{\mathrm{i}'} \in \mathcal{N}_{c^1}$ there is at least one unsatisfied check node of the form $v^1_{\mathrm{i}'} c^2_{\mathrm{j}}$, for some $c^2_{\mathrm{j}} \in \mathscr{T}(\mathcal{V}^2_{\mathrm{s}})$.
 \end{enumerate}
\end{lemma}
\begin{proof}
    Let $\operatorname{supp}(\bar{\mathbf{e}})$ denote the support of the error pattern $\bar{\mathbf{e}}$. Divide $\operatorname{supp}(\bar{\mathbf{e}})$ into two disjoint subsets $\operatorname{supp}(\bar{\mathbf{e}}_{\mathrm{vv}})$ and $\operatorname{supp}(\bar{\mathbf{e}}_{\mathrm{cc}})$ such that $\operatorname{supp}(\bar{\mathbf{e}}_{\mathrm{vv}})$ contains only VV-type variable nodes while $\operatorname{supp}(\bar{\mathbf{e}}_{\mathrm{cc}})$ contains only CC-type variable nodes.
    In Part~2 of Lemma~\ref{lemma:stab_induced_TS_containg_classical_TS}, we have seen that each VV-type variable node of $\mathscr{T}_{\mathrm{stab}}$ is contained in exactly one subgraph of the form $v^1_{\mathrm{i}'} \,\times \, \mathscr{T}(\mathcal{V}^2_{\mathrm{s}})$ for some $v^1_{\mathrm{i}'} \in \mathcal{N}_{c^1}$, and that these subgraphs are pairwise disjoint, i.e.,
\begin{equation*}
\begin{aligned}
(v^1_{\mathrm{i}'}\!\times\!\mathscr{T}(\mathcal{V}^2_{\mathrm{s}}))\!\cap\!(v^1_{\mathrm{k}'}\!\times\!\mathscr{T}(\mathcal{V}^2_{\mathrm{s}})) &= \emptyset,\\
&\forall\, v^1_{\mathrm{i}'}, v^1_{\mathrm{k}'} \in \mathcal{N}_{c^1},\; i' \neq k'.
\end{aligned}
\end{equation*}
This implies that $\operatorname{supp}(\bar{\mathbf{e}}_{\mathrm{vv}})$ can be further divided into $|\mathcal{N}_{c^1}|$ disjoint subsets, say $\operatorname{supp}(\bar{\mathbf{e}}_{\mathrm{vv},v^1_{\mathrm{i}'}})$ for $v^1_{\mathrm{i}'}\in \mathcal{N}_{c^1}$, such that $\operatorname{supp}(\bar{\mathbf{e}}_{\mathrm{vv},v^1_{\mathrm{i}'}})$ is contained in the subgraph $v^1_{\mathrm{i}'} \,\times \, \mathscr{T}(\mathcal{V}^2_{\mathrm{s}})$.
Recall that we want to establish a condition that guarantees an error pattern whose support lies entirely on VV-type variable nodes and that is consistent with the measured syndrome. If $\operatorname{supp}(\bar{\mathbf{e}}_{\mathrm{cc}}) = \emptyset$, the claim is trivial. Therefore, in the remainder of the proof, we assume that $\operatorname{supp}(\bar{\mathbf{e}}_{\mathrm{cc}}) \neq \emptyset$.

Let us examine when there exists an index $v^1_{\mathrm{i}'}\in\mathcal{N}_{c^1}$ such that the subgraph $v^1_{\mathrm{i}'} \times \mathscr{T}(\mathcal{V}^2_{\mathrm{s}})$ has no unsatisfied check nodes.
Since $\operatorname{supp}(\bar{\mathbf{e}}_{\mathrm{cc}}) \neq \emptyset$, there exists at least one erroneous CC-type variable node.
By Part~6 of Lemma~\ref{lemma:stab_induced_TS_containg_classical_TS}, every CC-type variable node in $\mathscr{T}_{\mathrm{stab}}$ is adjacent to exactly one check node in $v^1_{\mathrm{i}'} \times \mathscr{T}(\mathcal{V}^2_{\mathrm{s}})$ for each $v^1_{\mathrm{i}'} \in \mathcal{N}_{c^1}$.
Therefore, if $\operatorname{supp}(\bar{\mathbf{e}}_{\mathrm{vv}, v^1_{\mathrm{i}'}})=\emptyset$, then the only way for $v^1_{\mathrm{i}'} \times \mathscr{T}(\mathcal{V}^2_{\mathrm{s}})$ to contain an unsatisfied check is through an erroneous CC-type variable node adjacent to one of its checks.
Consequently, if $\operatorname{supp}(\bar{\mathbf{e}}_{\mathrm{vv}, v^1_{\mathrm{i}'}})=\emptyset$ and there exists an erroneous CC-type variable node adjacent to a check in $v^1_{\mathrm{i}'} \times \mathscr{T}(\mathcal{V}^2_{\mathrm{s}})$, then this subgraph must contain an unsatisfied check node. Hence, if $v^1_{\mathrm{i}'} \times \mathscr{T}(\mathcal{V}^2_{\mathrm{s}})$ has no unsatisfied check nodes, we must have $\operatorname{supp}(\bar{\mathbf{e}}_{\mathrm{vv}, v^1_{\mathrm{i}'}}) \neq \emptyset$.

Now consider the case where $v_{\mathrm{i}'}^1 \times \mathscr{T}(\mathcal{V}^2_{\mathrm{s}})$ has no unsatisfied check nodes and $\operatorname{supp}(\bar{\mathbf{e}}_{\mathrm{vv},v^1_{\mathrm{i}'}}) \neq \emptyset$.
Define an error pattern $\mathbf{f}\in \mathbb{F}_2^q$ with non-empty support such that $\operatorname{supp}(\mathbf{f})=\operatorname{supp}(\bar{\mathbf{e}}_{\mathrm{vv},v^1_{\mathrm{i}'}})$. Since we consider the full syndrome in $\mathscr{G}_{\mathrm{z}}$, a zero-syndrome pattern must be a stabilizer or a logical operator. By assumption, $\mathbf{f}$ cannot be a logical operator. Moreover, $\operatorname{supp}(\mathbf{f})$ is contained in the single subgraph $v^1_{\mathrm{i}'} \times \mathscr{T}(\mathcal{V}^2_{\mathrm{s}})$, and this subgraph does not contain a CV-type stabilizer: by the definition of the neighbor set of a CV-type check node, any CV-type stabilizer involves VV-type variable nodes with more than one value of the first coordinate. 
Hence, the syndrome corresponding to $\mathbf{f}$ is non-zero.
Next, we examine how erroneous variable nodes that are not contained in $v^1_{\mathrm{i}'} \times \mathscr{T}(\mathcal{V}^2_{\mathrm{s}})$, but are contained in $\mathscr{T}_{\mathrm{stab}}$, alter the set of unsatisfied check nodes associated with the error pattern $\mathbf{f}$.
The unsatisfied check nodes caused by the error vector $\mathbf{f}$ can be written as $v^1_{\mathrm{i}'} \times \mathcal{C}_{\mathrm{o}}^2$, where $\mathcal{C}_{\mathrm{o}}^2$ is the subset of check nodes in $\mathscr{T}(\mathcal{V}^2_{\mathrm{s}})$ whose parity-check equations are violated by the VV-type errors in $\operatorname{supp}(\mathbf{f})$.
As $(v^1_{\mathrm{i}'}\times \mathcal{V}^2_{\mathrm{s}}) \cap (v^1_{\mathrm{k}'}\times \mathcal{V}^2_{\mathrm{s}}) = \emptyset$ for any distinct $v^1_{\mathrm{i}'},v^1_{\mathrm{k}'}\in \mathcal{N}_{c^1}$, erroneous variable nodes in $v^1_{\mathrm{k}'}\times \mathcal{V}^2_{\mathrm{s}}$ do not affect the number of unsatisfied check nodes in $v^1_{\mathrm{i}'}\times \mathcal{V}^2_{\mathrm{s}}$.
This further implies that, given the error pattern $\bar{\mathbf{e}}$, the set of unsatisfied check nodes in $v^1_{\mathrm{i}'} \times \mathscr{T}(\mathcal{V}^2_{\mathrm{s}})$ depends only on (i) the VV-type variable nodes in the subgraph $v^1_{\mathrm{i}'} \times \mathscr{T}(\mathcal{V}^2_{\mathrm{s}})$ and (ii) the CC-type variable nodes in $\mathscr{T}_{\mathrm{stab}}$, and is unaffected by VV-type variable nodes in $v^1_{\mathrm{k}'} \times \mathscr{T}(\mathcal{V}^2_{\mathrm{s}})$ for any $v^1_{\mathrm{k}'} \in \mathcal{N}_{c^1}$ with $\mathrm{k}'\neq \mathrm{i}'$.
Each check node $v^1_{\mathrm{i}'}c^2_{\mathrm{j}}$ in $v^1_{\mathrm{i}'} \times \mathcal{C}_{\mathrm{o}}^2$, for $c^2_{\mathrm{j}} \in \mathcal{C}_{\mathrm{o}}^2$, has only one CC-type variable node as a neighbor, namely $c^1c_{\mathrm{j}}^2$. This follows from the fact that the CC-type neighbors of $v^1_{\mathrm{i}'}c_{\mathrm{j}}^2$ in $\mathscr{T}_{\mathrm{stab}}$ are of the form $c^1c^2_{\mathrm{j}}$.
Moreover, $c^1c^2_{\mathrm{j}}$ has only one neighbor in subgraph $v^1_{\mathrm{i}'} \times \mathscr{T}(\mathcal{V}^2_{\mathrm{s}})$. Therefore, for subgraph $v^1_{\mathrm{i}'} \times \mathscr{T}(\mathcal{V}^2_{\mathrm{s}})$ to have no unsatisfied check nodes, all CC-type variable nodes in $c^1\times \mathcal{C}_{\mathrm{o}}^2$ must be erroneous.
Moreover, the fact that $v^1_{\mathrm{i}'} \times \mathscr{T}(\mathcal{V}^2_{\mathrm{s}})$ does not have any unsatisfied check nodes implies that the only erroneous CC-type variable nodes in $\mathscr{T}_{\mathrm{stab}}$ are those in $c^1 \times \mathcal{C}_{\mathrm{o}}^2$.
Any other CC-type variable node outside this set would create an additional unsatisfied check node in $v^1_{\mathrm{i}'} \times \mathscr{T}(\mathcal{V}^2_{\mathrm{s}})$, which would contradict the assumption that this subgraph has no unsatisfied check nodes.

Define the projection of $\operatorname{supp}(\mathbf{f})$ onto $\mathcal{V}^2$, denoted by $\operatorname{P}^2(\operatorname{supp}(\mathbf{f}))$, as
\begin{equation*}
\operatorname{P}^2(\operatorname{supp}(\mathbf{f}))
= \bigl\{ v^2_{\mathrm{k}} : \text{there exists } v^1_{\mathrm{k}'} v^2_{\mathrm{k}} \in \operatorname{supp}(\mathbf{f}) \bigr\}.
\end{equation*}
Recall that $\operatorname{supp}(\mathbf{f})=\operatorname{supp}(\bar{\mathbf{e}}_{\mathrm{vv},v^1_{\mathrm{i}'}})$. Now consider the stabilizer induced by CV-type stabilizer generators
\begin{equation*}
 c^1 \times \mathscr{T}\bigl(\operatorname{P}^2(\operatorname{supp}(\mathbf{f}))\bigr).
\end{equation*}
Let $\mathbf{s}$ be the binary representation of this induced stabilizer. 
From Part~5 of Lemma~\ref{lemma:stab_induced_TS_containg_classical_TS}, it follows that the set of CC-type variable nodes in $\operatorname{supp}(\mathbf{s})$ is $c^1 \times \mathcal{C}_{\mathrm{o}}^2$. We have shown earlier that $\operatorname{supp}(\bar{\mathbf{e}}_{\mathrm{cc}}) = c^1 \times \mathcal{C}_{\mathrm{o}}^2$. Therefore, the support of the vector $\tilde{\mathbf{e}} = \bar{\mathbf{e}} + \mathbf{s}$ contains only VV-type variable nodes, which proves the claim in Part~1. Since the statement in Part~2 is exactly the contrapositive of the statement in Part~1, its validity follows immediately and does not require a separate proof.
\end{proof}
Using arguments similar to those in the proof of Lemma~\ref{lemma:error_exclusively_on_VV_condition}, we can derive conditions that guarantee the existence of an error pattern supported only on CC-type variable nodes within the \emph{stabilizer-induced} TS $\mathscr{T}(\mathcal{C}^1_{\mathrm{s}}) \times \mathscr{T}(v^2)$, where $\mathcal{C}^1_{\mathrm{s}} \subset \mathcal{C}^1$ and $v^2 \in \mathcal{V}^2$. 
In particular, if there exists a check node $c^2_{\mathrm{j}} \in \mathcal{N}_{v^2}$ such that there are no unsatisfied check nodes in $\mathscr{T}(\mathcal{C}^1_{\mathrm{s}}) \times c^2_{\mathrm{j}}$, then there exists an error pattern, supported only on the CC-type variable nodes in $\mathscr{T}(\mathcal{C}^1_{\mathrm{s}}) \times \mathscr{T}(v^2)$, that matches the measured syndrome.

\begin{lemma}

\label{lemma:decoder_resolving_classical_plus_stab_TS}

Let $\mathscr{G}_{\mathrm{z}}$ be the $Z$-Tanner graph of an HGP/LP code obtained from the parent Tanner graphs $\mathscr{G}_1$ and $\mathscr{G}_2$. Fix $c^1_{\mathrm{\ell}'}\in\mathcal{C}^1$ and $\mathcal{V}^2_{\mathrm{s}}\subseteq\mathcal{V}^2$, and assume that a decoder $D_2$ on $\mathscr{G}_2$ successfully decodes every error pattern whose support is contained in the TS $\mathscr{T}(\mathcal{V}^2_{\mathrm{s}})$. Then, for any error pattern supported within the \emph{stabilizer-induced} TS $\mathscr{T}(c^1_{\mathrm{\ell}'})\times\mathscr{T}(\mathcal{V}^2_{\mathrm{s}})$, Algorithm~\ref{alg:derived_decoders_from_parent_decoders} (with preprocessing Steps~6--12) produces a syndrome-compatible error pattern supported only on the VV-type variable nodes in this TS, and a subsequent run of $D_{\mathrm{vv}}$ (derived from $D_2$) corrects the error.

Similarly, fix $v^2_{\mathrm{k}}\in\mathcal{V}^2$ and $\mathcal{C}^1_{\mathrm{s}}\subseteq\mathcal{C}^1$, and assume that a decoder $D_1^{\mathsf{T}}$ on $\mathscr{G}_1^{\mathsf{T}}$ successfully decodes every error pattern whose support is contained in the TS $\mathscr{T}(\mathcal{C}^1_{\mathrm{s}})$. Then, for any error pattern supported within $\mathscr{T}(\mathcal{C}^1_{\mathrm{s}})\times\mathscr{T}(v^2_{\mathrm{k}})$, preprocessing Steps~18--24 produces a syndrome-compatible error pattern supported only on the CC-type variable nodes, and one subsequent run of $D_{\mathrm{cc}}$ (derived from $D_1^{\mathsf{T}}$) corrects the error.
\end{lemma}

\begin{proof}

We prove the lemma for the case \(\mathscr{T}(c^1_{\mathrm{\ell}'}) \times \mathscr{T}(\mathcal{V}^2_{\mathrm{s}})\); the case \(\mathscr{T}(\mathcal{C}^1_{\mathrm{s}}) \times \mathscr{T}(v^2_{\mathrm{k}})\) is analogous.
According to Lemma~\ref{lemma:error_exclusively_on_VV_condition}, to complete the proof, it suffices to show that, after the preprocessing steps (6--12) in Algorithm~\ref{alg:derived_decoders_from_parent_decoders}, there exists a node \(v^1 \in \mathcal{N}_{c^1_{\mathrm{\ell}'}}\) such that all checks of the form \(v^1 c^2_{\mathrm{j}}\), with \(c^2_{\mathrm{j}} \in \mathscr{T}(\mathcal{V}^2_{\mathrm{s}})\), are satisfied. Indeed, by that lemma, the existence of such a node is equivalent to the existence of an error pattern supported only on VV-type variable nodes that is compatible with the measured syndrome.
Moreover, since the error patterns considered here are supported within \(\mathscr{T}(c^1_{\mathrm{\ell}'}) \times \mathscr{T}(\mathcal{V}^2_{\mathrm{s}})\), the syndrome is zero outside this induced subgraph; for the decoders considered, this implies that no variable nodes outside this induced subgraph are flipped, and hence preprocessing does not create unsatisfied VC-type checks outside \(\mathscr{T}(c^1_{\mathrm{\ell}'}) \times \mathscr{T}(\mathcal{V}^2_{\mathrm{s}})\).

We now verify that the preprocessing steps (6--12) guarantee the existence of such a node.
In Step~7, the algorithm examines whether there exist more than \(\left\lfloor\lvert\mathcal{N}_{c^1_{\mathrm{\ell}'}}\rvert/2\right\rfloor\) variable nodes \(v^1_{\mathrm{i}'} \in \mathcal{N}_{c^1_{\mathrm{\ell}'}}\) such that, for each such variable node, there exists at least one unsatisfied VC-type check node of the form \(v^1_{\mathrm{i}'} c^2_{\mathrm{j}}\), with \(c^2_{\mathrm{j}} \in \mathscr{T}(\mathcal{V}^2_{\mathrm{s}})\). 
If there are at most \(\left\lfloor\lvert\mathcal{N}_{c^1_{\mathrm{\ell}'}}\rvert/2\right\rfloor\) such variable nodes, then there exists a node \(v^1\in \mathcal{N}_{c^1_{\mathrm{\ell}'}}\) for which all checks \(v^1 c^2_{\mathrm{j}}\) are satisfied (cf. Lemma~\ref{lemma:error_exclusively_on_VV_condition}), and we are done.
Otherwise, assume that there are more than \(\left\lfloor\lvert\mathcal{N}_{c^1_{\mathrm{\ell}'}}\rvert/2\right\rfloor\) variable nodes \(v^1_{\mathrm{i}'} \in \mathcal{N}_{c^1_{\mathrm{\ell}'}}\) such that, for each such variable node, there exists at least one unsatisfied VC-type check node of the form \(v^1_{\mathrm{i}'} c^2_{\mathrm{j}}\), with \(c^2_{\mathrm{j}} \in \mathscr{T}(\mathcal{V}^2_{\mathrm{s}})\).
Note that this subsumes the condition in Part~2 of Lemma~\ref{lemma:error_exclusively_on_VV_condition} concerning the non-existence of error patterns solely supported on VV-type variable nodes.  
In such a scenario, the algorithm proceeds to Step~8: it selects the unsatisfied VC-type check node $v^1_{\mathrm{i}'}c^2_{\mathrm{j}}$ for \(v^1_{\mathrm{i}'} \in \mathcal{N}_{c^1_{\mathrm{\ell}'}}\). Note that there could be multiple of such unsatisfied VC-type check nodes. In this case, Step~8 randomly selects one of those. In what follows, we denote the chosen VC-type check node by $v^1_{\mathrm{i}'}c^2_{\mathrm{j}}.$

Fix the variable node \(v^1_{\mathrm{i}'}\) and for each unsatisfied check of the form \(v^1_{\mathrm{i}'} c^2_{\mathrm{j}}\), Step~11 flips the corresponding CC-type variable node \(c^1_{\mathrm{\ell}'} c^2_{\mathrm{j}}\).
By Part~6 of Lemma~\ref{lemma:stab_induced_TS_containg_classical_TS}, for each fixed $v^1_{\mathrm{i}'}\in\mathcal{N}_{c^1_{\mathrm{\ell}'}}$ and each $c^2_{\mathrm{j}}\in\mathscr{T}(\mathcal{V}^2_{\mathrm{s}})$, the CC-type variable node $c^1_{\mathrm{\ell}'}c^2_{\mathrm{j}}$ is adjacent to exactly one VC-type check node of the form $v^1_{\mathrm{i}'}c^2_{\mathrm{j}}$ in $v^1_{\mathrm{i}'}\times\mathscr{T}(\mathcal{V}^2_{\mathrm{s}})$. 
Since flipping a variable node toggles the parity of every incident check node, flipping $c^1_{\mathrm{\ell}'}c^2_{\mathrm{j}}$ toggles the syndrome of $v^1_{\mathrm{i}'}c^2_{\mathrm{j}}$, and it does not affect any check of the form $v^1_{\mathrm{i}'}c^2_{\mathrm{\ell}}$ with $\ell\neq j$. (Such a flip may toggle checks of the form $v^1_{\mathrm{k}'}c^2_{\mathrm{j}}$ for $v^1_{\mathrm{k}'}\neq v^1_{\mathrm{i}'}$, but this is irrelevant to our goal, which is to ensure the existence of at least one $v^1$---namely $v^1_{\mathrm{i}'}$---for which all checks $v^1_{\mathrm{i}'}c^2_{\mathrm{j}}$ are satisfied.)
Therefore, since Step~11 flips $c^1_{\mathrm{\ell}'}c^2_{\mathrm{j}}$ for every unsatisfied check $v^1_{\mathrm{i}'}c^2_{\mathrm{j}}$, after all these flips are performed, all checks of the form $v^1_{\mathrm{i}'}c^2_{\mathrm{j}}$ are satisfied. Thus, after preprocessing, for the specific node $v^1_{\mathrm{i}'}$ chosen in Step~8, there are no unsatisfied checks of the form $v^1_{\mathrm{i}'}c^2_{\mathrm{j}}$. Equivalently, after preprocessing there always exists a node $v^1\in\mathcal{N}_{c^1_{\mathrm{\ell}'}}$, namely $v^1=v^1_{\mathrm{i}'}$, such that no check of the form $v^1c^2_{\mathrm{j}}$ is unsatisfied.
By Lemma~\ref{lemma:error_exclusively_on_VV_condition}, the existence of this node \(v^1=v^1_{\mathrm{i}'}\) implies the existence of an error pattern supported only on VV-type variable nodes that is compatible with the measured syndrome. By the assumption of the lemma, the decoder \(D_{\mathrm{vv}}\) corrects any such error pattern within \(\mathscr{T}(c^1_{\mathrm{\ell}'}) \times \mathscr{T}(\mathcal{V}^2_{\mathrm{s}})\). Therefore, once preprocessing has produced such a node \(v^1=v^1_{\mathrm{i}'}\), Algorithm~\ref{alg:derived_decoders_from_parent_decoders} will successfully correct every error pattern in \(\mathscr{T}(c^1_{\mathrm{\ell}'}) \times \mathscr{T}(\mathcal{V}^2_{\mathrm{s}})\). This completes the proof.
\end{proof}

\begin{remark}
\label{remark:alternate_dec_startegy}
Lemma~\ref{lemma:decoder_resolving_classical_plus_stab_TS} is stated for a single decoder $D_2$ that resolves every error pattern supported on $\mathscr{T}(\mathcal{V}^2_{\mathrm{s}})$. When no single decoder achieves this property, the lemma extends directly with $D_2$ replaced by a set of decoders, under the assumption that every such error pattern is resolved by at least one member of the set. The preprocessing steps are executed once and the resulting estimates are shared across each decoder in the set. Then the member decoders in the set run on each of the embedded copies of the parent code's Tanner graph in $\mathscr{G}_{\mathrm{z}}.$ The guarantees of this section, and the adversarial results reported in Table~\ref{tab:cc_success}, are with respect to this assumption.
\end{remark}

\subsection{Stabilizer-induced TSs as graph products of a tree-like subgraph and a TS}
\label{sec:tree_subgraph_graph_product}

Next, we study \emph{stabilizer-induced} TSs of the form $\mathscr{T}(\mathcal{C}_{\mathrm{s}}^1) \times \mathscr{T}(\mathcal{V}^2_{\mathrm{s}})$. Here, $\mathscr{T}(\mathcal{C}_{\mathrm{s}}^1)$ is a subgraph of $\mathscr{G}_1^{\mathsf{T}}$ with $\mathcal{C}_{\mathrm{s}}^1 \subseteq \mathcal{C}_1$, and $\mathscr{T}(\mathcal{V}^2_{\mathrm{s}})$ is a subgraph of $\mathscr{G}_2$ with $\mathcal{V}^2_{\mathrm{s}} \subseteq \mathcal{V}_2$.
We first analyze the case in which $\mathscr{T}(\mathcal{C}_{\mathrm{s}}^1)$ is a tree in $\mathscr{G}_1^{\mathsf{T}}$ and $\mathscr{T}(\mathcal{V}^2_{\mathrm{s}})$ is a TS in $\mathscr{G}_2$. This setting reduces to cases studied above when $\lvert\mathcal{C}_{\mathrm{s}}^1\rvert=1$ (i.e., $\mathscr{T}(\mathcal{C}^1_{\mathrm{s}})=c^1$) or when $\lvert\mathcal{V}^2_{\mathrm{s}}\rvert=1$ (i.e., $\mathscr{T}(\mathcal{V}^2_{\mathrm{s}})=v^2$).
We now explain why these TSs must be analyzed. In the low-error regime, the slope of the logical error-rate curve is governed by the smallest TS on which the decoder fails to converge. Suppose that we wish to guarantee correction of every TS induced by an error pattern of weight at most $\mu$ (including both gate failures and data errors due to decoherence).
Part~4 of Lemma~\ref{lemma:stab_induced_TS_containg_classical_TS} implies that if a weight-$\mu$ \emph{data-error} pattern activates $\mathscr{T}(\mathcal{V}^2_{\mathrm{s}})$ in $\mathscr{G}_2$, then there exists a set of at most $\mu$ CNOT failures that activates $\mathscr{T}(c^1) \times \mathscr{T}(\mathcal{V}^2_{\mathrm{s}})$ in $\mathscr{G}_{\mathrm{z}}$. We showed above that Algorithm~\ref{alg:derived_decoders_from_parent_decoders} resolves $\mathscr{T}(c^1) \times \mathscr{T}(\mathcal{V}^2_{\mathrm{s}})$ for any $c^1\in\mathcal{C}_1$, provided that we have a decoder that resolves $\mathscr{T}(\mathcal{V}^2_{\mathrm{s}})$ in $\mathscr{G}_2$.
However, this condition is necessary but not sufficient to guarantee correction of all TSs induced by error patterns of weight at most $\mu$. 
For example, for some $\mathcal{V}^2_{\mathrm{s}}\subseteq \mathcal{V}_2$, if $\mathscr{T}(\mathcal{V}^2_{\mathrm{s}})$ is induced in $\mathscr{G}_2$ by a \emph{data-error} pattern of weight less than $\mu/2$, then by an argument analogous to Part~4 of Lemma~\ref{lemma:stab_induced_TS_containg_classical_TS} the TS $\mathscr{T}(c^1_{\mathrm{j}'} \cup c^1_{\mathrm{k}'}) \times \mathscr{T}(\mathcal{V}^2_{\mathrm{s}})$ can be induced by a set of CNOT failures of weight less than $\mu$ in the \emph{circuit-level} noise model. Such TSs are not addressed in Lemma~\ref{lemma:decoder_resolving_classical_plus_stab_TS}. We must therefore analyze the more general TSs of the form $\mathscr{T}(\mathcal{C}_{\mathrm{s}}^1) \times \mathscr{T}(\mathcal{V}^2_{\mathrm{s}})$.
Fortunately, Algorithm~\ref{alg:derived_decoders_from_parent_decoders}, with an amended Step~8, can correct these error patterns when either $\mathscr{T}(\mathcal{C}_{\mathrm{s}}^1)$ or $\mathscr{T}(\mathcal{V}^2_{\mathrm{s}})$ is a tree. We prove this formally in Lemma~\ref{lemma:composite_trapping_set_decoding_condition_tree}. Before that, we discuss the intuition behind the need to amend Step~8 and how the amendment enables the decoder to correct such TSs. In the next section, we consider the case where $\mathscr{T}(\mathcal{C}_{\mathrm{s}}^1)$ is not a tree.

Recall the role of Steps~6--7 in Algorithm~\ref{alg:derived_decoders_from_parent_decoders}. For each check node $c^1_{\mathrm{\ell}'}\in\mathcal{C}^1$, Step~7 bounds the number of subgraphs of the form $v^1\times\mathscr{T}(\mathcal{V}^2_{\mathrm{s}})$, for $v^1 \in \mathcal{N}_{c^1_{\mathrm{\ell}'}}$, that contain at least one unsatisfied check by $\big\lfloor |\mathcal{N}_{c^1_{\mathrm{\ell}'}}|/2\big\rfloor$. If this bound is not satisfied, Step~8 selects one such subgraph, and the subsequent steps toggle the error estimate on the CC-type variable nodes so that the selected subgraph has no unsatisfied checks.
Under this condition, Lemma~\ref{lemma:error_exclusively_on_VV_condition} implies that the syndrome is consistent with an error pattern supported only on VV-type variable nodes. This implication is the key step used in Lemma~\ref{lemma:decoder_resolving_classical_plus_stab_TS} to prove the correction of errors supported on $\mathscr{T}(c^1_{\mathrm{\ell}'})\times\mathscr{T}(\mathcal{V}^2_{\mathrm{s}})$.

For a TS of the form $\mathscr{T}(\mathcal{C}^1_{\mathrm{s}})\times\mathscr{T}(\mathcal{V}^2_{\mathrm{s}})$ with $\mathscr{T}(\mathcal{C}^1_{\mathrm{s}})$ a tree, one would like to apply the same argument by focusing on a leaf check node $c^1_{\mathrm{j}'}\in\mathscr{T}(\mathcal{C}^1_{\mathrm{s}})$. This leaf check node has exactly one neighbor $v^1_{\mathrm{sh}}\in\mathcal{N}_{c^1_{\mathrm{j}'}}$ that is also adjacent to another check node $c^1_{\mathrm{nbr}}$ in the tree.
Because $v^1_{\mathrm{sh}}\times\mathscr{T}(\mathcal{V}^2_{\mathrm{s}})$ lies in both $\mathscr{T}(c^1_{\mathrm{j}'})\times\mathscr{T}(\mathcal{V}^2_{\mathrm{s}})$ and $\mathscr{T}(c^1_{\mathrm{nbr}})\times\mathscr{T}(\mathcal{V}^2_{\mathrm{s}})$, the syndrome on the shared subgraph $v^1_{\mathrm{sh}}\times\mathscr{T}(\mathcal{V}^2_{\mathrm{s}})$ can be influenced by errors in both $\mathscr{T}(c^1_{\mathrm{j}'})\times\mathscr{T}(\mathcal{V}^2_{\mathrm{s}})$ and $\mathscr{T}(c^1_{\mathrm{nbr}})\times\mathscr{T}(\mathcal{V}^2_{\mathrm{s}})$. Consequently, even if this shared subgraph has no unsatisfied checks after flipping CC-type variable nodes, Lemma~\ref{lemma:error_exclusively_on_VV_condition} does not apply.
\begin{algorithm}[!ht]
\footnotesize
\DontPrintSemicolon
   \KwInput{Measured syndrome $\boldsymbol{\sigma}$, Maximum number of iterations $Maxiter$}
  \KwOutput{Error estimate $\hat{\mathbf{e}}$}
  \KwData{Tanner graph $\mathscr{G}_{\mathrm{z}}$, Decoders $D_{\mathrm{vv}}(\cdot)$ and $D_{\mathrm{cc}}(\cdot)$. }
  \KwVars{Intermediate error estimate on VV-type (respectively, CC-type) variable nodes $\tilde{\mathbf{e}}_{\mathrm{vv}}$ (respectively, $\tilde{\mathbf{e}}_{\mathrm{cc}}$),  Intermediate error estimate $\tilde{\mathbf{e}}=\begin{bmatrix}\tilde{\mathbf{e}}_{\mathrm{vv}} & \tilde{\mathbf{e}}_{\mathrm{cc}}\end{bmatrix}$, Mismatched syndrome $\tilde{\boldsymbol{\sigma}}$, Estimated syndrome $\hat{\boldsymbol{\sigma}}$}
  \KwInits{$\tilde{\boldsymbol{\sigma}} =\boldsymbol{\sigma}$, $\hat{\mathbf{e}}=\boldsymbol{0}$}
  
 Steps 1--5 is same as Algorithm~\ref{alg:derived_decoders_from_parent_decoders}.\;
  \For{ each $c^1_{\mathrm{\ell}'}\in \mathcal{C}^1$}
  {
    \If{there exist more than \(\left\lfloor\left\lvert\mathcal{N}_{c^1_{\mathrm{\ell}'}}\right\rvert/2\right\rfloor\) variable nodes \(v^1_{\mathrm{i}'} \in \mathcal{N}_{c^1_{\mathrm{\ell}'}}\) such that, for each of these nodes, there exists at least one VC-type check of the form \(v^1_{\mathrm{i}'} c^2_\mathrm{j}\), for some \(c^2_\mathrm{j} \in \mathcal{C}^2\), that is unsatisfied}
    {
        Choose two \emph{distinct} unsatisfied VC-type checks $v^1_{a'} c^2_{b}$ and $v^1_{\bar{a}'} c^2_{\bar{b}}$ with $v^1_{a'},v^1_{\bar{a}'}\in\mathcal{N}_{c^1_{\mathrm{\ell}'}}$\;
        Save $(\hat{\mathbf{e}}^{(0)},\tilde{\boldsymbol{\sigma}}^{(0)})\leftarrow (\hat{\mathbf{e}},\tilde{\boldsymbol{\sigma}})$\;
        \For{each $c^2\in\mathcal{C}^2$ with $\tilde{\boldsymbol{\sigma}}(v^1_{a'} c^2)=1$}
        {
            Flip CC-type variable node $c^1_{\mathrm{\ell}'}c^2$ and update $\tilde{\boldsymbol{\sigma}}$ accordingly (cf. Algorithm~\ref{alg:derived_decoders_from_parent_decoders})\;
        }
        $\hat{\mathbf{e}}_{\mathrm{vv}}\leftarrow D_{\mathrm{vv}}(\tilde{\boldsymbol{\sigma}})$ and update $(\hat{\mathbf{e}},\tilde{\boldsymbol{\sigma}})$ accordingly\;
        \If{$\tilde{\boldsymbol{\sigma}}(v^1 c^2)=0$ for all $v^1 c^2\in (\mathcal{N}_{c^1_{\mathrm{\ell}'}}\setminus\{v^1_{a'}\})\times\mathcal{C}^2$}
        {
            \tcp*[f]{First attempt succeeds; skip the second attempt}
        }
        \Else
        {
            Restore $(\hat{\mathbf{e}},\tilde{\boldsymbol{\sigma}})\leftarrow(\hat{\mathbf{e}}^{(0)},\tilde{\boldsymbol{\sigma}}^{(0)})$\;
            \For{each $c^2\in\mathcal{C}^2$ with $\tilde{\boldsymbol{\sigma}}(v^1_{\bar{a}'} c^2)=1$}
            {
                Flip CC-type variable node $c^1_{\mathrm{\ell}'}c^2$ and update $\tilde{\boldsymbol{\sigma}}$ accordingly (cf. Algorithm~\ref{alg:derived_decoders_from_parent_decoders})\;
            }
            $\hat{\mathbf{e}}_{\mathrm{vv}}\leftarrow D_{\mathrm{vv}}(\tilde{\boldsymbol{\sigma}})$ and update $(\hat{\mathbf{e}},\tilde{\boldsymbol{\sigma}})$ accordingly\;
        }
    }
  }
 \For{ each $v^2_{\mathrm{k}}\in \mathcal{V}^2$}
  {
    \If{there exist more than \(\left\lfloor\left\lvert\mathcal{N}_{v^2_{\mathrm{k}}}\right\rvert/2\right\rfloor\) check nodes \(c^2_{\mathrm{j}} \in \mathcal{N}_{v^2_{\mathrm{k}}}\) such that each $c^2_{\mathrm{j}}$ participates in at least one unsatisfied VC-type check of the form $v^1c^2_{\mathrm{j}}$ (for some $v^1\in\mathcal{V}^1$)}
    {
        Choose two \emph{distinct} unsatisfied VC-type checks $v^1_{a'} c^2_{b}$ and $v^1_{\bar{a}'} c^2_{\bar{b}}$ with $c^2_{b},c^2_{\bar{b}}\in\mathcal{N}_{v^2_{\mathrm{k}}}$\;
        Save $(\hat{\mathbf{e}}^{(0)},\tilde{\boldsymbol{\sigma}}^{(0)})\leftarrow (\hat{\mathbf{e}},\tilde{\boldsymbol{\sigma}})$\;
        \For{each $v^1\in\mathcal{V}^1$ with $\tilde{\boldsymbol{\sigma}}(v^1 c^2_{b})=1$}
        {
            Flip VV-type variable node $v^1v^2_{\mathrm{k}}$ and update $\tilde{\boldsymbol{\sigma}}$ accordingly (cf. Algorithm~\ref{alg:derived_decoders_from_parent_decoders})\;
        }
        $\hat{\mathbf{e}}_{\mathrm{cc}}\leftarrow D_{\mathrm{cc}}(\tilde{\boldsymbol{\sigma}})$ and update $(\hat{\mathbf{e}},\tilde{\boldsymbol{\sigma}})$ accordingly\;
        \If{$\tilde{\boldsymbol{\sigma}}(v^1 c^2)=0$ for all $v^1 c^2\in \mathcal{V}^1\times(\mathcal{N}_{v^2_{\mathrm{k}}}\setminus\{c^2_{b}\})$}
        {
            \tcp*[f]{First attempt succeeds; skip the second attempt}
        }
        \Else
        {
            Restore $(\hat{\mathbf{e}},\tilde{\boldsymbol{\sigma}})\leftarrow(\hat{\mathbf{e}}^{(0)},\tilde{\boldsymbol{\sigma}}^{(0)})$\;
            \For{each $v^1\in\mathcal{V}^1$ with $\tilde{\boldsymbol{\sigma}}(v^1 c^2_{\bar{b}})=1$}
            {
                Flip VV-type variable node $v^1v^2_{\mathrm{k}}$ and update $\tilde{\boldsymbol{\sigma}}$ accordingly (cf. Algorithm~\ref{alg:derived_decoders_from_parent_decoders})\;
            }
            $\hat{\mathbf{e}}_{\mathrm{cc}}\leftarrow D_{\mathrm{cc}}(\tilde{\boldsymbol{\sigma}})$ and update $(\hat{\mathbf{e}},\tilde{\boldsymbol{\sigma}})$ accordingly\;
        }
    }
  }
  \caption{Amended decoder that resolves \emph{stabilizer-induced} TSs when the TS is the graph product of a tree and a subgraph from the two parent codes.}
  \label{alg:amended_derived_decoders_from_parent_decoders}
\end{algorithm}
\begin{remark}[Notational differences between Algorithms~\ref{alg:derived_decoders_from_parent_decoders} and~\ref{alg:amended_derived_decoders_from_parent_decoders}]
\label{remark:alg1_alg2_notational_difference}
Compared with Algorithm~\ref{alg:derived_decoders_from_parent_decoders}, Algorithm~\ref{alg:amended_derived_decoders_from_parent_decoders} suppresses explicit index-level updates of the syndrome and error estimate to avoid clutter. The exact positions of the flipped variable nodes and the corresponding syndrome updates can be mapped in the same way as in Algorithm~\ref{alg:derived_decoders_from_parent_decoders}.
\end{remark}

To address this, we amend Step~8 of Algorithm~\ref{alg:derived_decoders_from_parent_decoders} as follows. In the pass of the outer loop corresponding to a leaf check node $c^1_{\mathrm{j}'}$ (with shared neighbor $v^1_{\mathrm{sh}}$), instead of choosing a single unsatisfied VC-type check, we choose two distinct unsatisfied VC-type checks, say $v^1_{a'}c^2_{b}$ and $v^1_{\bar{a}'}c^2_{\bar{b}}$, with $v^1_{a'},v^1_{\bar{a}'}\in\mathcal{N}_{c^1_{\mathrm{j}'}}$ and $c^2_{b},c^2_{\bar{b}} \in \mathcal{C}^2$, if they exist.
We perform the preprocessing steps (Steps~8--12) using $v^1_{a'}c^2_{b}$, then run $D_{\mathrm{vv}}$, and test whether all VC-type checks in $(\mathcal{N}_{c^1_{\mathrm{j}'}}\setminus\{v^1_{\mathrm{sh}}\})\times\mathscr{T}(\mathcal{V}^2_{\mathrm{s}})$ are satisfied, i.e., whether $\tilde{\boldsymbol{\sigma}}$ is zero on those checks. If the test fails, we undo the CC-type flips performed in this attempt and repeat the same procedure using $v^1_{\bar{a}'}c^2_{\bar{b}}$. In Lemma~\ref{lemma:composite_trapping_set_decoding_condition_tree}, we show that one of these two attempts succeeds in each pass associated with a leaf check node; the proof then peels the verified leaf from $\mathscr{T}(\mathcal{C}^1_{\mathrm{s}})$ and iterates on the remaining smaller tree. By symmetry, Step~20 of Algorithm~\ref{alg:derived_decoders_from_parent_decoders} admits an analogous amendment when $\mathscr{T}(\mathcal{V}^2_{\mathrm{s}})$ is a tree and $\mathscr{T}(\mathcal{C}^1_{\mathrm{s}})$ is a TS. The decoder with these amended steps is described in Algorithm~\ref{alg:amended_derived_decoders_from_parent_decoders}.
\begin{lemma}
    \label{lemma:composite_trapping_set_decoding_condition_tree}
    Consider the Tanner graph $\mathscr{G}_{\mathrm{z}}$ obtained by the HGP/LP construction in Lemma~\ref{lemma:decoder_resolving_classical_plus_stab_TS}. Define
    \[
        \mathscr{T}_{\mathrm{stab}} \triangleq \mathscr{T}(\mathcal{C}_{\mathrm{s}}^1) \times \mathscr{T}(\mathcal{V}^2_{\mathrm{s}}),
    \]
    where $\mathscr{T}(\mathcal{C}_{\mathrm{s}}^1)$ and $\mathscr{T}(\mathcal{V}^2_{\mathrm{s}})$ are subgraphs in $\mathscr{G}_1^{\mathsf{T}}$ and $\mathscr{G}_2$, respectively. Assume that one of these two subgraphs is a tree. In addition, assume that $\mathscr{T}_{\mathrm{stab}}$ does not contain the support of any nontrivial logical operator (equivalently, if any error pattern supported in $\mathscr{T}_{\mathrm{stab}}$ has all-zero syndrome, then it is a stabilizer).
    
    \begin{enumerate}
        \item Assume that every $c^1\in\mathcal{C}_{\mathrm{s}}^1$ has a degree at least $5$ in $\mathscr{G}_1^{\mathsf{T}}$, i.e., $\lvert\mathcal{N}_{c^1}\rvert\ge 5$, and that $D_2$ resolves $\mathscr{T}(\mathcal{V}^2_{\mathrm{s}})$. If $\mathscr{T}(\mathcal{C}_{\mathrm{s}}^1)$ is a tree and $\mathscr{T}(\mathcal{V}^2_{\mathrm{s}})$ is a TS, then Algorithm~\ref{alg:amended_derived_decoders_from_parent_decoders} successfully decodes $\mathscr{T}_{\mathrm{stab}}$.

        \item Assume that every $v^2\in\mathcal{V}^2_{\mathrm{s}}$ has a degree at least $5$ in $\mathscr{G}_2$, i.e., $\lvert\mathcal{N}_{v^2}\rvert\ge 5$, and that $D_1^\mathsf{T}$ resolves $\mathscr{T}(\mathcal{C}_{\mathrm{s}}^1)$. If $\mathscr{T}(\mathcal{V}^2_{\mathrm{s}})$ is a tree and $\mathscr{T}(\mathcal{C}_{\mathrm{s}}^1)$ is a TS, then Algorithm~\ref{alg:amended_derived_decoders_from_parent_decoders} successfully decodes $\mathscr{T}_{\mathrm{stab}}$.
    \end{enumerate}
\end{lemma}

\begin{proof}
We prove only Part~(1) of the lemma; Part~(2) follows by symmetry, swapping the roles of the two factors.

Let $\mathscr{T}(\mathcal{C}_{\mathrm{s}}^1)$ be a tree in $\mathscr{G}_1^{\mathsf{T}}$, and let $c^1_{\mathrm{j}'}$ be a leaf check node in this tree.
Then there exists a unique shared neighbor $v^1_{\mathrm{sh}}\in\mathcal{N}_{c^1_{\mathrm{j}'}}$ such that $v^1_{\mathrm{sh}}$ is also adjacent to some other check node $c^1_{\mathrm{nbr}}\neq c^1_{\mathrm{j}'}$ in $\mathscr{T}(\mathcal{C}_{\mathrm{s}}^1)$. Moreover, every $v^1\in\mathcal{N}_{c^1_{\mathrm{j}'}}\setminus\{v^1_{\mathrm{sh}}\}$ is not adjacent to any other check node in $\mathscr{T}(\mathcal{C}_{\mathrm{s}}^1)$.
Fix the pass of the outer loop in Steps~3--14 of Algorithm~\ref{alg:amended_derived_decoders_from_parent_decoders} indexed by $c^1_{\mathrm{\ell}'}=c^1_{\mathrm{j}'}$. 
We claim that after this pass and the subsequent run of $D_{\mathrm{vv}}$, all VC-type checks of the form $v^1c^2$, with $v^1\in \mathcal{N}_{c^1_{\mathrm{j}'}}\setminus\{v^1_{\mathrm{sh}}\}$ and $c^2\in \mathscr{T}(\mathcal{V}^2_{\mathrm{s}})\cap\mathcal{C}^2$, are satisfied, and the decoder's estimate on $\bigl(\mathscr{T}(c^1_{\mathrm{j}'}) \setminus \{v^1_{\mathrm{sh}}\}\bigr) \times \mathscr{T}(\mathcal{V}^2_{\mathrm{s}})$ matches the actual error up to a stabilizer.
Contingent upon this claim being true, it suffices to peel $c^1_{\mathrm{j}'}$ from $\mathscr{T}(\mathcal{C}^1_{\mathrm{s}})$ and analyze the decoder on $\mathscr{T}(\mathcal{C}^1_{\mathrm{s}}\setminus\{c^1_{\mathrm{j}'}\}) \times\mathscr{T}(\mathcal{V}^2_{\mathrm{s}})$. Since $\mathscr{T}(\mathcal{C}^1_{\mathrm{s}}\setminus\{c^1_{\mathrm{j}'}\})$ is a tree, the peeling procedure can be iterated to show that the decoding is successful on the whole of $\mathscr{T}_{\mathrm{stab}}$.
To complete the proof of the lemma, we need to prove the claim.
We divide the proof of the claim into two cases depending on whether the condition in Step~3 of Algorithm~\ref{alg:amended_derived_decoders_from_parent_decoders} is satisfied.

If the condition in Step~3 is satisfied, then by definition there are at most $\left\lfloor\lvert\mathcal{N}_{c^1_{\mathrm{j}'}}\rvert/2\right\rfloor$ nodes $v^1\in\mathcal{N}_{c^1_{\mathrm{j}'}}$ that participate in an unsatisfied check of the form $v^1c^2$, where $c^2\in \mathscr{T}(\mathcal{V}^2_{\mathrm{s}})\cap\mathcal{C}^2$.
Since $\lvert\mathcal{N}_{c^1_{\mathrm{j}'}}\rvert\ge 5$, the number of nodes in $\mathcal{N}_{c^1_{\mathrm{j}'}}$ that do not participate in any such unsatisfied check is at least
$\lvert\mathcal{N}_{c^1_{\mathrm{j}'}}\rvert-\left\lfloor\lvert\mathcal{N}_{c^1_{\mathrm{j}'}}\rvert/2\right\rfloor \ge 3$.
In particular, there exist at least two variable nodes $v^1_{a'}, v^1_{\bar{a}'}\in\mathcal{N}_{c^1_{\mathrm{j}'}}$ such that all checks of the form $v^1_{\mathrm{t}'}c^2$, with $\mathrm{t}' \in\{a',\bar{a}'\}$ and $c^2\in \mathscr{T}(\mathcal{V}^2_{\mathrm{s}})\cap\mathcal{C}^2$, are satisfied. Equivalently, all VC-type checks in $v^1_{\mathrm{t}'} \times \mathscr{T}(\mathcal{V}^2_{\mathrm{s}})$ are satisfied.
Recall that $v^1_{\mathrm{sh}}$ is the only neighbor of the leaf check node $c^1_{\mathrm{j}'}$ that is also adjacent to another check node in the tree $\mathscr{T}(\mathcal{C}^1_{\mathrm{s}})$. Hence, at least one of $v^1_{a'}$ and $v^1_{\bar{a}'}$ is different from $v^1_{\mathrm{sh}}$; without loss of generality, assume $v^1_{a'} \neq v^1_{\mathrm{sh}}$. The decoder is agnostic to which neighbor of $c^1_{\mathrm{j}'}$ is shared; the proof only uses the existence of such a choice.

For such a choice of $v^1_{a'}\neq v^1_{\mathrm{sh}}$, subgraph $v^1_{a'}\times\mathscr{T}(\mathcal{V}^2_{\mathrm{s}})$ does not contain unsatisfied VC-type checks. Moreover, since $v^1_{a'}$ is not shared with any other check node in $\mathscr{T}(\mathcal{C}^1_{\mathrm{s}})$, every VC-type check node in $v^1_{a'}\times\mathscr{T}(\mathcal{V}^2_{\mathrm{s}})$ has all its neighbors contained in $\mathscr{T}(c^1_{\mathrm{j}'})\times\mathscr{T}(\mathcal{V}^2_{\mathrm{s}})$. Hence, the syndrome restricted to $v^1_{a'}\times\mathscr{T}(\mathcal{V}^2_{\mathrm{s}})$ is not affected by errors in $\left(\mathscr{T}\left(\mathcal{C}^1_{\mathrm{s}}\setminus\{c^1_{\mathrm{j}'}\}\right)\right)\times\mathscr{T}(\mathcal{V}^2_{\mathrm{s}})$.
Therefore, we may regard the subproblem as an example of the setting of Lemma~\ref{lemma:error_exclusively_on_VV_condition} for stabilizer-induced TS $\mathscr{T}(c^1_{\mathrm{j}'})\times\mathscr{T}(\mathcal{V}^2_{\mathrm{s}})$. Since all VC-type checks in $v^1_{a'}\times\mathscr{T}(\mathcal{V}^2_{\mathrm{s}})$ are satisfied, Lemma~\ref{lemma:error_exclusively_on_VV_condition} implies that there exists a syndrome-consistent error pattern supported only on VV-type variable nodes within $\bigl(\mathscr{T}(c^1_{\mathrm{j}'})\setminus\{v^1_{\mathrm{sh}}\}\bigr)\times\mathscr{T}(\mathcal{V}^2_{\mathrm{s}})$.
Since the code is degenerate, syndrome-based decoding can determine the error only up to a stabilizer. Any additional discrepancy supported within $\bigl(\mathscr{T}(c^1_{\mathrm{j}'})\setminus\{v^1_{\mathrm{sh}}\}\bigr)\times\mathscr{T}(\mathcal{V}^2_{\mathrm{s}})\subseteq\mathscr{T}_{\mathrm{stab}}$ can be a nontrivial logical operator only if $\mathscr{T}_{\mathrm{stab}}$ contains the support of a nontrivial logical operator, which is excluded by the assumption in the lemma.
Since $D_2$ resolves $\mathscr{T}(\mathcal{V}^2_{\mathrm{s}})$, the subsequent run of $D_{\mathrm{vv}}$ corrects this error pattern. In particular, all VC-type checks in $\bigl(\mathcal{N}_{c^1_{\mathrm{j}'}}\setminus\{v^1_{\mathrm{sh}}\}\bigr)\times\mathscr{T}(\mathcal{V}^2_{\mathrm{s}})$ are satisfied, and the resulting estimate matches the actual error up to a stabilizer on $\bigl(\mathscr{T}(c^1_{\mathrm{j}'})\setminus\{v^1_{\mathrm{sh}}\}\bigr)\times\mathscr{T}(\mathcal{V}^2_{\mathrm{s}})$.

Otherwise, when the condition in Step~3 is not satisfied, the algorithm chooses two unsatisfied VC-type checks of the form $v^1_{a'}c^2_{b}$ and $v^1_{\bar{a}'}c^2_{\bar{b}}$ with distinct $v^1_{a'},v^1_{\bar{a}'}\in\mathcal{N}_{c^1_{\mathrm{j}'}}$.
Since the condition in Step~3 is not satisfied, at least $\lfloor|\mathcal{N}_{c^1_{\mathrm{j}'}}|/2\rfloor+1$ nodes in $\mathcal{N}_{c^1_{\mathrm{j}'}}$ participate in an unsatisfied check of the form $v^1c^2$ with $c^2\in\mathscr{T}(\mathcal{V}^2_{\mathrm{s}})\cap\mathcal{C}^2$. Because $|\mathcal{N}_{c^1_{\mathrm{j}'}}|\ge 5$, we have $\lfloor|\mathcal{N}_{c^1_{\mathrm{j}'}}|/2\rfloor+1\ge 3$, and therefore the algorithm can select two such unsatisfied checks with distinct endpoints $v^1_{a'}$ and $v^1_{\bar{a}'}$.

After choosing two VC-type checks that were not satisfied, the algorithm first performs the preprocessing steps (Steps~3--14) using the check $v^1_{a'}c^2_{b}$, then runs $D_{\mathrm{vv}}$. If after this attempt, all VC-type checks in $(\mathcal{N}_{c^1_{\mathrm{j}'}}\setminus\{v^1_{\mathrm{sh}}\})\times\mathscr{T}(\mathcal{V}^2_{\mathrm{s}})$ are satisfied, then it stops. Otherwise, it reverts the CC-type flips performed in this attempt and repeats the same procedure using the check $v^1_{\bar{a}'}c^2_{\bar{b}}$. Note that this test can succeed only if there exists an error pattern supported exclusively on VV-type variable nodes whose syndrome agrees with the measured syndrome on $(\mathcal{N}_{c^1_{\mathrm{j}'}}\setminus\{v^1_{\mathrm{sh}}\})\times\mathscr{T}(\mathcal{V}^2_{\mathrm{s}})$, since $D_{\mathrm{vv}}$ runs on the subgraph induced by VV-type variable nodes.
It remains to show that at least one of these two attempts succeeds. After the preprocessing steps (Steps~3--14) using $v^1_{a'}c^2_{b}$, the subgraph $v^1_{a'}\times\mathscr{T}(\mathcal{V}^2_{\mathrm{s}})$ does not contain unsatisfied VC-type checks (by construction of Steps~3--14). If $v^1_{a'}\neq v^1_{\mathrm{sh}}$, then $v^1_{a'}$ is not shared with any other check node in $\mathscr{T}(\mathcal{C}^1_{\mathrm{s}})$, and therefore every VC-type check in $v^1_{a'}\times\mathscr{T}(\mathcal{V}^2_{\mathrm{s}})$ has all its neighbors contained in $\mathscr{T}(c^1_{\mathrm{j}'})\times\mathscr{T}(\mathcal{V}^2_{\mathrm{s}})$. Hence, the syndrome on $v^1_{a'}\times\mathscr{T}(\mathcal{V}^2_{\mathrm{s}})$ is unaffected by errors in $(\mathscr{T}(\mathcal{C}_{\mathrm{s}}^1)\setminus\{c^1_{\mathrm{j}'}\})\times\mathscr{T}(\mathcal{V}^2_{\mathrm{s}})$. In this situation, we can apply Lemma~\ref{lemma:error_exclusively_on_VV_condition} to stabilizer-induced TS $\mathscr{T}(c^1_{\mathrm{j}'})\times\mathscr{T}(\mathcal{V}^2_{\mathrm{s}})$: the absence of unsatisfied VC-type checks in the slice $v^1_{a'}\times\mathscr{T}(\mathcal{V}^2_{\mathrm{s}})$ implies the existence of a syndrome-consistent error pattern supported only on VV-type variable nodes within $\bigl(\mathscr{T}(c^1_{\mathrm{j}'})\setminus\{v^1_{\mathrm{sh}}\}\bigr)\times\mathscr{T}(\mathcal{V}^2_{\mathrm{s}})$. Since $D_2$ resolves $\mathscr{T}(\mathcal{V}^2_{\mathrm{s}})$, the subsequent run of $D_{\mathrm{vv}}$ corrects this pattern.

Therefore, if the first attempt fails, we must have $v^1_{a'}=v^1_{\mathrm{sh}}$, and since $v^1_{\bar{a}'}\neq v^1_{a'}$ it follows that $v^1_{\bar{a}'}\neq v^1_{\mathrm{sh}}$. In the second attempt, the syndrome on $v^1_{\bar{a}'}\times\mathscr{T}(\mathcal{V}^2_{\mathrm{s}})$ is not affected by errors in $(\mathscr{T}(\mathcal{C}_{\mathrm{s}}^1\setminus\{c^1_{\mathrm{j}'}\}))\times\mathscr{T}(\mathcal{V}^2_{\mathrm{s}})$. Therefore, the same application of Lemma~\ref{lemma:error_exclusively_on_VV_condition} together with a run of $D_{\mathrm{vv}}$ shows that all VC-type checks in $(\mathcal{N}_{c^1_{\mathrm{j}'}}\setminus\{v^1_{\mathrm{sh}}\})\times\mathscr{T}(\mathcal{V}^2_{\mathrm{s}})$ are satisfied, and by the additional assumption of the lemma the resulting estimate matches the actual error up to a stabilizer on $\bigl(\mathscr{T}(c^1_{\mathrm{j}'})\setminus\{v^1_{\mathrm{sh}}\}\bigr)\times\mathscr{T}(\mathcal{V}^2_{\mathrm{s}})$.
We have thus shown that, after completing the pass corresponding to the leaf check node $c^1_{\mathrm{j}'}$ (including the subsequent run of $D_{\mathrm{vv}}$), all VC-type checks in $\bigl(\mathscr{T}(c^1_{\mathrm{j}'})\setminus\{v^1_{\mathrm{sh}}\}\bigr)\times\mathscr{T}(\mathcal{V}^2_{\mathrm{s}})$ are satisfied. Removing this verified leaf reduces the problem to a strictly smaller tree; repeating the same argument completes the proof.

\end{proof}
\begin{remark}[On the degree constraints in Lemma~\ref{lemma:composite_trapping_set_decoding_condition_tree}]
\label{remark:degree_constraint_artifact}
The check/variable-node degree constraints in Lemma~\ref{lemma:composite_trapping_set_decoding_condition_tree} are artifacts of the proof technique. They are used to guarantee the existence of two candidate VC-type checks with distinct endpoints, at least one of which avoids the shared node. In practice, we observe that the amended decoder succeeds in most cases already when the relevant degrees are at least $3$. This is an important consideration in code design, since increasing degrees typically decreases the code rate and can increase the number of gates in the syndrome-measurement circuit, and hence, the exposure to noise in the \emph{circuit-level} noise model.
\end{remark}

In Lemma~\ref{lemma:decoder_resolving_classical_plus_stab_TS} and ~\ref{lemma:composite_trapping_set_decoding_condition_tree}, we showed that, under the stated degree/tree hypotheses, and using Algorithm~\ref{alg:amended_derived_decoders_from_parent_decoders}, decoders derived from the parent classical decoders can resolve certain \emph{stabilizer-induced} TSs that contain multiple copies of a classical TS, provided that the parent decoder already resolves the underlying classical TS. Thus, at least for this family of large \emph{stabilizer-induced} structures, one can avoid decoding failures without designing an entirely new decoder from scratch. This is notable because in the \emph{circuit-level} noise model considered here, such large TSs can be activated by a small number of CNOT failures despite involving many variable nodes.

\section{TSs in HGP and LP Codes from TSs of Two Parent LDPC Codes}
\label{sec:comp_TS}
In the previous section, we considered families of \emph{stabilizer-induced} TSs obtained as the graph product of a TS from one parent code and a tree-like subgraph from the other. We showed that the decoding failures caused by these \emph{stabilizer-induced} TSs can be avoided provided that the chosen parent-code decoder already resolves the corresponding parent-code TS.
In this section, we study \emph{stabilizer-induced} TSs obtained as the graph product of two TSs, one from each parent code. We refer to such TSs as composite \emph{stabilizer-induced} TSs.
Resolving the \emph{stabilizer-induced} TSs considered in the previous section can reduce the logical error rate, but it does not change the error-floor slope of the logical error rate curve versus the physical error rate, since that slope is determined by the minimum-weight uncorrectable pattern.
Such minimum-weight uncorrectable patterns may be supported on composite \emph{stabilizer-induced} TSs. In general, these composite \emph{stabilizer-induced} TSs are not handled by a straightforward reuse of the parent decoders, especially when those decoders are not robust to syndrome noise.
The following lemma describes several families of composite \emph{stabilizer-induced} TSs and characterizes their structure.

\begin{lemma}

\label{lemma:composite_TS_containing_single_TS_from_parent_code}
Let an HGP/LP code be constructed from two classical LDPC codes with Tanner graphs $\mathscr{G}_1$ and $\mathscr{G}_2$. Let $\mathscr{G}_1^{\mathsf{T}}$ and $\mathscr{G}_2^{\mathsf{T}}$ denote the graphs obtained from $\mathscr{G}_1$ and $\mathscr{G}_2$, respectively, by exchanging the roles of the variable and the check nodes. Denote by $\mathscr{G}_{\mathrm{x}}$ and $\mathscr{G}_{\mathrm{z}}$ the Tanner graphs corresponding to the $X$- and $Z$-stabilizer generators of the quantum code.
Consider a TS $\mathscr{T}(\mathcal{C}_{\mathrm{s}}^1)$ in $\mathscr{G}_1^{\mathsf{T}}$ and a TS $\mathscr{T}(\mathcal{V}_{\mathrm{s}}^1)$ in $\mathscr{G}_1$. Similarly, consider a TS $\mathscr{T}(\mathcal{V}_{\mathrm{s}}^2)$ in $\mathscr{G}_2$ and a TS $\mathscr{T}(\mathcal{C}_{\mathrm{s}}^2)$ in $\mathscr{G}_2^{\mathsf{T}}$. Here, $\mathcal{C}_{\mathrm{s}}^1$ and $\mathcal{C}_{\mathrm{s}}^2$ are subsets of check nodes, and $\mathcal{V}_{\mathrm{s}}^1$ and $\mathcal{V}_{\mathrm{s}}^2$ are subsets of variable nodes, in $\mathscr{G}_1$ and $\mathscr{G}_2$, respectively. (In Parts~1--3 below we only use $\mathscr{T}(\mathcal{C}_{\mathrm{s}}^1)$ and $\mathscr{T}(\mathcal{V}_{\mathrm{s}}^2)$; the other two TSs are introduced for symmetry.)

\begin{enumerate}
\item The subgraphs $v_{\mathrm{i}'}^1 \times \mathscr{T}(\mathcal{V}_{\mathrm{s}}^2)$ and $\mathscr{T}(\mathcal{C}_{\mathrm{s}}^1) \times c_{\mathrm{j}}^2$, for $v_{\mathrm{i}'}^1 \in \mathscr{T}(\mathcal{C}^1_\mathrm{s})$ and $c_{\mathrm{j}}^2 \in \mathscr{T}(\mathcal{V}^2_\mathrm{s})$, in the Tanner graph $\mathscr{G}_{\mathrm{z}}$ have only one common node, namely the VC-type check node $v_{\mathrm{i}'}^1 c_{\mathrm{j}}^2$.
\item Define the subgraph
\[
\mathscr{T}_{\mathrm{comp}} \triangleq
\bigl( \mathcal{V}_{\mathrm{s}^\star}^1 \times \mathscr{T}(\mathcal{V}_{\mathrm{s}}^2) \bigr)
\,\cup\,
\bigl( \mathscr{T}(\mathcal{C}_{\mathrm{s}}^1) \times \mathcal{C}_{\mathrm{s}^\star}^2 \bigr)
\]
inside $\mathscr{G}_{\mathrm{z}}$, where
\begin{itemize}

\item $\mathcal{V}_{\mathrm{s}^\star}^1$ is chosen as a subset of the variable nodes in $\mathscr{T}(\mathcal{C}_{\mathrm{s}}^1)$ such that no two nodes in $\mathcal{V}_{\mathrm{s}^\star}^1$ share a common neighboring check node in $\mathscr{G}_1$, and
\item $\mathcal{C}_{\mathrm{s}^\star}^2$ is chosen as a subset of the check nodes in $\mathscr{T}(\mathcal{V}_{\mathrm{s}}^2)$ such that no two nodes in $\mathcal{C}_{\mathrm{s}^\star}^2$ share a common neighboring variable node in $\mathscr{G}_2$.
\end{itemize}
Then the support of no $X$-type stabilizer generator, equivalently no CV-type check, contains more than one variable node of $\mathscr{T}_{\mathrm{comp}}$. Hence, the number of CNOT faults required to activate $\mathscr{T}_{\mathrm{comp}}$ equals the number of data errors required to activate it.

\item Redefine the subgraph
\[
\mathscr{T}_{\mathrm{comp}} \triangleq
\bigl( \mathscr{T}(\mathcal{C}^{1}_{\mathrm{s^\star}}) \times \mathscr{T}(\mathcal{V}^{2}_{\mathrm{s}}) \bigr)
\,\cup\,
\bigl( \mathscr{T}(\mathcal{C}^{1}_{\mathrm{s}}) \times \mathscr{T}(\mathcal{V}^{2}_{\mathrm{s^\star}}) \bigr)
\]
inside \(\mathscr{G}_{\mathrm{z}}\), where \(\mathcal{C}^{1}_{\mathrm{s^\star}} \subset \mathcal{C}^{1}_{\mathrm{s}}\) and \(\mathcal{V}^{2}_{\mathrm{s^\star}} \subset \mathcal{V}^{2}_{\mathrm{s}}\). The following statements hold.
\begin{itemize}
\item[(a)] The subgraph \(\mathscr{T}_{\mathrm{comp}}\) is induced by the CV-type checks (i.e., $X$-type) in the set
$
\bigl( \mathcal{C}^{1}_{\mathrm{s^\star}} \times \mathcal{V}^{2}_{\mathrm{s}} \bigr)
\,\cup\,
\bigl( \mathcal{C}^{1}_{\mathrm{s}} \times \mathcal{V}^{2}_{\mathrm{s^\star}} \bigr)$.

\item[(b)] There exist syndrome-measurement schedules (in general, different ones for different expressions) such that the number of CNOT faults required to activate \(\mathscr{T}_{\mathrm{comp}}\) is upper-bounded by the minimum of the following quantities:
\begin{align*}
& e\bigl(\mathscr{T}(\mathcal{C}^{1}_{\mathrm{s}}),\mathscr{G}_{1}^{\mathsf{T}}\bigr)\,\lvert\mathcal{V}^2_{\mathrm{s^\star}}\rvert
 + e\bigl(\mathscr{T}(\mathcal{V}^{2}_{\mathrm{s}}),\mathscr{G}_{2}\bigr)\,\lvert\mathcal{C}^1_{\mathrm{s^\star}}\rvert,\\
& e\bigl(\mathscr{T}(\mathcal{C}^{1}_{\mathrm{s^\star}}),\mathscr{G}_{1}^{\mathsf{T}}\bigr)\,\lvert\mathcal{V}^2_{\mathrm{s}}\rvert
 + e\bigl(\mathscr{T}(\mathcal{V}^{2}_{\mathrm{s^\star}}),\mathscr{G}_{2}\bigr)\,\lvert\mathcal{C}^1_{\mathrm{s}}\rvert,\\
& e\bigl(\mathscr{T}(\mathcal{C}^{1}_{\mathrm{s^\star}}),\mathscr{G}_{1}^{\mathsf{T}}\bigr)\,\lvert\mathcal{V}^2_{\mathrm{s}}\rvert
 + e\bigl(\mathscr{T}(\mathcal{C}^{1}_{\mathrm{s}}),\mathscr{G}_{1}^{\mathsf{T}}\bigr)\,\lvert\mathcal{V}^2_{\mathrm{s^\star}}\rvert,\\
& e\bigl(\mathscr{T}(\mathcal{V}^{2}_{\mathrm{s}}),\mathscr{G}_{2}\bigr)\,\lvert\mathcal{C}^1_{\mathrm{s^\star}}\rvert
 + e\bigl(\mathscr{T}(\mathcal{V}^{2}_{\mathrm{s^\star}}),\mathscr{G}_{2}\bigr)\,\lvert\mathcal{C}^1_{\mathrm{s}}\rvert.
\end{align*}
where $e(\mathscr{T},\mathscr{G})$ denotes the number of data errors required to activate the trapping set $\mathscr{T}$ in the Tanner graph $\mathscr{G}$.
\end{itemize}

\end{enumerate}

\end{lemma}
\begin{proof}
    To complete the proof of Part~1, we must show that for any \(v^1_{\mathrm{i}'} \in \mathscr{T}(\mathcal{C}_\mathrm{s}^1)\) and any
    \(c^2_{\mathrm{j}} \in \mathscr{T}(\mathcal{V}^2_{\mathrm{s}})\), the subgraphs
    \[
        v_{\mathrm{i}'}^1 \times \mathscr{T}(\mathcal{V}^2_{\mathrm{s}})
        \quad\text{and}\quad
        \mathscr{T}(\mathcal{C}_{\mathrm{s}}^1) \times c_{\mathrm{j}}^2
    \]
    intersect in \(\mathscr{G}_{\mathrm{z}}\) only in the VC-type check node \(v^1_{\mathrm{i}'}c^2_{\mathrm{j}}\).
First, consider the variable nodes. By construction, all variable nodes in the subgraph
    \(v_{\mathrm{i}'}^1 \times \mathscr{T}(\mathcal{V}^2_{\mathrm{s}})\)
    are of VV-type, while all the variable nodes in the subgraph
    \(\mathscr{T}(\mathcal{C}_{\mathrm{s}}^1) \times c_{\mathrm{j}}^2\)
    are of CC-type. Hence, these two subgraphs do not share any variable nodes.
 Next, consider the check nodes. In subgraph
    \(v_{\mathrm{i}'}^1 \times \mathscr{T}(\mathcal{V}^2_{\mathrm{s}})\),
    each check node has the form \(v^1_{\mathrm{i}'}c^2_{\mathrm{\ell}}\) with
    \(c^2_{\mathrm{\ell}} \in \mathscr{T}(\mathcal{V}^2_{\mathrm{s}})\). In subgraph
    \(\mathscr{T}(\mathcal{C}_{\mathrm{s}}^1) \times c_{\mathrm{j}}^2\),
    each check node has the form \(v^1_{\mathrm{k}'}c^2_{\mathrm{j}}\) with
    \(v^1_{\mathrm{k}'} \in \mathscr{T}(\mathcal{C}_\mathrm{s}^1)\). Therefore, the only check node that can possibly lie in the intersection of these two subgraphs is
    \(v^1_{\mathrm{i}'}c^2_{\mathrm{j}}\). 
Finally, since the variable node \(v^1_{\mathrm{i}'}\) belongs to
    \(\mathscr{T}(\mathcal{C}_\mathrm{s}^1)\) and the check node \(c^2_{\mathrm{j}}\) belongs to
    \(\mathscr{T}(\mathcal{V}^2_{\mathrm{s}})\), the check node \(v^1_{\mathrm{i}'}c^2_{\mathrm{j}}\)
    actually lies in both
    \(v_{\mathrm{i}'}^1 \times \mathscr{T}(\mathcal{V}^2_{\mathrm{s}})\) and
    \(\mathscr{T}(\mathcal{C}_{\mathrm{s}}^1) \times c_{\mathrm{j}}^2\).
    Thus, these two subgraphs intersect exactly in this VC-type check node, which completes the proof of Part~1.
    
    We now prove Part~2. Throughout the remainder of this paragraph, \(\mathscr{T}_{\mathrm{comp}}\) refers to
    the subgraph defined in Part~2. Now consider two VV-type variable nodes, say \(v^1_{\mathrm{i}'}v^2_{\mathrm{i}}\)
    and \(v^1_{\mathrm{k}'}v^2_{\mathrm{k}}\), in \(\mathscr{T}_{\mathrm{comp}}\). By construction,
    we have \(v^1_{\mathrm{i}'},v^1_{\mathrm{k}'} \in \mathcal{V}_{\mathrm{s}^\star}^1\). The lemma assumes that no
    two variable nodes in \(\mathcal{V}_{\mathrm{s}^\star}^1\) share a common neighboring check node in
    \(\mathscr{G}_1\). Therefore, there is no check node \(c^1\) such that
    \(c^1 \in \mathcal{N}_{v^1_{\mathrm{i}'}} \cap \mathcal{N}_{v^1_{\mathrm{k}'}}\). Since a CV-type check node
    \(c^1v^2\) fixes the second coordinate of the neighboring VV-type variable nodes, this implies that no CV-type check can contain two distinct VV-type
    variable nodes of \(\mathscr{T}_{\mathrm{comp}}\) in its support.

    An analogous argument applies to the CC-type variable nodes in \(\mathscr{T}_{\mathrm{comp}}\). The
    assumption that no two check nodes in \(\mathcal{C}_{\mathrm{s}^\star}^2\) share a common neighboring
    variable node in \(\mathscr{G}_2\) implies the following: for any \(v^2\in\mathcal{V}^2\), there is at most one
    check node \(c^2\in\mathcal{C}_{\mathrm{s}^\star}^2\) adjacent to \(v^2\). Hence, no CV-type check \(c^1v^2\)
    can contain two distinct CC-type variable nodes of \(\mathscr{T}_{\mathrm{comp}}\) in its support.

    Under the assumptions of Part~2, each CV-type check intersects \(\mathscr{T}_{\mathrm{comp}}\) in at most one
    VV-type variable node and at most one CC-type variable node (where which VV/CC nodes are relevant depends on
    whether the check lies in \(\mathcal{C}^{1}_{\mathrm{s^\star}}\times\mathcal{V}^{2}_{\mathrm{s}}\) or in
    \(\mathcal{C}^{1}_{\mathrm{s}}\times\mathcal{V}^{2}_{\mathrm{s^\star}}\)). Let $e_{\mathrm{comp}}$ denote the number of data errors
    required to activate \(\mathscr{T}_{\mathrm{comp}}\). Consider an activation pattern of
    weight $e_{\mathrm{comp}}$ supported on variable nodes of \(\mathscr{T}_{\mathrm{comp}}\). We claim that one
    can choose, for each erroneous variable node in this pattern, a CV-type check whose support contains that
    variable node and contains no other variable node of \(\mathscr{T}_{\mathrm{comp}}\), and such that these
    chosen CV-type checks are all distinct.

    For each erroneous VV-type variable node \(v^1v^2\), choose a CV-type check of the form \(c^1v^2\) with
    \(c^1\in\mathcal{N}_{v^1}\). The assumption on \(\mathcal{V}^1_{\mathrm{s}^\star}\) guarantees that two
    distinct \(v^1\) values do not share a common neighboring \(c^1\), so these choices can be made without
    reusing the same CV-type check. For each erroneous CC-type variable node \(c^1c^2\), choose a CV-type check
    of the form \(c^1v^2\) with \(v^2\in\mathcal{N}_{c^2}\). The assumption on
    \(\mathcal{C}^2_{\mathrm{s}^\star}\) guarantees that two distinct \(c^2\) values do not share a common
    neighboring \(v^2\), so these choices can also be made using distinct CV-type checks.

    For each chosen CV-type check, choose a syndrome-measurement schedule, that is, an order of CNOTs in its
    measurement circuit, and place the faulty CNOT so that the resulting propagated error lands on the targeted
    variable node, as in Part~4 of Lemma~\ref{lemma:stab_induced_TS_containg_classical_TS}.
    Since different targeted variable nodes use different CV-type checks, the resulting data errors do not
    cancel. Therefore, \(\mathscr{T}_{\mathrm{comp}}\) can be activated using $e_{\mathrm{comp}}$ CNOT failures.
    On the other hand, under the assumptions of Part~2, any single CNOT failure can corrupt at most one variable node of
    \(\mathscr{T}_{\mathrm{comp}}\), since no CV-type check contains two distinct variable nodes of \(\mathscr{T}_{\mathrm{comp}}\) in its support.
    Therefore, activating an error pattern of weight $e_{\mathrm{comp}}$ requires at least $e_{\mathrm{comp}}$ CNOT failures.
    Consequently, the number of CNOT failures required
    to activate \(\mathscr{T}_{\mathrm{comp}}\) equals the number of data errors required to
    activate it, which completes the proof of Part~2.

For Part~3(a), let
\[
\mathcal{S}_{\mathrm{CV}}
\coloneqq
\bigl( \mathcal{C}^{1}_{\mathrm{s^\star}} \times \mathcal{V}^{2}_{\mathrm{s}} \bigr)
\,\cup\,
\bigl( \mathcal{C}^{1}_{\mathrm{s}} \times \mathcal{V}^{2}_{\mathrm{s^\star}} \bigr)
\]
denote the set of CV-type checks from Part~3. To show that \(\mathscr{T}_{\mathrm{comp}}\) is exactly the subgraph induced by \(\mathcal{S}_{\mathrm{CV}}\), it suffices to prove both inclusions: (i) every variable node in the support of a check in \(\mathcal{S}_{\mathrm{CV}}\) lies in \(\mathscr{T}_{\mathrm{comp}}\), and (ii) every variable node of \(\mathscr{T}_{\mathrm{comp}}\) appears in the support of at least one check in \(\mathcal{S}_{\mathrm{CV}}\). Together these imply that the variable-node set of \(\mathscr{T}_{\mathrm{comp}}\) coincides with the set of variable nodes touched by \(\mathcal{S}_{\mathrm{CV}}\), hence \(\mathscr{T}_{\mathrm{comp}}\) is the induced subgraph.
The verification is identical to the argument in Part~3 of Lemma~\ref{lemma:stab_induced_TS_containg_classical_TS}, so we omit the  details.

For Part~3(b), using the same scheduling idea as in Part~4 of Lemma~\ref{lemma:stab_induced_TS_containg_classical_TS}, we obtain an explicit upper bound on the number of CNOT faults needed to activate \(\mathscr{T}_{\mathrm{comp}}\). Recall that \(\mathscr{T}_{\mathrm{comp}}\) is the union of the two product subgraphs
\[
\bigl(\mathscr{T}(\mathcal{C}^{1}_{\mathrm{s^\star}}) \times \mathscr{T}(\mathcal{V}^{2}_{\mathrm{s}})\bigr)
\,\cup\,
\bigl(\mathscr{T}(\mathcal{C}^{1}_{\mathrm{s}}) \times \mathscr{T}(\mathcal{V}^{2}_{\mathrm{s^\star}})\bigr).
\]
We prove the first bound in Part~3(b). The remaining three bounds follow by the same argument.
By Part~3(a), \(\mathscr{T}_{\mathrm{comp}}\) is the subgraph induced by the checks in \(\mathcal{S}_{\mathrm{CV}}\). Fix any stabilizer generator \(c^1v^2\in\mathcal{S}_{\mathrm{CV}}\). We may choose the CNOT order in its measurement circuit so that the Pauli error resulting from a single faulty CNOT propagates onto any prescribed subset of its neighboring VV-type variable nodes.
To realize the first bound in Part~3(b), proceed as follows.
\begin{itemize}
    \item For each \(v^2\in\mathcal{V}^{2}_{\mathrm{s^\star}}\), use the above choice of CNOT ordering to activate  \(\mathscr{T}(\mathcal{C}^{1}_{\mathrm{s}}) \times \mathscr{T}(v^2)\), for $v^2 \in \mathcal{V}^2_{\mathrm{s}}$, inside \(\mathscr{T}(\mathcal{C}^{1}_{\mathrm{s}})\times\mathscr{T}(\mathcal{V}^{2}_{\mathrm{s^\star}})\). This costs $e\bigl(\mathscr{T}(\mathcal{C}^{1}_{\mathrm{s}}),\mathscr{G}_{1}^{\mathsf{T}}\bigr)$ CNOT faults for each such  \(\mathscr{T}(\mathcal{C}^{1}_{\mathrm{s}}) \times \mathscr{T}(v^2)\) inside \(\mathscr{T}(\mathcal{C}^{1}_{\mathrm{s}})\times\mathscr{T}(\mathcal{V}^{2}_{\mathrm{s^\star}})\) (indexed by \(v^2\in\mathcal{V}^{2}_{\mathrm{s^\star}}\)), and therefore at most $e\bigl(\mathscr{T}(\mathcal{C}^{1}_{\mathrm{s}}),\mathscr{G}_{1}^{\mathsf{T}}\bigr)\,\lvert\mathcal{V}^{2}_{\mathrm{s^\star}}\rvert$ faults in total.
    \item For each \(c^1\in\mathcal{C}^{1}_{\mathrm{s^\star}}\), use the above choice of CNOT ordering to activate \(c^1 \times \mathscr{T}(\mathcal{V}^{2}_{\mathrm{s}})\) inside \(\mathscr{T}(\mathcal{C}^{1}_{\mathrm{s^\star}})\times\mathscr{T}(\mathcal{V}^{2}_{\mathrm{s}})\). This costs at most $e\bigl(\mathscr{T}(\mathcal{V}^{2}_{\mathrm{s}}),\mathscr{G}_{2}\bigr)\,\lvert\mathcal{C}^{1}_{\mathrm{s^\star}}\rvert$ faults.
\end{itemize}
Summing gives the first expression.

The other three expressions follow from the same CNOT-ordering argument, but with a different choice of what is activated in each of the two product components of \(\mathscr{T}_{\mathrm{comp}}\): in \(\mathscr{T}(\mathcal{C}^{1}_{\mathrm{s^\star}})\times\mathscr{T}(\mathcal{V}^{2}_{\mathrm{s}})\) one may activate either the \(\mathscr{T}(\mathcal{C}^{1}_{\mathrm{s^\star}})\)-side or the \(\mathscr{T}(\mathcal{V}^{2}_{\mathrm{s}})\)-side, and likewise in \(\mathscr{T}(\mathcal{C}^{1}_{\mathrm{s}})\times\mathscr{T}(\mathcal{V}^{2}_{\mathrm{s^\star}})\). The four possible choices give the four candidate bounds in Part~3(b); taking their minimum yields a valid upper bound. This completes the proof of Part~3.
\end{proof}
\begin{figure*}
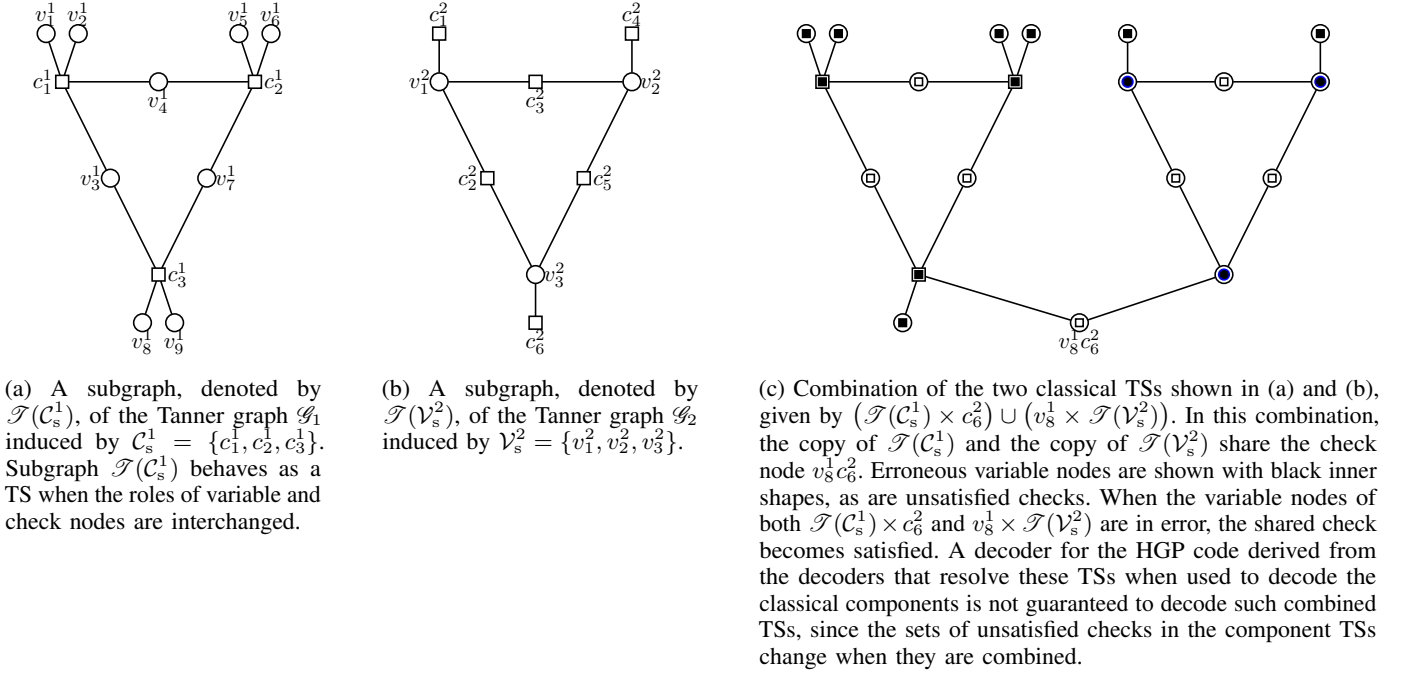

    \begin{subfigure}[t]{0.23\textwidth}
    \centering
        \input{Figures/TS_from_two_classical_TSs/classical_TS_from_parent_code_1}
        
        \subcaption{A subgraph, denoted by $\mathscr{T}(\mathcal{C}_{\mathrm{s}}^1)$, of the Tanner graph $\mathscr{G}_1$ induced by $\mathcal{C}_{\mathrm{s}}^1=\{c^1_1,c^1_2,c^1_3\}$. Subgraph $\mathscr{T}(\mathcal{C}^1_{\mathrm{s}})$ behaves as a TS when the roles of variable and check nodes are interchanged.}
        \label{fig:Classical_TS_from_code_1}
    \end{subfigure}
    \hfill
    \begin{subfigure}[t]{0.23\textwidth}
    \centering
        \input{Figures/TS_from_two_classical_TSs/classical_TS_from_parent_code2}
        \subcaption{A subgraph, denoted by $\mathscr{T}(\mathcal{V}_{\mathrm{s}}^2)$, of the Tanner graph $\mathscr{G}_2$ induced by $\mathcal{V}_{\mathrm{s}}^2=\{v^2_1,v^2_2,v^2_3\}$.}
       \label{fig:Classical_TS_from_code_2}
    \end{subfigure}
    \hfill
     \begin{subfigure}[t]{0.45\textwidth}
    \centering
        \input{Figures/TS_from_two_classical_TSs/combination_of_two_classical_TSs}
        
        \subcaption{Combination of the two classical TSs shown in (a) and (b), given by $\bigl(\mathscr{T}(\mathcal{C}^1_{\mathrm{s}}) \times c^2_6\bigr) \cup \bigl(v^1_8 \times \mathscr{T}(\mathcal{V}^2_{\mathrm{s}})\bigr)$. In this combination, the copy of $\mathscr{T}(\mathcal{C}^1_{\mathrm{s}})$ and the copy of $\mathscr{T}(\mathcal{V}^2_{\mathrm{s}})$ share the check node $v^1_8 c^2_6$. Erroneous variable nodes are shown with black inner shapes, as are unsatisfied checks. When the variable nodes of both $\mathscr{T}(\mathcal{C}^1_{\mathrm{s}}) \times c^2_6$ and $v^1_8 \times \mathscr{T}(\mathcal{V}^2_{\mathrm{s}})$ are in error, the shared check becomes satisfied. A decoder for the HGP code derived from the decoders that resolve these TSs when used to decode the classical components is not guaranteed to decode such combined TSs, since the sets of unsatisfied checks in the component TSs change when they are combined.}
        \label{fig:combination_of_two_classical_TS}
    \end{subfigure}
    \caption{Combining two TSs, each taken from a different parent code, to obtain a TS in the resulting HGP code.}
    \label{fig:combination_of_two_classical_TSs_main}
   \end{figure*}
Part~1 of Lemma~\ref{lemma:composite_TS_containing_single_TS_from_parent_code} explains how TSs inherited from the two parent codes can combine to form a composite TS in an HGP/LP code. In particular, when the TSs from both parents are simultaneously active, decoders derived from the parent-code decoders can fail even if each parent decoder corrects its own TS in isolation.
The reason is that the two component TSs can share check nodes. In that case, the syndrome bit at a shared check is determined jointly by errors from both components, so the syndrome pattern of the composite TS need not decompose into the syndrome patterns of the two parent TSs. The same mechanism underlies the composite TS families in Parts~2 and~3.
Parts~2 and~3 further distinguish composite TSs that can be activated by a number of CNOT failures equal to the number of data errors required to induce them from those for which these numbers can differ (depending on the measurement schedule). This distinction is useful when designing decoders that provably eliminate all error patterns up to a target Hamming weight, which in turn fixes the slope of the logical error-rate curve in the error-floor regime.

\begin{figure*}

 \centering
 \input{Figures/TSs_from_two_codes_and_circuit_level_error/decoding_illustration_iteration_0}    
    \caption{The \emph{stabilizer-induced} trapping set $\mathscr{T}_{\mathrm{stab}} \coloneqq \bigl(\mathscr{T}(\mathcal{C}^1_{\mathrm{s}}) \times \mathscr{T}(v^2_3)\bigr)\cup\bigl(\mathscr{T}(c^1_3) \times \mathscr{T}(\mathcal{V}^2_{\mathrm{s}})\bigr)$, where $\mathscr{T}(\mathcal{C}^1_{\mathrm{s}})$ and $\mathscr{T}(\mathcal{V}^2_{\mathrm{s}})$ are TSs in $\mathscr{G}_1^{\mathsf{T}}$ and $\mathscr{G}_2$, respectively, $c^1_3$ is a check node in $\mathscr{G}_1$, $v^2_3$ is a variable node in $\mathscr{G}_2$, and $\mathscr{T}(c^1_3)$, $\mathscr{T}(v^2_3)$ denote the subgraphs induced by these single nodes. Because $v^2_3$ has degree $3$ in $\mathscr{G}_2$, the first product term contributes three copies of $\mathscr{T}(\mathcal{C}^1_{\mathrm{s}})$; because $c^1_3$ has degree $4$ in $\mathscr{G}_1$, the second contributes four copies of $\mathscr{T}(\mathcal{V}^2_{\mathrm{s}})$. These copies are interconnected through the check nodes adjacent to (i)~the CC-type variable nodes $c_3^1 c_\mathrm{j}^2$ with $c^2_\mathrm{j} \in \mathcal{N}_{v_3^2}$, and (ii)~the VV-type variable nodes $v^1_{\mathrm{i}'} v^2_3$ with $v^1_{\mathrm{i}'} \in \mathcal{N}_{c^1_3}$. To reduce clutter, only the nodes connecting the copies of $\mathscr{T}(\mathcal{C}^1_{\mathrm{s}})$ to those of $\mathscr{T}(\mathcal{V}^2_{\mathrm{s}})$ are labeled. By arguments analogous to Part~3 of Lemma~\ref{lemma:stab_induced_TS_containg_classical_TS}, $\mathscr{T}_{\mathrm{stab}}$ is induced by the CV-type stabilizer generators in $(\mathcal{C}^1_{\mathrm{s}} \times v^2_3) \cup (c^1_3 \times \mathcal{V}^2_{\mathrm{s}})$.}
\label{fig:combination_of_two_classical_TS_circuit_level_noise}
 \end{figure*}

We now illustrate the conclusions of Lemma~\ref{lemma:composite_TS_containing_single_TS_from_parent_code} using the composite trapping-set examples in Fig.~\ref{fig:combination_of_two_classical_TSs_main} and Fig.~\ref{fig:combination_of_two_classical_TS_circuit_level_noise}. In particular, Fig.~\ref{fig:Classical_TS_from_code_1} shows a $(3,6)$ TS $\mathscr{T}(\mathcal{C}_\mathrm{s}^1)$ in $\mathscr{G}_1^\mathsf{T}$, and Fig.~\ref{fig:Classical_TS_from_code_2} shows a $(3,3)$ TS $\mathscr{T}(\mathcal{V}_\mathrm{s}^2)$ in $\mathscr{G}_2$.
The composite TS $\bigl(\mathscr{T}(\mathcal{C}^1_{\mathrm{s}}) \times c_6^2\bigr) \cup \bigl(v^1_8 \times \mathscr{T}(\mathcal{V}^2_{\mathrm{s}})\bigr)$ in Fig.~\ref{fig:combination_of_two_classical_TS} illustrates Part~1: it contains one copy of each parent TS, and these two copies share the VC-type check node $v^1_8 c^2_6$. When both copies are active, the parity bit at the shared check can change (or even cancel), so the composite syndrome no longer decomposes into the two component syndromes. Figure~\ref{fig:combination_of_two_classical_TS_circuit_level_noise} illustrates the structural characterization in Parts~2--3: the composite TS contains multiple interconnected copies of the parent TSs and can be described as the subgraph induced by an explicit set of CV-type checks (equivalently, by a set of $X$-type stabilizer generators). This induced-subgraph viewpoint is what enables schedule-dependent counting arguments for how many CNOT faults are needed to realize such structures.
Even when constructing decoders that fully resolve composite TSs directly from classical decoders is challenging, using decoders derived from the parent codes as an initial design step significantly simplifies the process. When run as described in Algorithm~\ref{alg:amended_derived_decoders_from_parent_decoders}, these derived decoders already eliminate decoding failures due to many such composite \emph{stabilizer-induced} TSs, thereby reducing the number of remaining cases that must be addressed explicitly in the decoder design.

The results of Sections~\ref{sec:derived_decoders} and~\ref{sec:comp_TS} together yield a design procedure targeting the dominant failure configurations in the error-floor regime. Given the operating regime of interest: (i) enumerate the TSs of the parent codes with small critical number; (ii) the arguments in Sections~\ref{sec:derived_decoders}--\ref{sec:comp_TS} bound the number of circuit faults required to activate the corresponding product and composite TS families, and thereby identify which families can arise with few faults---families requiring substantially more faults, or error weights whose probability is negligible in the operating regime, require no treatment; (iii) design decoders for the dominant families via Lemmas~\ref{lemma:decoder_resolving_classical_plus_stab_TS} and~\ref{lemma:composite_trapping_set_decoding_condition_tree}; (iv) verify by adversarial error injection as in Section~\ref{sec:adv_noise_model}. We note that decoding approaches relying on post-processing or ensemble diversity possess no analogue of Step~(ii). Without such a characterization of the failure mechanisms, one cannot identify which configurations dominate or certify that they have been addressed.

\section{Handling CNOT-Failure-Induced Syndrome Errors Without Additional Nodes}
\label{sec:syndrome_error_inducing_CNOT_faults}
We now turn from hook-error correlations to CNOT faults that manifest primarily as \emph{syndrome measurement errors} (in addition to possible data errors), as described in Section~\ref{sec:cnot_faults_and_syndrome_error}. The question is whether such faults can also be handled without introducing a dedicated variable node for each CNOT-fault location.

To illustrate the issue, consider two rounds of $X$-stabilizer measurement in Fig.~\ref{fig:discrepancy_between_error_and_syndrome}. If a faulty CNOT between the data qubit $d_7$ and ancilla qubit $a_2$ introduces a $Z$ error on $d_7$ and the readout of ancilla qubit $a_2$ can also be faulty, then a natural decoding graph is the \emph{phenomenological} Tanner graph in Fig.~\ref{fig:phenomenological_Tanner_graph}, which includes variable nodes representing errors on the data qubits and the ancilla qubits.
The suitability of \emph{the phenomenological} Tanner graph as the decoding graph comes from the key observation that a $Z$ error on $d_7$ introduced during measurement can have the same measured syndrome as a $Z$ error present before measurement together with a syndrome-bit error on $a_2$. More generally, CNOT failures that appear as syndrome errors can often be reinterpreted as a data error plus one or more syndrome-bit errors. This suggests that explicit CNOT-fault nodes might be avoidable, but only at the cost of an increased effective rate of \emph{syndrome} noise.
\begin{figure*}
 \centering
    \usetikzlibrary{patterns,snakes}
\tikzstyle{cnode}=[circle,minimum size=0.4 cm,draw,black,scale=0.7]
\tikzstyle{gray_cnode}=[circle,minimum size=0.4 cm,draw,black,fill=gray,scale=0.7]
\tikzstyle{blue_cnode}=[circle,minimum size=0.4 cm,draw,black,fill=blue,scale=0.7]
\tikzstyle{scnode}=[circle,minimum size=0.15 cm,draw,black,scale=0.33]
\tikzstyle{sb_cnode}=[circle,minimum size=0.15 cm,draw,black,fill=gray,scale=0.28]
\tikzstyle{rnode}=[rectangle,draw,minimum width=0.35 cm,minimum height=0.35cm,outer sep=0pt,scale=0.7]
\tikzstyle{rrnode}=[rectangle,red,draw,minimum width=0.35 cm,minimum height=0.35cm,outer sep=0pt,scale=0.4]
\tikzstyle{srnode}=[rectangle,draw,minimum width=0.1 cm,minimum height=0.1cm,outer sep=0pt,scale=0.33]
\tikzstyle{sb_rnode}=[rectangle,draw,minimum width=0.1 cm,minimum height=0.1cm,outer sep=0pt,fill=gray,scale=0.33]
\tikzstyle{overbrace_style}=[decorate,decoration={brace,raise=2mm,amplitude=3pt}]
\tikzstyle{hexagon}=[regular polygon, regular polygon sides=6, shape aspect=0.5, minimum width=1.5cm, minimum height=1cm, draw,shape border rotate=60,scale=0.33]
\tikzstyle{Ellipse}=[ellipse,minimum width=3cm,minimum height=1.25cm,draw,scale=0.33]
\tikzstyle{smallEllipse}=[ellipse,minimum width=1.75cm,minimum height=1.25cm,draw,scale=0.33]
\definecolor{color1}{rgb}{1,0.2,0.3}
\definecolor{color2}{rgb}{0.4,0.5,0.7}
\definecolor{color3}{rgb}{0.1,0.8,0.5}
\definecolor{color4}{rgb}{0.5,0.3,1}
\definecolor{color5}{rgb}{0.5,1,1}
\definecolor{color6}{rgb}{0.8,0.3,0.6}
\definecolor{color7}{rgb}{0.6,0.4,0.3}

\begin{tikzpicture}[every node/.style={scale=1}]
\begin{scope}[node distance=1.5 cm,semithick]

\node [cnode] (v1){};
\node [cnode] (v2) [right of=v1]{};
\node (v1_lable) [left of=v1,xshift=1.1cm]{$d_1$};
\node (v2_lable) [left of=v2,xshift=1.1cm]{$d_2$};
\node [cnode] (v3) [right of=v2]{};
\node (v3_lable) [left of=v3,xshift=1.1cm]{$d_3$};
\node [cnode] (v4) [right of=v3]{};
\node (v4_lable) [left of=v4,xshift=1.1cm]{$d_4$};
\node [cnode] (v5) [right of=v4]{};
\node (v5_lable) [left of=v5,xshift=1.1cm]{$d_5$};
\node [cnode] (v6) [right of=v5]{};
\node (v6_lable) [left of=v6,xshift=1.1cm]{$d_6$};
\node [cnode] (v7) [right of=v6]{};
\node (v7_lable) [left of=v7,xshift=1.1cm]{$d_7$};
\draw[overbrace_style] (v1) -- (v7);
\node (label_first_round) [above of=v4,yshift=-1cm]{\textbf{First round of syndrome measurement}};
\node [rnode] (c1) [below of=v3,xshift=0.5cm,yshift=-1cm]{};
\node [rnode] (c2) [right of=c1]{};
\node [rnode] (c3) [right of=c2]{};
\draw (c1) -- (v4);
\draw (c1) -- (v5);
\draw (c1) -- (v6);
\draw (c1) -- (v7);
\draw (c2) -- (v2);
\draw (c2) -- (v3);
\draw (c2) -- (v6);
\draw (c2) -- (v7);
\draw (c3) -- (v1);
\draw (c3) -- (v3);
\draw (c3) -- (v5);
\draw (c3) -- (v7);
\node [cnode] (v8) [right of=v7,xshift=2cm]{};
\node (v8_lable) [left of=v8,xshift=1.1cm]{$d_1$};
\node [cnode] (v9) [right of=v8]{};
\node (v9_lable) [left of=v9,xshift=1.1cm]{$d_2$};
\node [cnode] (v10) [right of=v9]{};
\node (v10_lable) [left of=v10,xshift=1.1cm]{$d_3$};
\node [cnode] (v11) [right of=v10]{};
\node (v11_lable) [left of=v11,xshift=1.1cm]{$d_4$};
\node [cnode] (v12) [right of=v11]{};
\node (v12_lable) [left of=v12,xshift=1.1cm]{$d_5$};
\node [cnode] (v13) [right of=v12]{};
\node (v13_lable) [left of=v13,xshift=1.1cm]{$d_6$};
\node [cnode] (v14) [right of=v13]{};
\node (v14_lable) [left of=v14,xshift=1.1cm]{$d_7$};
\draw[overbrace_style] (v8) -- (v14);
\node (label_second_round) [above of=v11,yshift=-1cm]{\textbf{Second round of syndrome measurement}};
\node [rnode] (c4) [below of=v10,xshift=0.5cm,yshift=-1cm]{};
\node [rnode] (c5) [right of=c4]{};
\node [rnode] (c6) [right of=c5]{};
\draw (c4) -- (v11);
\draw (c4) -- (v12);
\draw (c4) -- (v13);
\draw (c4) -- (v14);
\draw (c5) -- (v9);
\draw (c5) -- (v10);
\draw (c5) -- (v13);
\draw (c5) -- (v14);
\draw (c6) -- (v8);
\draw (c6) -- (v10);
\draw (c6) -- (v12);
\draw (c6) -- (v14);
\node[blue_cnode] (cnot_fail) [below of=c3,xshift=2cm,yshift=-1cm]{};
\draw (cnot_fail)--(c1);
\draw (cnot_fail)--(c3);
\draw (cnot_fail)--(c5);
\node [gray_cnode] (s1) [right of=cnot_fail,xshift=0.5cm]{};
\node (s1_label) [left of=s1,xshift=1.1cm]{$s_1$};
\node [gray_cnode] (s2) [right of=s1,xshift=0cm]{};
\node (s2_lable) [left of=s2,xshift=1.1cm]{$s_2$};
\node [gray_cnode] (s3) [right of=s2,xshift=0cm]{};
\node (s3_lable) [left of=s3,xshift=1.1cm]{$s_3$};
\draw (s1)--(c1);
\draw (s1)--(c4);
\draw (s2)--(c2);
\draw (s2)--(c5);
\draw (s3)--(c3);
\draw (s3)--(c6);
\end{scope}    
\end{tikzpicture}
     \caption{Tanner graph for two rounds of $X$-stabilizer measurements. Unfilled circles represent data-qubit errors, and rectangles represent $X$-stabilizer generators. The blue node indicates a possible $Z$ error caused by a faulty CNOT between data qubit $d_7$ and ancilla $a_2$ in the first round, while gray nodes represent measurement (syndrome-bit) errors in the syndrome; for clarity, not all syndrome-bit error nodes are shown. The syndrome pattern resulting from a CNOT failure between $d_7$ and $a_2$ is identical to the pattern produced by an initial $Z$ error on $d_7$ together with a measurement error on $a_2$, so the explicit CNOT-failure node is redundant. This figure is based on a CSS code rather than a QLDPC code, so the usual variable-node and check-node conventions for QLDPC Tanner graphs are not followed. The discussion applies to any CSS code, not only to QLDPC codes.}
     \label{fig:phenomenological_Tanner_graph}
 \end{figure*}
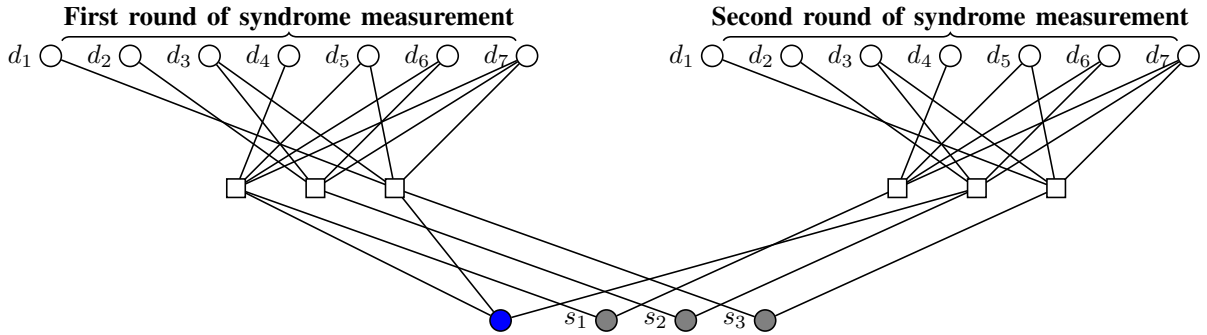

This tradeoff is different from the \emph{hook-error} case, where explicit CNOT-fault nodes can be removed without increasing the effective noise. If measurement-induced CNOT faults are absorbed into syndrome-bit errors, then the decoder must tolerate additional syndrome noise, which typically requires increased temporal redundancy (more rounds of syndrome measurements) or increased spatial redundancy (a lower-rate data--syndrome code~\cite{ashikhmin2020quantum,campbell2019theory,guttentag2024robust,quantum_expander_codes,bombin2015single,gu2024single}). If, instead, one keeps explicit CNOT-fault nodes, the resulting graph typically contains additional short cycles (in particular, 4-cycles) that can degrade the performance of the iterative decoder; post-processing can mitigate this but increases decoding complexity and latency.
\begin{figure}
    \centering
    \input{Figures/TS_phenomenologic_tanner_graph}
    \caption{TSs in the \emph{phenomenological} Tanner graph. The nodes corresponding to data qubits and stabilizer generators are shown as white circles and rectangles, respectively. The nodes corresponding to syndrome-bit errors are shown in gray; for clarity, not all syndrome-bit error nodes are shown. Black variable nodes indicate positions where the decoder's current hard-decision estimate does not match the actual error (a \emph{mismatched} estimate), and black check nodes indicate unsatisfied checks under this estimate. Note that two distinct error estimates can match the same set of unsatisfied checks: (1) one places errors on the data-qubit nodes that are actually in error, and (2) the other puts the estimate on  nodes corresponding to syndrome-bit errors together with data nodes that are not actually in error. Due to the existence of multiple explanations of the same syndrome, a message-passing decoder might fail to converge.}
    \label{fig:TS_phenomenological_tanner_graph}
\end{figure}
We next discuss how introducing syndrome-bit error nodes changes the TS picture in the resulting space--time or \emph{phenomenological} Tanner graph.

\subsection{TSs in Space and Time}
In previous sections, we discussed \emph{stabilizer-induced} TSs and how to mitigate the decoding failures they cause, and we also argued that explicit variable nodes for every possible circuit fault are often unnecessary. In a \emph{phenomenological} Tanner graph, however, variable nodes representing measurement (syndrome-bit) errors are unavoidable, and they couple the per-round copies of the Tanner graph across time. This creates two distinct failure mechanisms: (i) the added syndrome-bit nodes can introduce additional short cycles and multiple competing explanations of the same measured syndrome, and (ii) even when a spatial TS is not fully active within any single time slice, errors distributed across different time slices can interact through syndrome-bit nodes and appear to the decoder as a TS.

The \emph{phenomenological} Tanner graph is obtained by taking several copies of the code's Tanner graph---one per measurement round $\tau$---and connecting consecutive rounds by variable nodes representing measurement (syndrome-bit) errors (see Fig.~\ref{fig:phenomenological_Tanner_graph}). We refer to the copy corresponding to a fixed round $\tau$ as \emph{time slice $\tau$}. Consequently, for a given measured syndrome, there can be multiple distinct error patterns consistent with that syndrome, which can prevent message-passing decoders from converging (Fig.~\ref{fig:TS_phenomenological_tanner_graph}).
Such competing explanations are characterized in detail in subsequent work \cite{chytas2026temporal}.
To reduce the risk of such failures, we use a time-sliced message-update schedule. In one stage, we update all messages associated with variable and check nodes in even slices while keeping the messages associated with odd slices fixed; in the next stage, we swap the roles of even and odd slices. This is in addition to any within-slice schedule used to avoid decoding failures on TSs inside a single slice.
The purpose of this alternation is to reduce non-convergence caused by TSs or degeneracy in the phenomenological graph. Because syndrome-bit error nodes connect two consecutive slices, the same measured syndrome can be explained by different candidate error patterns that place the estimated error associated with a specific syndrome mismatch in different time slices. For example, a mismatch involving slices $\tau$ and $\tau+1$ can often be matched either by flipping certain data-variable nodes in slice $\tau$ or by flipping certain data-variable nodes in slice $\tau+1$ together with a corresponding choice of syndrome-bit error nodes between the slices. If all slices are updated simultaneously, the decoder may fail to settle on one of these candidate patterns, preventing convergence even when no single slice contains a TS.
We use two schedules to mitigate the effect of these. In \emph{strict alternation}, we swap even/odd every iteration; this is useful when the error pattern is small and localized. In \emph{batched alternation}, we run several decoding iterations updating only the odd slices, then several iterations updating only the even slices. \emph{Batched alternation} is useful for chain-like patterns of erroneous variable nodes within a single slice. In both cases, when a slice is being updated, we still use the within-slice message-update schedule required to avoid decoding failures on TSs inside that slice.

To make (ii) precise, fix a TS subgraph $\mathscr{T}$ in one time slice of the \emph{phenomenological} Tanner graph (Fig.~\ref{fig:time_spread_TS}). Each variable node $v$ of $\mathscr{T}$ appears in every slice; let $v(\tau)$ denote its copy in slice $\tau$, and let $\mathscr{T}(\tau)$ denote the corresponding copy of the TS. An error pattern spread over time can place different subsets of the variable nodes in $\mathscr{T}(\tau)$ in error across the slices, so that for every $\tau$, its restriction to slice $\tau$ is not itself a TS. Nevertheless, syndrome-bit error nodes can couple these partially erroneous copies across time, creating an effective TS that is not captured by analyzing each slice in isolation.
\begin{figure*}[t]
    \centering
    \begin{subfigure}[t]{0.9\textwidth}
        \centering
\tikzstyle{cnode}=[circle,minimum size=0.4 cm,draw,black,scale=0.7]
\tikzstyle{black_cnode}=[circle,minimum size=0.4 cm,draw,black,fill=black,scale=0.7]
\tikzstyle{gray_cnode}=[circle,minimum size=0.4 cm,draw,black,fill=gray,scale=0.7]
\tikzstyle{rnode}=[rectangle,draw,minimum width=0.35 cm,minimum height=0.35cm,outer sep=0pt,scale=0.7]
\tikzstyle{black_rnode}=[rectangle,draw,minimum width=0.35 cm,minimum height=0.35cm,outer sep=0pt,fill=black,scale=0.7]

\begin{tikzpicture}[every node/.style={scale=1}]
\begin{scope}[node distance=1 cm,semithick]
\node[black_cnode] (v1){};
\node[black_rnode] (c1) [above of=v1,xshift=-0cm]{};
\node[black_rnode] (c2) [below of = v1,xshift=1 cm]{};
\node[black_rnode] (c3) [right of = v1,xshift=1cm]{};
\draw (v1)--(c1);
\draw (v1)--(c2);
\draw (v1)--(c3);
\node[cnode] (v2) [right of = c3,xshift=1 cm]{};
\node[rnode] (c4) [above of=v2,xshift=-0cm]{};
\node[rnode] (c5) [below of = v2,xshift=-1 cm,yshift=0cm]{};
\draw (v2)--(c3);
\draw (v2)--(c4);
\draw (v2)--(c5);
\node[cnode] (v3) [below of = c2,xshift=1cm,yshift=0cm]{};
\node[rnode] (c6) [below of=v3,xshift=-0cm]{};
\draw (v3)--(c2);
\draw (v3)--(c5);
\draw (v3)--(c6);
\node[gray_cnode] (s1) [right of=c4,yshift=1cm,xshift=1cm]{};
\node[gray_cnode] (s2) [right of=s1]{};
\node[gray_cnode] (s3) [right of=c6,yshift=-1cm,xshift=3cm]{};
\node[gray_cnode] (s4) [above of=s1,yshift=0.5cm]{};
\node[gray_cnode] (s5) [below of=c2,yshift=-3cm]{};
\draw (c1)--(s1);
\draw (c4)--(s2);
\draw (c6)--(s3);
\draw (c3)--(s4);
\draw (c2)--(s5);
\node[rnode] (c7) [right of=s1,xshift=1cm,yshift=-1cm]{};
\node[cnode] (v4) [below of=c7]{};
\node[rnode] (c8) [below of = v4,xshift=1 cm]{};
\node[black_rnode] (c9) [right of = v4,xshift=1cm]{};
\draw (v4)--(c7);
\draw (v4)--(c8);
\draw (v4)--(c9);
\node[black_cnode] (v5) [right of = c9,xshift=1 cm]{};
\node[black_rnode] (c10) [above of=v5,xshift=-0cm]{};
\node[black_rnode] (c11) [below of = v5,xshift=-1 cm,yshift=0cm]{};
\draw (v5)--(c9);
\draw (v5)--(c10);
\draw (v5)--(c11);
\node[cnode] (v6) [below of = c8,xshift=1cm,yshift=0cm]{};
\node[rnode] (c12) [below of=v6,xshift=-0cm]{};
\draw (v6)--(c8);
\draw (v6)--(c11);
\draw (v6)--(c12);
\draw (c7)--(s1);
\draw (c10)--(s2);
\draw (c12)--(s3);
\draw (c9)--(s4);
\draw (c8)--(s5);
\node[gray_cnode] (s7) [right of=c10,yshift=1cm,xshift=1cm]{};
\node[gray_cnode] (s8) [right of=s7]{};
\node[gray_cnode] (s9) [right of=c12,yshift=-1cm,xshift=3cm]{};
\node[gray_cnode] (s11) [below of=c8,yshift=-3cm]{};
\draw (c7)--(s7);
\draw (c10)--(s8);
\draw (c12)--(s9);
\draw (c8)--(s11);
\node[rnode] (c13) [right of=s7,xshift=1cm,yshift=-1cm]{};
\node[cnode] (v7) [below of=c13]{};
\node[black_rnode] (c14) [below of = v7,xshift=1 cm]{};
\node[rnode] (c15) [right of = v7,xshift=1cm]{};
\draw (v7)--(c13);
\draw (v7)--(c14);
\draw (v7)--(c15);
\node[cnode] (v8) [right of = c15,xshift=1 cm]{};
\node[rnode] (c16) [above of=v8,xshift=-0cm]{};
\node[black_rnode] (c17) [below of = v8,xshift=-1 cm,yshift=0cm]{};
\draw (v8)--(c15);
\draw (v8)--(c16);
\draw (v8)--(c17);
\node[black_cnode] (v9) [below of = c14,xshift=1cm,yshift=0cm]{};
\node[black_rnode] (c18) [below of=v9,xshift=-0cm]{};
\draw (v9)--(c14);
\draw (v9)--(c17);
\draw (v9)--(c18);
\node[gray_cnode] (s12) [below of=c17,yshift=-3cm]{};
\draw (c13)--(s7);
\draw (c16)--(s8);
\draw (c18)--(s9);
\draw (c14)--(s11);
\draw (c11)--(s12);
\draw (c17)--(s12);
\node (t1) [below of=c6,yshift=-0.85cm]{slice $\tau$};
\node (t2) [below of=c12,yshift=-0.85cm]{slice $\tau+1$};
\node (t3) [below of=c18,yshift=-0.85cm]{slice $\tau+2$};
\end{scope}
\end{tikzpicture}
        \subcaption{A $(3,3)$ TS whose three erroneous variable nodes are spread over three consecutive time slices. The spread is enabled by erroneous syndrome-bit error variable nodes, which couple the slices; as a result, no single slice contains the full $(3,3)$ TS, yet the combined space--time pattern can behave as a TS.}
        \label{fig:time_spread_TS}
    \end{subfigure}

    \vspace{1ex}

    \begin{subfigure}[t]{0.9\textwidth}
        \centering
\tikzstyle{cnode}=[circle,minimum size=0.4 cm,draw,black,scale=0.7]
\tikzstyle{black_cnode}=[circle,minimum size=0.4 cm,draw,black,fill=black,scale=0.7]
\tikzstyle{gray_cnode}=[circle,minimum size=0.4 cm,draw,black,fill=gray,scale=0.7]
\tikzstyle{flipped_snode}=[circle,minimum size=0.4 cm,draw=red,line width=0.9pt,fill=black,scale=0.7]
\tikzstyle{rnode}=[rectangle,draw,minimum width=0.35 cm,minimum height=0.35cm,outer sep=0pt,scale=0.7]
\tikzstyle{black_rnode}=[rectangle,draw,minimum width=0.35 cm,minimum height=0.35cm,outer sep=0pt,fill=black,scale=0.7]

\begin{tikzpicture}[every node/.style={scale=1}]
\begin{scope}[node distance=1 cm,semithick]
\node[cnode] (v1){};
\node[rnode] (c1) [above of=v1,xshift=-0cm]{};
\node[rnode] (c2) [below of = v1,xshift=1 cm]{};
\node[rnode] (c3) [right of = v1,xshift=1cm]{};
\draw (v1)--(c1);
\draw (v1)--(c2);
\draw (v1)--(c3);
\node[cnode] (v2) [right of = c3,xshift=1 cm]{};
\node[rnode] (c4) [above of=v2,xshift=-0cm]{};
\node[rnode] (c5) [below of = v2,xshift=-1 cm,yshift=0cm]{};
\draw (v2)--(c3);
\draw (v2)--(c4);
\draw (v2)--(c5);
\node[cnode] (v3) [below of = c2,xshift=1cm,yshift=0cm]{};
\node[rnode] (c6) [below of=v3,xshift=-0cm]{};
\draw (v3)--(c2);
\draw (v3)--(c5);
\draw (v3)--(c6);
\node[flipped_snode] (s1) [right of=c4,yshift=1cm,xshift=1cm]{};
\node[gray_cnode] (s2) [right of=s1]{};
\node[gray_cnode] (s3) [right of=c6,yshift=-1cm,xshift=3cm]{};
\node[flipped_snode] (s4) [above of=s1,yshift=0.5cm]{};
\node[flipped_snode] (s5) [below of=c2,yshift=-3cm]{};
\draw (c1)--(s1);
\draw (c4)--(s2);
\draw (c6)--(s3);
\draw (c3)--(s4);
\draw (c2)--(s5);
\node[rnode] (c7) [right of=s1,xshift=1cm,yshift=-1cm]{};
\node[cnode] (v4) [below of=c7]{};
\node[rnode] (c8) [below of = v4,xshift=1 cm]{};
\node[rnode] (c9) [right of = v4,xshift=1cm]{};
\draw (v4)--(c7);
\draw (v4)--(c8);
\draw (v4)--(c9);
\node[cnode] (v5) [right of = c9,xshift=1 cm]{};
\node[rnode] (c10) [above of=v5,xshift=-0cm]{};
\node[rnode] (c11) [below of = v5,xshift=-1 cm,yshift=0cm]{};
\draw (v5)--(c9);
\draw (v5)--(c10);
\draw (v5)--(c11);
\node[cnode] (v6) [below of = c8,xshift=1cm,yshift=0cm]{};
\node[rnode] (c12) [below of=v6,xshift=-0cm]{};
\draw (v6)--(c8);
\draw (v6)--(c11);
\draw (v6)--(c12);
\draw (c7)--(s1);
\draw (c10)--(s2);
\draw (c12)--(s3);
\draw (c9)--(s4);
\draw (c8)--(s5);
\node[flipped_snode] (s7) [right of=c10,yshift=1cm,xshift=1cm]{};
\node[flipped_snode] (s8) [right of=s7]{};
\node[gray_cnode] (s9) [right of=c12,yshift=-1cm,xshift=3cm]{};
\node[flipped_snode] (s11) [below of=c8,yshift=-3cm]{};
\draw (c7)--(s7);
\draw (c10)--(s8);
\draw (c12)--(s9);
\draw (c8)--(s11);
\node[black_rnode] (c13) [right of=s7,xshift=1cm,yshift=-1cm]{};
\node[black_cnode] (v7) [below of=c13]{};
\node[rnode] (c14) [below of = v7,xshift=1 cm]{};
\node[rnode] (c15) [right of = v7,xshift=1cm]{};
\draw (v7)--(c13);
\draw (v7)--(c14);
\draw (v7)--(c15);
\node[black_cnode] (v8) [right of = c15,xshift=1 cm]{};
\node[black_rnode] (c16) [above of=v8,xshift=-0cm]{};
\node[rnode] (c17) [below of = v8,xshift=-1 cm,yshift=0cm]{};
\draw (v8)--(c15);
\draw (v8)--(c16);
\draw (v8)--(c17);
\node[black_cnode] (v9) [below of = c14,xshift=1cm,yshift=0cm]{};
\node[black_rnode] (c18) [below of=v9,xshift=-0cm]{};
\draw (v9)--(c14);
\draw (v9)--(c17);
\draw (v9)--(c18);
\node[flipped_snode] (s12) [below of=c17,yshift=-3cm]{};
\draw (c13)--(s7);
\draw (c16)--(s8);
\draw (c18)--(s9);
\draw (c14)--(s11);
\draw (c11)--(s12);
\draw (c17)--(s12);
\node (t1) [below of=c6,yshift=-0.85cm]{slice $\tau$};
\node (t2) [below of=c12,yshift=-0.85cm]{slice $\tau+1$};
\node (t3) [below of=c18,yshift=-0.85cm]{slice $\tau_0=\tau+2$};
\end{scope}
\end{tikzpicture}
        \subcaption{Projection to a reference slice: starting from the error pattern spread across all three slices in (a), we flip selected syndrome-bit error variable nodes to push the unsatisfied checks forward in time so that all remaining unsatisfied checks lie in the third slice (i.e., $\tau_0=3$). After this projection, a decoder designed to resolve a $(3,3)$ TS within a single slice can decode successfully.}
        \label{fig:projection_to_tau0}
    \end{subfigure}
    \caption{Illustrations supporting the discussion of TS mechanisms in the \emph{phenomenological} Tanner graph: a TS structure spread across three time slices (top) and projection of unsatisfied checks to a reference slice (bottom). In both panels, black variable nodes indicate positions where the decoder's current hard-decision estimate does not match the actual error (a \emph{mismatched} estimate), and black check nodes indicate unsatisfied checks under this estimate. Syndrome-bit error nodes are shown in gray; for clarity, not all syndrome-bit error nodes are shown. In (b) only, the syndrome-bit error nodes whose estimates are flipped during projection are highlighted with a red outline.}
    \label{fig:time_spread_and_projection}
\end{figure*}

Assume that we have a decoder that, when run on the code's single-slice Tanner graph, avoids decoding failures on the relevant spatial TSs (i.e., TSs contained within a single time slice). One way to mitigate the decoding failures due to TSs that are spread across time is to project the current syndrome mismatch onto a single reference slice $\tau_0$ (Fig.~\ref{fig:projection_to_tau0}).
Concretely, after running message passing for a fixed number of iterations on the \emph{phenomenological} Tanner graph, we form a hard decision on all variable nodes (including syndrome-bit error nodes) and compute the set of unsatisfied checks (where the predicted syndrome disagrees with the measured syndrome).
Flipping the hard decision on a syndrome-bit error node toggles its adjacent checks in both of the slices it connects. By flipping appropriate syndrome-bit error nodes between slices $\tau-1$ and $\tau$, an unsatisfied check in slice $\tau$ can be moved into slice $\tau-1$ (or $\tau+1$). Algorithmically, one can sweep over slices from the last slice down to $\tau_0$ and, for each $\tau>\tau_0$, flip exactly those syndrome-bit error nodes needed to eliminate all unsatisfied checks in slice $\tau$. After this sweep, all remaining unsatisfied checks lie in slice $\tau_0$ (with an updated syndrome-bit error pattern across time), so a single-slice TS-resolving decoder can be applied to the projected instance.

This projection is not always desirable (see Fig.~\ref{fig:time_spread_and_projection}). Pushing all unsatisfied checks into one slice can increase the effective weight seen by the decoder in that slice and can transform an error pattern spread over time slices into a spatial pattern in one slice that lies outside the TS families for which the decoder is designed (e.g., the projected pattern may have weight exceeding roughly $d/2$, or it may create a larger TS that the decoder does not reliably resolve), even though the per-slice restrictions were individually low weight.
A way to mitigate this is to encourage temporal separation during message passing (Fig.~\ref{fig:push_before_after}). Let $\gamma_\mathrm{i}$ denote the reliability value used by the message-passing decoder for variable node $v_\mathrm{i}$. After a fixed number of iterations, we temporarily scale $\gamma_\mathrm{i}\leftarrow \alpha\gamma_\mathrm{i}$ for variable nodes $v_\mathrm{i}$ that are within one hop of any syndrome-bit error node currently estimated as $1$, with $\alpha\in(0,1)$, and then gradually increase $\alpha$ back to $1$ over subsequent iterations. Here, by one hop nodes, we refer to variable nodes that are connected to a common check node. This reduces the tendency of the decoder to merge temporally separated error patterns into a single TS.

\begin{figure}[t]
    \centering
\tikzstyle{cnode}=[circle,minimum size=0.4 cm,draw,black,scale=0.7]
\tikzstyle{black_cnode}=[circle,minimum size=0.4 cm,draw,black,fill=black,scale=0.7]
\tikzstyle{gray_cnode}=[circle,minimum size=0.4 cm,draw,black,fill=gray,scale=0.7]
\tikzstyle{llr_cnode}=[circle,minimum size=0.4 cm,draw=red,line width=0.9pt,fill=white,scale=0.7]
\tikzstyle{llr_black_cnode}=[circle,minimum size=0.4 cm,draw=red,line width=0.9pt,fill=black,scale=0.7]
\tikzstyle{rnode}=[rectangle,draw,minimum width=0.35 cm,minimum height=0.35cm,outer sep=0pt,scale=0.7]
\tikzstyle{black_rnode}=[rectangle,draw,minimum width=0.35 cm,minimum height=0.35cm,outer sep=0pt,fill=black,scale=0.7]

\begin{tikzpicture}[every node/.style={scale=1}]
\begin{scope}[node distance=1 cm,semithick]
\node[llr_cnode] (v1){};
\node[black_rnode] (c1) [above of=v1,xshift=-0cm]{};
\node[black_rnode] (c2) [below of = v1,xshift=1 cm]{};
\node[black_rnode] (c3) [right of = v1,xshift=1cm]{};
\draw (v1)--(c1);
\draw (v1)--(c2);
\draw (v1)--(c3);
\node[black_cnode] (v2) [right of = c3,xshift=1 cm]{};
\node[black_rnode] (c4) [above of=v2,xshift=-0cm]{};
\node[black_rnode] (c5) [below of = v2,xshift=-1 cm,yshift=0cm]{};
\draw (v2)--(c3);
\draw (v2)--(c4);
\draw (v2)--(c5);
\node[black_cnode] (v3) [below of = c2,xshift=1cm,yshift=0cm]{};
\node[black_rnode] (c6) [below of=v3,xshift=-0cm]{};
\draw (v3)--(c2);
\draw (v3)--(c5);
\draw (v3)--(c6);
\node[black_cnode] (s1) [right of=c4,yshift=1cm,xshift=1cm]{};
\node[gray_cnode] (s2) [right of=s1]{};
\node[gray_cnode] (s3) [right of=c6,yshift=-1cm,xshift=3cm]{};
\node[gray_cnode] (s4) [above of=s1,yshift=0.5cm]{};
\node[gray_cnode] (s5) [below of=c2,yshift=-3cm]{};
\draw (c1)--(s1);
\draw (c4)--(s2);
\draw (c6)--(s3);
\draw (c3)--(s4);
\draw (c2)--(s5);
\node[rnode] (c7) [right of=s1,xshift=1cm,yshift=-1cm]{};
\node[llr_black_cnode] (v4) [below of=c7]{};
\node[rnode] (c8) [below of = v4,xshift=1 cm]{};
\node[black_rnode] (c9) [right of = v4,xshift=1cm]{};
\draw (v4)--(c7);
\draw (v4)--(c8);
\draw (v4)--(c9);
\node[black_cnode] (v5) [right of = c9,xshift=1 cm]{};
\node[black_rnode] (c10) [above of=v5,xshift=-0cm]{};
\node[black_rnode] (c11) [below of = v5,xshift=-1 cm,yshift=0cm]{};
\draw (v5)--(c9);
\draw (v5)--(c10);
\draw (v5)--(c11);
\node[cnode] (v6) [below of = c8,xshift=1cm,yshift=0cm]{};
\node[rnode] (c12) [below of=v6,xshift=-0cm]{};
\draw (v6)--(c8);
\draw (v6)--(c11);
\draw (v6)--(c12);
\draw (c7)--(s1);
\draw (c10)--(s2);
\draw (c12)--(s3);
\draw (c9)--(s4);
\draw (c8)--(s5);
\node (t1) [below of=c6,yshift=-0.85cm]{slice $\tau$};
\node (t2) [below of=c12,yshift=-0.85cm]{slice $\tau+1$};
\end{scope}
\end{tikzpicture}
    \caption{Illustration of the reliability reassignment step. Colors indicate the decoder the current estimate after hard decision: black variable nodes indicate positions where the estimate does not match the actual error (a \emph{mismatched} estimate), and black check nodes indicate unsatisfied checks under this estimate. Variable nodes whose reliabilities (LLRs) are reassigned by the heuristic---i.e., data nodes within one hop (through a check) of a syndrome-bit error node currently estimated as $1$---are highlighted with a red outline.}
    \label{fig:push_before_after}
\end{figure}

\begin{remark}
If the message-passing decoder is sufficiently robust to syndrome-bit noise, then the above heuristic that reduces LLR of certain variable nodes may be unnecessary, since the decoder can already tolerate the corrupted syndrome, which ensure temporal separation. However, projecting the syndrome to a reference slice remains useful in our setting: to resolve TSs with small-$b$, it is often difficult to simultaneously design decoders that resolve the TS and also robust to syndrome-bit error.
\end{remark}

\section{Simulation Results}
\label{sec:simulation_results}
In this section, we first validate our earlier claims using adversarial simulations: we inject errors confined to a specific TS and verify that the decoder, designed according to the principles in the previous sections, successfully avoids the corresponding decoding failure.
We then simulate the logical error rate under stochastic noise to compare the proposed decoder with the state-of-the-art BP+OSD and \emph{beam search} decoders. For these simulations, we use two LP codes, a $[[1054,140,d\leq 20]]$ code and a $[[775,43,d\leq 20]]$ code. The $[[1054,140,d\leq 20]]$ LP code is obtained by using
\begin{equation*}
\label{eq:weight_matrix_Tanner_code}
\mathbf{B}_1=\mathbf{B}_2=\mathbf{B}=\begin{bmatrix}
x & x^2 & x^3 & x^4 & x^5\\
x^5 & x^{10} & x^{20} & x^{9} & x^{18}\\
x^{25} & x^{19} & x^7 & x^{14} & x^{28}
\end{bmatrix},
\end{equation*}
while the $[[775,43,d\leq 20]]$ LP code is obtained by using
\begin{equation*}
\label{eq:weight_matrix_Tanner_code_2}
\mathbf{B}_1=\mathbf{B}_2=\mathbf{B}=\begin{bmatrix}
x & x^2 & x^3 & x^4 \\
x^5 & x^{10} & x^{20} & x^{9} \\
x^{25} & x^{19} & x^7 & x^{14}
\end{bmatrix},
\end{equation*}
in~\eqref{eq:LP_construction}. Let us denote the classical Tanner graph obtained by lifting the base matrix $\mathbf{B}$ by $\mathscr{G}$, and the Tanner graph obtained by exchanging the roles of variable and check nodes in $\mathscr{G}$ by $\mathscr{G}^\mathsf{T}$.
In general, LP and HGP codes are built from two (possibly different) classical codes, so we used superscripts to distinguish nodes belonging to each code’s Tanner graph. In this section, however, we take both classical codes to be the same. Consequently, the superscripts are no longer necessary, and we omit them from the notation in what follows.

\subsection{Adversarial noise model}
\label{sec:adv_noise_model}

\subsubsection{Graph product of tree-like subgraphs}
\label{sec:adv_stabilizer_induced_ts}
We next consider \emph{stabilizer-induced} TSs, which arise from the graph product of
tree-like subgraphs of the parent code's Tanner graph. In the Tanner graph corresponding to
the $Z$-type stabilizer generators, we enumerate error patterns supported on TSs induced by one,
two, or three $X$-type stabilizer generators, restricting each pattern to lie
either entirely on the VV-type qubits or entirely on the CC-type qubits. For two
or more generators, we define the support as the union of their supports,
i.e., a qubit is included if it is covered by at least one generator; qubits common to multiple stabilizer generators
are therefore retained. We restrict each adversarial pattern to VV-type or CC-type variable nodes because, for the \emph{stabilizer-induced} TS families illustrated in Section~\ref{sec:motivation_for_proposed_decoder}, these two restrictions capture the two canonical degenerate failure modes: the erroneous variables lie entirely on one side (among VV and CC) of $\mathscr{G}_{\mathrm{z}}$ and all checks in the TS are  unsatisfied.

We compare against a flooding-schedule min-sum decoder with OSD-CS post-processing of order $10$. Its performance strongly depends on the soft output passed to OSD-CS. Using the \emph{a posteriori} LLRs from the iteration with the fewest unsatisfied checks is the most favorable choice; even then, min-sum + OSD-CS $(10)$ leaves approximately two fifths of the error patterns supported on a single stabilizer generator, and roughly one quarter of the error patterns supported on three stabilizer generators, uncorrected for the $[[1054,140,d\leq 20]]$ code. Passing instead the LLRs after a fixed budget of $50$ iterations is markedly worse: on the same code, essentially none of the error patterns supported on a single stabilizer generator and only $1\%$ of the error patterns supported on three stabilizer generators are corrected. Similar failure rates are also observed for the $[[775,43,d\leq 20]]$ code.

In contrast, a min-sum decoder run under the alternating (VV/CC)
message-update schedule of~\cite{pradhan2025_TS_left_right_schedule_decoder}
corrects all of these error patterns on both codes, without post-processing, and
typically within two iterations. 
The \emph{beam-search} decoder proposed in \cite{ye2025beam} would converge better if it were run with this alternating schedule on the \emph{phenomenological} Tanner graph, which can resolve these \emph{stabilizer-induced} TSs without needing additional rounds.
The correlations introduced during syndrome
extraction can alternatively be represented by augmenting the decoding matrix with
one variable node per \emph{hook error}---the correlated multi-qubit error left by a
single ancilla fault. As discussed in Section~\ref{sec:motivation_for_proposed_decoder},
this kind of augmentation is expected to improve convergence because it makes
previously correlated faults explicit in the decoding graph. This also corrects
every pattern considered, but at a considerably
higher cost: the augmentation multiplies the number of variable nodes and raises
the check degrees, increasing the complexity of both message passing and when
used the Gaussian elimination and combination sweep in OSD. The alternating
schedule reaches the same outcome on the unmodified graph and with no
post-processing at all. Table~\ref{tab:stab_induced_tree_summary} summarizes these results.

\begin{table}[t]
\centering
\scriptsize
\caption{Decoding success rates (\%) for adversarial error patterns supported on stabilizer-induced trapping sets formed as graph products of tree-like subgraphs (Section~\ref{sec:adv_stabilizer_induced_ts}). Results are similar for both LP codes ($[[1054,140,d\leq 20]]$ and $[[775,43,d\leq 20]]$). MS denotes flooding min-sum; OSD-CS(10) denotes order-10 OSD combination-sweep post-processing. $L^{\star}$ (resp. $L^{(50)}$) denotes using the best-so-far LLRs (resp. the LLRs after 50 iterations) as input to OSD-CS. ``VV/CC schedule'' denotes min-sum run with the alternating VV/CC layered update schedule of~\cite{pradhan2025_TS_left_right_schedule_decoder} (no post-processing).}
\label{tab:stab_induced_tree_summary}
\setlength{\tabcolsep}{3.5pt}
\renewcommand{\arraystretch}{1.05}
\begin{tabular}{|c|c|c|c|}
\hline
\multirow{2}{*}{\# stabilizer generators} & \multicolumn{2}{c|}{MS $+$ OSD-CS(10)} & \multirow{2}{*}{VV/CC schedule}\\
\cline{2-3}
 & $L^{\star}$ & $L^{(50)}$ & \\
\hline\hline
1 & $\approx 60\%$ & $\approx 0\%$ & $100\%$\\
\hline
2 & $\approx 70\%$ & $\approx 1\%$ & $100\%$\\
\hline
3 & $\approx 75\%$ & $\approx 1\%$ & $100\%$\\
\hline
\end{tabular}
\end{table}

\subsubsection{Graph product of a tree-like subgraph with a TS}
\label{sec:adv_tree_times_ts}
Next, we consider decoding failures in the Tanner graph $\mathscr{G}_\mathrm{z}$ caused by TSs of the forms $\mathscr{T}(c) \times \mathscr{T}(\mathcal{V}_\mathrm{s})$ and $\mathscr{T}(\mathcal{C}_\mathrm{s}) \times \mathscr{T}(v)$ (see Section~\ref{sec:proposed_decoders}). Here, $\mathscr{T}(\mathcal{V}_\mathrm{s})$ and $\mathscr{T}(\mathcal{C}_\mathrm{s})$ are TSs in $\mathscr{G}$ and $\mathscr{G}^\mathsf{T}$, respectively, while $c$ and $v$ denote a check node and a variable node in $\mathscr{G}$.
In the adversarial setting, we test Algorithm~\ref{alg:derived_decoders_from_parent_decoders} using the derived decoders $D_{\mathrm{vv}}$ and $D_{\mathrm{cc}}$ constructed in Section~\ref{sec:derived_decoders}. For TSs of the form $\mathscr{T}(c) \times \mathscr{T}(\mathcal{V}_\mathrm{s})$, we take $D_{\mathrm{cc}}$ to be the min-sum decoder and use, as $D_{\mathrm{vv}}$, a set of decoders derived from parent-code decoders. For TSs of the form $\mathscr{T}(\mathcal{C}_\mathrm{s}) \times \mathscr{T}(v)$, the roles of $D_{\mathrm{vv}}$ and $D_{\mathrm{cc}}$ are swapped, as in Section~\ref{sec:derived_decoders}.
For our experiments, we select $\mathscr{T}(\mathcal{V}_\mathrm{s})$ to be the $(5,3)$, $(6,4)$, and $(7,3)$ TSs in $\mathscr{G}$ for both LP codes $[[1054,140,d\leq 20]]$ and $[[775,43,d\leq 20]]$.
In our experiments, it suffices to fix the check node $c$ in $\mathscr{T}(c) \times \mathscr{T}(\mathcal{V}_\mathrm{s})$ and to vary only $\mathscr{T}(\mathcal{V}_\mathrm{s})$. The reason is that the rule in Step~10, together with $D_{\mathrm{vv}}$ (derived from a parent-code decoder that resolves $\mathscr{T}(\mathcal{V}_\mathrm{s})$), resolves the resulting TS independently of the specific choice of $c$.
We use the finite-alphabet iterative decoders (FAIDs)~\cite{FAID_diversity}, described in Appendix~\ref{sec:FAID}, to construct 4-bit decoders that avoid decoding failures in $\mathscr{G}_{\mathrm{z}}$ induced by these TSs. When ten 4-bit FAIDs are run as a set, they correct all enumerated error patterns for the $[[1054,140,d\leq 20]]$ code; for the $[[775,43,d\leq 20]]$ code, a two-decoder FAID set suffices.
The decoder-set sizes reported in this section (two for the $[[775, 43, d \leq 20]]$ code and ten for the $[[1054, 140, d \leq 20]]$ code) reflect the differing TS neighborhood of the two parent Tanner graphs used to construct these two LP codes, not the blocklengths. Because the ability of a finite-precision decoder to resolve a TS depends on the neighborhood in which the TS is embedded, the search must be repeated for each code; it is, however, performed entirely offline and contributes nothing to decoding latency.
For comparison, Table~\ref{tab:cc_success} reports success rates for BP (implemented here as flooding min-sum) combined with OSD-CS$(10)$ under four variants: decoding on the original graph vs. the \emph{detector-error} Tanner graph, and passing to OSD either the best \emph{a posteriori} LLRs $L^{\star}$ (early stopping---fewest unsatisfied checks) or the final-iteration LLRs $L^{(50)}$. We contrast these baselines with the proposed decoder (FAID set + Algorithm~\ref{alg:derived_decoders_from_parent_decoders}), which requires neither decoding on the \emph{detector-error} Tanner graph nor post-processing.
We also compare the success rates of the \emph{beam-search} decoder when its internal min-sum decoder is implemented using the alternating VV/CC schedule proposed in \cite{pradhan2025_TS_left_right_schedule_decoder}.
\begin{table*}[t]
\centering
\caption{Decoding success rates (\%) for error patterns supported inside a
\emph{stabilizer-induced} trapping set $\mathcal{T}(c)\times\mathcal{T}(\mathcal{V}_\mathrm{s})$ of an LP
code, as a function of the number of simultaneously active copies of the
classical TS $\mathcal{T}(V_s)$, the error weight within each active copy, and
the number of erroneous CC-type variable nodes drawn from the same TS. MS
denotes the min-sum decoder, run with a flooding schedule for 50 iterations.
For each TS, decoding success rate is reported for only the error patterns with $a-1$ and $a$. \emph{circuit-level} denotes decoding on the
detector--error Tanner graph; \emph{code-level} denotes decoding on the code
Tanner graph $\mathbf{H}_\mathrm{z}$ itself. $L^{\star}$ and $L^{(50)}$ denote the \emph{a posteriori}
LLRs passed to the post-processor, taken respectively from the iteration with
the fewest unsatisfied checks and from the last iteration. The subscript
``cl'' marks that the decoder is run on the \emph{detector--error} Tanner graph.
The last column is a scheduled beam decoder, run on the
\emph{phenomenological} detector graph ($R=3$ rounds) with a layered
schedule with beam width $16$, and maximum number of rounds $24$;
a decoding counts as a success only if the returned residual is a stabilizer.}
\label{tab:cc_success}
\resizebox{\textwidth}{!}{%
\begin{tabular}{|l|c|c|c|c|c|c|c|c|c|c|}
\hline
Code & $\mathscr{T}(\mathcal{V}_\mathrm{s})$ & Weight & Copies & CC err. &
\multicolumn{4}{c|}{MS $+$ OSD-CS(10)} &
\multirow{2}{*}{\shortstack{FAID set\\$+$ Alg.~1}} &
\multirow{2}{*}{\shortstack{Scheduled\\beam}} \\
\cline{6-9}
 & & & & & code-level $L^{\star}$ & code-level $L^{(50)}$ &
 circuit-level $L^{\star}_{\mathrm{cl}}$ & circuit-level $L^{(50)}_{\mathrm{cl}}$ & & \\
\hline
\multirow{24}{*}{$[[1054,140]]$} & \multirow{8}{*}{$(5,3)$} & \multirow{4}{*}{$4$} & \multirow{2}{*}{$1$} & $0$ & 100.00 & 100.00 & 64.00 & 96.00 & 100.00 & 100.00 \\
 &  &  &  & $1$ & 100.00 & 100.00 & 100.00 & 100.00 & 100.00 & 100.00 \\
\cline{4-11}
 &  &  & \multirow{2}{*}{$3$} & $0$ & 100.00 & 92.00 & 100.00 & 76.00 & 100.00 & 100.00 \\
 &  &  &  & $1$ & 100.00 & 0.00 & 92.00 & 100.00 & 100.00 & 100.00 \\
\cline{3-11}
 &  & \multirow{4}{*}{$5$} & \multirow{2}{*}{$1$} & $0$ & 100.00 & 100.00 & 60.00 & 100.00 & 100.00 & 100.00 \\
 &  &  &  & $1$ & 100.00 & 20.00 & 100.00 & 100.00 & 100.00 & 100.00 \\
\cline{4-11}
 &  &  & \multirow{2}{*}{$3$} & $0$ & 100.00 & 0.00 & 100.00 & 60.00 & 100.00 & 40.00 \\
 &  &  &  & $1$ & 100.00 & 0.00 & 80.00 & 60.00 & 100.00 & 50.00 \\
\cline{2-11}
 & \multirow{8}{*}{$(6,4)$} & \multirow{4}{*}{$5$} & \multirow{2}{*}{$1$} & $0$ & 97.22 & 29.26 & 48.33 & 98.89 & 100.00 & 100.00 \\
 &  &  &  & $1$ & 98.52 & 33.89 & 97.59 & 98.89 & 100.00 & 100.00 \\
\cline{4-11}
 &  &  & \multirow{2}{*}{$3$} & $0$ & 81.30 & 0.00 & 77.78 & 55.93 & 100.00 & 96.67 \\
 &  &  &  & $1$ & 79.07 & 0.19 & 84.26 & 69.07 & 100.00 & 88.33 \\
\cline{3-11}
 &  & \multirow{4}{*}{$6$} & \multirow{2}{*}{$1$} & $0$ & 100.00 & 16.67 & 22.22 & 100.00 & 100.00 & 90.00 \\
 &  &  &  & $1$ & 100.00 & 13.33 & 86.67 & 91.11 & 100.00 & 90.00 \\
\cline{4-11}
 &  &  & \multirow{2}{*}{$3$} & $0$ & 83.33 & 0.00 & 41.11 & 12.22 & 100.00 & 10.00 \\
 &  &  &  & $1$ & 61.11 & 0.00 & 50.00 & 24.44 & 100.00 & 20.00 \\
\cline{2-11}
 & \multirow{8}{*}{$(7,3)$} & \multirow{4}{*}{$6$} & \multirow{2}{*}{$1$} & $0$ & 100.00 & 13.81 & 34.29 & 88.57 & 100.00 & 90.00 \\
 &  &  &  & $1$ & 100.00 & 5.24 & 89.52 & 88.10 & 100.00 & 87.14 \\
\cline{4-11}
 &  &  & \multirow{2}{*}{$3$} & $0$ & 69.05 & 0.00 & 43.81 & 19.05 & 100.00 & 42.86 \\
 &  &  &  & $1$ & 61.90 & 0.00 & 64.29 & 38.10 & 100.00 & 30.00 \\
\cline{3-11}
 &  & \multirow{4}{*}{$7$} & \multirow{2}{*}{$1$} & $0$ & 100.00 & 0.00 & 10.00 & 50.00 & 100.00 & 70.00 \\
 &  &  &  & $1$ & 96.67 & 0.00 & 66.67 & 33.33 & 100.00 & 50.00 \\
\cline{4-11}
 &  &  & \multirow{2}{*}{$3$} & $0$ & 6.67 & 0.00 & 26.67 & 0.00 & 100.00 & 0.00 \\
 &  &  &  & $1$ & 26.67 & 0.00 & 33.33 & 3.33 & 100.00 & 0.00 \\
\hline
\multirow{24}{*}{$[[775,43]]$} & \multirow{8}{*}{$(5,3)$} & \multirow{4}{*}{$4$} & \multirow{2}{*}{$1$} & $0$ & 100.00 & 100.00 & 50.00 & 100.00 & 100.00 & 100.00 \\
 &  &  &  & $1$ & 100.00 & 100.00 & 100.00 & 100.00 & 100.00 & 100.00 \\
\cline{4-11}
 &  &  & \multirow{2}{*}{$3$} & $0$ & 100.00 & 90.00 & 90.00 & 90.00 & 100.00 & 100.00 \\
 &  &  &  & $1$ & 100.00 & 100.00 & 100.00 & 90.00 & 100.00 & 100.00 \\
\cline{3-11}
 &  & \multirow{4}{*}{$5$} & \multirow{2}{*}{$1$} & $0$ & 100.00 & 100.00 & 50.00 & 100.00 & 100.00 & 100.00 \\
 &  &  &  & $1$ & 100.00 & 100.00 & 100.00 & 100.00 & 100.00 & 100.00 \\
\cline{4-11}
 &  &  & \multirow{2}{*}{$3$} & $0$ & 100.00 & 100.00 & 100.00 & 100.00 & 100.00 & 0.00 \\
 &  &  &  & $1$ & 100.00 & 50.00 & 100.00 & 100.00 & 100.00 & 60.00 \\
\cline{2-11}
 & \multirow{8}{*}{$(6,4)$} & \multirow{4}{*}{$5$} & \multirow{2}{*}{$1$} & $0$ & 100.00 & 100.00 & 50.00 & 100.00 & 100.00 & 100.00 \\
 &  &  &  & $1$ & 100.00 & 99.38 & 99.38 & 99.38 & 100.00 & 100.00 \\
\cline{4-11}
 &  &  & \multirow{2}{*}{$3$} & $0$ & 90.74 & 90.74 & 91.98 & 78.40 & 100.00 & 93.33 \\
 &  &  &  & $1$ & 93.21 & 93.83 & 95.68 & 88.27 & 100.00 & 81.67 \\
\cline{3-11}
 &  & \multirow{4}{*}{$6$} & \multirow{2}{*}{$1$} & $0$ & 100.00 & 77.78 & 22.22 & 100.00 & 100.00 & 90.00 \\
 &  &  &  & $1$ & 100.00 & 88.89 & 88.89 & 100.00 & 100.00 & 90.00 \\
\cline{4-11}
 &  &  & \multirow{2}{*}{$3$} & $0$ & 74.07 & 88.89 & 77.78 & 55.56 & 100.00 & 50.00 \\
 &  &  &  & $1$ & 88.89 & 92.59 & 85.19 & 66.67 & 100.00 & 50.00 \\
\cline{2-11}
 & \multirow{8}{*}{$(7,3)$} & \multirow{4}{*}{$6$} & \multirow{2}{*}{$1$} & $0$ & 100.00 & 61.90 & 23.81 & 96.83 & 100.00 & 97.14 \\
 &  &  &  & $1$ & 100.00 & 71.43 & 87.30 & 98.41 & 100.00 & 97.14 \\
\cline{4-11}
 &  &  & \multirow{2}{*}{$3$} & $0$ & 85.71 & 84.13 & 73.02 & 47.62 & 100.00 & 30.00 \\
 &  &  &  & $1$ & 95.24 & 73.02 & 76.19 & 50.79 & 100.00 & 28.57 \\
\cline{3-11}
 &  & \multirow{4}{*}{$7$} & \multirow{2}{*}{$1$} & $0$ & 66.67 & 66.67 & 0.00 & 100.00 & 100.00 & 90.00 \\
 &  &  &  & $1$ & 66.67 & 33.33 & 66.67 & 88.89 & 100.00 & 90.00 \\
\cline{4-11}
 &  &  & \multirow{2}{*}{$3$} & $0$ & 66.67 & 22.22 & 22.22 & 11.11 & 100.00 & 20.00 \\
 &  &  &  & $1$ & 66.67 & 33.33 & 55.56 & 44.44 & 100.00 & 20.00 \\
\hline
\end{tabular}%
}
\end{table*}

Table~\ref{tab:cc_success} highlights four takeaways consistent with the design principles developed in Sections~III--V. First, on the \emph{code-level} Tanner graph, BP+OSD can fail on these \emph{stabilizer-induced} TS families, and the outcome can depend strongly on whether OSD is fed the final-iteration LLRs $L^{(50)}$ or the best-so-far LLRs $L^{\star}$. Second, decoding on the \emph{detector-error} Tanner graph typically improves BP+OSD by making correlations due to \emph{hook errors} explicit, although it does not uniformly dominate in every setting. Third, the proposed decoder achieves $100\%$ success throughout the table because it is designed to avoid these failures, following the principles established in Section~\ref{sec:derived_decoders}. Fourth, \emph{beam search} achieves performance comparable to min-sum+OSD-CS, despite operating on the \emph{phenomenological} Tanner graph (without additional variable nodes representing CNOT faults), consistent with its ability to mitigate failures associated with classical TSs.

\subsubsection{Graph product of two TSs from the two parent codes}
\label{sec:adv_ts_times_ts}
We now consider decoding failures in the Tanner graph $\mathscr{G}_\mathrm{z}$ caused by combined TSs.
\[
\big(\mathscr{T}(\mathcal{C}_\mathrm{s}) \times \mathscr{T}(v)\big)\;\cup\;\big(\mathscr{T}(c) \times \mathscr{T}(\mathcal{V}_\mathrm{s})\big)
\]
introduced in Section~\ref{sec:comp_TS}. Unlike the TSs considered previously, Algorithm~\ref{alg:derived_decoders_from_parent_decoders} is not guaranteed (see Section~\ref{sec:comp_TS}) to avoid failures on these combined structures, even when $D_{\mathrm{vv}}(\cdot)$ and $D_{\mathrm{cc}}(\cdot)$ are derived from parent-code decoders that individually correct all error patterns on $\mathscr{T}(\mathcal{V}_\mathrm{s})$ and $\mathscr{T}(\mathcal{C}_\mathrm{s})$.
As discussed in Section~\ref{sec:comp_TS}, the combined TS is still resolved if the underlying parent-code decoder used to derive $D_{\mathrm{vv}}$ or $D_{\mathrm{cc}}$ remains robust on $\mathscr{T}(\mathcal{V}_\mathrm{s})$ or $\mathscr{T}(\mathcal{C}_\mathrm{s})$, respectively, even in the presence of a single erroneous syndrome bit.
For the TSs $\mathscr{T}(\mathcal{C}_\mathrm{s})$ that live in the transpose Tanner graph $\mathscr{G}^\mathsf{T}$, we find empirically that min-sum is already robust to a single syndrome-bit error. We attribute this to the high row rank of the transpose parity-check matrix (for the parent codes underlying the LP constructions studied here) and the large minimum distance of the transpose code. For example, for the $[155,64,20]$ parent code underlying the $[[1054,140,d\leq 20]]$ LP code, the transpose code has only four codewords and minimum distance $62$, and $\mathscr{G}^\mathsf{T}$ does not exhibit small harmful TSs that cause the min-sum decoder to fail on $\mathscr{G}$.
Systematically enumerating TSs in $\mathscr{G}^\mathsf{T}$ is difficult, so we instead test elementary candidates built from a small number of overlapping cycles (which typically have a larger number of unsatisfied checks than comparable TSs on the $\mathscr{G}$ side). We verify convergence by running the min-sum directly on $\mathscr{G}^\mathsf{T}$, including with one flipped syndrome bit, without designing a specialized FAID for the transpose side.
Therefore, as described in Section~\ref{sec:comp_TS}, Algorithm~\ref{alg:derived_decoders_from_parent_decoders} resolves TS of the form $\left(\mathscr{T}(c) \times \mathscr{T}(\mathcal{V}_{\mathrm{s}})\right) \cup \left(\mathscr{T}(v) \times \mathscr{T}(\mathcal{C}_{\mathrm{s}})\right).$
Apart from this, we tested that these composite TSs get resolved in the presence of a one-bit syndrome error when we start with a round of $D_{\mathrm{cc}}(\cdot)$ followed by a round of $D_{\mathrm{vv}}(\cdot)$. Both $D_{\mathrm{vv}}(\cdot)$ and $D_{\mathrm{cc}}(\cdot)$ are derived from min-sum decoders.
This observation would help us minimize the number of decoders while using an ensemble of decoders without compromising the logical error rate, as we will see in Section~\ref{sec:adv_noise_model}.
This robustness is also true for the BP+OSD results on the \emph{detector-error} Tanner graph: the OSD post-processor can tolerate a small number of syndrome-bit errors as the \emph{detector-error} Tanner graph makes the relevant faults explicit.

Motivated by Section~\ref{sec:syndrome_error_inducing_CNOT_faults}, we also tested the convergence of the decoder on the \emph{phenomenological} Tanner graph, which includes nodes for data errors and syndrome-measurement errors but no explicit nodes for CNOT-induced faults. In all of the adversarial settings mentioned above, we flip one quarter of the syndrome bits associated with the TSs and observe that the decoder still converges to the correct estimate.

\begin{remark}
Classical TSs denoted by $\mathscr{T}(\mathcal{C}_{\mathrm{s}})$ in $\mathscr{G}^\mathsf{T}$ also tend to be easier to resolve in the presence of a small number of syndrome-bit errors because they often involve higher check degrees (here, $4$ and $5$) and therefore a larger number of odd-degree checks than TSs denoted by $\mathscr{T}(\mathcal{V}_\mathrm{s})$ in $\mathscr{G}$.
If $\mathscr{T}(\mathcal{C}_\mathrm{s})$ had only a small number of odd-degree checks, then one would likely need to optimize both $D_{\mathrm{vv}}(\cdot)$ and $D_{\mathrm{cc}}(\cdot)$; this is most relevant for higher-distance codes at modest blocklengths (a few hundred qubits).
Finally, the decoder proposed in Section~\ref{sec:derived_decoders} relies on the product structure of HGP/LP codes (disjoint isomorphic copies of the parent Tanner graphs). For code families without such parent-code decomposition (e.g., bivariate-bicycle codes~\cite{bravyi2024high}), the same high-level principles still apply (TS enumeration, targeted TS-resolving decoder design, and fault-node elimination in Section~\ref{sec:syndrome_error_inducing_CNOT_faults}); see also the BB-code simulations in Section~\ref{sec:sim_bb_codes}. A corresponding structural characterization of stabilizer-induced TSs for such families remains an open problem.
\end{remark}

\subsection{Stochastic noise model}
\label{sec:stoch_noise_model}
In this section, we use the \emph{circuit-level} depolarizing noise model of Section~\ref{sec:circuit_level_noise_model}. We perform three rounds of syndrome measurement and assume that the final round is perfect. Faults are sampled via the corresponding \emph{detector--error} matrix (equivalently, the \emph{detector--error} Tanner graph) of the measurement circuit. In practice, syndrome extraction spans multiple clock cycles set by the circuit depth, during which data qubits may idle and accrue additional decoherence errors; we do not explicitly include such between-cycle idling noise in our simulations. Since this omission applies uniformly to all codes and decoders considered here, the resulting performance comparisons remain meaningful.
The baseline BP+OSD (min-sum + OSD-0) and the \emph{beam-search} decoder of~\cite{ye2025beam} run on the \emph{detector--error} Tanner graph. We also evaluate two decoders that run on the \emph{phenomenological} Tanner graph with variable nodes for data-qubit errors and syndrome-bit errors, but without a dedicated variable node for each circuit-fault location: (i) an ensemble of layered min-sum decoders and (ii) a layered \emph{beam-search} decoder.
Each decoder in the ensemble of min-sum decoders is characterized by the message-update schedule designed to resolve TSs. We discuss the exact schedules later in the section.

For decoders running on the \emph{phenomenological} Tanner graph, we assign variable-node priors using the noise model of Section~\ref{sec:circuit_level_noise_model}. These priors depend on the syndrome-measurement circuit and are obtained by marginalizing the circuit-fault probabilities represented in the \emph{detector--error} Tanner graph.
We focus on computing the priors for $X$ errors while running the decoder on $\mathscr{G}_{\mathrm{z}}$; the corresponding priors for $Z$ errors on $\mathscr{G}_{\mathrm{x}}$ can be computed analogously.
A data qubit can acquire an $X$ error from multiple independent fault mechanisms. For example, a direct depolarizing fault on that qubit or a failure of one of the CNOT gates during measurement of a stabilizer generator that acts non-trivially on it. The corresponding data-qubit variable node is in error if an odd number of such faults contribute an $X$ component. We compute the resulting reliability (LLR) $\gamma_{\mathrm{i}}$ corresponding to variable node $v_{\mathrm{i}}$ by combining these contributions using the standard sum-product rule (check-node update)~\cite{mct}. For simplicity, we ignore additional contributions from ancilla faults that later propagate as hook errors, since these contributions are schedule-dependent.
Let $\mathcal{N}_{v_{\mathrm{i}}}^{z}$ denote the neighbors of $v_{\mathrm{i}}$ in $\mathscr{G}_z$. We assume the standard depolarizing noise model at each location: (i) each data-qubit location applies $X,Y,Z$ with probability $p/3$ each, so the probability of an $X$ component is $p_{\mathrm{i}}^{\mathrm{dep}}=2p/3$ and $\gamma_{\mathrm{i}}^{\mathrm{dep}}\coloneqq \log\frac{1-2p/3}{2p/3}$; and (ii) each incident CNOT gate fails with probability $p$ and, conditioned on failure, applies a uniformly random non-identity two-qubit Pauli (each with probability $p/15$). In (ii), each incident CNOT contributes an $X$ component on $v_{\mathrm{i}}$ with probability $8p/15$, so the corresponding LLR is $\gamma_{\mathrm{CNOT}}\coloneqq \log\frac{1-8p/15}{8p/15}$. We then set
\begin{equation}
\label{eq:phenomenological_llr_parity}
\gamma_{\mathrm{i}}
\,=\,2\operatorname{atanh}\!\Biggl( \tanh\!\Bigl(\frac{\gamma_{\mathrm{i}}^{\mathrm{dep}}}{2}\Bigr)
\,\tanh\!\Bigl(\frac{\gamma_{\mathrm{CNOT}}}{2}\Bigr)^{\lvert\mathcal{N}_{v_{\mathrm{i}}}^{z}\rvert}\Biggr).
\end{equation}

Similarly, for a syndrome-bit variable node corresponding to a check node $c_{\mathrm{j}}$ in $\mathscr{G}_{\mathrm{z}}$, we model a syndrome-bit flip as the XOR of (i) a direct measurement error on the corresponding ancilla and (ii) ancilla bit-flip contributions from the CNOT gates used to measure $c_{\mathrm{j}}$. Assuming a measurement-bit flip with probability $p$, the corresponding LLR is $\gamma_{\mathrm{meas}}\coloneqq \log\frac{1-p}{p}$. Let $\mathcal{N}_{c_{\mathrm{j}}}^{z}$ denote the neighbors of $c_{\mathrm{j}}$ in $\mathscr{G}_{\mathrm{z}}$. Using the same per-CNOT ancilla bit-flip LLR $\gamma_{\mathrm{CNOT}}$ as above, we set
\begin{equation}
\label{eq:syndrome_bit_llr_parity}
\gamma_{c_{\mathrm{j}}}
\,=\,2\operatorname{atanh}\!\Biggl( \tanh\!\Bigl(\frac{\gamma_{\mathrm{meas}}}{2}\Bigr)
\,\tanh\!\Bigl(\frac{\gamma_{\mathrm{CNOT}}}{2}\Bigr)^{\lvert\mathcal{N}_{c_{\mathrm{j}}}^{z}\rvert}\Biggr).
\end{equation}

The min-sum ensemble consists of six decoders that differ only in their message-update schedules. We first describe the schedule used by the first decoder; the remaining five schedules are small variants of it.
All variants operate on the \emph{phenomenological} Tanner graph (Section~\ref{sec:syndrome_error_inducing_CNOT_faults}) using a layered min-sum update and the time-sliced scheduling ideas described in the discussion of TSs in space and time (Subsection~VI-A). We refer to the per-round copies of the code-level Tanner graph as \emph{time slices}. Throughout this work, for any decoder using a layered schedule, one iteration denotes one full pass through the prescribed layer ordering: every variable node is updated at least once per iteration, while check-node updates occur within each layer (and therefore may be applied multiple times per iteration).
In each min-sum iteration, the first decoder uses the following update order:
{\renewcommand{\labelenumi}{(\roman{enumi})}
\begin{enumerate}
\item all variable nodes representing syndrome-bit errors;
\item CC-type variable nodes in the odd slices;
\item VV-type variable nodes in the odd slices;
\item again, all variable nodes representing syndrome-bit errors;
\item CC-type variable nodes in the even slices;
\item VV-type variable nodes in the even slices.
\end{enumerate}
}
The first decoder runs min-sum for $60$ iterations under this schedule (Schedule~1).
This ordering separates message updates between odd and even copies. As discussed in Section~\ref{sec:syndrome_error_inducing_CNOT_faults}, a given syndrome can admit different explanations between these copies, and alternating odd/even updates helps to avoid associated failures. Within each copy, alternating CC and VV-type updates mitigate \emph{stabilizer-induced} TS (cf. Section~\ref{sec:comp_TS}).
Schedule~2 interchanges the VV- and CC-type variable-node updates (i.e., steps (ii)\,--\,(iii) and (v)\,--\,(vi) are swapped). This swap helps when a large part of the error's support is on CC-type variable nodes. 
Schedules~3 and~4 are the corresponding \emph{batched} variants of Schedules~1 and~2, respectively, in which we perform several updates within the odd (resp. even) time slices before switching to the other, following the batched update schedule in Section~\ref{sec:syndrome_error_inducing_CNOT_faults}. This helps when the support of the error pattern is concentrated in a specific slice and the decoder takes longer to converge on that slice.
Schedules~5 and~6 match Schedules~1 and~2, respectively, but include the re-initialization step every ten iterations from Section~\ref{sec:syndrome_error_inducing_CNOT_faults} to enforce temporal separation.
After running each decoder for $60$ iterations, we project the residual unsatisfied checks to the last round and then run a set of TS-avoiding FAID decoders on the corresponding (projected) time slice, following Algorithm~\ref{alg:derived_decoders_from_parent_decoders}. Recall that Algorithm~\ref{alg:derived_decoders_from_parent_decoders} requires two decoders, $D_{\mathrm{vv}}(\cdot)$ and $D_{\mathrm{cc}}(\cdot)$, derived from a parent-code decoder running on $\mathscr{G}$ and $\mathscr{G}^{\mathsf{T}}$, respectively. In our implementation, the post-projection decoders use $D_{\mathrm{vv}}$ instances derived from a small library of FAIDs designed to resolve TSs on $\mathscr{G}$. Specifically, this FAID set targets $(5,3)$-, $(6,4)$-, and $(7,3)$-type TSs for both the $[[1054,140,d \leq 20]]$ and $[[775,43,d \leq 20]]$ codes. The corresponding $D_{\mathrm{cc}}$ decoder is derived from a quantized min-sum decoder on $\mathscr{G}_1^{\mathsf{T}}$. We use floating-point message updates before the projection step, and fixed-point updates after projection (since FAIDs are inherently fixed-point).

We also simulate the \emph{beam-search} decoder on the \emph{phenomenological} Tanner graph, using the same layered update schedule as Schedule~1 of the ensemble. Because \emph{beam-search} is already robust to several of the small, problematic error patterns inherited from the parent code, we use a single \emph{beam-search} decoder (rather than an ensemble). This schedule improves convergence and reduces the number of rounds needed for \emph{beam-search} to settle. For reference, the \emph{beam-search} baseline on the \emph{detector--error} Tanner graph does not use any post-processing and therefore has a lower overall decoding cost than BP+OSD; however, each iteration is more expensive because message passing runs on the larger detector--error graph.

In Figures~\ref{fig:logical_error_rate_1054} and~\ref{fig:logical_error_rate_775}, we compare our layered decoders with two baselines: BP+OSD and \emph{beam-search} with a flooding schedule. Our layered min-sum ensemble uses an additional short cleanup stage (the \emph{FAID stage}) after min-sum; this improves the logical error rate while keeping the per-iteration cost low because message passing runs on the smaller \emph{phenomenological} Tanner graph and requires no OSD post-processing. Overall, \emph{beam-search} on the \emph{phenomenological} Tanner graph, combined with the proposed layered schedule, achieves the lowest logical error rate among the methods we test.

More broadly, these results support the design principle that, once the dominant inherited failure mechanisms are mitigated, an iterative decoder on the \emph{phenomenological} Tanner graph with a suitable schedule can be competitive with decoding directly on the expanded \emph{detector--error} Tanner graph. The remaining failures at the three lowest simulated physical error rates are dominated by a specific small error pattern family (the $(8,2)$ patterns). Since the scheduled \emph{beam-search} decoder corrects most of these patterns (Table~\ref{tab:cc_success}), it achieves a better logical error rate than the ensemble decoder in the \emph{low-physical-error-rate} regime, at the cost of higher decoding latency. The scheduled \emph{beam-search} decoder uses $24$ rounds and a beam width of $16$, with $50$ message-passing iterations per round, while all decoders in the ensemble run for $60$ iterations followed by the FAID stage on a single slice.

Although \emph{beam} search cannot resolve many of the stabilizer-induced failure patterns reported in Table~\ref{tab:cc_success}, it performs better in this regime because the $(8,2)$ family dominates the residual failures. The ensemble's ability to resolve additional stabilizer-induced failures would become relevant only once those events become more likely than the smallest unresolved inherited family (the single-copy case).

\begin{figure*}
\begin{subfigure}{0.48\textwidth}
    \centering
    \begin{tikzpicture}
\definecolor{mycolor1}{rgb}{0.63529,0.07843,0.18431}%
\definecolor{mycolor2}{rgb}{0.00000,0.44706,0.74118}%
\definecolor{mycolor3}{rgb}{0.00000,0.49804,0.00000}%
\definecolor{mycolor4}{rgb}{0.87059,0.49020,0.00000}%
\definecolor{mycolor5}{rgb}{0.00000,0.44700,0.74100}%
\definecolor{mycolor6}{rgb}{0.74902,0.00000,0.74902}%

\begin{axis}[
scaled x ticks=false,
xticklabel={%
  \pgfmathparse{\tick*100}%
  \pgfmathprintnumber[fixed,precision=1]{\pgfmathresult}%
},
xmode=normal,ymode=log,
xmin=0.001, xmax=0.006,width=7.5cm,height=7 cm,
xtick distance=0.001,
ytickten={-8,-7,-6,-5,-4,-3,-2,-1,0},
ymin=7e-9, ymax=1,
minor x tick num=1,
minor y tick num=9,
minor grid style={draw=gray!50, line width=0.2pt},
major grid style={draw=gray!75, line width=0.3pt},
grid=both,
xminorgrids,
yminorgrids,
xlabel={Depolarizing probability ($\times 10^{-2}$)}, ylabel= Logical error rate,legend pos= north west,legend cell align=left,legend style={nodes={scale=0.5, transform shape}}]

\addplot [color=black,solid,line width=1pt,mark size=1.4pt,mark=square,mark options={solid}]
  table[row sep=crcr]{
  0.006 3.75e-1\\
  0.005 7.17e-2\\
  0.004 5e-3\\
  0.0035 9.766e-4\\
  0.003 1.56e-4\\
  0.0025 5.7e-5\\
  0.002 2.3e-5\\
};
\addlegendentry{\Large{BP+OSD-0 (det.-err.)}};
\addplot [color=mycolor4,solid,line width=1pt,mark size=1.4pt,mark=pentagon,mark options={solid}]
  table[row sep=crcr]{
  0.006 3.90e-1\\
  0.005 7.75e-2\\
  0.004 3.75e-3\\
  0.0035 6.25e-4\\
  0.003 3.79e-5\\
  0.0025 1.163e-5\\
  0.002 2.8e-6\\
};
\addlegendentry{\Large{\emph{Beam-search} (det.-err.)}};
\addplot [color=blue,solid,line width=1pt,mark size=1.4pt,mark=diamond,mark options={solid}]
  table[row sep=crcr]{
  0.006 2.75e-2\\
  0.005 1.6e-3\\
  0.004 1.93e-4\\
  0.003 1.8e-5\\
  0.002 5.99e-7\\
  0.001 8e-9\\
};
\addlegendentry{\Large{Layered MS ensemble (phen.)}};
\addplot [color=mycolor2,solid,line width=1pt,mark size=1.4pt,mark=star,mark options={solid}]
  table[row sep=crcr]{
  0.006 7.3e-2\\
  0.005 5.8e-3\\
  0.004 1.11e-4\\
  0.003 2.54e-6\\
  0.002 1.46e-7\\
};
\addlegendentry{\Large{Layered \emph{beam-search} (phen.)}};
\end{axis}
\end{tikzpicture}%
    \subcaption{Logical error rate as a function of the depolarizing noise rate for the $[[1054,140,d\leq 20]]$ LP code.}
    \label{fig:logical_error_rate_1054}
\end{subfigure}
\hfill
    \begin{subfigure}{0.48\textwidth}
    \centering
    \begin{tikzpicture}
\definecolor{mycolor1}{rgb}{0.63529,0.07843,0.18431}%
\definecolor{mycolor2}{rgb}{0.00000,0.44706,0.74118}%
\definecolor{mycolor3}{rgb}{0.00000,0.49804,0.00000}%
\definecolor{mycolor4}{rgb}{0.87059,0.49020,0.00000}%
\definecolor{mycolor5}{rgb}{0.00000,0.44700,0.74100}%
\definecolor{mycolor6}{rgb}{0.74902,0.00000,0.74902}%

\begin{axis}[
scaled x ticks=false,
xticklabel={%
  \pgfmathparse{\tick*100}%
  \pgfmathprintnumber[fixed,precision=1]{\pgfmathresult}%
},
xmode=normal,ymode=log,
xmin=0.002, xmax=0.007,width=7.5cm,height=7 cm,
xtick distance=0.001,
ytickten={-8,-7,-6,-5,-4,-3,-2,-1,0},
ymin=7e-9, ymax=1,
minor x tick num=1,
minor y tick num=9,
minor grid style={draw=gray!50, line width=0.2pt},
major grid style={draw=gray!75, line width=0.3pt},
grid=both,
xminorgrids,
yminorgrids,
xlabel={Depolarizing probability ($\times 10^{-2}$)}, ylabel= Logical error rate,legend pos= north west,legend cell align=left,legend style={nodes={scale=0.5, transform shape}}]


\addplot [color=red,solid,line width=1pt,mark size=1.4pt,mark=triangle,mark options={solid}]
  table[row sep=crcr]{
  0.007 1.25e-1\\
  0.006 2.10e-2\\
  0.0055 7.08e-3\\
  0.005 1.9e-3\\
  0.0045 3.55e-4\\
  0.004 8.75e-5\\
  0.0035 2.6e-5\\
  0.003 7.14e-6\\
};
\addlegendentry{\Large{BP+OSD-0 (det.-err.)}};
\addplot [color=mycolor1,solid,line width=1pt,mark size=1.4pt,mark=+,mark options={solid}]
  table[row sep=crcr]{
  0.007 1.98e-1\\
  0.006 3.70e-2\\
  0.0055 1.23e-2\\
  0.005 3.10e-3\\
  0.0045 6.07e-4\\
  0.004 8.88e-5\\
  0.0035 1.297e-5\\
   0.0025 1.23e-6\\
};
\addlegendentry{\Large{\emph{Beam-search} (det.-err.)}};
\addplot [color=green,solid,line width=1pt,mark size=1.4pt,mark=asterisk,mark options={solid}]
  table[row sep=crcr]{
  0.007 8e-3\\
  0.006 5e-4\\
  0.005 2.05e-5\\
  0.004 1.85e-6\\
  0.003 1.79e-7\\
  0.002 8e-9\\
};
\addlegendentry{\Large{MS ensemble (phen.)}};

\addplot [color=mycolor3,solid,line width=1pt,mark size=1.4pt,mark=pentagon,mark options={solid}]
  table[row sep=crcr]{
   0.007 1.4e-2\\
  0.006 1.4e-3\\
  0.005 4e-5\\
  0.004 7.57e-7\\
  0.003 8.5e-8\\
};
\addlegendentry{\Large{Layered \emph{beam-search} (phen.)}};
\end{axis}
\end{tikzpicture}%
    \subcaption{Logical error rate as a function of the depolarizing noise rate for the $[[775,43,d\leq 20]]$ code.}
    \label{fig:logical_error_rate_775}
\end{subfigure}
    \caption{Comparison of logical error rates under the \emph{circuit-level} depolarizing noise model. The BP+OSD and \emph{beam-search} baselines run message passing on the larger \emph{detector--error} Tanner graph. We also evaluate two approaches that run message passing on the smaller \emph{phenomenological} Tanner graph (with nodes for data-qubit and syndrome-bit errors, but not for each circuit-fault location): (i) a layered min-sum \emph{ensemble} consisting of six decoders that differ only in their update schedules (Schedules~1--6), followed by a short rule-based cleanup stage after the $60$ min-sum iterations; and (ii) a layered \emph{beam-search} decoder using the same layered schedule as Schedule~1. The min-sum ensemble achieves logical error rates competitive with the scheduled \emph{beam-search} decoder while retaining lower per-iteration complexity by avoiding decoding on the expanded detector--error graph and OSD post-processing. In the legend, \textsf{MS} denotes \emph{min-sum}, and the parenthetical tag (\textsf{det-er} or \textsf{phen}) indicates the Tanner graph used for message passing. In all cases, priors come from the same \emph{circuit-level} noise model (with \emph{phenomenological}-graph priors obtained by marginalizing circuit-fault probabilities from the \emph{detector--error} graph). We use three rounds of syndrome measurement, with the final round assumed perfect.}
    \label{fig:circuit_level_simulation}
\end{figure*}

\section{conclusion}
This work provides guidelines for designing decoders that operate on \emph{phenomenological} Tanner graphs, which do not include nodes for correlated errors. This avoids the additional decoding complexity introduced by such nodes. Although the logical error rates of the layered proposed decoders are better than those of BP+OSD, they can be further improved by optimizing the parent code decoders to resolve larger TSs. It would also be interesting to extend these design approaches from the quantum memory setting to the quantum computing setting.
\begin{appendices}


  \section{Stabilizer-induced TS in $\mathscr{G}^\mathsf{T}$}
  \label{sec:lemma2}

  \begin{lemma}
\label{lemma:stab_induced_TS_containg_classical_TS_CC}
Consider an HGP code constructed from two classical LDPC codes with Tanner graphs $\mathscr{G}_1$ and $\mathscr{G}_2$. Let $\mathscr{G}_{\mathrm{z}}$ denote the Tanner graph of the $Z$-checks. Let $\mathcal{C}_{\mathrm{s}}^1 \subseteq \mathcal{C}_1$ be a subset of check nodes of $\mathscr{G}_1$, and let $\mathscr{T}(\mathcal{C}_{\mathrm{s}}^1)$ be the subgraph of $\mathscr{G}_1$ induced by these check nodes (equivalently, the corresponding trapping-set subgraph in $\mathscr{G}_1^{\mathsf{T}}$). For a variable node $v^2\in\mathcal{V}_2$, let $\mathscr{T}(v^2)$ be the subgraph it induces. Define
\[
\mathscr{T}_{\mathrm{stab}} \triangleq \bigl(\mathscr{T}(\mathcal{C}^1_{\mathrm{s}}) \times \mathscr{T}(v^2)\bigr)\cap \mathscr{G}_{\mathrm{z}}.
\]
Then:
\begin{enumerate}
\item The subgraph $\mathscr{T}_{\mathrm{stab}}$ contains $\lvert\mathcal{N}_{v^2}\rvert$ pairwise isomorphic copies of $\mathscr{T}(\mathcal{C}^1_{\mathrm{s}})$, one for each $c^2_{\mathrm{j}}\in\mathcal{N}_{v^2}$. For a fixed $c^2_{\mathrm{j}}\in\mathcal{N}_{v^2}$, each check node $c^1\in \mathscr{T}(\mathcal{C}^1_{\mathrm{s}})$ corresponds to the CC-type variable node $c^1c^2_{\mathrm{j}}$, and each variable node $v^1$ adjacent to $c^1$ corresponds to the VC-type check node $v^1c^2_{\mathrm{j}}$.

\item These isomorphic copies are pairwise disjoint, i.e.,
\[
\bigl(\mathscr{T}(\mathcal{C}^1_{\mathrm{s}})\times c^2_{\mathrm{j}}\bigr)\cap\bigl(\mathscr{T}(\mathcal{C}^1_{\mathrm{s}})\times c^2_{\mathrm{\ell}'}\bigr)=\emptyset
\]
for all distinct $c^2_{\mathrm{j}},c^2_{\mathrm{\ell}'}\in\mathcal{N}_{v^2}$.

\item The subgraph $\mathscr{T}_{\mathrm{stab}}$ is induced by the CV-type stabilizer generators contained in $\mathcal{C}^1_{\mathrm{s}}\times v^2$.

\item Assume that $\mathscr{T}(\mathcal{C}^1_{\mathrm{s}})$ behaves as a TS in $\mathscr{G}_1^{\mathsf{T}}$. Then there exists a measurement schedule such that a single CNOT failure in the measurement circuit of each CV-type stabilizer generator in $\mathcal{C}^1_{\mathrm{s}}\times v^2$ can corrupt all variable nodes in $\ell$ of the isomorphic copies of $\mathscr{T}(\mathcal{C}^1_{\mathrm{s}})$ contained in $\mathscr{T}_{\mathrm{stab}}$, for some $\ell\in\{1,2,\dots,\lvert\mathcal{N}_{v^2}\rvert\}$.

\item Let $\mathcal{V}^1_{\mathrm{o}}$ denote the set of variable nodes of $\mathscr{G}_1$ that have odd degree in the subgraph $\mathscr{T}(\mathcal{C}^1_{\mathrm{s}})$. The stabilizer support generated by linear combinations of stabilizer generators in $\mathcal{C}^1_{\mathrm{s}}\times v^2$ then contains all CC-type variable nodes in $\mathscr{T}_{\mathrm{stab}}$, as well as all VV-type variable nodes in $\mathcal{V}^1_{\mathrm{o}}\times v^2$.

\item Each VV-type variable node $v^1v^2$ in $\mathscr{T}_{\mathrm{stab}}$ is adjacent to the VC-type check node $v^1c^2$ in every isomorphic copy $\mathscr{T}(\mathcal{C}^1_{\mathrm{s}})\times c^2$, for all $c^2\in\mathcal{N}_{v^2}$.
\end{enumerate}
\end{lemma}

\begin{proof}
Parts~1, 2, 3, and~6 follow from the corresponding parts of Lemma~\ref{lemma:stab_induced_TS_containg_classical_TS} by exchanging the roles of $\mathscr{G}_2$ and $\mathscr{G}_1^{\mathsf{T}}$; equivalently, by exchanging the roles of VV-type and CC-type variable nodes and of VC-type checks indexed by (variable, check) and (check, variable) pairs. The arguments carry over verbatim under this exchange.

For Part~4, the construction of Part~4 of Lemma~\ref{lemma:stab_induced_TS_containg_classical_TS} applies with the following modification: for each generator $c^1_{\mathrm{i}'} v^2\in \mathcal{C}^1_{\mathrm{s}}\times v^2$, the fault is scheduled among the CNOT gates connecting the ancilla to the data qubits in $c^1_{\mathrm{i}'}\times \mathcal{N}_{v^2}$, i.e., the CC-type portion of $\mathcal{N}_{c^1_{\mathrm{i}'} v^2}$, rather than the VV-type portion used in Lemma~\ref{lemma:stab_induced_TS_containg_classical_TS}. Choosing the schedule so that the $\ell$ gates $\mathrm{CNOT}(\text{ancilla}\rightarrow c^1_{\mathrm{i}'}c^2_{\mathrm{j}_t})$, $t\in[\ell]$, are executed last, a single $X$ fault on the ancilla immediately before these gates corrupts the CC-type nodes $c^1_{\mathrm{i}'}c^2_{\mathrm{j}_1},\dots,c^1_{\mathrm{i}'}c^2_{\mathrm{j}_\ell}$. Repeating this for every $c^1_{\mathrm{i}'}\in\mathcal{C}^1_{\mathrm{s}}$ corrupts all CC-type variable nodes in the $\ell$ isomorphic copies $\mathscr{T}(\mathcal{C}^1_{\mathrm{s}})\times c^2_{\mathrm{j}_t}$.

For Part~5, a VV-type node $v^1 v^2$ lies in the support of the generator $c^1_{\mathrm{i}'} v^2$ if and only if $c^1_{\mathrm{i}'}\in \mathcal{N}_{v^1}$; the number of generators in $\mathcal{C}^1_{\mathrm{s}}\times v^2$ containing it therefore equals the degree of $v^1$ in $\mathscr{T}(\mathcal{C}^1_{\mathrm{s}})$, which is odd exactly when $v^1\in\mathcal{V}^1_{\mathrm{o}}$. The CC-type supports $c^1_{\mathrm{i}'}\times \mathcal{N}_{v^2}$ are pairwise disjoint across distinct $c^1_{\mathrm{i}'}\in\mathcal{C}^1_{\mathrm{s}}$, since the first coordinates differ; hence each CC-type variable node of $\mathscr{T}_{\mathrm{stab}}$ appears in exactly one generator and lies in the stabilizer support.
\end{proof}


\section{FAID}
\label{sec:FAID}
A $b$-bit FAID, denoted by $\mathcal{D}_{\text{FAID}}$, is characterized by four parameters $\left(\mathcal{M},\mathcal{Y}, m_{c \rightarrow v}, m_{v \rightarrow c}\right)$, where
 $\mathcal{M}$ is the alphabet of possible message values,
$\mathcal{Y} = \{C, -C\}$ is the set of possible channel values,
 $m_{v \rightarrow c}$ is the function that determines the message from variable node $v$ to check node $c$ based on the incoming messages at $v$,
$m_{c \rightarrow v}$ is the function that determines the message from check node $c$ to variable node $v$ based on the incoming messages at $c$.
The message alphabet is defined as
\[
\mathcal{M} = \{0, \pm L_1, \pm L_2, \ldots, \pm L_s\},
\]
where $L_i$, for $1 \leq i \leq s$, are non-negative real numbers, and $s = 2^{b-1} - 1$, with the ordering $L_i \leq L_j$ for $i < j$.
For a check node $c$, the outgoing message from $c$ to a variable node $v$ is given by
\[
m_{c \rightarrow v}
= (1 - 2\sigma_c)
\prod_{v' \in \mathcal{N}_c \setminus \{v\}} \mathrm{sgn}(m_{v' \rightarrow c})
\min_{v' \in \mathcal{N}_c \setminus \{v\}} \lvert m_{v' \rightarrow c}\rvert,
\]
where $\mathcal{N}_c$ is the set of variable nodes connected to check node $c$, $\mathrm{sgn}(\cdot)$ is the sign function, and $\sigma_c$ is the syndrome bit corresponding to check node $c$.
For a variable node $v$ of degree $d_v$, the outgoing message from $v$ to a check node $c$ is given by
\[
m_{v \rightarrow c} = Q\left( \sum_{c' \in \mathcal{N}_v \setminus \{c\}} m_{c' \rightarrow v} + w(\mathbf{m}_v) C \right),
\]
where $\mathbf{m}_v$ denotes the vector of extrinsic incoming messages to $v$, with $\lvert\mathbf{m}_v\rvert = d_v - 1$ and $\mathbf{m}_v \in \mathcal{M}^{d_v - 1}$. The function $Q(\cdot)$ is the quantizer defined as
\[
Q(y) =
\begin{cases}
\mathrm{sgn}(y) L_i, & \text{if } T_i \leq |y| < T_{i+1}, \\[3pt]
0, & \text{if } |y| < T_1,
\end{cases}
\]
where $\{T_1, \ldots, T_s, T_{s+1} = \infty\}$ is a set of thresholds with $T_i \in \mathbb{R}^+$ for $1 \leq i \leq s$ and $T_i > T_j$ for any $i > j$.
The coefficient $w(\mathbf{m}_v)$ depends on the incoming messages. When $w(\mathbf{m}_v)$ is a constant, the decoder reduces to the offset min-sum decoder.
At the end of each iteration, the bit estimate associated with each variable node is determined by evaluating the sign of all incoming messages together with the channel value $C$; we denote this hard decision by $\hat{e}_v\in\{0,1\}$. At the beginning of the decoding process, all messages from variable nodes to check nodes are initialized to $C$.
In this work, we choose $w(\mathbf{m}_v)$ to be constant when $v$ is a CC-type variable node, while it depends on incoming messages when $v$ is a VV-type variable node.

\section{Simulation Results for BB Codes}
\label{sec:sim_bb_codes}
For HGP/LP codes, we argued (via the decoding-dynamics discussion in Section~\ref{sec:node_TS_graph_product}) that message passing on the \emph{phenomenological} Tanner graph can match the logical-error-rate performance of decoding on the expanded \emph{circuit-level} (detector--error) Tanner graph, while operating on a smaller graph. The simulations of the LP codes in Section~\ref{sec:stoch_noise_model} support this conclusion. Here, we ask whether the same phenomenon persists for bivariate-bicycle (BB) codes and, if so, in what parameter regime. To this end, we simulate two BB codes, $[[144,12,12]]$ and $[[288,12,18]]$, under the circuit-level noise model; we leave a corresponding analytical study for future work.

To test this hypothesis on BB codes, we simulate the \emph{beam-search} decoder on both the \emph{phenomenological} Tanner graph and the expanded \emph{circuit-level} (detector--error) Tanner graph. In both cases, we use beam width $16$, $24$ rounds, and $50$ message-passing iterations per round. Since our goal here is to isolate the effect of the decoding graph, we focus on \emph{beam-search} and do not additionally simulate BP+OSD.
We also do not simulate the layered min-sum ensemble of Section~\ref{sec:stoch_noise_model} because it relies on a set of code-specific FAIDs as a cleanup stage and we have not designed the corresponding FAIDs for these BB codes. As in Section~\ref{sec:stoch_noise_model}, error patterns are sampled from the circuit-level noise model, and the phenomenological-graph priors are obtained by marginalizing circuit-fault probabilities using \eqref{eq:phenomenological_llr_parity} and \eqref{eq:syndrome_bit_llr_parity}.

The logical error-rate curves for $[[144,12,12]]$ and $[[288,12,18]]$ are shown in Figures~\ref{fig:bb144} and~\ref{fig:bb288}, respectively. For $[[144,12,12]]$, decoding on the \emph{circuit-level} (\emph{detector--error}) Tanner graph yields a lower logical error rate than decoding on the phenomenological graph. For $[[288,12,18]]$, we observe the opposite trend.
\begin{figure*}[t]
\begin{subfigure}{0.48\textwidth}
    \centering
    \begin{tikzpicture}
\definecolor{mycolor1}{rgb}{0.63529,0.07843,0.18431}%
\definecolor{mycolor2}{rgb}{0.00000,0.44706,0.74118}%
\definecolor{mycolor3}{rgb}{0.00000,0.49804,0.00000}%
\definecolor{mycolor4}{rgb}{0.87059,0.49020,0.00000}%
\definecolor{mycolor5}{rgb}{0.00000,0.44700,0.74100}%
\definecolor{mycolor6}{rgb}{0.74902,0.00000,0.74902}%

\begin{axis}[
scaled x ticks=false,
xticklabel={%
  \pgfmathparse{\tick*100}%
  \pgfmathprintnumber[fixed,precision=1]{\pgfmathresult}%
},
xmode=normal,ymode=log,
xmin=0.001, xmax=0.005,width=7.5cm,height=7 cm,
xtick distance=0.001,
ytickten={-8,-7,-6,-5,-4,-3,-2,-1,0},
ymin=1e-7, ymax=1e-1,
minor x tick num=1,
minor y tick num=9,
minor grid style={draw=gray!50, line width=0.2pt},
major grid style={draw=gray!75, line width=0.3pt},
grid=both,
xminorgrids,
yminorgrids,
xlabel={Depolarizing probability ($\times 10^{-2}$)}, ylabel= Logical error rate,legend pos= north west,legend cell align=left,legend style={nodes={scale=0.5, transform shape}}]


\addplot [color=mycolor1,solid,line width=1pt,mark size=1.4pt,mark=+,mark options={solid}]
  table[row sep=crcr]{
  0.005 7.5e-3\\
  0.004 1.8e-3\\
  0.003 2.8e-4\\
  0.002 2e-5\\
  0.001 2.8e-7\\
};
\addlegendentry{\Large{\emph{Beam-search} (det.-err.)}};

\addplot [color=mycolor3,solid,line width=1pt,mark size=1.4pt,mark=pentagon,mark options={solid}]
  table[row sep=crcr]{
  0.005 1.6e-2\\
  0.004 3.8e-3\\
  0.003 8.7e-4\\
  0.002 7.9e-5\\
  0.001 2.19e-6\\
};
\addlegendentry{\Large{Layered \emph{beam-search} (phen.)}};
\end{axis}
\end{tikzpicture}%
    \subcaption{\emph{Beam-search} decoding of the $[[144,12,12]]$ BB code under circuit-level depolarizing noise. Decoding on the \emph{detector--error} Tanner graph reduces the logical error rate relative to decoding on the \emph{phenomenological} Tanner graph, consistent with the role of explicit fault nodes in representing \emph{hook} errors.}
    \label{fig:bb144}
\end{subfigure}
\hfill
\begin{subfigure}{0.48\textwidth}
    \centering
    \begin{tikzpicture}
\definecolor{mycolor1}{rgb}{0.63529,0.07843,0.18431}%
\definecolor{mycolor2}{rgb}{0.00000,0.44706,0.74118}%
\definecolor{mycolor3}{rgb}{0.00000,0.49804,0.00000}%
\definecolor{mycolor4}{rgb}{0.87059,0.49020,0.00000}%
\definecolor{mycolor5}{rgb}{0.00000,0.44700,0.74100}%
\definecolor{mycolor6}{rgb}{0.74902,0.00000,0.74902}%

\begin{axis}[
scaled x ticks=false,
xticklabel={%
  \pgfmathparse{\tick*100}%
  \pgfmathprintnumber[fixed,precision=1]{\pgfmathresult}%
},
xmode=normal,ymode=log,
xmin=0.002, xmax=0.005,width=7.5cm,height=7 cm,
xtick distance=0.001,
ytickten={-8,-7,-6,-5,-4,-3,-2,-1,0},
ymin=7e-9, ymax=1e-2,
minor x tick num=1,
minor y tick num=9,
minor grid style={draw=gray!50, line width=0.2pt},
major grid style={draw=gray!75, line width=0.3pt},
grid=both,
xminorgrids,
yminorgrids,
xlabel={Depolarizing probability ($\times 10^{-2}$)}, ylabel= Logical error rate,legend pos= north west,legend cell align=left,legend style={nodes={scale=0.5, transform shape}}]


\addplot [color=mycolor1,solid,line width=1pt,mark size=1.4pt,mark=+,mark options={solid}]
  table[row sep=crcr]{
  0.005 3.1e-3\\
  0.004 1.9e-4\\
  0.003 4.3e-6\\
  0.002 4e-8\\
};
\addlegendentry{\Large{\emph{Beam-search} (det.-err.)}};

\addplot [color=mycolor3,solid,line width=1pt,mark size=1.4pt,mark=pentagon,mark options={solid}]
  table[row sep=crcr]{
  0.005 9.3e-4\\
  0.004 5e-5\\
  0.003 9.59e-7\\
  0.002 1e-8\\
};
\addlegendentry{\Large{Layered \emph{beam-search} (phen.)}};
\end{axis}
\end{tikzpicture}%
    \subcaption{\emph{Beam-search} decoding of the $[[288,12,18]]$ BB code under the same noise model. Here, decoding on the \emph{phenomenological} Tanner graph achieves a lower logical error rate, suggesting that the additional fault nodes in the \emph{detector--error} graph can introduce short cycles that hinder convergence in this setting.}
    \label{fig:bb288}
\end{subfigure}
\caption{Logical error-rate curves for BB codes using \emph{beam-search} decoding on two alternative decoding graphs: the \emph{phenomenological} Tanner graph and the \emph{detector--error} (\emph{circuit-level}) Tanner graph. Logical error counts both decoding failures and convergence to the wrong coset.}
\label{fig:bb_sim_results}
\end{figure*}
In these plots, a logical error includes both (i) a decoding failure and (ii) convergence to the wrong coset. For $[[144,12,12]]$, wrong-coset events contribute a substantial fraction (almost half) of the logical error rate. For $[[288,12,18]]$, we observe comparatively fewer (less than a quarter) wrong-coset events in our simulations. For physical error rates below $4\times 10^{-3}$, the remaining logical errors are fully accounted for by decoding failures (i.e., non-convergence within the iteration budget); this breakdown is computed over the same $p<4\times 10^{-3}$ range for both decoding graphs.

One way to interpret the gap between the two decoding graphs is through the weight of syndrome-consistent explanations (error patterns). When an error pattern places multiple qubits on the support of a logical operator, and some of those qubits are introduced together by one (or a few) \emph{hook} errors, the \emph{detector--error} Tanner graph can represent that correlation explicitly via fault nodes. This can reduce the effective weight of the syndrome-consistent error pattern by treating a correlated multi-qubit event as a single fault, which can help beam search rank the correct hypothesis more highly; as discussed in Section~\ref{sec:motivation_for_proposed_decoder}, making such correlations explicit can also benefit OSD-based post-processing. Empirically, for $[[144,12,12]]$, many harmful cases when decoding on the \emph{phenomenological} Tanner graph involve hook errors that place multiple qubits on the support of a logical operator, consistent with the observed advantage of the \emph{detector--error} Tanner graph for this code.

For $[[288,12,18]]$, fewer wrong-coset events suggest that, in the tested regime, logical errors are dominated by TS-related decoding failures before correlated hook-error events become the main limiting mechanism. In this setting, the additional fault nodes in the \emph{detector--error} graph may provide limited benefit while introducing short cycles that hinder convergence and increase decoding failures. Overall, these results suggest a trade-off between representing correlations that can change the logical state and avoiding extra short cycles that can create additional TSs; we leave a more systematic study of this trade-off for future work.
For HGP/LP codes, this trade-off is more transparent: for HGP codes, hook errors do not propagate to multiple qubits on the support of a logical operator~\cite{manes2025distance}, so explicit fault nodes are typically unnecessary. For LP codes, this property need not hold for every logical operator, but it holds for a large fraction of them due to their similarity to the structure of HGP codes.

\end{appendices}
\bibliographystyle{IEEEtran}
\bibliography{ref_fault_tolerant_decoding}
\end{document}